\documentclass[11pt,twoside,openright,a4paper]{report}
\usepackage[english]{babel}
\usepackage{young,youngtab}
\usepackage{url,epstopdf}
\usepackage{amsfonts,amsmath,amssymb,mathrsfs,amsthm,multirow}
\usepackage{dsfont,bbm}
\usepackage{graphicx}
\usepackage{color}
\usepackage{setspace}
\usepackage{floatflt}
\usepackage{float}
\usepackage{wrapfig}
\usepackage{geometry}
\newtheorem{thm}{Theorem}[section]
\newtheorem{cor}[thm]{Corollary}
\newtheorem{lem}[thm]{Lemma}

\newtheorem{conj}[thm]{Conjecture}

\newtheorem{defn}[thm]{Definition}
\newtheorem{ex}[thm]{Example}
\newtheorem{rem}[thm]{Remark}

\numberwithin{equation}{section}
\newcommand{\K}{\mathbb K}
\newcommand{\R}{\mathbb R}
\newcommand{\C}{\mathbb C}
\newcommand{\Z}{\mathbb Z}
\newcommand{\N}{\mathbb N}
\newcommand{\NN}{\mathbb{N}_0}

\newcommand{\bd}{\boldsymbol}
\newcommand{\captionC}[1]{\caption{\emph{#1}}}
\newcommand{\z}[6]{\left(
\setlength{\arraycolsep}{0.06cm}
\renewcommand{\arraystretch}{0.60}
\begin{array}{ccc}
 \ifnum#1=1 {\cdot} \else {\,}  \fi
&\ifnum#3=1 {\cdot} \else {\,}  \fi
&\ifnum#5=1 {\cdot} \else {\,}  \fi  \\
 \ifnum#2=1 {\cdot} \else {\,}  \fi
&\ifnum#4=1 {\cdot} \else {\,}  \fi
&\ifnum#6=1 {\cdot} \else {\,}  \fi \\
\vspace{-0.35cm}
\end{array}\right)_Q}
\newcommand{\zX}[6]{\left(
\setlength{\arraycolsep}{0.06cm}
\renewcommand{\arraystretch}{0.60}
\begin{array}{ccc}
 \ifnum#1=1 {\cdot} \else {\,}  \fi
&\ifnum#3=1 {\cdot} \else {\,}  \fi
&\ifnum#5=1 {\cdot} \else {\,}  \fi  \\
 \ifnum#2=1 {\cdot} \else {\,}  \fi
&\ifnum#4=1 {\cdot} \else {\,}  \fi
&\ifnum#6=1 {\cdot} \else {\,}  \fi \\
\vspace{-0.35cm}
\end{array}\right)_X}
\newcommand{\zz}[8]{\left(
\setlength{\arraycolsep}{0.06cm}
\renewcommand{\arraystretch}{0.60}
\begin{array}{cccc}
 \ifnum#1=1 {\cdot} \else {\,}  \fi
&\ifnum#3=1 {\cdot} \else {\,}  \fi
&\ifnum#5=1 {\cdot} \else {\,}  \fi
&\ifnum#7=1 {\cdot} \else {\,}  \fi  \\
 \ifnum#2=1 {\cdot} \else {\,}  \fi
&\ifnum#4=1 {\cdot} \else {\,}  \fi
&\ifnum#6=1 {\cdot} \else {\,}  \fi
&\ifnum#8=1 {\cdot} \else {\,}  \fi \\
\vspace{-0.35cm}
\end{array}\right)_Q}
\DeclareMathOperator{\arcsinh}{arcsinh}
\usepackage{abstract}
\usepackage{fancyhdr}
\title{\vspace{-3cm}\begin{center}
\normalsize{DISS.~ETH No.~21748}
\end{center}
\vspace{1.5cm}
\textbf{Quantum Marginal Problem}\\
\textbf{and its Physical Relevance}}
\date{\,}
\author{\small{dissertation}\vspace{0.2cm}\\
\small{ETH Zurich}\\
\\
\\
\small{\,}\vspace{0.2cm}\\
\small{\textsc{\,}}\\
\\
\small{\,}\\
\\
\textbf{Christian Schilling}\\
\\
\small{\,}\\
\\
\\
\small{\,}\\
\small{\,}\\
\\
\small{}\\
\\
\\
\\
\\
\small{accepted on the recommendation of} \vspace{0.2cm}\\
\small{Prof.~Dr.~Matthias Christandl, examiner} \vspace{0.2cm}\\
\small{Prof.~Dr.~Dieter Jaksch, co-examiner} \vspace{0.2cm}\\
\small{Prof.~Dr.~Manfred Sigrist, co-examiner}\\
\\
\\
\\
\\
\\
2014
}

\begin{document}
\maketitle
\newgeometry{left=39mm, right=39mm, top=3.8cm, bottom=4.1cm}

\renewcommand{\abstractname}{}%
\begin{abstract}

\end{abstract}

\renewcommand{\abstractname}{Abstract}%
\begin{abstract}
The Pauli exclusion principle is a constraint on occupation numbers of fermionic quantum systems. It can be identified as a consequence of a much deeper
mathematical condition, the antisymmetry of the $N$-fermion wavefunction under particle exchange. Just recently, it was shown by Klyachko that this antisymmetry leads to further restrictions on natural occupation numbers. These so-called generalized Pauli constraints significantly strengthen Pauli's exclusion principle. Our first goal is to develop an understanding of the mathematical concepts behind Klyachko's work, in the context of quantum marginal problems. Afterwards, we explore the physical relevance of the generalized Pauli constraints and study concrete physical systems from that new viewpoint.

In the first part we review Klyachko's solution of the univariant quantum marginal problem. In particular we break his abstract derivation based on algebraic topology down to a more elementary level and reveal the geometrical picture behind it.

The second part explores the possible physical relevance of the generalized Pauli constraints.
We review the effect of pinning suggested by Klyachko. There one observes natural occupation numbers, which are pinned by the generalized Pauli constraints to the boundary of the allowed region. Although this effect would be quite spectacular and could imply strong restrictions for the corresponding system, we argue that pinning is unnatural. Instead, we conjecture the effect of quasi-pinning, defined by occupation numbers close to the boundary but not exactly on it. Furthermore, we find strong evidence that quasi-pinning as an effect in the $1$-particle picture corresponds to very specific and simplified structures of the corresponding $N$-fermion quantum state. In that sense quasi-pinning is highly physically relevant. After all, we develop the concept of a truncated pinning analysis, which allows to systematically investigate and quantify quasi-pinning.

In the third part we study concrete fermionic quantum systems from the new viewpoint of generalized Pauli constraints.
We compute the natural occupation numbers for the ground state of a family of interacting fermions in a harmonic potential.
Intriguingly, we find that the occupation numbers are strongly quasi-pinned, even up to medium interaction strengths. We identify this as an effect of the lowest few energy eigenstates, which provides first insights into the mechanism behind quasi-pinning. As a second model we analyze the Hubbard model with three electrons on three lattice sites and investigate the relation of symmetries and pinning. We find exact ground state pinning, which only seems possible whenever the physical model is very elementary and exhibits sufficiently many symmetries.
\end{abstract}

\renewcommand{\abstractname}{Zusammenfassung}%

\begin{abstract}
Das Pauli Ausschlussprinzip ist eine Bedingung an die Besetzungszahlen fermionischer Quantensysteme. Es kann als Konsequenz einer deutlich tieferen mathematischen Bedingung, der Antisymmetrie der $N$-Teilchen Wellenfunktion unter Teilchenaustausch, identifiziert werden. Erst k\"urzlich wurde von Klyachko gezeigt, dass diese Antisymmetrie noch zu weiteren Einschr\"ankungen der fermionischen Besetzungszahlen f\"uhrt. Diese sogenannten verallgemeinerten Pauli Bedingungen verst\"arken Paulis Ausschlussprinzip wesentlich. Unser erstes Ziel ist es ein Verst\"andniss f\"ur die mathematischen Konzepte Klyachkos Arbeit zu entwickeln, im Kontext des Quantenmarginalproblems. Danach untersuchen wir die physikalische Relevanz der verallgemeinerten Pauli Bedingungen und studieren konkrete physikalische Systeme von diesem neuen Gesichtspunkt.

Im ersten Teil reviewen wir Klyachkos L\"osung des univarianten Quantenmarginalproblems. Insbesondere brechen wir seine abstrakte Herleitung, basie-\\rend auf algebraischer Topologie, auf ein elementareres Level herunter und decken das geometrische Bild dahinter auf.

Der zweite Teil untersucht die m\"ogliche physikalische Relevanz der verallgemeinerten Pauli Bedingungen. Wir geben den Pinning Effek, vorgeschlagen von Klyachko, wieder. Dort beobachtet man nat\"urliche Besetzungszahlen, welche durch die verallgemeinerten Pauli Bedingungen auf den Rand der erlaubten Region gepinnt sind. Obwohl dieser Effekt sehr spektakul\"ar w\"are und starke Einschr\"ankungen f\"ur das entsprechende System implizieren k\"onnte argumentieren wir, dass Pinning unnat\"urlich ist. Stattdessen schlagen wir den Effekt des Quasi-Pinnings vor, definiert durch Besetzungszahlen nahe, aber eben nicht exakt auf dem Rand. Desweiteren finden wir starke Evidenz, dass Quasi-Pinning als Effekt im $1$-Teilchenbild sehr spezifischen und vereinfachten Strukturen des entsprechenden $N$-Fermionen Quantenzustand entspricht. In diesem Sinne ist Quasi-Pinning physikalisch h\"ochst relevant. Schlussendlich entwickeln wir noch das Konzept der trunkierten Pinning Analyse, welches erlaubt Quasi-Pinning systematisch zu untersuchen und zu quantifizieren.

Im dritten Teil betrachten wir konkrete fermionische Quantensysteme von dem neuen Gesichtspunkt der verallgemeinerten Pauli Bedingungen. Wir berechnen die nat\"urlichen Besetzungszahlen f\"ur den Grundzustand einer Famillie wehselwirkender Fermionen in einem harmonischen Potential. Interessanterweise finden wir, dass die Besetzungszahlen starkt quasi-gepinnt sind, sogar im Regime mittlerer Wechselwirkung. Wir identifizieren dies als Effekt der tiefsten Energiezust\"ande. Dies bietet erste Einsicht in den Mechanismus von Quasi-Pinning. Als zweites Modell analysieren wir das Hubbard-Modell mit drei Elektronen auf drei Gitterpl\"atzen und untersuchen die Relation von Symmetrien und Pinning. Wir finden exaktes Pinning f\"ur den Grundzustand, welches jedoch nur m\"oglich ist, wenn das physikalische Modell sehr elementar ist und gen\"ugend Symmetrien besitzt.
\end{abstract}

\renewcommand{\abstractname}{Acknowledgment}%

\begin{abstract}
In the first place I would like to thank very much my PhD supervisor, \emph{Prof.~Matthias Christandl}, for the guidance and all the scientific support I obtained during the last four years. Besides the scientific expertise I also benefited a lot from his interactive and professional research personality. I learned to efficiently present and publish scientific results and enjoyed very much collaborating with him and his co-workers.
I am also very grateful to him for having introduced to me the quantum marginal problem. By starting to work on the mathematical concepts behind that problem he gave me the great opportunity to work on the forehand of a new research field and to develop my own viewpoint on the physical relevance of the generalized Pauli constraints. I cannot imagine any more exciting research project.

My thanks also go to \emph{Prof.~David Gross} for all his scientific support and for having hosted me several times at his new place in Freiburg. I enjoyed very much the discussions with him and appreciated in particular his excellent intuitive understanding of involved mathematical concepts. The collaboration with him has always been very stimulating.

My long-time collaborators, \emph{Alex Lopes} and \emph{Dr.~James Whitfield}, I would like to express my gratitude for all their effort and stimulating ideas. I also would like to thank \emph{Dr.~James Whitfield} for the detailed remarks concerning my PhD thesis.

I also would like to thank \emph{Prof.~Dieter Jaksch} from the University of Oxford and \emph{Prof.~Manfred Sigrist} from our department for their willingness being co-referees of my PhD thesis.

Moreover, I would like to thank again my former teachers, \emph{Prof.~Volker Bach}, \emph{Prof.~Peter van Dongen} and \emph{Prof.~Martin Reuter} during my time at the Johannes-Gutenberg University of Mainz, as well as, \emph{Prof.~J\"urg Fr\"ohlich} at ETH Zurich for their great lectures, all the scientific support and the fruitful discussions.

In the same way I am very indebted to my father for all the inspiring scientific discussions and his advices concerning my academic career.

Finally, I would like to thank very much my family. Without the mental support of my parents a challenging project as a PhD thesis would not have been possible.
\end{abstract}
\newgeometry{left=30mm,right=37mm, top=3.8cm, bottom=4.1cm}
\tableofcontents
\chapter{Introduction and Background}\label{chap:Intro}
In the formulation of nonrelativistic quantum mechanics by Schr\"odinger the state of an $N$-particle system is described by a wave function $\Psi(\vec{x}_1,\ldots,\vec{x}_N)$ depending on the positions $\{\vec{x}_i\}$ of the particles in the $N$-particle configuration space.
Its time evolution is described by the Schr\"odinger equation,
\begin{equation}
i\hbar\partial_t \Psi_t = \hat{H}\Psi_t\,,
\end{equation}
where $\hat{H}$ is the Hamiltonian for that system. To understand the time evolution for any quantum state it suffices to determine the stationary states, the solutions of the time-independent Schr\"odinger equation
\begin{equation}
\hat{H}\Psi = E \Psi\,,
\end{equation}
which also provides the eigenenergies $E$. A prominent role is played by the ground state, the state with lowest energy. For instance, for macroscopic systems its properties have an influence on the excitations and therefore also on the behavior of that system. Consequently, much effort is put into the calculation of the ground state wave function. If the particles do not interact with each other, the ground state is just given by a product state of $N$ one-particle wave functions, i.e.~every particle is ``sitting'' in one state and the total system is described by a single configuration. Turning on an interaction, the ground state changes and cannot be described by a single configuration anymore. A well-known approximation scheme for calculating these more involved states is to use diagrammatic techniques, or as frequently be done in quantum chemistry, to use numerical methods.

For the case of indistinguishable particles, a quite different approach has been suggested by Hohenberg and Kohn \cite{Hohenberg}. Instead of solving the $N$-particle Schr\"odinger equation, they have proven that there exists a functional $E[n_1(\vec{x})]$ depending on the $1$-particle density $n_1(\vec{x})$ such that its minimizer $n_1^{(0)}(\vec{x})$ yields the ground state energy $E^{(0)}= E[n_1^{(0)}(\vec{x})]$ with corresponding one-particle density $n_1^{(0)}(\vec{x})$.

$n_1(\vec{x})$, as well as the $p$-particle density $n_p(\vec{x}_1,\ldots,\vec{x}_p)$ are the diagonal entries (in spatial representation) of the (kernel of the) so-called $p$-particle reduced density operator $\rho_p(\vec{x}_1,\ldots,\vec{x}_p,\vec{y}_1,\ldots,\vec{y}_p)$, which is obtained by integrating the $N$-particle density operator over the configuration space of $N-p$ particles,
\begin{eqnarray}
\rho_p(\vec{x}_1,\ldots,\vec{x}_p,\vec{y}_1,\ldots,\vec{y}_p) &\equiv& \int \mathrm{d}z_{p+1}\ldots \mathrm{d}z_{N}\,\Psi(\vec{x}_1,\ldots,\vec{x}_p,\vec{z}_{p+1},\ldots,\vec{z}_N)^\ast \nonumber \\
 &&\qquad \times \Psi(\vec{y}_1,\ldots,\vec{y}_p,\vec{z}_{p+1},\ldots,\vec{z}_N)\,.
\end{eqnarray}

In case that only one- and two-particle interactions exist, the energy expectation value in state $\Psi$ depends on $\rho_1$ and $\rho_2$, only. This is the reason why reduced density operators have been studied intensively \cite{Col2, Dav}. A couple of years ago they gained even more interest due to the development of the density matrix renormalization group theory \cite{Schoel}.

%Another important feature of quantum mechanics is the symmetry type of wave functions under particle exchange. Respecting the indistinguishable character of identical particles implies for systems of dimension larger than two that there are only two particle types, fermions and bosons. Fermionic wave functions are antisymmetric and bosonic ones are symmetric under particle exchange.

In 1925 the study of atomic transitions led to Pauli's famous exclusion principle \cite{Pauli1925}. It states that two identical fermions can never occupy the same quantum state at the same time. Mathematically, by introducing $n_i$ as the expectation value of the occupation number operator of a state labeled by $i$ we can formulate this principle as a linear inequality,
\begin{equation}\label{PauliConstraintIntro1}
0\leq n_i \leq1\,.
\end{equation}
The exclusion principle is particularly relevant for a concrete quantum system whenever one observes occupation numbers $n_i$ exactly or at least very close to the maximum $1$ or minimum $0$ of the allowed interval $[0,1]$. This incidence of occupation numbers pinned by the exclusion principle constraint to $1$ or $0$ appears for all systems of non-interacting fermions. As a consequence their theoretical description is quite simple. It suffices to consider the elementary $1$-particle picture and the energy eigenstates can be obtained by distributing the $N$ fermions to the available $1$-particle energy shells by respecting the exclusion principle. Also for interacting fermions one often observes occupation numbers close to $1$ or $0$. Famous examples are on the microscopic scale the electrons in an atom, and solid materials on the macroscopic level. Also on the astronomical scale the Pauli exclusion principle is important, since it is necessary for the stability of neutron stars. In that sense it plays an important role in nature.

However, already in 1926 Pauli's exclusion principle was identified by Dirac \cite{Dirac1926} and Heisenberg \cite{Heis1926} as a consequence of a much stronger mathematical statement, the antisymmetry of the $N$-fermion wave function under particle exchange. Hence, a natural question arises: Are there further restrictions on fermionic occupation numbers, beyond Pauli's exclusion principle? Since properties of atoms, molecules and even macroscopic systems like solid state materials mainly depend on the electronic degrees of freedom, this question is of broad relevance.

First evidence for the existence of such generalized Pauli constraints was provided in 1972 by Borland and Dennis \cite{Borl1972}. For the setting of three fermions with an underlying $6$-dimensional $1$-particle Hilbert space they found that only those natural occupation numbers (NONs) $\vec{\lambda}\equiv (\lambda_1,\ldots,\lambda_6)$, the eigenvalues of the $1$-particle reduced density operator, can arise from a pure and antisymmetric $3$-fermion state $|\Psi_3\rangle$, which fulfill
\begin{eqnarray}
&&\lambda_1+\lambda_6 = \lambda_2+\lambda_5 = \lambda_3+\lambda_4 = 1 \label{d=6a} \\
&&D^{(3,6)}(\vec{\lambda}) \equiv 2-(\lambda_1 +\lambda_2+\lambda_4) \geq 0 \label{d=6b} \,.
\end{eqnarray}
Here, we always order the $\lambda_i$ decreasingly and normalize their sum to the particle number. Notice that constraint (\ref{d=6b}) indeed strengthens Pauli's exclusion principle, which states $2-(\lambda_1+\lambda_2)\geq 0$.
It is also quite remarkable that even for this very small setting of three fermions and six dimensions it is very difficult to derive these conditions. Borland and Dennis also found the generalized Pauli constraints for $3$ fermions and the case of a $7$-dimensional $1$-particle Hilbert space,
\begin{eqnarray}
D_1^{(3,7)}(\vec{\lambda})\equiv 2-(\lambda_1+\lambda_2+\lambda_5+\lambda_6) \geq 0 &&\label{d=7a} \\
D_2^{(3,7)}(\vec{\lambda})\equiv 2-(\lambda_1+\lambda_3+\lambda_4+\lambda_6) \geq 0 &&\label{d=7b} \\
D_3^{(3,7)}(\vec{\lambda})\equiv 2-(\lambda_2+\lambda_3+\lambda_4+\lambda_5) \geq 0 &&\label{d=7c} \\
D_4^{(3,7)}(\vec{\lambda})\equiv 2-(\lambda_1+\lambda_2+\lambda_4+\lambda_7) \geq 0 &&\label{d=7d}\,.
\end{eqnarray}
They also expected that by increasing the dimension $d$ the generalized Pauli constraints are relaxing to the Pauli exclusion principle, which takes here the elegant form
 \begin{equation}\label{PauliConstraintIntro}
0\leq \lambda_i \leq1\,.
\end{equation}
However, just recently, in a ground-breaking mathematical work, Klyachko \cite{Kly3} has shown that this is not true. For each setting with $N$ fermions and a $d$-dimensional $1$-particle Hilbert space there exists generalized Pauli constraints, which are strengthening Pauli's exclusion principle. In addition, he provides an algorithm which allows for each fixed $N$ and $d$ to calculate these constraints. They always take the form of affine linear inequalities
\begin{equation}
D_j^{(N,d)}(\vec{\lambda}) \equiv \kappa_0^{(j)}+\kappa_1^{(j)} \lambda_1+\ldots + \kappa_d^{(j)} \lambda_d \geq 0\,\,\,\,,\,j=1,2,\ldots,r^{(N,d)}\,.
\end{equation}
and give rise to a polytope $\mathcal{P}_{N,d} \subset \R^d$ of possible NONs $\vec{\lambda}$, a proper subset of the ``Pauli hypercube'' (\ref{PauliConstraintIntro}).

Although our main goal is to explore the physical relevance of the generalized Pauli constraints it is important to understand the mathematical concepts behind them and their derivation. The task to determine all possible NONs can be reformulated as quantum marginal problem (QMP). A QMP asks when density operators (marginals) of subsystems are compatible in the sense that they can arise from a common state of the total system. One of the most important QMPs is the $r$-body pure/ensemble-$N$-representability problem \cite{Col2}: Which $r$-particle density operators can arise from a pure/ensemble $N$-fermion state? Although it was shown that this problem is not efficiently solvable for $r>1$ \cite{Liu} this is not true for the case $r=1$. There, the problem simplifies due to a unitary equivalence and one can immediately infer that only the eigenvalues $\vec{\lambda}$ of the $1$-particle reduced density operator are relevant. Hence, the case $r=1$ is identical to the original task of finding \emph{all} restrictions of fermionic occupation numbers. In 2004 by building on recent progress in invariant and representation theory Klyachko has solved the univariant QMPs (these are the QMPs, for which only the spectra of the marginals are relevant) \cite{Kly4}. Since his solution is very abstract it is our first goal to understand it on a more elementary level.

Due to the importance of Pauli's exclusion principle we expect that the generalized Pauli constraints also play an important role. Can we use them to improve our physical understanding? Do they have an influence on the structure of fermionic quantum states and can they restrict or even explain the behavior of fermionic systems? Klyachko claimed first evidence for their relevance for ground states by suggesting the so-called \emph{pinning-effect} \cite{Kly1}. He expected that NONs $\vec{\lambda}$ of at least some fermionic ground states are pinned by the generalized Pauli constraints to the boundary of the polytope. In that case some of the generalized Pauli constraints would be active for the energy minimization in the sense that any further minimization of the energy expectation value would violate them. He also provides \cite{Kly1} first numerical evidence for pinning by analyzing variational data \cite{Nakata2001} obtained for a specific Beryllium state. However, Klyachko's analysis is inconclusive since using data with finite accuracy (as e.g.~numerical data) does never allow to verify the effect of pinning. In addition, pinning seems to be also quite unnatural from an intuitive viewpoint. Why should NONs of some \emph{interacting} many-fermion system exactly saturate some of those $1$-particle constraints? It is an important goal of this thesis to verify or falsify Klyachko's conjecture of pinning. For this we need to study \emph{analytically} concrete quantum systems of interacting fermions. However, exploring possible pinning is very challenging since one does not only need to find the $N$-fermion ground state but also should calculate and diagonalize the $1$-particle reduced density operator.

Analyzing concrete physical systems, microscopic ones like atoms and molecules, as well as macroscopic solid state materials, from the new viewpoint of generalized Pauli constraints is an interesting but so far purely academic task. What can one conclude from possible pinning as effect in the $1$-particle picture? It is e.g.~well-known that NONs $\vec{\lambda}=(1,\ldots,1,0,\ldots)$ imply that the corresponding $N$-fermion state $|\Psi_N\rangle$ can be written as one single Slater determinant, $|\Psi_N\rangle = |i_1,\ldots,i_N\rangle$ of $1$-particle states $|i_k\rangle$.
Does the information of pinning provides in a similar way structural information on $|\Psi_N\rangle$?
Yes, that is the case as indicated in \cite{Kly1}. Whenever NONs $\vec{\lambda}$ are pinned to the boundary of the polytope, $|\Psi_N\rangle$ has a very specific and simplified structure. As an example consider a three fermion state $|\Psi_3\rangle \in \wedge^3[\mathcal{H}_1^{(6)}]$, which is pinned by (\ref{d=6b}). Then the structural relation implies the form
\begin{equation}
|\Psi_3\rangle = \alpha \,|1,2,3\rangle + \beta \,|1,4,5\rangle + \gamma \,|2,4,6\rangle\,
\end{equation}
with appropriate $1$-particle states $|i\rangle$. The structure of that state is indeed much simpler then that of generic states, linear combination of $\binom{6}{3} = 20$ Slater determinants. An important question that this thesis addresses is whether this relation of pinning and structure of $|\Psi_N\rangle$ is stable. Do NONs, which are approximately saturating a generalized Pauli constraint, correspond to $N$-fermion states $|\Psi_N\rangle$ with approximately this structure? Since quantum properties of atoms and molecules can be (approximately) obtained by truncation of the infinite-dimensional Hilbert space to a low-dimensional one
the study of the role of the generalized Pauli constraints for few fermions and a one-particle Hilbert space of dimension $O(1)$ is justified. Of course, the investigation of their role for solid state materials with a macroscopic number of electrons will be not straightforward but a great challenge.

As a consequence the present thesis contains the following three conceptually quite different main parts,
\begin{enumerate}
\item \textbf{The Quantum Marginal Problem.}\,\,We introduce in detail the quantum marginal problem and review Klyachko's abstract solution. Moreover, we break it down to a more elementary level and reveal the geometric picture behind it.
\item \textbf{Generalized Pauli Constraints and Concept of Pinning.}\,\,
     We shed some light on the pinning effect suggested by Klyachko. However, we argue that it is more natural to observe \emph{quasi-pinning}, a very similar effect but with a conceptually different origin. Moreover, we investigate the physical relevance of quasi-pinning, by investigating possible structural implications for the corresponding $N$-fermion quantum state $|\Psi_N\rangle$. For concrete applications we will develop the concept of a truncated pinning analysis, which will allow to systematically investigate possible pinning and to quantify it.
\item \textbf{Pinning Analysis for Specific Physical Systems.}
    In this part we study concrete physical systems from the new viewpoint of generalized Pauli constraints. Are ground states (quasi-)pinned?
    For this, we will intensively study a model of a few harmonically coupled fermions confined to a harmonic trap. As a second system we will study the Hubbard model for three fermions on three sites. We will also explore the connection between symmetries and possible pinning.
\end{enumerate}
To support the reader we present at the beginning of each of the corresponding three chapters a detailed motivation together with a summary of the main results we will find.

\chapter{The Quantum Marginal Problem}\label{chap:Math}
\section{Motivation and Summary}\label{sec:MotSumMath}
Most physical effect have their origin in the interaction between two or more physical systems. From a theoretical viewpoint one treats them as subsystems of a total multipartite quantum system. Examples are given by electrons and a nucleus forming an atom, macroscopically many atoms building up crystal or also a solid material, which is coupled to a heat bath or an external magnetic field.

Although all subsystems interact with each other and are therefore relevant for the physical behavior, some of them are of more interest than others. E.g.~by studying an atom $\mathcal{A}$ coupled to a heat bath $\mathcal{R}$ only the properties of the atom are of interest. If the total system is described by a quantum state $\rho_{\mathcal{A} \mathcal{R}}$ all relevant information on $\mathcal{A}$ is provided by the reduced density operator of the atom (also called marginal),
\begin{equation}
\rho_{\mathcal{A}} \equiv \mbox{Tr}_{\mathcal{R}}[\rho_{\mathcal{A} \mathcal{R}}]\,,
\end{equation}
obtained by tracing out system $\mathcal{R}$.

In that context a natural mathematical question arises: Given marginals for some subsystems of a multipartite quantum system. Are they compatible to each other in the sense that they can arise from a \emph{common} total state? The task to determine all compatible tuples of marginals and to describe this set in an elegant way is called \emph{quantum marginal problem} (QMP). The most prominent example of an QMP is the $2$-body $N$-representability problem. There, the total system is given by $N$-identical fermions and one asks whether a given $2$-particle reduced density operator can arise from an $N$-fermion state. Solving this problem would allow to significantly improve variational optimizations for fermionic ground states.

%Although most QMPs seem to be not efficiently solvable this is not true for the subclass of univariant QMPs. There all subsystems of interest are disjoint in the sense that they do not contain any further common subsystems. As example consider a multipartite system $ABCD$ containing the subsystems $\{A,B,C,D\}$. If the subsystems of interest are given by $\{AB,D\}$ it is a univariant QMP, in contrast to $\{AB,AC,D\}$, where the subsystems $AB$ and $AC$ share the common subsystem $A$.

In this chapter we study in detail the prototype of a QMP. We consider a total system $AB$, which consists of two subsystems, $A$ and $B$. We ask which triple of spectra $(\vec{\lambda}_A,\vec{\lambda}_B,\vec{\lambda}_{AB})$ are compatible in the sense that there exists a total state $\rho_{AB}$ with eigenvalues $\vec{\lambda}_{AB}$ and corresponding marginals $\rho_A$ and $\rho_B$ with eigenvalues $\vec{\lambda}_A$ and $\vec{\lambda}_B$. Building up on
recent progress in invariant and representation theory Klyachko has solved this QMP \cite{Kly4}. He derived so-called \emph{marginal constraints}, \emph{linear} constraints on the eigenvalues of the marginals. Klyachko did also prove that they are not only necessary but also sufficient for the compatibility of $(\vec{\lambda}_A,\vec{\lambda}_B,\vec{\lambda}_{AB})$. We review his abstract derivation and break it down to a more elementary level. Moreover, we reveal the geometrical picture behind it.

As a first step we will find a variational principle, which is the origin of every marginal constraint. It takes the form
\begin{equation}\label{introVarPrin}
\lambda_{i_1}+\ldots +\lambda_{i_r} = \min_{V \in X_{\bd{i}}(\rho)}(\mbox{Tr}[P_V \rho])\,,
\end{equation}
where $\rho$ is a $d\times d$ -hermitian matrix with eigenvalues $\lambda_i$, $P_V$ denotes the orthogonal projection operator onto a linear subspace $V \leq \C^d$ belonging to the so-called \emph{Schubert variety} $X_{\bd{i}}(\rho)$, $\bd{i}\equiv (i_1,\ldots,i_r)$. Applying this variational principle in an appropriate way to a triple $(\rho_A,\rho_B,\rho_{AB})$ of compatible marginals with spectra $(\vec{\lambda}_A,\vec{\lambda}_B,\vec{\lambda}_{AB})$ allows to derive one marginal constraint whenever the three underlying Schubert varieties $X_{\bd{i}}(\rho_A)$, $X_{\bd{j}}(\rho_B)$ and $X_{\bd{k}}(\rho_{AB})$ fulfill a very specific and quite involved intersection property. Deriving all marginal constraints would require to check this intersection property for all triple of indices $(\bd{i},\bd{j},\bd{k})$. However, since this is too complicated one needs a more elegant approach. This is based on algebraic topology. We explain, that the Schubert varieties have a beautiful mathematical structure. They are complex projective algebraic varieties, i.e.~they are in particular manifolds and subvarieties of the so-called \emph{flag variety}. An algebraic variety is defined as the solution of polynomial equations, as e.g.~the unit circle given by points $(x,y)$ fulfilling $x^2+y^2-1 = 0$.
This allows to lift the intersection property to an algebraic level where one can check it systematically.
However, this lifting process is quite involved, since one needs to calculate the cohomology ring of the flag variety. We explain in detail how this works. As a central step we use that the flag variety can be obtained as a union of all its Schubert varieties. Moreover, we verify that this union has the substructure of a so-called cell complex. This abstract statement is verified by representing Schubert varieties as a family of matrices (reduced echelon form). Then, by using the cell structure of the flag variety elementary theorems of algebraic topology are used to calculate its cohomology ring.

After determining the cohomology ring one still needs to find the algebraic description of the intersection property. Since this turns out to be too involved we develop the geometric picture behind it. The three Schubert varieties of interest, $X_{\bd{i}}(\rho_A)$, $X_{\bd{j}}(\rho_B)$ and $X_{\bd{k}}(\rho_{AB})$, are submanifolds of the flag variety.
Since the intersection property is homeomorphically invariant we can deform all three Schubert varieties continuously without changing the intersection property. After applying an appropriate deformation the intersection property can be studied easily.
In that way we will derive marginal constraints for some settings with small dimensional Hilbert spaces for system $A$ and $B$.

\section{Definition of the problem}\label{sec:QMPdef}
In this part we define the quantum marginal problem (QMP) in its most general form.

\begin{defn}\label{QMPdefinition}
Given a multipartite quantum system. The quantum marginal problem is the problem of determining the mathematical conditions on density operators belonging to different subsystems of interest ensuring that all are belonging via partial trace to the same quantum state of the total system.
\end{defn}
\noindent Since the reduced density operators are called \emph{marginals}, the name quantum marginal problem is natural.
Now, we like to make Definition \ref{QMPdefinition} more precise from a mathematical perspective.
Let's consider finitely many quantum systems, labeled by indices $A,B,C,\ldots$ belonging to the finite set $\mathcal{J}=\{A,B,C,\ldots\}$
of indexes. Each system $J \in \mathcal{J} $ is represented by a separable Hilbert space $\mathcal{H}_{J}$, which might be a priori finite
or also infinite dimensional. We may think of $J \in \mathcal{J}$ as one particle, a few particles or a system of macroscopically many particles or also as just the spin degree of freedom of a single electron. Our label set $\mathcal{J}$ is chosen as a `basis', i.e.~every physical degree of interest is contained in exactly one of these systems $A,B,\ldots$.
The state of the system $J$ is then described by an element (density operator) in the state space $\mathcal{S}(\mathcal{H}_J)$, which is defined as
\begin{equation}
\mathcal{S}(\mathcal{H})= \{\rho \in \mathcal{B}(\mathcal{H})| \rho \ge 0 \wedge \mbox{Tr}[\rho]=1\} \,,
\end{equation}
where $\mathcal{B}(\mathcal{H})$ is the space of bounded linear operators on $\mathcal{H}$.
Every subset $\mathcal{I} \subset\mathcal{J}$ then labels a new physical system, namely the multipartite quantum system built up by all systems $I \in \mathcal{I}$. Let's denote the power set of $\mathcal{J}$ by $\mathcal{P}(\mathcal{J})$, which is then to be understood as the family of all physical subsystems contained in the total multipartite quantum system  $\mathcal{J}$. The system $\mathcal{I} \in \mathcal{P}(\mathcal{J})$ is then represented by the Hilbert space $\mathcal{H}_{\mathcal{I}}=\otimes_{I \in \mathcal{I}} \mathcal{H}_I$ and its states are described by elements in $\mathcal{S}(\mathcal{H}_{\mathcal{I}})$. Let $\mathcal{I}' \subset \mathcal{I} \subset \mathcal{J}$ and $\rho \in \mathcal{S}(\mathcal{H}_{\mathcal{I}})$. The partial trace
$\rho'$ of $\rho$ over the Hilbert spaces referring to the system $\mathcal{I} \setminus \mathcal{I}' \in \mathcal{P}(\mathcal{J})$ is denoted by $\rho'= \mbox{Tr}_{\mathcal{I} \setminus \mathcal{I}'}[\rho]$. Due to the properties of the partial trace it is clear that $\rho' \in \mathcal{S}(\mathcal{H}_{\mathcal{I}'})$ if $\rho \in \mathcal{S}(\mathcal{H}_{\mathcal{I}})$ .
For a given subset $\mathbf{\mathcal{K}} \subset \mathcal{P}(\mathcal{J}) $ the quantum marginal problem $\mathcal{M}_{\mathbf{\mathcal{K}}}$ is the problem of determining and describing the set $ \Sigma_{\mathbf{\mathcal{K}}} \subset \prod_{\mathcal{I} \in \mathbf{\mathcal{K}}} \mathcal{S}(\mathcal{H}_{\mathcal{I}})$ of tuples of compatible marginals. This is a subset of the cartesian product of the spaces $\mathcal{S}(\mathcal{H}_{\mathcal{I}})$ such that for each `point' $(\rho_{\mathcal{I}})_{\mathcal{I} \in \mathcal{K}} \in \Sigma_\mathcal{K}$ there exists a state $\rho_{\mathcal{J}} \in \mathcal{S}(\mathcal{H}_{\mathcal{J}})$ for the total system fulfilling
\begin{equation}
\rho_{\mathcal{I}}= \mbox{Tr}_{\mathcal{J} \setminus \mathcal{I}}[\rho_{\mathcal{J}}]   \qquad  \forall \mathcal{I} \in \mathcal{K} \, .
\end{equation}
\begin{floatingfigure}[r]{5cm}
\hspace{-0.5cm}
\centering
\includegraphics[width=4.9cm]{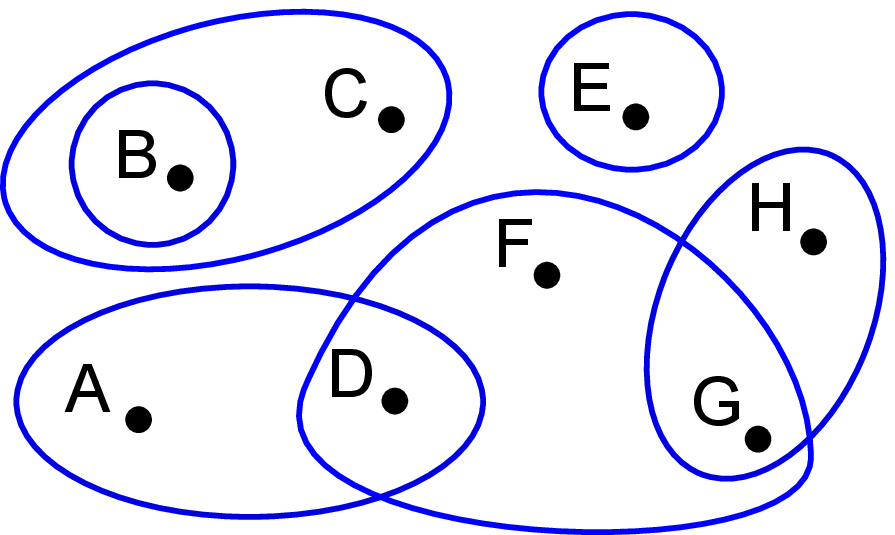}
\captionC{Illustration of the general marginal problem (see text).}
\label{QMPoverlapping1}
\end{floatingfigure}
The challenging task is not only to find this subset, but also to \emph{describe it in a compact way}.
This means that the quantum marginal problem is the problem of determining the conditions on density operators of subsystems of the total multipartite quantum systems arising from the condition that they all are compatible in the sense that there exists a state for the total system with these reduced density operators. This mathematical problem is also illustrated in Fig.~\ref{QMPoverlapping1}. There, every physical system $I  \in \mathcal{J}$ is symbolically described by a black dot.
Moreover, the family $\mathbf{\mathcal{K}}$ of certain subsets of $\mathcal{J}$ is illustrated by `blue islands' containing physical systems, where every island describes one subset $\mathcal{I} \in \mathbf{\mathcal{K}}$.
As emphasized in this figure these islands may overlap. If this is the case we call the corresponding marginal problem overlapping marginal problem and in the same way if two marginals are belonging to two
systems containing a common subsystem they are said to overlap. There is a first general observation:
\begin{lem}\label{convexsolution}
Given a multipartite quantum system described by $\mathcal{J}$ and for $\mathcal{K} \subset \mathcal{P}(\mathcal{J})$ we consider the marginal problem $\mathcal{M}_{\mathcal{K}}$. Assume that $\sigma^{(1)}=(\rho_{\mathcal{I}}^{(1)})_{\mathcal{I}\in\mathcal{K} }$ and $\sigma^{(2)}=(\rho_{\mathcal{I}}^{(2)})_{\mathcal{I}\in\mathcal{K} }$ are both compatible, i.e.~$\sigma^{(1)}, \sigma^{(2)} \in \Sigma_{\mathcal{K}}$. Then for all $0\leq \lambda \leq 1: \lambda \sigma^{(1)} + (1-\lambda) \sigma^{(2)} \in \Sigma_{\mathcal{K}}$. This means that $\Sigma_{\mathcal{K}}$ is a convex set.
\end{lem}
\begin{proof}
Let $\rho_{\mathcal{J}}^{(1)}$ and $\rho_{\mathcal{J}}^{(2)}$ be states of the total system with marginals $\sigma^{(1)}=(\rho_{\mathcal{I}}^{(1)})_{\mathcal{I}\in\mathcal{K} }$ and $\sigma^{(2)}=(\rho_{\mathcal{I}}^{(2)})_{\mathcal{I}\in\mathcal{K} }$, respectively.
Then, the linearity of partial traces ensures that the state $\rho_{\mathcal{J}}^{\lambda}:= \lambda \rho_{\mathcal{J}}^{(1)} + (1-\lambda) \rho_{\mathcal{J}}^{(2)}$ has the marginals $\lambda \sigma^{(1)} + (1-\lambda) \sigma^{(2)}= (\lambda \rho_{\mathcal{I}}^{(1)}+ (1-\lambda) \rho_{\mathcal{I}}^{(2)})_{\mathcal{I}\in\mathcal{K} } $.
\end{proof}

Additionally, one can restrict the QMP to a smaller set of solutions by demanding further constraints, as e.g
\begin{enumerate}
\item a pure total state
\item if all quantum systems are identical: total state is symmetric or antisymmetric under particle exchange
\item maximal or minimal values for entropies of some marginals
\end{enumerate}
The overlapping marginal problem is very difficult and not yet solved. In the next part we present a certain class of marginal problems that were solved by Klyachko in 2004 (see \cite{Kly1}, \cite{Kly2}, \cite{Kly3} and \cite{Kly4}). Further important contributions also came from Knutson \cite{Knut2}, Christandl and Mitchison \cite{MC} and from the work \cite{Daftuar200580} by Daftuar and Hayden.

\subsection{Pure univariant quantum marginal problem}\label{sec:QMPpureuni}
\begin{floatingfigure}[h]{5cm}
\hspace{-0.5cm}
\includegraphics[width=4.9cm]{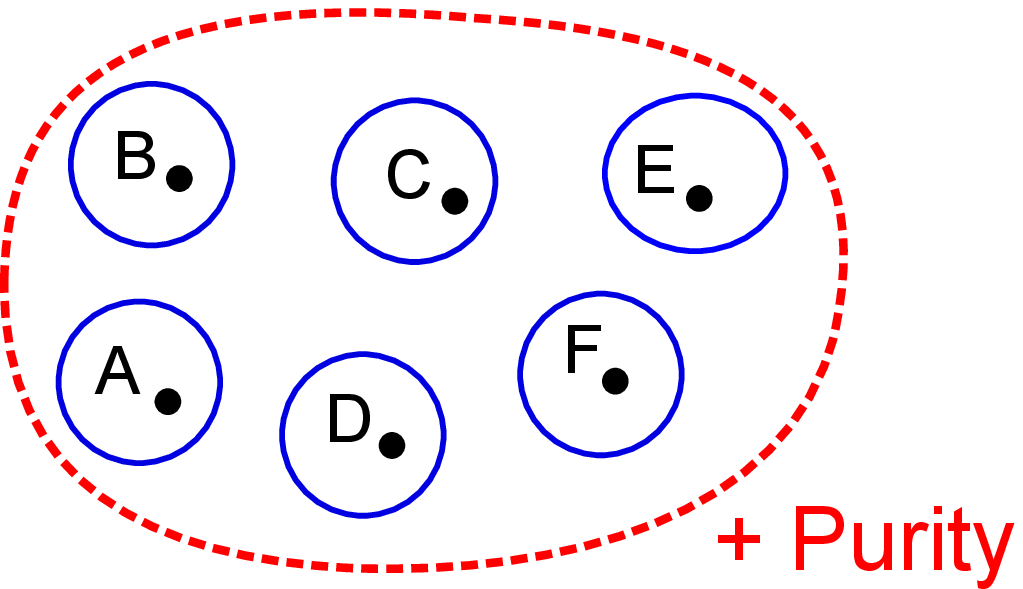}
\captionC{Illustration of the univariant marginal problem (see text).}
\label{QMPnonoverlapping}
\end{floatingfigure}
\noindent Due to the overwhelming complexity of the overlapping marginal problem, it belongs to the class of QMA complete problems \cite{Liu}, the quantum generalization of the NP-class, one in particular focussed to the special case of the so-called pure and non-overlapping QMP $\mathcal{M}_{\mathcal{K}}$. They have the property that the marginals are non-overlapping, i.e.~$\forall \mathcal{I}, \mathcal{I}' \in \mathcal{K}, \mathcal{I} \neq \mathcal{I}': \mathcal{I}\cap \mathcal{I}' = \emptyset $ and that the total state is pure. Fig.~\ref{QMPnonoverlapping} illustrates this:
The blue islands have no common black dots, i.e.~they do not overlap. Due to this structure we can formulate these marginal problems by drawing only one dot (instead of several black dots) in each blue island: If there were several black dots in one island we could define the union of them as a new black dot (physical system). The purity constraint as a property of the total state is indicated by a red dashed island. We are typically not interested in how this state looks like (hence, we don't draw a blue island) but we require purity as one of its properties.

The big advantages here is that the set of solutions can be described by spectral conditions only. This is an essential simplification compared to the overlapping version that refers to the large cartesian set $\prod_{\mathcal{I} \in \mathbf{\mathcal{K}}} \mathcal{S}(\mathcal{H}_{\mathcal{I}})$.
The reason for this is the unitary equivalences of non-overlapping marginals (also for the case without purity constraint):
\\
\begin{lem}\label{unitaryequiv}
 Let $\mathcal{J}$ be a multipartite quantum system and consider a non-overlapping QMP $\mathcal{M}_{\mathcal{K}}$ for a given $\mathcal{K} \in \mathcal{P}(\mathcal{J})$ (i.e.~all $\mathcal{I} \in \mathcal{K}$ are disjoint) and assume that the tuple $(\rho_{\mathcal{I}})_{\mathcal{I}\in\mathcal{K} }$ is compatible, i.e.~there exists a total state $\rho_{\mathcal{J}}$, such that $\forall \mathcal{I}\in\mathcal{K}:$ $\mbox{Tr}_{\mathcal{J}\setminus \mathcal{I} }[\rho_{\mathcal{J}}] = \rho_{\mathcal{I}}$. Then for every family $\{U_{\mathcal{I}}\}_{\mathcal{I}\in\mathcal{K}}$ of unitaries $U_{\mathcal{I}}: \mathcal{H}_{\mathcal{I}} \rightarrow \mathcal{H}_{\mathcal{I}}$
 the set $(U_{\mathcal{I}} \rho_{\mathcal{I}} U_{\mathcal{I}}^\dagger)_{\mathcal{I}\in\mathcal{K} }$ is also compatible.
\end{lem}
\begin{proof}
The state $\tilde{\rho}_{\mathcal{J}} := (\otimes_{\mathcal{I}\in \mathcal{K}} U_{\mathcal{I}}) \, \rho_{\mathcal{J}} \, (\otimes_{\mathcal{I}\in \mathcal{K}} U_{\mathcal{I}}^\dagger) $ is in $\mathcal{S}(\mathcal{H}_{\mathcal{J}})$ and it has the marginals $(U_{\mathcal{I}}  \rho_{\mathcal{I}} U_{\mathcal{I}} ^\dagger)_{\mathcal{I}\in\mathcal{K} }$.
\end{proof}

\noindent Lem.~(\ref{unitaryequiv}) implies that the solution of pure and non-overlapping QMP is given by conditions on the spectra of the marginals, only. In contrast to the more general overlapping version the `orientation' of the marginals to each other, namely the orientation of their eigenspaces do not matter here anymore.
Moreover one often considers only finite dimensional Hilbert spaces $\mathcal{H}_{\mathcal{I}}, \mbox{dim}(\mathcal{H}_{\mathcal{I}})= d_{\mathcal{I}} < \infty$, which also holds for this thesis as far as it is not stated differently. In that case the solution of the QMP is a set in the high (but finite) dimensional Euclidean space $\prod_{\mathcal{I}\in \mathcal{K}}\mathbb{R}^{d_{\mathcal{I}}}$, the set of tuples of eigenvalues of all density operators of interest.

Due to normalization $\mbox{Tr}[\rho_{\mathcal{I}}]=1$ and by arranging the corresponding eigenvalues non-increasingly we can restrict to the cartesian product of simplexes $\Delta_{d}:= \{\lambda \in \mathbb{R}^d| 1\ge \lambda_1\ge\ldots\ge \lambda_d\ge 0, \,\sum_{j=1}^d \lambda_j=1\}$, $\prod_{\mathcal{I}\in \mathcal{K}}\Delta_{d_{\mathcal{I}}}$, but from time to time mathematical elegance requires a more general setting, namely the one of hermitian operators in the larger set $\mathcal{B}(\mathcal{H})$, without any positivity condition or trace normalization to unity. Thus, we'd like to keep our notation as general as possible.

\section{Variational principle}\label{variationalprin}
In this section we present the so-called Hersch-Zwahlen variational principle, which is later used to find conditions on the eigenvalues of the marginals.
A part of our presentation is similar to the one given in \cite{Daftuar200580}.
A simpler and well-known variational principle is the Ky-Fan principle:
\begin{lem}[Ky-Fan's Min Max principle]\label{KyFan} Let $\mathcal{H}$ be a d-dimensional Hilbert space, $\rho$ a hermitian operator (e.g.~$\rho \in
\mathcal{S}(\mathcal{H})$) with spectrum $\lambda$, which is arranged in non-increasing order. Then for all $k=0,1,\ldots d$:
\begin{eqnarray}
\sum_{j=1}^k \lambda_j &=& \max \limits_{V \leq \mathcal{H},\mbox{dim}(V)=k}(\mbox{Tr}[P_V \rho]) \label{KyFan1}\\
\sum_{j=d-k+1}^d \lambda_j &=& \min \limits_{V \leq \mathcal{H},\mbox{dim}(V)=k}(\mbox{Tr}[P_V \rho]) \label{KyFan2} \,.
\end{eqnarray}
\end{lem}
\noindent Here $P_V$ is the orthogonal projector onto the linear subspace $V$. The proof of Lem.~\ref{KyFan} is given in the Appendix \ref{trivial}.
Eq.\! \ref{KyFan1} states in particular that the sum of the $k$ largest diagonal elements of $\rho$ represented w.r.t.~to an arbitrary orthonormal basis is never larger than the sum of the $k$ largest eigenvalues of $\rho$.
For the generalization of Lem.~\ref{KyFan} we first introduce Schubert cells of the Grassmannian and the variety of complete flags (which are introduced with all details later in Sec.~\ref{flagvarieties} and Sec.~\ref{grassmannian}) that will play a very important role in solving the univariant QMP in an elegant way. We define
\begin{defn}\label{defflag}
Let $\mathcal{H}$ be a $d$-dimensional Hilbert space. A \emph{complete flag} $F_{\bullet}$ is a maximal sequence of nested linear subspaces, i.e.
\begin{equation}
F_{\bullet}:=[{0}= F_0 \lneq F_1 \lneq \ldots \lneq F_{d-1}\lneq F_d = \mathcal{H}]
\end{equation}
\end{defn}
\noindent An important concept is the one of complete flags induced by a non-degenerate hermitian operator.
\begin{defn}\label{flagoperator}
Given a hermitian operator $A$ on a $d$-dimensional Hilbert space with non-degenerate spectrum $a=(a_1,\ldots,a_d)$ arranged in decreasing order. $A$ then induces a complete flag $F_{\bullet}(A)$ according
\begin{equation}
F_{i}(A) = \langle v_1,\ldots, v_i\rangle\qquad,\, \forall i=0,1,\ldots,d\,,
\end{equation}
where $v_j$ is the eigenvector corresponding to the eigenvalue $a_j$ and $\langle \cdot\rangle$ denotes the span of vectors.
\end{defn}
\begin{defn}\label{SchubertcellGr}
Let $\mathcal{H}$ be a $d$-dimensional Hilbert space and $F_{\bullet}$ a complete flag. Then for every binary sequence $\pi \in \{0,1\}^d$ we define the Grassmannian Schubert cell $S_{\pi}^{\circ}(F_{\bullet})$ by
\begin{equation}
S_{\pi}^{\circ}(F_{\bullet}):= \{ V\leq\mathcal{H}\,|\, \forall i=1,\ldots, d: \mbox{dim}((V\cap F_i)/ (V \cap F_{i-1}))= \pi_i \}
\end{equation}
and for every permutation $\alpha=(\alpha_1,\ldots,\alpha_d)$ the complete flag variety Schubert cell $X_{\alpha}^{\circ}(F_{\bullet})$ by
\begin{equation}
X_{\alpha}^{\circ}(F_{\bullet}):= \{V_{\bullet}\,|\,\forall \, 1\leq p,q \leq d:\,\mbox{dim}(V_p\cap F_q)= \#\{i\leq p: \alpha_i \leq q\} \}\,.
\end{equation}
These cells are subsets of the Grassmannian $\mbox{Gr}_{\|\pi\|_1,d}$ and the flag variety $\mbox{\emph{Fl}}(\mathcal{H})$, respectively, which are defined as (see also Sec.~\ref{flagvarieties} and Sec.~\ref{grassmannian} for more details)
\begin{eqnarray}\label{GrFlagPredefinition}
\mbox{Gr}_{\|\pi\|_1,d} &:=& \{V\leq \mathcal{H} \,|\, \mbox{dim}(V)= \|\pi\|_1\} \\
\mbox{\emph{Fl}}_d \equiv \mbox{\emph{Fl}}(\mathcal{H}) &:=& \{F_{\bullet}\}
\end{eqnarray}
and $\|\pi\|_1 \equiv \sum_{k=1}^d \pi_k$\,.
\end{defn}
\begin{rem}
For the case of the Grassmannian Schubert cell the binary sequence defines the indices at which the components (vector spaces) of the sequence $V\cap F_0 \leq V\cap F_1 \leq \ldots \leq V\cap F_d$ increase their dimension. The label $^{\circ}$ indicates that the Schubert cells are open w.r.t.~a natural topology (see Sec.~\ref{flagvarieties} and Sec.~\ref{grassmannian}). The closures of these Schubert cells will later also play a very important role.
\end{rem}
\begin{rem}
In the following, by using the concept of Schubert cells for the solution of the QMP, the flags $F_{\bullet}$ will always arise in a natural way, namely induced by a density operator according to Definition \ref{flagoperator}.
\end{rem}
\noindent Now, we can express sums of \emph{arbitrary} eigenvalues by more advanced variational principles:
\begin{lem}[Hersch-Zwahlen 1]\label{HerschZwahlen} Let $\rho$ be a hermitian operator with non-degenerate spectrum $\lambda$ arranged in decreasing order, $\pi \in \{0,1\}^d$ a binary sequence of length $d$. Then
\begin{equation}
\sum_{j=1}^d \pi_j \lambda_j = \min \limits_{V \in S_{\pi}^{\circ}(\rho)}(\mbox{Tr}[P_V \rho])\,. \label{HerschZwahlen1}
\end{equation}
\end{lem}
\begin{proof}
For given $\rho$ and $\pi$, let $i_1 < \ldots < i_k$ denote the indices with $\pi_{i}=1$, where $k = \|\pi\|_1$. For $V \in S_{\pi}^{\circ}(\rho) $ we choose an orthonormal basis $v_1, \ldots, v_k$ such that for all $l=1,\ldots, k$: $v_l \in V\cap F_{i_l}$, which means in particular $v_{l} \in F_{i_l}$. Then,
\begin{equation}
\mbox{Tr}[P_V \rho] = \sum_{l=1}^k \langle v_l, \rho v_l\rangle \geq \sum_{l=1}^k \min_{v \in F_{i_l}}\langle v, \rho v\rangle   =\sum_{l=1}^k \lambda_{i_l} = \sum_{j=1}^d \pi_j \lambda_j \,,
\end{equation}
where we used that $F_k$ is the eigenspace of $\rho$ corresponding to the $k$ largest eigenvalues.
To finish the proof we find a minimizer for the right handed side in (\ref{HerschZwahlen1}). By denoting the eigenvectors of $\rho$ by $e_k$ we consider $V_0 :=\mbox{span}(\{e_{i_l}\}_{l=1}^k) \in S_{\pi}^{\circ}(\rho)$, which leads to $\mbox{Tr}[P_{V_0} \rho] = \sum_{j=1}^d \pi_j \lambda_j$.
\end{proof}
\noindent Moreover
\begin{lem}[Hersch-Zwahlen 2]\label{HerschZwahlenAdv} Let $\rho$ be a hermitian operator with non-degenerate spectrum $\lambda$ arranged in decreasing order, $\alpha=(\alpha_1,\ldots, \alpha_d)$ a permutation and $a =(a_1,\ldots,a_d)\in \mathbb{R}^d, a_1> \ldots> a_d$ a so-called test spectrum. Then
\begin{equation}
\sum_{j=1}^d  \lambda_{\alpha_j} a_j = \min \limits_{\begin{array}{c}F_{\bullet}(A) \in X_{\alpha}^{\circ}(\rho)\\ spec(A)=a\end{array}}(\mbox{Tr}[\rho A])\,, \label{HerschZwahlen2}\,
\end{equation}
where on the rhs one minimizes over flags induced by hermitian operators $A$ with spectrum $a$ belonging to the corresponding flag variety Schubert cell $X_{\alpha}^{\circ}(\rho)$ (see also Definition \ref{flagoperator}).
\end{lem}
\begin{proof}
For given $A$ and $\rho$ we denote their eigenvectors by $v_1,\ldots,v_d$ and $f_1,\ldots,f_d$, respectively and define $F_j := \langle f_1,\ldots,f_j\rangle$. $A$ induces a flag according Definition \ref{flagoperator}.
We find
\begin{equation}
\mbox{Tr}[\rho A] = \sum_{i=1}^d \, \langle v_i, \rho v_i \rangle a_i =: \mu \cdot a\,.
\end{equation}
where $\mu := (\langle v_i, \rho v_i \rangle)$ is a probability vector, i.e.~$\mu_i \geq 0$ and $\sum_{i=1}^d \,\mu_i =1$. Moreover for all $k=0,1,\ldots,d$
\begin{equation}
\sum_{i=1}^k \mu_i = \mbox{Tr}[\rho P_{V_k}] \qquad , \, V_k:= \langle v_1,\ldots,v_k\rangle.
\end{equation}
Since $V_{\bullet} \in X_{\alpha}^{\circ}(\rho)$, $\lambda = \mbox{spec}(\rho)$ is arranged in decreasing order and by using $\gamma_d:=\lambda_d\geq 0$ and $\gamma_i := \lambda_i-\lambda_{i+1} > 0$ for $i=1,\ldots,d-1$ we find (recall Definition \ref{SchubertcellGr})
\begin{eqnarray}
\sum_{i=1}^k \mu_i &=& \mbox{Tr}[\rho P_{V_k}] \nonumber \\
&=&  \sum_{i=1}^d \gamma_i \, \mbox{Tr}[P_{F_i} P_{V_k}] \nonumber \\
&\geq&  \sum_{i=1}^d \gamma_i \, \mbox{dim}(F_i \cap V_k) \nonumber \\
&=&  \sum_{i=1}^d \gamma_i \, \#(j\leq k\,:\, \alpha_j\leq i) \nonumber \\
&=& \sum_{i=1}^d \lambda_i\, \left( \,\#(j\leq k\,:\, \alpha_j\leq i)-\#(j\leq k\,:\, \alpha_j\leq i-1)  \,\right) \nonumber \\
&=& \sum_{i=1}^d \lambda_i \, \chi(i \in \{\alpha_1,\ldots,\alpha_k\}) \nonumber \\
&=& \sum_{i=1}^k \, \lambda_{\alpha_i} \,,
\end{eqnarray}
where $\chi$ denotes the characteristic function and in the fourth step we have used Definition \ref{SchubertcellGr}. Hence, the vector $\lambda_{\alpha}:=(\lambda_{\alpha_1},\ldots,\lambda_{\alpha_d})$ is majorized by the vector $\mu$.
Then
\begin{equation}
(\mu- \lambda_{\alpha}) \cdot a = (\mu- \lambda_{\alpha}) \cdot (\sum_{i=1}^d b_i 1_i) \geq 0
\end{equation}
with $1_i =(\underbrace{1,\ldots 1}_i,0,\ldots,0)$ and $b_d := a_d$, for $k<d,$ $b_k := a_k- a_{k+1} \geq 0$ . Hence
\begin{equation}
\mu\cdot a \geq \lambda_{\alpha}\cdot a \Leftrightarrow \mbox{Tr}[\rho A] \geq \sum_{i=1}^d \, \lambda_{\alpha_i} a_i\,.
\end{equation}
\end{proof}
\noindent The Hersch-Zwahlen variational principles (\ref{HerschZwahlen}) and (\ref{HerschZwahlenAdv}) are the basic tools for deriving spectral conditions on reduced density operators. To show how this works we consider the two basic marginal problems, namely first $\{A,AB\}$ and later $\{A,B,AB\}$:
\begin{enumerate}
\item
Given two finite dimensional Hilbert spaces $\mathcal{H}^{(A)}$ and $\mathcal{H}^{(B)}$ with dimensions $d_A$ and $d_B$, respectively. The morphisms $\Phi_k, k=1,2,\ldots,d_A$ of Grassmannians are defined as
\begin{equation}\label{morphism1}
\Phi_k: \mbox{Gr}_{k,d_A} \rightarrow \mbox{Gr}_{k d_B,d_A d_B}\qquad,\,\,\, V \mapsto V \otimes \mathcal{H}^{(B)}\,.
\end{equation}
Given a state $\rho_{AB}$ for the total system $AB$ and let
$\rho_A = \mbox{Tr}_B[\rho_{AB}]$ be the marginal w.r.t system $A$. Moreover, $\lambda^{(AB)}:= \mbox{spec}(\rho_{AB})^{\downarrow}$, $\lambda^{(A)}:= \mbox{spec}(\rho_{A})^{\downarrow}$ and choose binary sequences $\pi \in \{0,1\}^{d_A}, \sigma \in \{0,1\}^{d_A d_B}$. Then, whenever the intersection property
(the \emph{dual binary sequence} $\hat{\sigma}$ of a sequence $\sigma$ of length $d$ is defined by $\hat{\sigma}_k := \sigma_{d-k+1}$)
\begin{equation}\label{intersectionprop1}
\Phi_{\|\pi\|_1}(S_{\pi}^{\circ}(\rho_A)) \cap S_{\hat \sigma}^{\circ}(\rho_{AB}) = \left(S_{\pi}^{\circ}(\rho_A) \otimes \mathcal{H}^{(B)} \right) \cap S_{\hat \sigma}^{\circ}(\rho_{AB}) \neq \emptyset
\end{equation}
holds, we obtain
\begin{eqnarray}\label{deviationspecineq1}
\lefteqn{\sum_{j=1}^{d_A}\pi_j \lambda_j^{(A)} - \sum_{i=1}^{d_A d_B} \sigma_i \lambda_i^{(AB)}}&&\nonumber \\
&=& \sum_{j=1}^{d_A} \pi_j \lambda_j^{(A)} + \sum_{i=1}^{d_A d_B} \sigma_i (-\lambda_i^{(AB)}) \nonumber \\
&=& \min \limits_{V \in S_{\pi}^{\circ}(\rho_A)}(\mbox{Tr}_A[P_V \rho_A]) + \min \limits_{W \in S_{\hat \sigma}^{\circ}(-\rho_{AB})}(\mbox{Tr}_{AB}[P_W (-\rho_{AB})]) \nonumber \\
&=& \min \limits_{V\otimes \mathcal{H}^{(B)} \in S_{\pi}^{\circ}(\rho_A)\otimes \mathcal{H}^{(B)}}(\mbox{Tr}_{AB}[P_{V\otimes \mathcal{H}^{(B)}} \rho_{AB}]) + \min \limits_{W \in S_{\hat \sigma}^{\circ}(-\rho_{AB})}(\mbox{Tr}_{AB}[P_W (-\rho_{AB})]) \nonumber \\
&\leq& \mbox{Tr}_{AB}[P_{W_0} \rho_{AB}] + \mbox{Tr}_{AB}[P_{W_0} (-\rho_{AB})] \nonumber \\
&=&0 \,,
\end{eqnarray}
where we applied $S_{\sigma}^{\circ}(-\rho) = S_{\hat{\sigma}}^{\circ}(\rho) $ in the third line and
(\ref{intersectionprop1}) was used in the second last line, with an element $W_0 \in  \left(S_{\pi}^{\circ}(\rho_A) \otimes \mathcal{H}^{(B)}\right) \cap S_{\hat \sigma}^{\circ}(\rho_{AB})$.
Hence, we obtain a spectral inequality
\begin{equation}\label{spectralineq1}
\boxed{
\sum_{j=1}^{d_A}\pi_j \lambda_j^{(A)} \leq \sum_{i=1}^{d_A d_B} \sigma_i \lambda_i^{(AB)}} \,.
\end{equation}
\item Given two finite dimensional Hilbert spaces $\mathcal{H}^{(A)}$ and $\mathcal{H}^{(B)}$ with dimensions $d_A$ and $d_B$, respectively and let $a \in \mathbb{R}^{d_A}$ and $b \in \mathbb{R}^{d_B}$ be two non-degenerate test spectra (always arranged in decreasing order) such that their sum $\{a+b\}:= \{a_i+b_j\,|\, 1\leq i\leq d_A\,,\,1\leq j\leq d_B\}$ is also non-degenerate. This pair of test spectra induces index maps $i_{a,b}$ and $j_{a,b}$
    \begin{eqnarray}\label{indexmaps}
    (i_{a,b},j_{a,b}) : && \{1,\ldots,d_A d_B\} \rightarrow \{1,\ldots,d_A\} \times \{1,\ldots,d_B\}  \nonumber \\
    && k<k' \Leftrightarrow a_{i_{a,b}(k)} + b_{j_{a,b}(k)} > a_{i_{a,b}(k')} + b_{j_{a,b}(k')}\,.
    \end{eqnarray}
\begin{rem}
Note that the index map is well-defined. For given $k \in \{1,\ldots, d_A d_B\}$ the indices $(i_{a,b}(k),j_{a,b}(k))$ indicate which two components of the test spectra $a$ and $b$ one has to sum up to get the $k$-th largest element in the list $\{a_i+b_j\}$ (namely the element $a_{i_{a,b}(k)}$ and the element $b_{j_{a,b}(k)}$).
\end{rem}
The morphisms $\Phi_{a,b}$ of flag varieties are defined as
\begin{eqnarray}\label{morphism2}
\Phi_{a,b}: &&\mbox{\emph{Fl}}(\mathcal{H}^{(A)}) \times \mbox{\emph{Fl}}(\mathcal{H}^{(B)}) \rightarrow \mbox{\emph{Fl}}(\mathcal{H}^{(A)}\otimes\mathcal{H}^{(B)}) \nonumber \\
&& (F_{\bullet},G_{\bullet}) \mapsto H_{\bullet} \qquad, \,\, H_k := F_{i_{a,b}(k)} \otimes G_{j_{a,b}(k)}\,\,.
\end{eqnarray}
Given a state $\rho_{AB}$ for the total system $AB$ and let
$\rho_A = \mbox{Tr}_B[\rho_{AB}]$ and $\rho_B = \mbox{Tr}_A[\rho_{AB}]$ be the marginal w.r.t system $A$ and system $B$, respectively. Moreover, $\lambda^{(AB)}:= \mbox{spec}(\rho_{AB})^{\downarrow}$, $\lambda^{(A)}:= \mbox{spec}(\rho_{A})^{\downarrow}$, $\lambda^{(B)}:= \mbox{spec}(\rho_{B})^{\downarrow}$ and choose non-degenerate test spectra $a, b$ such that their sum $a+b$ is also non-degenerate and permutations $\alpha \in \mathcal{S}_{d_A}$, $\beta \in \mathcal{S}_{d_B}$ and $\gamma \in \mathcal{S}_{d_A d_B}$. Then, whenever the intersection property
($\omega_0 :=(d_A d_B,\ldots,2,1)$ is the permutation of maximal length)
\begin{equation}\label{intersectionprop2}
\Phi_{a,b}(X_{\alpha}^{\circ}(\rho_A),X_{\beta}^{\circ}(\rho_B)) \cap X_{\gamma \omega_0}^{\circ}(-\rho_{AB}) \neq \emptyset
\end{equation}
holds, i.e.~there exists a hermitian operator
\begin{equation}
T_{AB}^{(0)} = T_A^{(0)}\otimes \mathds{1}_B + \mathds{1}_A \otimes T_B^{(0)}
\end{equation}
with
\begin{equation}
F(T_{AB}^{(0)}) \in \Phi_{a,b}(X_{\alpha}^{\circ}(\rho_A),X_{\beta}^{\circ}(\rho_B)) \cap X_{\gamma \omega_0}^{\circ}(-\rho_{AB})\,\,,
\end{equation}
we obtain
\begin{eqnarray}\label{deviationspecineq2}
\lefteqn{\sum_{i=1}^{d_A} \lambda_{\alpha_i}^{(A)}a_i + \sum_{j=1}^{d_B} \lambda_{\beta_j}^{(B)}b_j}&&\nonumber \\
&=& \min \limits_{\small{\begin{array}{c}F_{\bullet}(T_A) \in X_{\alpha}^{\circ}(\rho_A)\\ spec(T_A)=a\end{array}}}(\mbox{Tr}[\rho_A T_A])+\min \limits_{\small{\begin{array}{c}F_{\bullet}(T_B) \in X_{\beta}^{\circ}(\rho_B)\\ spec(T_B)=b\end{array}}}(\mbox{Tr}[\rho_B T_B]) \nonumber \\
&=& \min \limits_{\tiny{\begin{array}{c}(F_{\bullet}(T_A),F_{\bullet}(T_B)) \in X_{\alpha}^{\circ}(\rho_A) \times X_{\beta}^{\circ}(\rho_B)\\ spec(T_A)=a\,,\,spec(T_B)=b\end{array}}}(\mbox{Tr}[\rho_{AB} \left(T_A\otimes \mathds{1}_B + \mathds{1}_A \otimes T_B \right)]) \nonumber \\
&\leq& - \mbox{Tr}[(-\rho_{AB}) \left(T_A^{(0)}\otimes \mathds{1}_B + \mathds{1}_A \otimes T_B^{(0)}\right)] \nonumber \\
&\leq& -\min \limits_{\small{\begin{array}{c}F_{\bullet}(T_{AB}) \in X_{\gamma \omega_0}^{\circ}(\rho_{AB})\\ spec(T_{AB})=(a+b)^{\downarrow}\end{array}}}(\mbox{Tr}[(-\rho_{AB}) T_{AB}]) \nonumber \\
&=& - \sum_{k=1}^{d_A d_B} \left(\mbox{spec}(-\rho_{AB})^{\downarrow}\right)_{(\gamma \omega_0)_k}\, (a+b)^{\downarrow}_k \nonumber \\
&=&  \sum_{k=1}^{d_A d_B} \left(\mbox{spec}(\rho_{AB})^{\uparrow}\right)_{(\gamma \omega_0)_k}\, (a+b)^{\downarrow}_k \nonumber \\
&=&\sum_{k=1}^{d_A d_B} \left(\mbox{spec}(\rho_{AB})^{\downarrow}\right)_{\gamma_k}\, (a+b)^{\downarrow}_k\,.
\end{eqnarray}
Here, we used the second Hersch-Zwahlen variational principle \ref{HerschZwahlenAdv} in step 1 and 5 and for step 3 and 4 we used the fact that the intersection property (\ref{intersectionprop2}) is fulfilled. Moreover the symbols $\uparrow$ and $\downarrow$ indicate that the entries are arranged in increasing and decreasing order. Hence, whenever the intersection property (\ref{intersectionprop2}) holds we find an inequality
\begin{equation}\label{spectralineq2}
\boxed{\sum_{i=1}^{d_A} \lambda_{\alpha_i}^{(A)}a_i + \sum_{j=1}^{d_B} \lambda_{\beta_j}^{(B)}b_j \leq \sum_{k=1}^{d_A d_B} \lambda^{(AB)}_{\gamma_k}\, (a+b)^{\downarrow}_k }\,.
\end{equation}

\end{enumerate}
\textbf{For both basic versions $\{A,AB\}$ and $\{A,B,AB\}$ of the quantum marginal problem the main task is now to find all pairs of binary sequences $\pi $ and $\sigma$ and test spectra $a , b$ and permutations $\alpha, \beta, \gamma$, respectively, such that the corresponding intersection property (\ref{intersectionprop1}) and (\ref{intersectionprop2}), respectively, is fulfilled.}
In the following we will introduce an enormous machinery, the Schubert calculus, a discipline in algebraic topology to study these intersection properties more systematically. This significant additional effort is justified by the following reasons
\begin{enumerate}
\item An elegant and systematic approach is in general favorable and might lead to a deeper understanding of the structure behind the problem. This can also help to apply results from simpler versions of the QMP, as e.g.~$\mathcal{M}_{A,B,AB}$, to more general settings (those with a larger underlying multipartite quantum systems).
\item Since we should check the intersection property (\ref{intersectionprop1}) for all possible binary sequences occurring in the corresponding marginal problem, the effort in doing so increases rapidly as the dimensions of the underlying Hilbert spaces increase. Hence, a systematic approach that yields directly all tuples of binary sequences fulfilling the intersection property is preferable.
\item To find all inequalities for the problem $\{A,B,AB\}$ is even more difficult since there are uncountably many pairs of test spectra.
\item To verify (\ref{intersectionprop1}) for one single pair $(\sigma, \pi)$ is already difficult. The same also holds for (\ref{intersectionprop2}). It may be very convincing that (co)homology was developed exactly to determine dimensions of intersections of euclidian spaces embedded in a larger total euclidian space. But this task is very close to our question, whether two Schubert cells intersect or not. After all, the Schubert cells depend on given density operators and hence we cannot expect yet to find a solution of the QMP that can be applied to all possible tuples of density operators (but only for one single tupel/pair).
\item A systematic approach towards the solution of the QMP, i.e.~in particular an elegant way of describing the solution is necessary to develop computer programs which can calculate all marginal constraints
\end{enumerate}

\section{Generalized flag varieties}\label{genflagvarieties}
In this section we present the concept of generalized flag varieties in the language of Schubert calculus and later we will consider two special cases more detailed, the variety of complete flags (typically just called `flag variety') and the Grassmannian variety. For these two examples we also will prove several technical statements. We follow partially \cite{BlBr}, \cite{springerlink:10.1007/3-7643-7342-3_2} and \cite{Prag}.
In the following we consider the Hilbert space $\mathcal{H} \cong \mathbb{C}^d$ for some fixed $d$. For a given $r\leq d$ and sequence $(0=d_0< d_1< \ldots < d_r=d)$ we consider the family of nested linear subspaces
\begin{equation}
\mbox{\emph{Fl}}_{d_0,\ldots,d_r} := \{0=F_0 \lneq F_1 \lneq \ldots \lneq F_{r-1} \lneq F_r = \mathcal{H} \,|\, \mbox{dim}(F_k) = d_k \,,\, \forall k \}\,,
\end{equation}
which is called \textbf{generalized flag variety} and its elements
\begin{equation}
F_{\bullet} := [ F_1 \lneq \ldots \lneq F_r ]
\end{equation}
are called \textbf{generalized} or \textbf{partial flags}.
That these families of flags have the additional structure of a variety (see Appendix \ref{appvarieties} for a definition) is shown in Sec.\ref{pluecker} for the Grassmannian.
As already shown the variational principles \ref{HerschZwahlen} and \ref{HerschZwahlenAdv} for the two QMP $\{A,AB\}$ and $\{A,B,AB\}$ are closely related to two of these generalized flag varieties. In the first one we will deal with the \textbf{Grassmannian} $Gr_{k,d}$, defined by $(0<k<d)$
\begin{equation}
Gr_{k,d} := \mbox{\emph{Fl}}_{0,k,d}
\end{equation}
and in the second one with the \textbf{flag variety}
\begin{equation}
\mbox{\emph{Fl}}_d := \mbox{\emph{Fl}}_{0,1,\ldots,d} \,.
\end{equation}
The Grassmannian is the family of all $k$-dimensional linear subspaces of $\mathcal{H}$ and the flag variety $\mbox{\emph{Fl}}_d$ is the family of all complete flags $[ F_1 \lneq \ldots \lneq F_d ] $, where $\mbox{dim}(F_j)=j$.
Now, we explain the general concept how to equip these generalized flag varieties with the structure of a topological space and later with the one of a variety/manifold. Consider for fixed $r$ and $d_1<\ldots <d_r=d$ the group $Gl(d)$ of regular complex matrices. For all $k=1,\ldots,r$ the first $k$ column vectors of a given matrix $M \in Gl(d)$ define (w.r.t.~a fixed basis of $\mathcal{H}$) a linear subspace $F_k$ of dimension $k$. In that sense, every matrix $M$ defines a partial flag $F_{\bullet}(M)$. We introduce the equivalence relation $\sim$ (which of course depends on the fixed constants $r$ and $d_1,\ldots, d_r$) on $Gl(d)$ by saying that two regular matrices are equivalent if they define the same partial flag,
\begin{equation}
M \, \sim \, N \, :\Leftrightarrow F_{\bullet}(M) = F_{\bullet}(N) \,.
\end{equation}
Hence
\begin{equation}
\mbox{\emph{Fl}}_{d_0,\ldots,d_r} \equiv Gl(d)/ \sim \, .
\end{equation}
The last relation also defines the natural topology for the flag variety by using the natural topology for $Gl(d)$. These concepts are explained with all details in Secs.~\ref{flagvarieties}, \ref{grassmannian} for the flag variety and the Grassmannian, respectively.
It turns out that the equivalence relation $\sim$ can also be described by the corresponding parabolic subgroup $H \leq Gl(d)$ in the sense,
\begin{equation}
\mbox{\emph{Fl}}_{d_0,\ldots,d_r} \equiv Gl(d)/H \,,
\end{equation}
where $H \leq Gl(d)$ is the subgroup of all block upper triangle matrices with $r$ blocks of size $d_1, d_2-d_1, \ldots, d_r -d_{r-1}$.
In the next two sections, Secs.~\ref{flagvarieties}, \ref{grassmannian} we will study these concepts more detailed.

\subsection{Flag variety}\label{flagvarieties}
As already stated in (\ref{defflag}) in Sec.~\ref{variationalprin} a \emph{complete flag $F_{\bullet}$} is a nested sequence
\begin{equation}
F_{\bullet}= [0= F_0 \lneq F_1 \lneq \ldots \lneq F_d= \mathcal{H}]
\end{equation}
of linear subspaces $F_k$ with $\mbox{dim}(F_k)=k$. The family of all flags is called \emph{flag variety} $\mbox{\emph{Fl}}_d$. Later we will justify the term variety for this set. Two flags $F_{\bullet}$ and $G_{\bullet}$ are said to be \emph{transversal}, if for all $i,j=0,1,\ldots,d$: $\mbox{dim}(F_i \cap G_j)=\mbox{max}(i+j-d,0)$. This means two transversal flags have the property that each pair $(F_i, G_j)$ of linear subspaces has a minimal intersection dimension. Moreover flags are generically transversal. To illustrate this assumption consider the case $d=2$ and two flags $F_{\bullet}$ and $G_{\bullet}$. The only non-trivial intersection is $F_1 \cap G_1$. Since both subspaces have dimension $1$ and are embedded in a two dimensional space, it is clear that for generic choices for $F_1$ and $G_1$ their intersection has dimension $0$ and thus $F_{\bullet}$ and $G_{\bullet}$ are transversal, indeed. For a flag $F_{\bullet}$ we define its complementary flag $F_{\bullet}^{\perp}$ by setting $F_{k}^{\perp}:= (F_{d-k})^{\perp}$. This means that $F_{\bullet}$ and $F_{\bullet}^{\perp}$ are transverse. We define the subset $X_{\perp}(F_{\bullet})$ as the family of flags transverse to $F_{\bullet}$. Later it will be stated that these subsets are open sets in \emph{Fl}$_d$ w.r.t.~the natural topology that is introduced in Remark \ref{remtopology}. Intuitively, openness w.r.t.~to a reasonable topology is clear since deforming a flag a very little bit does not change its property to be transversal to a second given flag.

To work with flags it is helpful to represent them in a more concrete form. For the following we choose an orthonormal basis $\{e_1,\ldots,e_d\}$ that will be called standard basis and we express arbitrary vectors $v \in \mathcal{H}$ w.r.t to this basis in form of coordinate vectors $\vec{v}= (v_1,\ldots,v_d)^T$, i.e.~$v = v_1 e_1+ \ldots+ v_d e_d$ and in particular $(\vec{e_k})_i= \delta_{k i}$. Given a regular complex matrix $g \in Gl(d)$ we can understand it as built up by $d$ coordinate column vectors $\vec{g}_1,\ldots, \vec{g}_d$ representing vectors $g_1, \ldots, g_d \in \mathcal{H}$ w.r.t to this standard basis:
\begin{equation}
g = (\vec{g}_1,\ldots,\vec{g}_d) \,.
\end{equation}
Since $g$ is regular, these $d$ vectors $g_1,\ldots,g_d$ are linearly independent and therefore $g$ defines a (complete) flag $G_{\bullet}$ by setting for all $k=0,1,\ldots,d$: $G_k := \langle g_1, \ldots, g_k\rangle$. This defines a map $\Lambda$
\begin{eqnarray}
\Lambda: Gl(d) &\rightarrow & \mbox{\emph{Fl}}_d \,,\nonumber  \\
 g &\mapsto & G_{\bullet}(g) = [0 \lneq \langle g_1 \rangle \lneq \ldots \lneq \langle g_1, \ldots,g_d\rangle = \mathcal{H}] .
\end{eqnarray}
Obviously, $\Lambda$ is not injective and gives rise to an equivalence relation. We call elements $g, g' \in Gl(d)$ equivalent, $g \sim g'$, if
$\Lambda(g)= \Lambda(g')$, that means they represent the same flag. Thus,
\begin{equation}
\mbox{\emph{Fl}}_d \cong Gl(d)/\sim \,. \label{flagvarietyGL1}
\end{equation}
Mathematically, this equivalence relation means to take the $k$-th column of a given matrix $g \in Gl(d)$ modulo linear combinations of the first $k-1$ columns and normalization. The statement \ref{flagvarietyGL1} can also be expressed more elegantly. Let $B \subset Gl(d)$ be the set of regular (complex) upper triangle matrices. It is shown in the Appendix \ref{Schubertcalculus} that $B$ is also a subgroup of $Gl(d)$, the so called Borel subgroup of $Gl(d)$. We call two elements $g_1, g_2$ equivalent if $g_1 (g_2)^{-1} \in B$. This defines an equivalence relation on $Gl(d)$, which separates into disjoint equivalence classes, namely the left cosets $g B$ $(g \in Gl(d))$. Moreover, this equivalence relation is identical to the one mentioned above and (\ref{flagvarietyGL1}) can be rephrased as
\begin{equation}
\mbox{\emph{Fl}}_d \cong Gl(d)/B \,. \label{flagvarietyGL2}
\end{equation}
\begin{rem}\label{remtopology}
The flag variety $\mbox{\emph{Fl}}_d$ carries a natural topology induced by the natural bijection \ref{flagvarietyGL2}: The standard topology of $\mathbb{C}$ induces a topology on $Gl(d)$ and yields a natural topology (relative topology) on $Gl(d)/B$ and then also on $\mbox{\emph{Fl}}_d$. W.r.t.~this topology the $1-1$-correspondence in (\ref{flagvarietyGL1}) and (\ref{flagvarietyGL2}) is also a homeomorphism.
\end{rem}
\noindent The representation of flags by left cosets/equivalence classes is for our later purpose still too abstract and in the following we will find a `nice' rule for picking a unique representant from each left coset $gB$ to represent its flag. This is realized in the following.
Consider a matrix $g \in Gl(d)$ and express it with coordinate (column) vectors, as $g = \left(\vec{g}_1,\ldots,\vec{g}_d\right)$.
The method of transforming $g$ to its \emph{Column Echelon Form} (CEF) works as follows. In the first step we consider the first coordinate vector $\vec{g}_1$ and denote the position of the first (from below) non-vanishing coordinate by $\alpha_1$ and divide the whole vector $\vec{g}_1$ by the corresponding coordinate. This changes this coordinate vector to the new one $\vec{g}'_1$, that has the form
\begin{equation}
\vec{g}'_1= (\ast,\ldots,\ast,1,\underbrace{0,\ldots,0}_{d-\alpha_1})^T \, .
\end{equation}
Moreover, to finish the first step we add $\mathbb{C}$-multiples of this new vector $\vec{g}'_1$ to all the other $d-1$ coordinate vectors such that all
have a vanishing $(\alpha_1)$-components.
This first change of the matrix $g$ to the new matrix $g'$ can be expressed by multiplying $g$ from the right by an appropriate matrix $b_1 \in B$,
\begin{equation}
g' = (\vec{g}'_1,\ldots,\vec{g}'_d)= \left(\begin{array}{cccc} \vdots & & & \vdots \\ \ast &\ast&\cdots& \ast \\ 1& 0 & \cdots & 0 \\ 0&\ast& \cdots &\ast \\ \vdots &\vdots&&\vdots \\ 0&\ast&\cdots&\ast\end{array}\right) = g \,b_1
\end{equation}
In the second step we consider the vector $\vec{g}'_2$ of the new matrix $g'$ and denote the position of its first (from below) non-vanishing coordinate index by $\alpha_2$ and divide the whole vector by the corresponding coordinate. To finish the second step we add multiples of this new second coordinate vector to the
vectors $\vec{g}'_3,\ldots, \vec{g}'_d$ such that all of them have not only a vanishing $(\alpha_1)$-component but also a vanishing $(\alpha_2)$-component. This second step can also be rephrased by multiplying the matrix, here $g'$, by an appropriate matrix $b_2 \in B$. We end up with a new matrix $g''$ (in the following we present the case $\alpha_1 > \alpha_2$, for the other case $\alpha_1 < \alpha_2$ the matrix looks slightly different)
\begin{equation}
g'' =(\vec{g}''_1,\ldots,\vec{g}''_d)=  \left(\begin{array}{ccccc}
\,\,\vdots\,\, &\,\, \,\,&\,\,\,\,&\,\, \,\, &\,\, \vdots \,\,\\
\,\,\ast\,\, &\,\,\ast\,\,&\,\,\cdots\,\,&\,\, \,\,&\,\,\ast \,\,\\
\,\,1\,\,&\,\, 0 \,\,&\,\, \cdots \,\,&\, \,\,\,&\,\, 0\,\, \\
\,\,0\,\,&\,\,\ast\,\,&\,\, \cdots \,\,&\,\,\,\, &\,\,\ast\,\, \\
\,\,\vdots \,\,&\,\,\vdots\,\,&\,\,\,\,&\,\,\,\,&\,\,\vdots\,\, \\
\,\,0\,\, &\,\,1\,\,&\,\,0\,\,&\,\,\cdots \,\,&\,\, 0 \,\,\\
\,\,\,\,&\,\,0\,\,&\,\,\,\,&\,\,\,\,&\,\,\,\, \\
\,\,\vdots \,\,&\,\,\vdots\,\,&\,\,\,\,&\,\,\,\,&\,\,\vdots \,\,\\
\,\,0\,\,&\,\,0\,\,&\,\,\cdots\,\,&\,\,\,\,&\,\,\ast\,\,
\end{array}\right) = g \, b_1 \, b_2 \,.
\end{equation}
We go on with this procedure up to the very last column vector and end up with a final matrix $g^{c}$ given by
\begin{equation}
g^{c}= g \, b_1 \cdot \ldots \cdot b_{d} = g \, b
\end{equation}
with $b_1,\ldots, b_d \in B$ and hence $b_1 \cdot \ldots \cdot b_d =: b \in B$. It is worth making a comment on the possibility to finish this procedure up to the very last column vector. First of all, each of these $d$ steps leads to a unique result under the condition that the corresponding vector of investigation has still a non-vanishing component. That this is the case, in particular also for the last column vector (which consists of at least $d-1$ zeros) is clear since the procedure does not change the rank of the matrix $g$, which was maximal, namely $d$. Hence, none of the vectors could ever become identical to the zero vector.
\begin{floatingfigure}[h]{6cm}
\hspace{-0.6cm}
\includegraphics[width=5.9cm]{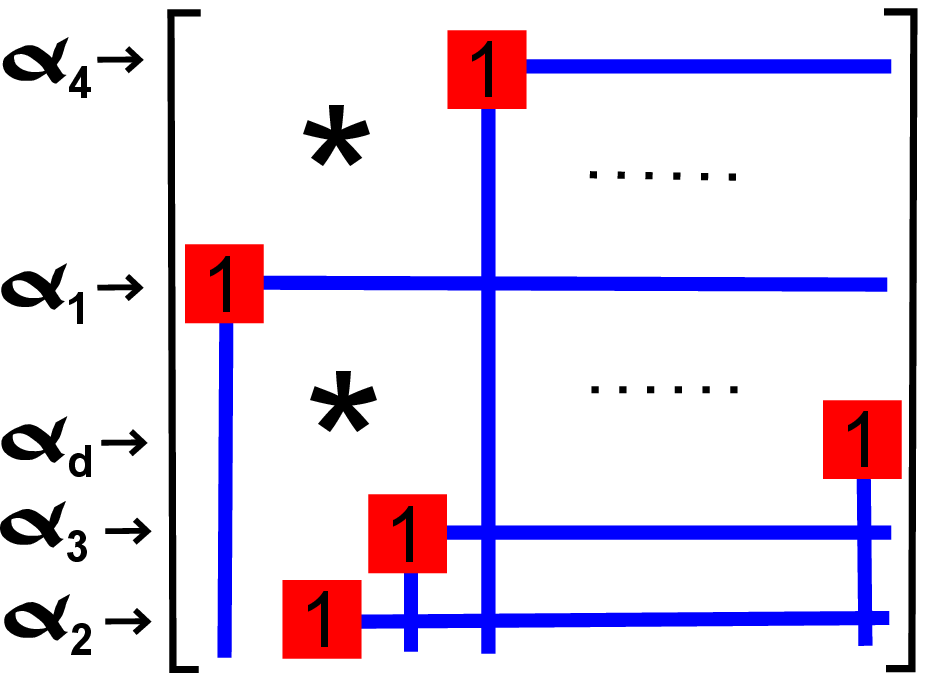}
\captionC{Illustration of the CEF (see text) and the concept of pivots.}
\label{figurepivotsFlag}
\end{floatingfigure}

The final form $g^C$ is called \emph{Column Echelon Form} (CEF) and has the property that every column has a characteristic $1$ that is followed only by zeros by going downwards in the column and also rightwards in the corresponding row. All the other entries are arbitrary $\mathbb{C}$-numbers (e.g.~also $1$ or $0$). The characteristic $1$'s are called \emph{pivots} and their row and column coordinate can be represented in form of a permutation
$\alpha = (\alpha_1, \ldots,\alpha_d)$, where $\alpha_k$ is the position of the pivot in the $k$-th column vector. All this aspects are also presented in Fig.~\ref{figurepivotsFlag}. There and in following we use squared brackets to refer to the equivalence class represented by the matrix form inside the brackets and round brackets to refer to matrices. The $d$ pivots are marked in red. They `look' to the right and also downwards, which is indicated by blue lines. All entries on these lines are zero.

Nevertheless the most important aspect is that this unique procedure maps a given matrix $g \in Gl(d)$ to another matrix $g^c \in Gl(d)$ belonging to the same equivalence class $g B$, i.e.~both represent the same flag! To conclude that the Column Echelon Form defines a \emph{unique} representant for every equivalence class $gB$ we still need the Lemma presented in the Appendix \ref{Schubertcalculus}, which states that all matrices belonging to the same equivalence class also have the same CEF, which then makes the $1-1$ correspondence (\ref{flagvarietyGL2}) more concrete. By introducing the subset $Gl(d)^{CEF} \subset Gl(d)$ of CEF's this means
\begin{equation}\label{flagvarietyGL3}
\mbox{\emph{Fl}}_d \cong Gl(d)^{CEF} \,.
\end{equation}
This means that every complete flag is in a $1-1$-correspondence to a CEF as shown in Fig.~\ref{figurepivotsFlag}, which is completely determined by a permutation $\alpha$ and $\mathbb{C}$-values for the free variables indicated by stars $\ast$.
\begin{rem}
The concept of CEF also allows us to determine the complex dimension of flag manifolds. $Gl(d)$ as a subspace of $\mathbb{C}^{d\times d}$ is $d^2$ dimensional.
Dividing by $B$ in form of the CEF means to reduce the dimension in the first column by 1, in the second column by 2 and so on. This leads to the dimension of $\emph{\mbox{Fl}}_d$, $\mbox{dim}(\emph{\mbox{Fl}}_d) = d^2-\binom{d+1}{2}=\binom{d}{2}$.
\end{rem}

\noindent Moreover we claim
\begin{lem}\label{manifoldatlas}
The flag variety $\mbox{\emph{Fl}}_d$ (with its natural topology) is a differentiable manifold and $A = \{X_{\perp}(F_{\bullet})|F_{\bullet} \in \mbox{\emph{Fl}}_d  \}$ is an atlas on $\mbox{\emph{Fl}}_d$ with complex dimension $\binom{d}{2}$.
\end{lem}
\noindent We present the main ideas of the proof.
\begin{proof}
Since $F_{\bullet} \in X_{\perp}(F_{\bullet}^{\perp})$, $\{X_{\perp}(F_{\bullet})|F_{\bullet} \in \mbox{\emph{Fl}}_d  \}$ is a covering of  $\mbox{\emph{Fl}}_d $ and it can be shown that $\{X_{\perp}(F_{\bullet})|F_{\bullet}\in \mbox{\emph{Fl}}_d  \}$ is open.  Now, we verify that $X_{\perp}(F_{\bullet})$ is homeomorphic to a subset of $\mathbb{C}^{\binom{d}{2}}$, by stating the homeomorphism explicitly. Let $\{e_1,\ldots,e_d\}$ be the orthonormal reference basis. Let $E_{\bullet}$ be the flag with linear subspaces $E_k=\langle e_1,\ldots,e_k \rangle$. The map
\begin{eqnarray}
\Phi:\mathbb{C}^{\binom{d}{2}} &\rightarrow&  X_{\perp}(F_{\bullet}) \nonumber \\
(\ast,\ldots,\ast) &\mapsto& \left( \begin{array}{cccc} \ast & \cdots & \ast & 1 \\ \vdots & &  $\reflectbox{$\ddots$}$  & 0 \\ \ast & 1 &   & \vdots \\ 1 & 0 & \cdots & 0 \end{array} \right)
\end{eqnarray}
is an bijection between $\mathbb{C}^{\binom{d}{2}}$ and $X_{\perp}(E_{\bullet})$ (in the spirit of (\ref{flagvarietyGL3})) and it turns out to be also homeomorphic. For given flags $F_{\bullet}, G_{\bullet}$ there exist $A_F, A_G \in Gl(d)$ such that $E_{\bullet} = A_F F_{\bullet} = A_G G_{\bullet}$ and the desired homeomorphisms are $A_F \circ \Phi : \mathbb{C}^{\binom{d}{2}}\rightarrow X_{\perp}(F_{\bullet})$ and $A_G \circ \Phi : \mathbb{C}^{\binom{d}{2}}\rightarrow X_{\perp}(G_{\bullet})$ and they are $\mathcal{C}^{\infty}$ on the preimage of the overlap of two open sets.
\end{proof}

As already pointed out in Sec.~\ref{variationalprin} we are in particular interested in Schubert cells of the flag variety, which play a central role in the variational principle \ref{HerschZwahlenAdv}. Therefore, we first would like to understand the concept of Schubert cells $X_{\alpha}^{\circ}(F_{\bullet})$ introduced in \ref{SchubertcellGr} in the language of CEF. It is obvious that $X_{\alpha}^{\circ}(F_{\bullet})$ is represented by the family of CEF with fixed pivot structure $\alpha$ (see Fig.~\ref{figurepivotsFlag}) and the stars are understood as complex variables. This means that after fixing $\alpha$ (pivot structure) in \ref{figurepivotsFlag} every flag $F_{\bullet} \in X_{\alpha}^{\circ}(F_{\bullet})$ is represented by a corresponding CEF with fixed complex values for the stars and every CEF with fixed values for the stars represents a concrete flag in this Schubert cell. Due to its presentational advantage we will often use the CEF-representation of Schubert cells.
\begin{lem}
Every Schubert cell $X_{\alpha}^{\circ}(F_{\bullet})$ is homeomorphic to $\mathbb{C}^{l(\alpha)}$, where $l$ denotes the length of permutations.
\end{lem}
The $1-1$-correspondence between $X_{\alpha}^{\circ}(F_{\bullet})$ and the CEF of Fig.~\ref{figurepivotsFlag} with stars as complex variables was already explained. Hence, it is clear that $X_{\alpha}^{\circ}(F_{\bullet})$ is homeomorphic to $\mathbb{C}^r$. It can easily be verified that the correct power is given by $r= l(\alpha)$, which is nothing else but the number of stars in the CEF for fixed pivot structure $\alpha$. Moreover we recall \ref{flagvarietyGL3},
\begin{lem}\label{partitionFl}
\begin{equation}
\mbox{\emph{Fl}}_d = \bigcup_{\alpha \in \mathcal{S}_d }^{\circ}  X_{\alpha}^{\circ}(F_{\bullet})\,.
\end{equation}
\end{lem}

\noindent For the later purpose (a motivation is given much later) we introduce the \emph{Flag Schubert varieties} $X_{\alpha}(F_{\bullet})$:
\begin{defn}\label{defnFlagSchubert}
The \emph{Flag Schubert varieties} $X_{\alpha}(F_{\bullet})$ are defined as the closures w.r.t.~to the natural topology of the
flag Schubert cells $X_{\alpha}^{\circ}(F_{\bullet})$.
\end{defn}

In the following we need to determine the closure of Schubert cells and present related concepts and results as e.g.~the Bruhat order. For the complete flag variety discussed in this part this is not that easy and we will keep the introduction short. In particular we will skip the proofs. Nevertheless, in the next section we will introduce similar concepts for the Grassmannian and there we will present some details and also some proofs.

Recall that the topology for $Gl(d)$ induces a topology for the flag variety according to Remark \ref{remtopology}. Nevertheless, e.g.~for given Schubert cell $X_{\alpha}^{\circ}(F_{\bullet})$ it is not that obvious how its closure looks like. We would like to give some intuitive understanding for the boundary of a set. Let us consider the family of regular matrices
\begin{equation}
\left( \begin{array}{cc} z & 1 \\ 1&0  \end{array}\right) \qquad,\,\, z\in \mathbb{C}\,,
\end{equation}
representing the elements of the Schubert cell $X_{(2,1)}^{\circ}(F_{\bullet})$. The boundary of this set can be obtained by considering the limit $|z| \rightarrow\infty$. This coincides with our understanding of closure and boundary for Euclidean spaces and hence also for $Gl(d)$. The flag/flags obtained by this limit are defined by all linear subspaces $F_0 \lneq F_1  \lneq F_2$. Here only $F_1$ is non-trivial and we find $F_1 = \left( \begin{array}{c} 1\\ 0 \end{array}\right)$. Hence
\begin{equation}
\lim_{z\rightarrow\infty} \,\left( \begin{array}{cc} z & 1 \\ 1&0  \end{array}\right) \,\,\,\mbox{mod}(B) =  \left( \begin{array}{cc} 1 & 0 \\ 0&1  \end{array}\right)\,\,\,\mbox{mod}(B)\,,
\end{equation}
and thus
\begin{equation}
\overline{\left[ \begin{array}{cc} \ast & 1 \\ 1&0  \end{array}\right]} = \left[ \begin{array}{cc} \ast & 1 \\ 1&0  \end{array}\right] \cup \left[ \begin{array}{cc} 1 & 0 \\ 0&1  \end{array}\right]\,.
\end{equation}
For Hilbert spaces with larger dimensions $d>2$ the form of the closure is not obvious anymore.
For later purpose we state an important but very technical result. Therefore, we first introduce the so-called Bruhat order $\leq$, which is a
partial order on the family of permutations $\alpha \in \mathcal{S}_d$ of $d$ elements. We follow essentially \cite{Fult}.
\begin{defn}\label{defbruhat}
For the group $\mathcal{S}_d$ of permutations we define the Bruhat order $\leq$ by one of the following equivalent definitions 1., 2. and 3.:
For $\alpha, \beta \in \mathcal{S}_d$ we define
\begin{equation}
\alpha \leq \beta \qquad \Leftrightarrow
\end{equation}
\begin{enumerate}
\item there exists a sequence of permutations,
\begin{equation}
\alpha \rightarrow \gamma_1\rightarrow\ldots \gamma_k \rightarrow \beta\,,
\end{equation}
where each permutations is given by applying an appropriate transposition to the previous one and the length of the permutation increases in every step exactly by $1$.
\item $\forall \,1\leq p\leq d:$
\begin{equation}
\{\alpha_1,\ldots,\alpha_p\}_{<}\leq \{\beta_1,\ldots,\beta_p\}_{<} \,,
\end{equation}
where both lists are arranged in increasing order (indicated by the index $<$)
and $\leq$ then means that the $i$-th element of the first set is smaller or equal than the $i$-th element of the second set for all $1\leq i\leq p$.
\item
\begin{equation}
r_{\alpha}(p,q) \geq r_{\beta}(p,q) \qquad \forall 1\leq p,q \leq d\,,
\end{equation}
where we defined
\begin{equation}
r_{\alpha}(p,q) := \#(i\leq p\,:\,\alpha_i \leq q)\,.
\end{equation}
\end{enumerate}
\end{defn}
\noindent We do not verify that all three definition of the Bruhat order given in Definition \ref{defbruhat} are equivalent.
To get an intuitive understanding for the Bruhat order we find possible relations for some concrete permutations.
First, note that
\begin{equation}
(1,2,\ldots,d)\leq \alpha \leq (d,d-1,\ldots,1)\qquad \forall  \alpha \in \mathcal{S}_d
\end{equation}
and
\begin{equation}
\alpha \leq \beta \Rightarrow l(\alpha) \leq l(\beta)\,.
\end{equation}
Consider now $\alpha = (2,1,4,3,5)$ and $\beta = (4,1,2,5,3)$. We find $l(\alpha)= 2 < 4 =l(\beta)$, which implies that if $\alpha$ and $\beta$ are related then
$\alpha \leq \beta$. By applying successively length-increasing transpositions we find the sequence
\begin{equation}
\alpha = (\underline{2},1,\underline{4},3,5) \rightarrow (4,1,2,\underline{3},\underline{5}) \rightarrow (4,1,2,5,3) = \beta\,\,,
\end{equation}
which means $\alpha \leq \beta$. Alternatively we can also confirm this by
\begin{equation}
(r_{\alpha}(p,q)) = \left(\begin{array}{ccccc} 0&1&1&1&1 \\ 1&2&2&2&2\\ 1&2&2&3&3 \\ 1&2&3&4&4\\ 1&2&3&4&5\\  \end{array}\right)      \geq   \left(\begin{array}{ccccc} 0&0&0&0&1 \\ 0&0&1&1&2\\ 0&0&1&2&3 \\ 1&1&2&3&4\\ 1&2&3&4&5\\  \end{array}\right)     = (r_{\beta}(p,q))\,.
\end{equation}
On the other hand if we consider $\gamma = (3,2,1,5,4)$ with length $4$ and the same $\alpha$ we find
\begin{equation}
(r_{\alpha}(p,q)) = \left(\begin{array}{ccccc} 0&1&1&1&1 \\ 1&2&2&2&2\\ 1&2&\boxed{2}&3&3 \\ 1&2&3&4&4\\ 1&2&3&4&5\\  \end{array}\right)      \ngeq   \left(\begin{array}{ccccc} 0&0&1&1&1 \\ 0&1&2&2&2\\ 1&2&\boxed{3}&3&3 \\ 1&2&3&3&4\\ 1&2&3&4&5\\  \end{array}\right)     = (r_{\gamma}(p,q))\,,
\end{equation}
which means $\alpha$ and $\beta$ are not related.
By the use of the Bruhat order we state (without any proof):
\begin{lem}\label{partitionFlSchubert}
Each Schubert variety $X_{\alpha}(F_{\bullet})$ can be decomposed according
\begin{equation}
X_{\alpha}(F_{\bullet}) = \bigcup_{\beta \leq \alpha}^{\circ}  X_{\beta}^{\circ}(F_{\bullet})\,.
\end{equation}
\end{lem}
\begin{rem}
Lem.~\ref{partitionFlSchubert} in particular states
\begin{equation}
\beta \leq \alpha \qquad  \Rightarrow \qquad   X_{\beta}(F_{\bullet}) \subset X_{\alpha}(F_{\bullet})\,.
\end{equation}
\end{rem}

\subsection{Grassmannians}\label{grassmannian}
In this section we recall the definition of the \emph{complex Grassmannian} (see also (\ref{GrFlagPredefinition}) in Sec.~\ref{variationalprin}) and study its mathematical structure. We follow partially \cite{BlBr} and \cite{Blas}. There are different ways of introducing the Grassmannian depending on own preferences. They may be based on set theoretical aspects or also on further geometric and algebraic aspects.
\begin{defn}
Let $\mathcal{H}$ be a $d$-dimensional complex Hilbert space. The (complex) Grassmannian $Gr_{k,d}$ is defined as the set of $k$-dimensional subspaces in $\mathcal{H}$,
\begin{equation}
Gr_{k,d} := \{V \leq \mathcal{H} | \mbox{dim}(V) = k\} \,.
\end{equation}
To endow the Grassmannian with a geometric structure we can alternatively define $Gr_{k,d}:= Gl(d)/H$, where $H$ is the corresponding stabilizer of the transitive group action given by the general linear group $Gl(d)$ on $\mathcal{H}$ keeping the corresponding linear subspace invariant or alternatively  $Gr_{k,d}:= SU(d)/(SU(k) \times SU(d-k))$, where $SU$ is the special unitary group (over the complex field).
\end{defn}
It is convenient to find a concrete representation of elements in $Gr_{k,d}$ by regular matrices $M \in Gl(d)$ as already pointed out in the introduction of generalized flag varieties at the beginning of Sec.~\ref{genflagvarieties}. The first $k$ column coordinate vectors then span (w.r.t.~to a fixed basis) the $k$-dimensional subspace $V \in Gr_{k,d}$. The remaining $d-k$ column coordinate vectors $\vec{g}_{k+1},\ldots, \vec{g}_d$ then complete $V$ to the whole ambient space $\mathcal{H}$. To represent $V$ they are irrelevant and therefore we skip them and represent the Grassmannian by the family $\mathcal{M}_{k,d} \subset \mathbb{C}^{d\times k}$ of rank $k$ matrices.
Since we already introduced the CEF with all details and proofs to describe complete flags, we present shortly the analogous concepts for the Grassmannian. As indicated above, after fixing an orthonormal basis for $\mathcal{H}$, every $k$-dimensional subspace $V$ can be represented by a $d \times k$ matrix $g$ of the form
\begin{equation}
g = \left( \vec{g}_1 , \ldots,  \vec{g}_k \right) \,,
\end{equation}
with column coordinate vectors representing a basis $g_1,\ldots, g_k$ for $V$. For a given point $V \in Gr_{k,d}$ this description is not unique: Changing the column vectors $\vec{g}_j$ by scalar multiplication and addition of arbitrary linear combination of the other column vectors does not change the subspace $V$. Identifying two matrices $g$ and $h$, $g \sim h$, if they can be transformed to each other by this algebraic process the Grassmannian is given by
\begin{equation}\label{matrixgras}
Gr_{k,d} \cong \mathcal{M}_{k,d} / \sim \,.
\end{equation}
\\
\begin{floatingfigure}[h]{5cm}
\hspace{-0.5cm}
\includegraphics[width=4.9cm]{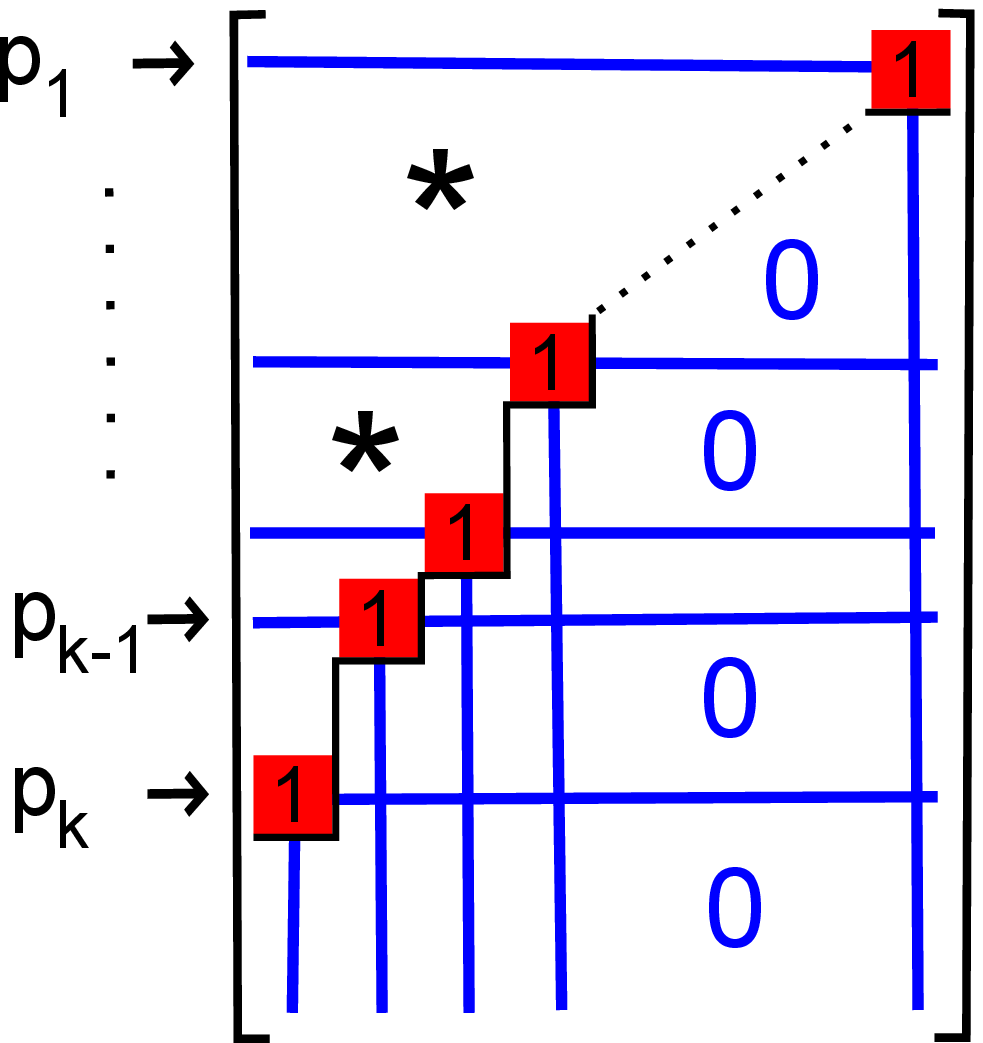}
\captionC{Illustration of the strict CEF (see text) and the concept of pivots.}
\label{figurepivots}
\end{floatingfigure}

We can define in a very similar way as done for the flag variety a column echelon form (see Fig.~\ref{figurepivots}). Here we call it \emph{strict Column Echelon Form} (sCEF) since we can additionally arrange the pivots in an increasing order. In every column $\vec{g}_i$ there is one characteristic $1$, called \emph{pivot}, at a position $p_{k-i+1} \in\{1,\ldots,d\}$, which is the first (from below) non-vanishing entry in the column. By adding column vector $\vec{g}_i$ with appropriate weights to all the other $k-1$ column vectors we can guarantee that they have a zero entry at the position $p_{k-i+1}$ of the pivot of vector $\vec{g}_i$. Moreover we can permute all column vectors to obtain the pivot structure shown in Fig.~\ref{figurepivots}. All these transformations of the representing matrix do not change the vector space that it is representing. We say that every pivot looks in three direction, to the right and left side and downwards (remember in the CEF the pivots were not looking to the left side), which should mean that in these directions the matrix entries are zeros.

It is easy to verify that the sCEF yields a unique representant for every equivalence class in $  \mathcal{M}_{k,d}/\sim$. The remaining entries are arbitrary fixed complex numbers represented by stars. In the following, for a given matrix in sCEF we describe its pivots structure by $k$ integers $1\leq p_1 < \ldots < p_k \leq d $, which means that the pivot of column $i$ is at position $p_{k-i+1}$. We simply call $p = (p_1,\ldots,p_k)$ pivot structure or position of the pivots.
Moreover we can define a natural forgetful map $P_k : \mbox{\emph{Fl}}_k \rightarrow Gr_{k,d},\,k\leq d$, which maps a flag $F_{\bullet}$ to its $k$-dimensional subspace $F_k$. This map is surjective and smooth and it also assigns a topology to the Grassmannian. Lem.~\ref{manifoldatlas} (see Sec.~\ref{flagvarieties}) then implies that the Grassmannian is also a differentiable manifold.

In the following we consider again the concept of (Grassmannian) Schubert cells introduced in Definition \ref{SchubertcellGr} and will see that all elements of a given Schubert cells represented in sCEF have the same pivot structure.
Consider an element $V \in Gr_{k,d}$ represented in sCEF with the pivot at positions $p_1,\ldots, p_k$, where the external orthonormal basis $\{f_1,\ldots,f_d\}$ is given by the complete flag $F_{\bullet}$ labeling the corresponding Schubert cell $S_{\pi}^{\circ}(F_{\bullet})$. This should mean that the vectors $g_i$ defining the $k$-dimensional subspace $V$ are given by
\begin{equation}
g_i = \sum_{j=1}^d \, (\vec{g}_i)_j \, f_j \,.
\end{equation}
By reconsidering Definition \ref{SchubertcellGr} we easily see that $V$ with the pivot structure $p = (p_1,\ldots,p_k)$ belongs to the Schubert cell $S_{\pi}^{\circ}(F_{\bullet})$ with $\pi = \pi(p)$,
\begin{equation}\label{pivotbinary1}
\pi_j=\begin{cases}
1,  & \mbox{if} \qquad j \in \{p_1,\ldots,p_k\}\\
  0 & \mbox{if} \qquad j \not \in \{p_1,\ldots,p_k\} \,.
\end{cases}
\end{equation}
Hence, all matrices with the same pivot structure belong to the same Schubert cell. The converse also holds since relation (\ref{pivotbinary1}) can be inverted
and we find
\begin{equation}
p_k = \mbox{position of the $k$-th $1$ in $\pi = (\pi_1,\ldots,\pi_d) $} \,,
\end{equation}
i.e.~the pivots have the same positions as the $1$'s in $\pi$.
By interpreting all the stars in the matrix shown in Fig.~\ref{figurepivots} as complex degrees of freedom instead of fixed numbers this matrix
is then exactly the sCEF-representation of the corresponding Schubert cell $S_{\pi(p)}^{\circ}(F_{\bullet})$.

For the later purpose we still introduce a bijection between all binary sequences $\pi = (\pi_1,\ldots,\pi_d)$ with fixed weight $k$ and the family of Young diagrams $\alpha= (\alpha_1,\ldots,\alpha_k)$ by
\begin{equation}
\pi \mapsto \alpha = (\alpha_1,\ldots, \alpha_k) \qquad,\, \mbox{with} \,\,\,\alpha_i := p_{k+1-i}(\pi) - (k+1-i) \,.
\end{equation}
This defines a $1-1$ map between $\{\pi \in \{0,1\}^d\,|\, \|\pi\|_1=k\}$ and the Young diagrams $\alpha$ that fit into a $k \times (d-k)$ rectangle. Then, we observe,
\begin{equation}
S_{\alpha}^{\circ}(F_{\bullet}):= S_{\pi(\alpha)}^{\circ}(F_{\bullet}) = \{ V\leq\mathcal{H}\,|\, \forall i, \alpha_{k+1-j}+j \leq i < \alpha_{k-j}+j+1: \mbox{dim}(V\cap F_i)= j \}
\end{equation}

Let us summarize these insights. We found a concrete description of Schubert cells by sCEF and fixed pivot structures and also a $1-1$ map between Young diagrams contained in the $k \times (d-k)$ rectangle and the family of Schubert cells of the Grassmannian $Gr_{k,d}$. Moreover we state
\begin{lem}\label{homeo}
Every Grassmannian Schubert cell $S_{\alpha}^{\circ}(F_{\bullet})$ is homeomorphic to $\mathbb{C}^{\|\alpha\|_1}$\,.
\end{lem}
The $1-1$ correspondence was already shown. The additional topological structure of the bijection is not that obvious and we skip its proof.
For the later purpose (a motivation is given much later) we introduce the \emph{Grassmannian Schubert varieties} $S_{\alpha}(F_{\bullet})$:
\begin{defn}\label{defnGrSchubert}
The \emph{Grassmannian Schubert varieties} $S_{\alpha}(F_{\bullet})$ are defined as the closures w.r.t.~to the natural topology of the
Grassmannian Schubert cells $S_{\alpha}^{\circ}(F_{\bullet})$.
\end{defn}
Since all Schubert cells are disjoint, the following lemma is trivial
\begin{lem}\label{partition}
\begin{equation}
Gr_{k,d} = \bigcup_{\begin{array}{c}\beta \subset k \times (d-k) %\\ \|\beta\|_1 = k
\end{array}}^{\circ} S_{\beta}^{\circ}(F_{\bullet})
\end{equation}
\end{lem}
\noindent Moreover, we state
\begin{lem}\label{Schubertvarieties}
\begin{eqnarray}
S_{\alpha}(F_{\bullet}) & =& \bigcup_{\beta \subset \alpha} S_{\beta}^{\circ}(F_{\bullet}) \nonumber \\
 &= & \{ V\leq\mathcal{H}\,|\, \forall j: \mbox{dim}(V\cap F_{\alpha_{k-1+j}+j}) \geq j \}
\end{eqnarray}
\end{lem}
\noindent We do not present a complete proof of this lemma, but a very strong motivation for it by referring to our intuitive understanding of the natural topology.
The idea for verifying the first line in Lem.~\ref{Schubertvarieties} is to think of the stars in Fig.~\ref{figurepivots}
as concrete complex numbers, vary some of them and understand which flags (represented in CEF) can
(at least asymptotically) be reached. To make it easier we restrict to the case $Gr_{1,d}$, i.e.~every CEF consists of a $1 \times d$ matrix and $\alpha =(\alpha_1)$.
Consider a point $(d_1,\ldots,d_{\beta_1},1,0,\ldots)$ in the Grassmannian with complex numbers $d_1,\ldots,d_{\beta_1}$ and a `pivot structure' described by the Young diagram $\beta =(\beta_1) \varsupsetneq \alpha$, i.e.~$\beta_1 > \alpha_1$. By varying the $\alpha_1$ complex entries $c_1,\ldots,c_{\alpha_1}$ in $(c_1,\ldots,c_{\alpha_1},1,0,\ldots)$ it won't be possible to reach asymptotically the vector $(d_1,\ldots,d_{\beta_1},1,0,\ldots)$. This is due to the characteristic $1$ at position $\beta_1+1>\alpha_1+1$. On the other hand, if $\beta_1\leq \alpha_1$ and again $d_1,\ldots, d_{\beta_1}$ complex numbers, we observe that
\begin{equation}
(\lambda d_1,\ldots,\lambda d_{\beta_1},\lambda,0,\ldots,0,1_{\alpha_1},0,\ldots) \rightarrow (d_1,\ldots,d_{\beta_1},1,0,\ldots)
\end{equation}
w.r.t.~the natural topology on $Gr_{1,d}$, as $\lambda \rightarrow \infty$.
This explanation can be extended to the Grassmannian $Gr_{k,d}$ for arbitrary $k$.
The second statement in Lem.~\ref{Schubertvarieties} is easy to show and we skip its proof.
\begin{ex} To illustrate the different ways of labeling Schubert cells and also Lem.~\ref{Schubertvarieties} we consider the Grassmannian $Gr_{2,4}$. There are in total six Schubert cells. They are listed in the Tab.~\ref{tab:Gr24}. There we show all the different ways of labeling them, namely by binary sequences $\pi$, partitions $\alpha$/Young diagrams and the strict Column Echelon Form.
\\
\\
\begin{table}[h]
\centering
\setlength{\arraycolsep}{0.2cm}
\renewcommand{\arraystretch}{2.7}
$\begin{array}{|c|c|c|c|c|}
\hline\pi & \alpha & \mbox{sCEF} & \mbox{Young} & \mbox{sCEF of closure} \\ \hline
(0,0,1,1)&(2,2)& \left[\setlength{\arraycolsep}{0.06cm}
\renewcommand{\arraystretch}{0.8
}\begin{array}{cc} \ast&\ast \\ \ast &\ast \\ 0& 1\\ 1&0\end{array}\right] & \begin{array}{c} \vspace{-0.2cm}\small{\yng(2,2)}\end{array} &
\left[\setlength{\arraycolsep}{0.06cm}\renewcommand{\arraystretch}{0.8}\begin{array}{cc} \ast& \ast\\ \ast &\ast \\ 0&1 \\  1&0  \end{array}\right] \cup \left[\setlength{\arraycolsep}{0.06cm}\renewcommand{\arraystretch}{0.8}
\renewcommand{\arraystretch}{0.8}\begin{array}{cc} \ast& \ast\\ 0 & 1\\ \ast& 0\\  1& 0 \end{array}\right]  \cup \left[\setlength{\arraycolsep}{0.06cm}\renewcommand{\arraystretch}{0.8}\begin{array}{cc} 0&1 \\ \ast&0 \\ \ast &0 \\ 1 &0  \end{array}\right] \cup \left[\setlength{\arraycolsep}{0.06cm}\renewcommand{\arraystretch}{0.8}\begin{array}{cc} \ast&\ast \\ 0 &1 \\ 1& 0\\ 0 &0  \end{array}\right]
\cup \left[\setlength{\arraycolsep}{0.06cm}\renewcommand{\arraystretch}{0.8}\begin{array}{cc} 0&1 \\  \ast&0 \\ 1& 0\\ 0 &0  \end{array}\right] \cup \left[\setlength{\arraycolsep}{0.06cm}\renewcommand{\arraystretch}{0.8}\begin{array}{cc} 0& 1\\ 1 &0 \\ 0& 0\\ 0 & 0 \end{array}\right]
\\\hline
(0,1,0,1)&(2,1)& \left[\setlength{\arraycolsep}{0.06cm}\renewcommand{\arraystretch}{0.8}\begin{array}{cc} \ast& \ast\\ \ast & 1\\ 0& 0\\ 1 & 0 \end{array}\right] & \begin{array}{c} \vspace{-0.2cm}\small{\yng(2,1)}\end{array}&
\left[\setlength{\arraycolsep}{0.06cm}\renewcommand{\arraystretch}{0.8}\begin{array}{cc} \ast&\ast \\  0&1 \\ \ast& 0\\ 1 &0  \end{array}\right] \cup \left[\setlength{\arraycolsep}{0.06cm}\renewcommand{\arraystretch}{0.8}\begin{array}{cc} 0&1 \\ \ast & 0\\ \ast& 0\\ 1 &0  \end{array}\right]\cup \left[\setlength{\arraycolsep}{0.06cm}\renewcommand{\arraystretch}{0.8}\begin{array}{cc} \ast&\ast \\ 0 &1 \\ 1&0 \\  0&0  \end{array}\right] \cup \left[\setlength{\arraycolsep}{0.06cm}\renewcommand{\arraystretch}{0.8}\begin{array}{cc} 0&1 \\  \ast& 0\\ 1&0 \\  0&0  \end{array}\right] \cup \left[\setlength{\arraycolsep}{0.06cm}\renewcommand{\arraystretch}{0.8}\begin{array}{cc} 0&1 \\ 1 & 0\\ 0&0 \\ 0 & 0 \end{array}\right]\\ \hline
(1,0,0,1)&(2,0)& \left[\setlength{\arraycolsep}{0.06cm}\renewcommand{\arraystretch}{0.8}\begin{array}{cc} \ast&1 \\ \ast&0 \\ 0&0 \\  1&0  \end{array}\right] & \begin{array}{c} \vspace{0.0cm}\small{\yng(2)}\end{array}&
\left[\setlength{\arraycolsep}{0.06cm}\renewcommand{\arraystretch}{0.8}\begin{array}{cc} 0&1 \\ \ast & 0\\ \ast& 0\\ 1 &0  \end{array}\right]\cup \left[\setlength{\arraycolsep}{0.06cm}\renewcommand{\arraystretch}{0.8}\begin{array}{cc} \ast&\ast \\ 0 &1 \\ 1&0 \\  0&0  \end{array}\right] \cup \left[\setlength{\arraycolsep}{0.06cm}\renewcommand{\arraystretch}{0.8}\begin{array}{cc} 0&1 \\  \ast& 0\\ 1&0 \\  0&0  \end{array}\right] \cup \left[\setlength{\arraycolsep}{0.06cm}\renewcommand{\arraystretch}{0.8}\begin{array}{cc} 0&1 \\ 1 & 0\\ 0&0 \\ 0 & 0 \end{array}\right]\\ \hline
(0,1,1,0)&(1,1)& \left[\setlength{\arraycolsep}{0.06cm}\renewcommand{\arraystretch}{0.8}\begin{array}{cc} \ast&\ast \\ 0 & 1\\ 1&0 \\ 0 & 0 \end{array}\right] & \begin{array}{c} \vspace{-0.17cm}\small{\yng(1,1)}\end{array}&
\left[\setlength{\arraycolsep}{0.06cm}\renewcommand{\arraystretch}{0.8}\begin{array}{cc} \ast&\ast \\ 0 &1 \\ 1&0 \\  0&0  \end{array}\right] \cup \left[\setlength{\arraycolsep}{0.06cm}\renewcommand{\arraystretch}{0.8}\begin{array}{cc} 0&1 \\  \ast& 0\\ 1&0 \\  0&0  \end{array}\right] \cup \left[\setlength{\arraycolsep}{0.06cm}\renewcommand{\arraystretch}{0.8}\begin{array}{cc} 0&1 \\ 1 & 0\\ 0&0 \\ 0 & 0 \end{array}\right] \\ \hline
(1,0,1,0)&(1,0)& \left[\setlength{\arraycolsep}{0.06cm}\renewcommand{\arraystretch}{0.8}\begin{array}{cc} 0&1 \\  \ast&0 \\ 1&0 \\  0&0  \end{array}\right] & \begin{array}{c} \vspace{-0.0cm}\small{\yng(1)}\end{array}& \left[\setlength{\arraycolsep}{0.06cm}\renewcommand{\arraystretch}{0.8}\begin{array}{cc} 0&1 \\  \ast& 0\\ 1&0 \\  0&0  \end{array}\right] \cup \left[\setlength{\arraycolsep}{0.06cm}\renewcommand{\arraystretch}{0.8}\begin{array}{cc} 0&1 \\ 1 & 0\\ 0&0 \\ 0 & 0 \end{array}\right] \\ \hline
(1,1,0,0)&(0,0)& \left[\setlength{\arraycolsep}{0.06cm}\renewcommand{\arraystretch}{0.8}\begin{array}{cc} 0&1 \\ 1 & 0\\ 0&0 \\ 0 & 0 \end{array}\right] & & \left[\setlength{\arraycolsep}{0.06cm}\renewcommand{\arraystretch}{0.8}\begin{array}{cc} 0&1 \\ 1 & 0\\ 0&0 \\ 0 & 0 \end{array}\right] \\\hline
\end{array}$
\captionC{All six Schubert cells of the Grassmannian $Gr_{2,4}$ (represented in all different ways) and their closures.}
\label{tab:Gr24}
\end{table}
\end{ex}

\subsection{The Pl\"ucker embedding}\label{pluecker}
In the following we would like to equip a homogeneous structure to the Grassmannian and to the flag variety. This is done by introducing the so-called Pl\"ucker embedding, i.e.~an embedding into a projective space (see also \cite{Prag}).
As stated in the last sections the flag variety and the Grassmannian variety are not only topological spaces but have in addition the structure of a complex, differentiable, orientated manifold. Moreover it turns out (shown in the recent section) that both have also the structure of a non-singular projective algebraic variety (over the field $\mathbb{C}$). To show this one has to embed both manifolds into a projective space and identify them as the zero locus of some appropriate ideal of polynomials. Then later we will use the fact that the closures of Schubert cells (in both cases) are indeed subvarieties. This has a deeper meaning from the point of view of (co)homology theory. Indeed, we can assign to each Schubert variety a distinct generator in the corresponding (additive) cohomology group of the Grassmannian and flag variety, respectively. Moreover also the ring structure of both cohomology theories could be analyzed by referring to the variety structure.

In the Appendix \ref{rest} several necessary definitions of the underlying mathematical objects, as e.g.~projective variety and Zariski topology are presented.
Let's fix a basis $B = \{e_1, \ldots, e_d\}$ for $V$. The \emph{Pl\"ucker embedding} is defined as the map $\theta$,\\
$
\begin{array}{cccccc}
\qquad \qquad&\theta &:& Gr_{k,d} & \rightarrow & \mathbb{P}[\bigwedge^k V] \\
\qquad \qquad&&& U &\mapsto & [u_1 \wedge \ldots \wedge u_k] \qquad ,
\end{array}$ \\
where $\{u_1 , \ldots , u_k\}$ is some basis for $U$ and $[w]$ is the equivalence class of $w \in \bigwedge ^k V$.
It can easily be proven that this map is well-defined and injective:
Let $B_1= \{u_1,\ldots,u_k\}$ and $B_2= \{v_1,\ldots,v_k\}$ be two basis sets for $U \in Gr_{k,d}$. Hence, there exists a $G \in Gl(k)$, such that $v_i = \sum_{j=1}^k G_{i j} u_j$. Herewith we find, by using explicit properties of the exterior product,
\begin{eqnarray}
v_1\wedge \ldots \wedge v_k &=& \sum_{j_1,\ldots,j_k=1}^d G_{1 j_1}\cdot \ldots \cdot G_{k j_k} u_{j_1}\wedge \ldots \wedge u_{j_k} \nonumber
\end{eqnarray}
\begin{eqnarray}
\qquad \qquad &=& \sum_{\pi \in \mathcal{S}^k} \mbox{sign}(\pi)G_{1 \pi(1)}\cdot \ldots \cdot G_{k \pi(k)} \, u_1\wedge \ldots \wedge u_k \nonumber\\
&=& \mbox{det}(G) \,u_1\wedge \ldots \wedge u_k
\end{eqnarray}
Since $G \in Gl(k)$, $\mbox{det}(G)\neq 0$ and thus $[v_1\wedge \ldots \wedge v_k] = [u_1\wedge \ldots \wedge u_k]$, which proves the well-definiteness of $\theta$. To show injectivity, assume $U,W \in Gr_{k,d}, \,U\neq W$. Since $U \neq W$, $\mbox{dim}(U \cap W) = r <k$ and we can find orthonormal bases $\{u_1,\ldots,u_k\}$ and
$\{u_1,\ldots,u_r, w_{r+1},\ldots,w_k\}$ for $U$ and $W$, i.e.~$U \cap W = \mbox{span}(u_1,\ldots,u_r)$, $u_{r+1},\ldots,u_k \in W^{\perp}$ and $w_{r+1},\ldots,w_k \in U^{\perp}$. Thus, $u_1\wedge \ldots u_r \wedge w_{r+1} \wedge \ldots \wedge w_k $ is not proportional to $u_1\wedge \ldots \wedge u_k $ and therefore $\theta(U) \neq \theta(W)$.

Now, we verify that $Gr_{k,d}$ has the structure of a projective variety. First, we still observe for some matrix $A = A(B_1, B)$
\begin{eqnarray}
\lefteqn{u_1\wedge \ldots \wedge u_d} && \nonumber \\
&=&\sum_{k_1, \ldots ,k_d =1}^n A_{k_1,1} \cdot\ldots \cdot A_{k_d,d}\, e_{k_1}\wedge \ldots \wedge e_{k_d}\\
&=& \sum_{1\leq k_1 \leq \ldots \leq k_d \leq n} \sum_{\pi \in \mathcal{S}^d} \mbox{sign}(\pi) A_{k_{\pi(1)},1} \cdot\ldots \cdot A_{k_{\pi (d)},d} \, e_{k_1}\wedge \ldots \wedge e_{k_d}
\end{eqnarray}
and by introducing the $(j_1,\ldots,j_k)$-minor $p_{j_1,\ldots,j_k}$ of $A$ we find
\begin{equation}
u_1\wedge \ldots \wedge u_k = \sum_{1\leq j_1 \leq \ldots \leq j_k \leq n} p_{j_1,\ldots,j_k} \,e_{j_1} \wedge \ldots \wedge e_{j_k} \, .
\end{equation}
These numbers  $p_{\underline{j}}(A) = p_{j_1,\ldots,j_k}(A)$ are called Pl\"ucker coordinates of the point $U \in Gr_{k,d}$ (w.r.t.~the given basis $B=\{e_1,\ldots,e_d\}$) and they are the homogeneous coordinates of $\theta(U)$. They also depend on the fixed basis $B$. Moreover the matrix $A$ is given by the matrix $(\vec{u}_1,\ldots,\vec{u}_k)$, where $\vec{u}_i$ is the coordinate vector of $u_i$ with respect to the basis $B$. Note that two regular matrices determine the same point in $Gr_{k,d}$ if and only if they are in the same $Gl(k)$-orbit and hence we find $Gr_{k,d}= M_{d \times k}/{Gl(k)}$, where $M_{d \times k}$ denotes the $d \times k$-matrices with rank $k$.
After all, we define for an arbitrary subset $\underline{j}=\{j_1,\ldots,j_k\} \subseteq \{1,\ldots,d\}$, $p_{\underline{j}}= 0$ if $|j|<k$ and $p_{\underline{j}}=\mbox{sign}(\pi) p_{\underline{j}^{\uparrow}}$ otherwise, where $\underline{j}^{\uparrow}$ means to increasingly order the components of $\underline{j}$ and $\pi$ is the corresponding permutation of $\underline{j}$ to $\underline{j}^{\uparrow}$.

By using the Pl\"ucker coordinates we state
\begin{lem}\label{zero}
The Grassmannian $Gr_{k,d}$ is the zero set of the system
\begin{equation}
\sum_{m=1}^{k+1} (-1)^m p_{i_1,\ldots,i_{k-1},j_m} p_{j_1,\ldots,\hat{j}_m,\ldots,j_{k+1}}\,
\end{equation}
where the hat means omitting the corresponding index and the sets $\{i_1,\ldots,i_{k-1}\}$ and $\{j_1,\ldots,j_{k+1}\}$ are subsets of $\{1,\ldots,d\}$.
\end{lem}
\noindent The proof can be found in the Appendix \ref{Schubertcalculus}.

The flag variety is indeed also a projective algebraic variety. The natural embedding into a projective space is induced by the corresponding embedding of the Grassmannian:
\begin{equation}
\mbox{\emph{Fl}}_d \rightarrow Gr_{1,d} \times \ldots \times Gr_{d-1,d} \rightarrow \prod_{j=1}^{d-1} \mathbb{P} \left( \bigwedge ^j [V] \right) \rightarrow \mathbb{P}\left[ \bigotimes_{j=1}^{d-1}\bigwedge ^j [V]  \right]\,,
\end{equation}
where the last map is the so-called Segre embedding (see also \cite{Knut1}).

\section{Introduction to homology and cohomology theory}\label{sec:introHomCohom}
For this section we followed \cite{Sato}, \cite{Blas} and \cite{BlBr} and recommend \cite{Hatch}, \cite{Wall1} and \cite{Wall2} as additional literature to the reader. Topological spaces are playing an important role in different mathematical fields since they are one of the most general
spaces with a minimal, but non-trivial structure (namely the notion of open sets). The fruitful field of algebraic topology has the aim of
analyzing these spaces from an algebraic point of view, i.e.~one attaches topologically invariant algebraic structure to them. Topologically invariant
means that homeomorphic spaces will have isomorphic algebraic structures.
\begin{floatingfigure}[!hb]{5cm}
\hspace{-0.5cm}
\includegraphics[width=4.9cm]{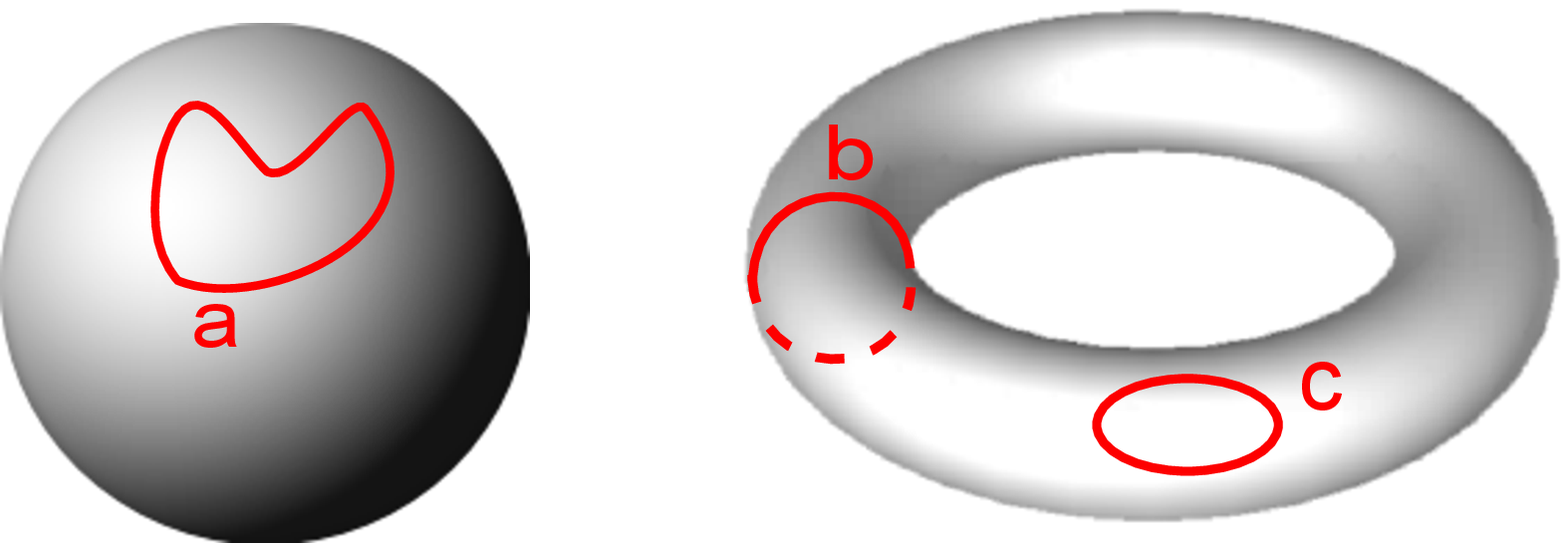}
\captionC{Sphere and Torus (see text).}
\label{spheretorus}
\end{floatingfigure}
A very primitive example is presented in Fig.~\ref{spheretorus}. We consider two 2-dimensional manifolds, a sphere and a torus as topological spaces governed with the relative topology of $\mathbb{R}^3$. Consider a closed and non-crossing curve on each of these manifolds and then cut them along those paths. The sphere has the algebraic property that independent of the path $a$ we end up with two pieces. Although one can also find paths on the torus, which lead to two parts (e.g.~path $c$) this is not the case for all paths (see e.g.~$b$).

Another example is the fundamental group for simple connected  and sufficiently `nice' subsets of $\mathbb{R}^2$ that appears in homotopy theory. This concept also analyzes these topological spaces. The fundamental group refers to the number of holes in the space and is closely related to the idea of winding numbers.
\par

A much deeper structure is the one given by the \emph{homology and cohomology theories of topological space}. There are several different theories, but all of them are fulfilling the standard axioms (e.g.~presented in \cite{Sato} and \cite{Wall1}):
\\
\\
\begin{defn}\label{homology}
A homology theory is a mathematical theory with the following properties:
\begin{enumerate}
\item It assigns to each pair of topological spaces $(E,F)$ and for all $p \in \mathds{N}_0$ an Abelian group $H_p(E,F)$ and to every continuous
map $f: (E,F) \rightarrow (E',F')$ a homomorphism $f_p: H_p(E,F) \rightarrow H_p(E',F')$  .
\item For a composition $f \circ g$ of continuous maps $f, g$ the formula
\begin{equation}
(f \circ g)_p = f_p \circ g_p
\end{equation}
holds for all $p$ and if $f$ is the identity then $f_p$ is also the identity.
\item If $E = pt$ is the singleton space, i.e.~consists of only one point and $F = \emptyset$ is empty then $H_p(E,F)= 0 \, \forall\, p>0$ and $H_0(E,F)= \mathbb{Z}$.
\item If $f,g: (E,F) \rightarrow (E',F')$ are homotopic then $f_p = g_p$.
\item For $E \supset F \supset G$ there exists a homomorphism $\partial_p : H_p(E,F)\rightarrow H_{p-1}(F,G)$, which commutes with the homomorphism associate  with continuous maps and the sequence
    \begin{equation}
    \rightarrow H_p(F,G) \xrightarrow{i_p} H_p(E,G) \xrightarrow{j_p} H_p(E,F) \xrightarrow{\partial_p} H_{p-1}(F,G)\rightarrow
    \end{equation}
    is exact, where $i, j$ are the natural inclusions.
\item If $(E,F)$ is a pair of spaces and $A \subset E$ such that $\overline A  \subset \mathring{F}$ then the homomorphism
\begin{equation}
i_p: H_p(E \setminus A, F \setminus A) \rightarrow H_p(E,F)
\end{equation}
induced by the inclusion map $i$ is an isomorphism for all $p$.
\end{enumerate}
\end{defn}

\begin{rem}
We denote the category of pairs of topological spaces by $\textbf{Top}^2$ and the one of graded modules by $\textbf{GradMod}$. A homology theory is a covariant functor $H_*: \mbox{Top}^2 \rightarrow \mbox{GradMod}$.
\end{rem}
\begin{rem}
There are several different ways of constructing such a homology theory. The most common are singular, CW and simplicial homology. For the class of so-called triangulated topological spaces it can be shown that the definition of a homology leads to a (up to isomorphisms) unique theory (see e.g.~\cite{Sato}).
\end{rem}
\begin{rem}
If the topological space $E$ is a subset of $\mathbb{R}^2$ and $F$ the empty set then $H_1(E,F)$ is isomorphic to the fundamental group of $E$.
\end{rem}
\begin{rem}
Despite the fact that the reader already wonders about the motivation for these a priori `strange' and abstract homology theories it is even more strange that
the homology groups depend on a pair of topological spaces instead of one single topological space. The motivation for this is a powerful generalization. Indeed, we can choose in particular the pair $(E, \emptyset)$ to deal with homologies of single spaces. The advantage of the generalization simplifies the computation
of homology groups. By expressing a topological space $H$ of interest as $H \cong E/F$ its homology groups $H_p(H)$ are strongly related to $H_p(E,F)$ (see e.g.~\cite{Sato}, p32).
\end{rem}

The concept of cohomology is very similar to the one of homology. The important difference is that cohomology theories are contravariant functors, i.e.~the composition $f \circ g$ of two morphisms $f, g$ will be mapped to the homomorphism $f^* \circ g^*$. Indeed, the axioms for cohomology theories read
\\
\begin{defn}\label{cohomology}
A cohomology theory is a mathematical theory with the following properties:
\begin{enumerate}
\item It assigns to each pair of topological spaces $(E,F)$ and for all $p \in \mathds{N}_0$ an Abelian group $H^p(E,F)$ and to every continuous
map $f: (E,F) \rightarrow (E',F')$ a homomorphism $f^p: H^p(E',F') \rightarrow H^p(E,F)$  .
\item For a composition $g \circ f$ of continuous maps $f, g$ the formula
\begin{equation}
(f \circ g)^p = g^p \circ f^p
\end{equation}
holds for all $p \in \mathds{N}_0$ and if $f$ is the identity then $f^p$ is also the identity.
\item If $E = pt$ is the singleton space, i.e.~consists of only one point and $F = \emptyset$ is empty then $H^p(E,F)= 0 \, \forall\, p>0$ and $H^0(E,F)= \mathbb{Z}$.
\item If $f,g: (E,F) \rightarrow (E',F')$ are homotopic then $f^p = g^p$.
\item For $E \supset F \supset G$ there exists a homomorphism $\partial^p : H^{p-1}(F,G)\rightarrow H^{p}(E,F)$, which commutes with the homomorphism associate with continuous maps and the sequence
    \begin{equation}
    \leftarrow H^p(F,G) \xleftarrow{i^p} H^p(E,G) \xleftarrow{j^p} H^p(E,F) \xleftarrow{\partial^p} H^{p-1}(F,G)\leftarrow
    \end{equation}
    is exact, where $i, j$ are the natural inclusions.
\item If $(E,F)$ is a pair of spaces and $A \subset E$ such that $\overline A  \subset \mathring{F}$ then the homomorphism
\begin{equation}
i^p:  H^p(E,F) \rightarrow H^p(E \setminus A, F \setminus A)
\end{equation}
induced by the inclusion map $i$ is an isomorphism for all $p$.
\end{enumerate}
\end{defn}
\begin{rem}
A cohomology theory is a contravariant functor $H^*: \mbox{Top}^2 \rightarrow \mbox{GradMod}$.
\end{rem}
\begin{rem}
Cohomology theories carry an additional (hidden) structure: There is a natural product of elements in the cohomology groups´, induced by the continuous map $E \rightarrow E \times E \,,\, x\mapsto (x,x)$. This induced map, called cup product gives the cohomology the structure of a ring (see e.g.~\cite{Wall1}).
\end{rem}

\subsection{Integral cohomology of Grassmannians and flag varieties}\label{sec:CohomGrass}
In this section we determine the cohomology ring of the Grassmannian, whose structure as manifold and algebraic variety
will be essential for this purpose, in particular their structure of a CW complex. First, we roughly present the main steps for the calculation of the cohomology ring:
\begin{enumerate}
\item We verify that the Grassmannian has the structure of a CW complex, and that the cell structure is given by the Schubert cells and Schubert varieties, respectively.
\item For CW complexes we determine the so-called CW-homology the most natural realization of the homology axioms
\item This CW-homology is isomorphic to the singular homology
\item The structure of the Grassmannian leads to a zero boundary homomorphism and thus the homology groups
      are very easy to determine
\item The Poincar\'e duality leads immediately to the cohomology groups
\item E.g.~by using the concept of Chern classes and fibre bundles is used to also calculate the multiplicative (i.e.~ring) structure of the cohomology.
\end{enumerate}
We follow step by step this outline.
\begin{enumerate}
\item We verify that the Grassmannian $Gr_{k,d}$ is a CW-complex. Lem.~\ref{partition} states that $Gr_{k,d}$ is a disjoint union of all Schubert cells $S_{\alpha}^{\circ}$ with $\alpha \subset k \times (d-k)$. Lem.~\ref{homeo} states that each Schubert cell is homeomorphic to $\mathbb{C}^{\|\alpha\|_1} \cong \mathbb{R}^{2 \|\alpha\|_1}$. Moreover, Lem.~\ref{Schubertvarieties} ensures that the boundary of every Schubert cell is contained in lower dimensional Schubert cells. Hence, (see Appendix \ref{CW-complex} and Definition \ref{defcwcomplex2}) the Grassmannian has the structure of a CW complex.
\item We follow \cite{Sato}. Let $X$ be a CW-complex of dimension $n$ and $X^r$ its $r$-skeleton, $X^n = X$. We introduce the relevant ideas to define the so-called CW-homology and state some elementary results without verifying them (for proofs see e.g.~\cite{Sato}).
    For a fixed homology theory and for all $k\leq n$ we find a long exact sequence for the triple $X^{k-2} \subset X^{k-1} \subset X^k $ according to the axioms of homology theories with boundary homomorphism $\partial^{(k)} = \oplus_j \partial_j^{(k)}$. By defining the Abelian `chain groups'
    \begin{equation}\label{chaingroup}
    C_k := H_k(X^k,X^{k-1})
    \end{equation}
    we obtain a sequence
    \begin{equation}\label{chaincomplex}
    \rightarrow C_k \xrightarrow{\partial^{(k)}} C_{k-1} \xrightarrow{\partial^{(k-1)}} \ldots \xrightarrow{\partial^{(2)}} C_1 \xrightarrow{\partial^{(1)}} C_0 \,.
    \end{equation}
    Since $\partial^{(k-1)} \circ \partial^{(k)} = 0$ (which is non-trivial), we find
    \begin{equation}
    \mbox{Im}(\partial^{k-1}) \leq \mbox{Ker}(\partial^{k})
    \end{equation}
    and thus the sequence in (\ref{chaincomplex}) is a chain complex, namely the chain complex $C = C(X)$ induced by the CW-complex $X$.
    We define
    \begin{defn}\label{Chain-homology}
    Let $C$ be a chain complex with chain groups ${C_k}$ and boundary homomorphisms $\partial^{(k)}: C_k \rightarrow  C_{k-1}$.
    The chain complex homology is defined by the homology groups
    \begin{equation}
    H_k^{cc}:=  \mbox{Ker}(\partial^{k-1}) / \mbox{Im}(\partial^{k-1}) \,.
    \end{equation}
    \end{defn}
    and
     \begin{defn}\label{CW-homology}
    The CW-homology $H^{CW}(X)$ of a CW complex $X$ is defined as the chain complex homology (see Definition \ref{Chain-homology}) of the chain complex
    $C = C(X)$ as constructed above (see in particular (\ref{chaingroup})).
    \end{defn}
    We emphasize the strength of the concept of cell complexes by
    \begin{lem}\label{CW-groups}
    Let $X$ be a CW-complex. Then
    \begin{enumerate}
     \item $H_k(X^n, X^{n-1})= 0$ whenever $k \neq n$, and is free Abelian for $k=n$, with a set of generators, which is in a $1-1$ correspondence with the $n$-cells of $X$.
     \item $H_k(X^n)=0$, whenever $k>n$.
     \item The inclusion $i_k: X^n \rightarrow X$ induces an isomorphism $H_k(X^n) \rightarrow H_k (X)$
     \end{enumerate}
    \end{lem}
    The proof of Lem.~\ref{CW-groups} is elementary and it is a nice and instructive exercise (or can alternatively be found in \cite{Sato}, p.40f).
    The last remaining step for the calculation of the CW-homology groups is to determine the boundary homomorphisms $\partial^{(k)}$. Fortunately, due to
    the instance described in point 4 of our outline this is not relevant for our objective.
\item We state a very important result that finally explains why one resorts to CW-homology to calculate e.g.~singular homology groups:
    \begin{thm}\label{CWcoincide}
    Let $X$ be a CW complex. Then
    \begin{equation}
    H_k(X) \cong H_k^{CW}(C(X)) \qquad \forall k \in \mathds{N}_0 \,,
    \end{equation}
    where $H$ is here some homology theory and $C(X)$ the corresponding chain complex of this homology theory constructed as explained above (see in particular \ref{chaingroup}).
    \end{thm}
\item Now, we consider again the complex Grassmannian $Gr_{k,d}$ and apply the last two steps (step 2 and 3) to $X = Gr_{k,d}$. First, we determine the chain groups $C_k(Gr_{k,d})$ of the chain complex induced by the cell-structure of $Gr_{k,d}$. As noted in point 1 the presented cell decomposition of $Gr_{k,d}$
    by Schubert cells contains only cells with even (real) dimension. This means that all chain groups $C_k$ with $k$ odd vanish. Hence, the boundary homomorphisms $\partial^{(k)}$ are trivial, namely $\partial^{(k)} = 0$ and with Definition \ref{Chain-homology} and \ref{CW-homology} and Theorem \ref{CWcoincide} we find
    \begin{lem}\label{homologyGr}
    The complex Grassmannian $Gr_{k,d}$ has the following integral homology groups:
    \begin{equation}
    H_m(Gr_{k,d}; \mathbb{Z}) = \begin{cases}
     \mathbb{Z}^{\nu_m},  & \text{if } m \text{ is even,}\\
     0, & \text{if } m \text{ is odd}
\end{cases}\,,
    \end{equation}
    where $\nu_m$ is the number of Schubert cells $S_{\alpha}^{\circ}$ of (real) dimension $2 \nu_m$ in the standard cell decomposition of $Gr_{k,d}$.
    \end{lem}
    \begin{rem}\label{countcells}
    For all $k$ the integer $\nu_k$ is given by the number of Young diagrams $\alpha$ contained in the rectangle $k \times (d-k)$ and with fixed weight $\|\alpha\|_1 = k$.
    \end{rem}
    To recap these first four steps, that lead to the homology groups of the complex Grassmannian we determine explicitly the homology groups for the example $Gr_{2,4}$.
    \begin{ex} We consider the complex Grassmannian $Gr_{2,4}$, whose Schubert cells were already presented at the end of Sec.~\ref{grassmannian}.
    To determine its homology groups we only need to count the number of $k$-cells in the corresponding cell decomposition of $Gr_{2,4}$ or even easier due to Remark \ref{countcells} the number of Young diagrams with weight $2$, contained in the $2\times 2$ rectangle. This is trivial. The young diagrams that we find corresponding to $k=0,1,2,3,4$ are shown in Tab.~\ref{tab:Gr24weights}.
    \begin{table}[h]
    \centering
    $
    \setlength{\arraycolsep}{0.2cm}
\renewcommand{\arraystretch}{2.0}
    \begin{array}{|c|c|c|c|c|c|}
    \hline
    k & 4&3&2&1&0 \\ \hline
    \mbox{Young diagrams}&
    \begin{array}{c} \vspace{-0.15cm}\small{\yng(2,2)}\end{array}&
    \begin{array}{c} \vspace{-0.15cm}\small{\yng(2,1)}\end{array}&
    \begin{array}{c} \vspace{-0.0cm}\small{\yng(2)}\end{array} \,\,\begin{array}{c} \vspace{-0.15cm}\small{\yng(1,1)}\end{array}&
    \begin{array}{c} \vspace{-0.0cm}\small{\yng(1)}\end{array}&\\\hline
    \nu_{2k}&1&1&2&1&1 \\\hline
    \end{array}$
    \captionC{Cell decomposition of the Grassmannian $Gr_{2,4}$ with corresponding weights.}
    \label{tab:Gr24weights}
    \end{table}

    \noindent Thus, we find
    \begin{equation}
    H_*(Gr_{2,4};\mathbb{Z}) = \mathbb{Z} \oplus 0 \oplus  \mathbb{Z} \oplus 0 \oplus  \mathbb{Z}^2 \oplus 0 \oplus  \mathbb{Z} \oplus 0 \oplus  \mathbb{Z}
    \end{equation}
    \end{ex}
\item The so-called Poincar\'e-duality, which is presented later more detailed states
     that for every compact, closed and orientated $n$-manifold $X$ we find for all $k=0,1,\ldots,n$:
     \begin{equation}
     H_{n-k} \cong H^k \,.
     \end{equation}
    Since the Grassmannian $Gr_{k,d}$ is indeed a compact closed and orientated manifold and since
    its cell decomposition has obviously the property $\nu_j = \nu_{2 k (d-k)-j}$, we find
    \begin{equation}
    H^j(Gr_{k,d}) \cong H_j(Gr_{k,d}) \qquad \forall j \in \mathds{N}_0 \,.
    \end{equation}
\item An important concept that we need to determine the cohomology ring is the one of fundamental classes, e.g.~presented in \cite{Hatch} p.235f. The relevant theorem states
    \begin{thm}\label{fundamentalclass}
     Let $X$ be a compact, closed and orientated n-manifold. Then, for any $x_0 \in X$
     \begin{equation}
     H_n(X;\mathbb{Z}) \cong H_n(X,X\setminus {x_0};\mathbb{Z}) \cong \mathbb{Z} \,.
     \end{equation}
     \end{thm}
     This theorem holds for every coefficient ring, but we already restricted to the integer ring. The statement $H_n(X,X\setminus {x_0};\mathbb{Z}) \cong \mathbb{Z}$ is an elementary result and is not part of the theorem as such.
     Due to Theorem \ref{fundamentalclass} there is a (modulo $\mathbb{Z}_2$) unique element in $H_n(X;\mathbb{Z})$ that (additively) generate $H_n(X;\mathbb{Z})$. Due to Poincar\'e's duality there is also a unique element in $H^0(X;\mathbb{Z}) \cong \mathbb{Z}$, denoted by $[X]$. It is called the \emph{fundamental class} $[X]$ of $X$.
     This leads to one of the most important results of algebraic topology (one can also say that (co)homology was constructed such that this holds)
     \begin{rem}\label{submanifold}
     Let $Y \subset X$ be a compact, closed and orientated submanifold of the manifold X. Then as stated above there exists a specific element $[Y] \in H(Y)$ and due to the covariant structure of homology theories also a unique element in H(X) that we also denote by $[Y]$. Due to Poincare's duality the same also holds for the contravariant cohomology theories.
     \end{rem}
     This means e.g.~that the closed pathes $b$ and $c$ in Fig.~\ref{spheretorus}, as compact closed manifolds give rise to elements $[a], [b] \in H_*(T)$ in the homology of the torus $T$ (or alternatively also to elements in the cohomology $H^*(T)$).

     Using these insights one can conclude that the subvarieties of any flag variety are in a $1-1$-correspondence to the generators of the corresponding
     cohomology ring. In particular, one can show \cite{BlBr} that the cohomology ring of the flag variety can be represented by Schubert polynomials.
\end{enumerate}

\section{From Schubert cells to Schubert varieties}\label{sec:Schuberts}
Now, we would like to  apply the strong machinery provided by Schubert calculus combined with a powerful intersection theory (see Sec.~\ref{GeometricApproach}) to study intersection properties of type (\ref{intersectionprop1}) and (\ref{intersectionprop2}). Note that there is still a subtle problem. The intersection properties (\ref{intersectionprop1}) and (\ref{intersectionprop2}) refer to Schubert cells of the Grassmannian and the flag variety, respectively. In contrast to their closures, the Schubert varieties, they do not have the structure of closed manifolds which is required to use the concepts developed in the previous sections. Fortunately, as we will show in this section we can reformulate the intersection properties into the required form.
Let us first consider the intersection property (\ref{intersectionprop1}) for the QMP $\{A,AB\}$.
In the Appendix \ref{Schubertcalculus} it is shown that the map
\begin{eqnarray}
\mbox{Tr}[\rho P_{\bullet}]:\qquad \mbox{Gr}_{k,d}  &\rightarrow& \mathbb{R} \nonumber \\
 V &\mapsto& \mbox{Tr}[\rho P_V] \label{tracecontinuity1}
\end{eqnarray}
is continuous and it is also shown that this together with the definition of the topology for the Grassmannian yields
\begin{lem}\label{intersection1varieties}
Consider the QMP $\mathcal{M}_{A, AB}$. Whenever the intersection property
\begin{equation}\label{intersectionprop3}
\Phi(S_{\pi}(\rho_A)) \cap S_{\hat \sigma}(\rho_{AB}) :=S_{\pi}(\rho_A) \otimes \mathcal{H}^{(B)} \cap S_{\hat \sigma}(\rho_{AB}) \neq \emptyset
\end{equation}
is fulfilled the following spectral inequality holds:
\begin{equation}\label{spectralineq3}
\boxed{
\sum_{j=1}^{d_A}\pi_j \lambda_j^{(A)} \leq \sum_{i=1}^{d_A d_B} \sigma_i \lambda_i^{(AB)}}    \,.
\end{equation}
\end{lem}
Now, let us consider the intersection property (\ref{intersectionprop2}), which is relevant for the QMP $\{A,B,AB\}$.
In the Appendix \ref{Schubertcalculus} it is shown that for every fixed test spectrum $a$ the map
\begin{eqnarray}
\mbox{Tr}[\rho A(\bullet)]:\qquad \mbox{\emph{Fl}}_d  &\rightarrow& \mathbb{R} \nonumber \\
 F_{\bullet} &\mapsto& \mbox{Tr}[\rho A(F_{\bullet})] \label{tracecontinuity2}
\end{eqnarray}
is continuous where $A(\bullet)$ is an operator uniquely determined by its spectrum $a$ and the corresponding eigenspaces described by the flag $F_{\bullet}$.
Moreover it is also shown that this continuity condition together with the definition of the topology for the flag variety yields
\begin{lem}\label{intersection2varieties}
Consider the QMP $\mathcal{M}_{A,B,AB}$. Whenever the intersection property
\begin{equation}\label{intersectionprop4}
\Phi_{a,b}(X_{\alpha}(\rho_A)),X_{\beta}(\rho_{B})) \cap X_{\gamma \omega_0}(-\rho_{AB}) \neq \emptyset
\end{equation}
is fulfilled the following spectral inequality holds:
\begin{equation}\label{spectralineq4}
\boxed{\sum_{i=1}^{d_A} \lambda_{\alpha_i}^{(A)}a_i + \sum_{j=1}^{d_B} \lambda_{\beta_j}^{(B)}b_j \leq \sum_{k=1}^{d_A d_B} \lambda^{(AB)}_{\gamma_k}\, (a+b)^{\downarrow}_k }\,.
\end{equation}
\end{lem}
Finally, it is worth noting
\begin{rem}
If we replace the Schubert cells in the intersection conditions (\ref{intersectionprop1}) and (\ref{intersectionprop2}) by its closures, i.e.~by larger sets, we definitively do not reduce the number of possible intersection incidences and hence by this transition from Schubert cells to Schubert variety derived in this section we do not reduce the family of spectral marginal constraints.
\end{rem}

\section{The geometric approach}\label{GeometricApproach}
\begin{floatingfigure}[]{5cm}
\hspace{-0.5cm}
\includegraphics[width=4.8cm]{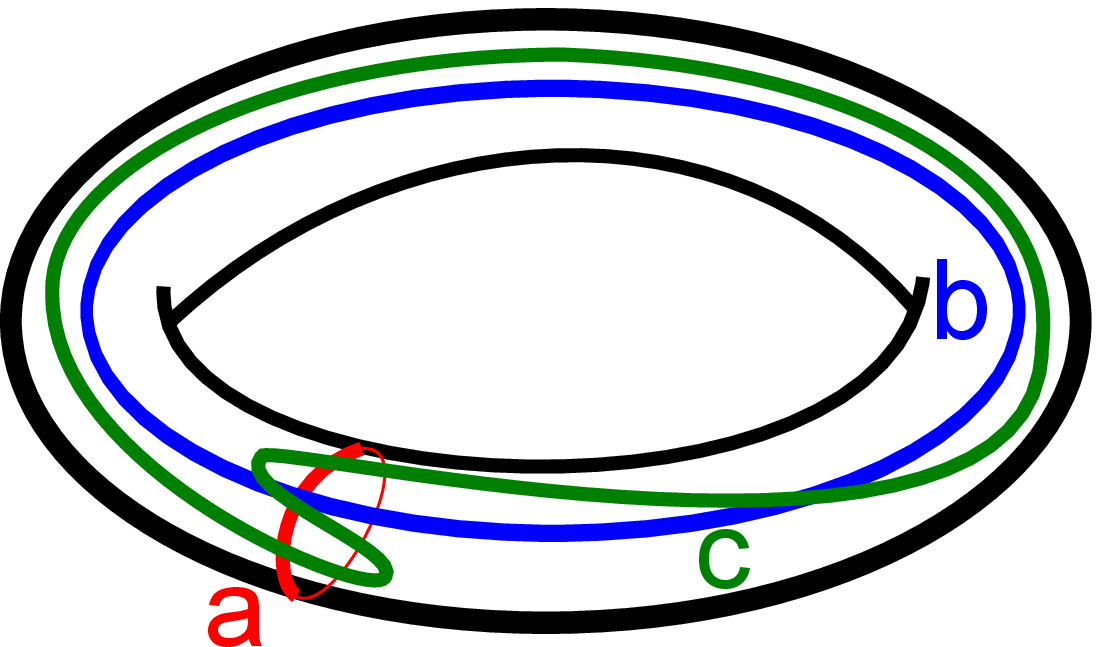}
\captionC{Torus and homologous cycles(see text).}
\label{torus}
\end{floatingfigure}
Since the algebraic methods in the intersection theory are quite difficult, we introduce the geometric motivation of this powerful algebraic machinery. Let us therefore consider Fig.~\ref{torus}: There is drawn a torus $T$. The closed pathes $a, b$ and $c$ are compact, closed and orientated submanifolds of $T$. Remark \ref{submanifold} states the existence of elements $[a], [b], [c] \in H_*(T)$. In that sense one calls these closed pathes cycles. In the same way, the Schubert varieties $S_{\alpha}(F_{\bullet})$ and $X_{\beta}(F_{\bullet})$ are irreducible subvarieties (this means at the end compact, closed and orientated submanifolds) of the Grassmannian variety $Gr_{k,d}$ and the flag variety $\mbox{\emph{Fl}}_d$, respectively and thus give rise to cycles $ [S_{\alpha}(F_{\bullet})] \in H_*(Gr_{k,d})$ and $[X_{\beta}(F_{\bullet})] \in H_*(\mbox{\emph{Fl}}_d)$, respectively. In Fig.~\ref{torus} the pathes $b$ and $c$ are homotopic equivalent. We can continuously deform path $b$ to path $c$. This is not possible for path $a$ and $b$. Due to the properties of (co)homology theory the fundamental classes of homotopic equivalent submanifolds are homologous, i.e.~$[b] = [c]$.

Now, we ask whether two closed pathes on $T$ intersect. E.g.~$a$ and $b$ intersect in exactly one point. If we deform $b$ continuously, e.g.~to path $c$ this new
path has more intersection points with $a$, namely $3$. But independent how we deform path $a$ and $b$ we never reach an even number of intersection points. We either obtain an odd number or an infinite number. The latter scenario is not that pleasant and we  would like to rule it out.
\begin{defn}
Let $X$ be a r-dimensional manifold and $M, N \subset X$ two submanifolds of dimension $m$ and $n$. If $M$ and $N$ intersect in a point $x \in X$ we say that they intersect \textbf{transversally} if their tangent spaces $T_x M, T_x N$ at that point span the whole tangent space $T_x X$ and $T_x M \cap T_x N = \{0\}$.
\end{defn}
\begin{rem}
Two submanifold $M, N \subset X$ of the manifold $X$ can only intersect transversally if their dimensions $\mbox{dim}(M)=m$ and $\mbox{dim}(N)=n$ add up to the dimension $r$ of the manifold $X$.
\end{rem}
By using this concept we observe
\begin{rem}
Let $X$ be a r-dimensional manifold and $M, M', N, N' \subset X$ submanifolds, where $M, M'$ have dimension $m$ and $N, N'$ have $n$ , with $m+n = r$ such that $M$ is homotopic equivalent to $M'$ and $N$ to $N'$. Then if $M$ intersects $N$ and $M'$ intersects $N'$ only transversally and finitely many times the number of intersection points modulo $2$ is the same for $M, N$ and $M', N'$.
\end{rem}
If there are given two fixed submanifolds e.g.~path $a$ and $c$ and we would like know whether they intersect we can deform both pathes continuously to new closed pathes $a'$ and $b'$. If we can end up in a situation where $a'$ and $b'$ intersect only in one point (and also transversally) then we know that $a \cap b \neq \emptyset$! This provides the tool for studying the intersection properties (\ref{intersectionprop3}) and (\ref{intersectionprop4}).

\section{Solving the quantum marginal problem \{A,AB\}}\label{sec:QMPAvsAB}
In this section we combine all the ideas and results from Schubert calculus to analyze the intersection property (\ref{intersectionprop3}) systematically.
Let's choose $\alpha \subset k \times (d_A-k)$ with weight $k$. The relevant Grassmannians are $Gr_{k,d_A}$ and $Gr_{k d_B, d_A d_B}$.
The map (see (\ref{morphism1}))
\begin{equation}
\Phi_k : Gr_{k,d_A} \rightarrow Gr_{k d_B, d_A d_B}
\end{equation}
induces a homomorphism
\begin{equation}
\Phi_k^* : H^*(Gr_{k d_B, d_A d_B}) \rightarrow H^*(Gr_{k,d_A}) \,.
\end{equation}
It is the goal to determine this homomorphism $\Phi^*$ and later by using intersection theory we will investigate the intersection property (\ref{intersectionprop3}) systematically.

We consider binary sequences $\pi$ and $\hat{\sigma}$ that lead to subvarieties $\Phi_{\|\pi\|_1}(S_{\pi}(\rho_A)) = S_{\pi}(\rho_A) \otimes \mathcal{H}^{(B)} $ and $ S_{\hat \sigma}(\rho_{AB})$ with complementary dimensions i.e.~they add up to the dimension of the Grassmannian $Gr_{k,d}$. Concretely speaking we restrict to the case
\begin{equation}
\mbox{dim}(S_{\pi}(\rho_A) \otimes \mathcal{H}^{(B)} ) +\mbox{dim}(S_{\hat \sigma}(\rho_{AB}))=\|\alpha(\pi)\|_1 + \|\alpha(\hat{\sigma})\|_1 = k (d-k) = \mbox{dim}(Gr_{k,d}) \,.
\end{equation}
Note, that we may miss some pairs $(\pi, \sigma)$ with nonempty intersection $S_{\pi}(\rho_A) \otimes \mathcal{H}^{(B)} \cap S_{\hat \sigma}(\rho_{AB})$, but at the moment we ignore this aspect.

We present the example $d_a = d_b =2$. A priori we have to consider three cases, namely $\|\pi\|_1 =0,1,2$. In the first case $\|\pi\|_1 =0$, i.e.~$\pi = (0,0)$ the corresponding Schubert cell is given by $S_{(0,0)}(\rho_A) =\{0\}$ and we find an intersection with $S_{\hat{\sigma}}(\rho_{AB})$ if $\sigma = (0,0,0,0)$, i.e.
$S_{\hat{\sigma}}(\rho_{AB}) =\{0\}$. This non-empty intersection yields a trivial inequality, namely $0\leq 0$. The case $\|\pi\|_1=2$ is similar. In that case we obtain $S_{\hat{\sigma}}(\rho_{AB}) \otimes \mathcal{H}^{(B)} = \{\mathcal{H}^{(AB)}\}$ and by choosing $\sigma = (1,1,1,1)$, i.e.~$S_{\hat{\sigma}}(\rho_{AB}) =\{\mathcal{H}^{(AB)}\}$ we find a non-empty intersection, which only yields a trivial inequality,
\begin{equation}
\lambda_1^{(A)} + \lambda_2^{(A)} \leq \lambda_1^{(AB)} + \lambda_2^{(AB)}+\lambda_3^{(AB)} + \lambda_4^{(AB)} \,,
\end{equation}
which is always true since both sides are identical to $1$ (trace normalization of the state $\rho_A$ and $\rho_{AB}$).
Only the case $\|\pi\|_1 =1$ is nontrivial. There are two subcases, $\pi = (1,0)$ and $\pi = (0,1)$. First, we consider the case $\pi = (1,0)$. To apply the geometric deformation concepts we still need some notation. The state $\rho_A$ defines a flag (of eigenspaces) and we denote the corresponding orthonormal basis by $\{a_1,a_2\}$. This means that $a_1$ is the eigenvector corresponding to the largest eigenvalue of $\rho_A$. Moreover, let $\{b_1,b_2\}$ be an orthonormal basis for $\mathcal{H}^{(B)}$ and $\{c_1,c_2,c_3,c_4\}$ the orthonormal basis induced by $\rho_{AB}$. By the use of the sCEF representation we find
\begin{equation}
S_{(1,0)}(\rho_A) = \left[\setlength{\arraycolsep}{0.06cm}\renewcommand{\arraystretch}{0.8}\begin{array}{c} 1 \\ 0  \end{array}\right]   \qquad \mbox{w.r.t} \, \{a_1,a_2\}\,.
\end{equation}
The specification `w.r.t $\{\ldots \}$' denotes the external reference bases of the sCEF. Furthermore,
\begin{equation}\label{example22}
S_{(1,0)}(\rho_A)\otimes \mathcal{H}^{(B)} = \left[\setlength{\arraycolsep}{0.09cm}\renewcommand{\arraystretch}{0.9}\begin{array}{cc} 0&1 \\  1& 0\\ 0&0 \\  0&0  \end{array}\right]   \,\,\mbox{w.r.t}\, \{a_1\otimes b_1,a_1\otimes b_2,a_2\otimes b_1,a_2\otimes b_2\}\,.
\end{equation}
To meet this variety transversally, the variety $S_{\hat{\sigma}}(\rho_{AB})$ must have maximal dimension, i.e.~dimension $4$. This means $\hat{\sigma}=(0,0,1,1)$ and thus
\begin{equation}
S_{(0,0,1,1)}(\rho_{AB}) = \left[\setlength{\arraycolsep}{0.09cm}\renewcommand{\arraystretch}{0.9}\begin{array}{cc} \ast& \ast \\  \ast& \ast\\ 0&1 \\  1&0  \end{array}\right] + \mbox{its boundary}   \,\,\mbox{w.r.t}\, \{c_1,c_2,c_3,c_4\}\,.
\end{equation}
The boundary is here not that relevant. For its conceptual description in the REF/CEF-representation we refer to Sec.~\ref{flagvarieties}. Since the group $Gl(4)$ is path connected and all elements have a non-vanishing determinant, we can deform the variety of the last equation by multiplying with $g \in Gl(4)$ from the right side. E.g.~we can continuously deform
\begin{equation}
\left[\setlength{\arraycolsep}{0.09cm}\renewcommand{\arraystretch}{0.9}\begin{array}{cc} \ast&\ast \\  \ast& \ast\\ 0&1 \\  1&0  \end{array}\right] + b. \rightarrow \left[\setlength{\arraycolsep}{0.09cm}\renewcommand{\arraystretch}{0.9}\begin{array}{cc} 0&1 \\ 1& 0\\ \ast&\ast \\  \ast&  \ast  \end{array}\right]+b. \, \,\,\mbox{w.r.t} \, \{c_1,c_2,c_3,c_4\}\,.
\end{equation}
The right side is not a sCEF but of course it is still well defined and describes a subset of $Gr_{2,4}$. In a second step one can continuously transform the basis $\{c_1,c_2,c_3,c_4\}$ to the basis $ \{a_1\otimes b_1,a_1\otimes b_2,a_2\otimes b_1,a_2\otimes b_2\}$, this means effectively
\begin{eqnarray}
\lefteqn{\left[\setlength{\arraycolsep}{0.09cm}\renewcommand{\arraystretch}{0.9}\begin{array}{cc} 0& 1 \\ 1&0 \\ \ast&\ast \\  \ast&\ast  \end{array}\right] + b. \, \,\,\mbox{w.r.t} \, \{c_1,c_2,c_3,c_4\}\,} &&\\ \nonumber
&\rightarrow& \left[\setlength{\arraycolsep}{0.09cm}\renewcommand{\arraystretch}{0.9}\begin{array}{cc} 0&1 \\  1& 0\\ \ast&\ast \\  \ast&\ast  \end{array}\right] +b. \, \,\,\mbox{w.r.t} \, \{a_1\otimes b_1,a_1\otimes b_2,a_2\otimes b_1,a_2\otimes b_2\}\,.
\end{eqnarray}
The variety on the right side of the last equation intersects with the variety shown in (\ref{example22}) only if all four complex stars take the value $0$. Then, the number of intersection points is $1$. This means that for $\pi=(1,0)$ and $\hat{\sigma}=(0,0,1,1)$ the corresponding intersection in (\ref{intersectionprop3}) is non-empty and we obtain the non-trivial spectral inequality (recall $\hat{\sigma} = (\hat{\sigma}_4,\hat{\sigma}_3,\hat{\sigma}_2,\hat{\sigma}_1)$)
\begin{equation}
\lambda_1^{(A)}  \leq \lambda_1^{(AB)} + \lambda_2^{(AB)}\,.
\end{equation}
It turns out that the remaining case $\pi = (0,1)$ can only lead to an inequality already implicitly contained in the previous one.
The spectral quantum marginal problem $\{A,AB\}$ with dimensions $d_A=d_B=2$ has the solution $\lambda_1^{(A)}  \leq \lambda_1^{(AB)} + \lambda_2^{(AB)}$.
Note, that due to the elegance of the deformation and intersection concepts it was at the end not relevant to know the states $\rho_A$ and $\rho_{AB}$.

\section{Solving the quantum marginal problem \{A,B,AB\}}\label{sec:QMPAvsBvsAB}
In this section we combine all the ideas and results from Schubert calculus to analyze the intersection property (\ref{intersectionprop4}) systematically.
One arising problem is the uncountability of the family $\Omega$ of test spectra $(a,b)$. Now, we first show how to get rid of this. The idea is to show that several of the spectral inequalities (\ref{spectralineq4}) are linearly dependent. Let us analyze the intersection property (\ref{intersectionprop4}) for given test spectrum $(a,b)$. We introduce an equivalence relation on the family $\Omega$ of pairs of test spectra by
\begin{equation}\label{testspecequiv}
(a,b) \sim  (a',b') \qquad :\Leftrightarrow \qquad \Phi_{a,b} = \Phi_{a',b'} \,,
\end{equation}
where $\Phi_{a,b}$ is the morphism of flag variety introduced in (\ref{morphism2}).
According to (\ref{testspecequiv}) two pairs of spectra are equivalent if and only if their induced index maps $i_{\bullet}$ and $j_{\bullet}$ are identical (recall their definition in (\ref{indexmaps})) or in other words equivalence means $1\leq k,m \leq d_A,\, 1\leq l,n\leq d_B\,:$
\begin{equation}
 a_k + b_m > a_l + b_n  \qquad \Leftrightarrow \qquad a'_k + b'_m > a'_l + b'_n \,.
\end{equation}
In particular we find for all $\lambda \in \mathbb{R}^{+}$ and $c, d \in \mathbb{R}$ that
\begin{equation}\label{spectraweyl}
(a,b)\,\, \sim \,\,(\lambda \,a + c\, 1_{d_A},\lambda \,b + d\, 1_{d_B})
\end{equation}
with $1_r = (1,\ldots,1)$ containing $r$ $1$'s.
Hence, $\Omega$ separates into disjoint subsets $Q_{a,b}$, where $(a,b)$ is a representant for the class $Q_{a,b}$.
For given permutations $\alpha, \beta$ and $\gamma$ we find intersection incidences for $(a,b) \in Q$ of intersection condition (\ref{intersectionprop4}) if and only if this is also the case for all the equivalent test spectra $(a',b')$. Hence, for given equivalence class (\ref{intersectionprop4}) holds for each element in this class or for none of them. This reduces the effort of checking (\ref{intersectionprop4}). Moreover independent of the question of intersection
we always have (due to trace normalization)
\begin{equation}\label{specineqtrace}
1_{d_A}\cdot \lambda^{(A)} + 1_{d_B}\cdot \lambda^{(B)} = 2 \, 1_{d_A d_B}\cdot \lambda^{(AB)} = (1_{d_A}+ 1_{d_B})^{\downarrow}\cdot\lambda^{(AB)}\,,
\end{equation}
where in this line the sum $(a+b)^{\downarrow}$ of two vectors is to be understood as $({a_i+b_j})^{\downarrow}$.
Assume that (\ref{intersectionprop4}) holds for given $(a,b)$ then
\begin{equation}
\sum_{i=1}^{d_A} \lambda_{\alpha_i}^{(A)}a_i + \sum_{j=1}^{d_B} \lambda_{\beta_j}^{(B)}b_j \leq \sum_{k=1}^{d_A d_B} \lambda^{(AB)}_{\gamma_k}\, (a+b)^{\downarrow}_k
\end{equation}
which can be reformulated by combining it with (\ref{specineqtrace}) as
\begin{equation}
\sum_{i=1}^{d_A} \lambda_{\alpha_i}^{(A)} (a_i-c) + \sum_{j=1}^{d_B} \lambda_{\beta_j}^{(B)} (b_j-d) \leq \sum_{k=1}^{d_A d_B} \lambda^{(AB)}_{\gamma_k}\, ((a- c \,1_{d_A}) + (b-d\,1_{d_B}))^{\downarrow}_k \,,
\end{equation}
for all $c,d \in \mathbb{R}$.
This means that the spectral solution set of the univariant QMP given by Lem.~\ref{intersection2varieties} is also obtained by checking (\ref{intersectionprop4}) only for test spectra $(a,b)$ fulfilling
\begin{equation}\label{testspecwyle}
\sum_i^{d_A} a_i = \sum_j^{d_B} b_j = 0 \,.
\end{equation}
We denote the subset of $\Omega$ of test spectra fulfilling (\ref{testspecwyle}) by $\Omega^{(0)}$ which also separates according to the equivalence relation (\ref{testspecequiv}) into disjoint subsets $Q^{(0)}_{a,b}$ called \textbf{cubicles}.
In the following we will study the geometry of the cubicles as subsets of Euclidean spaces, $Q_{a,b}^{(0)} \subset \mathbb{R}^{d_A+ d_B}$.
Each cubicle is concretely given by
\begin{eqnarray}\label{qubicle}
Q^{(0)}_{a,b} &=& \{(a,b) \in \mathbb{R}^{d_A+ d_B}\,|\, \sum_i^{d_A} a_i = \sum_j^{d_B} b_j = 0 \, \nonumber \\
&&\wedge\, a_{i(1)}+ b_{j(1)}> a_{i(2)}+ b_{j(2)}> \ldots > a_{i(d_A d_B)}+ b_{j(d_A d_B)} \}\,,
\end{eqnarray}
and $i$ and $j$ are the index maps given by (\ref{indexmaps}) (and we skipped their index since both only depend on the cubicle and not on the single representant). Obviously, every cubicle is not only a convex set but in particular has also the form of a high dimensional open polyhedral cone. Note that since $(a,b) \sim \lambda (a,b)$ it is not a polytope. Moreover it has dimension $d_A+ d_B -2$ and with respect to the subspace topology it is open. This is clear since a little perturbation of $a$ and $b$ under the condition that  (\ref{testspecwyle}) still holds does not change the hierarchy $a_{i(1)}+ b_{j(1)} > a_{i(2)}+ b_{j(2)} > \ldots > a_{i(d_A d_B)}+ b_{j(d_A d_B)}$.
Every polyhedral cones $Q$ has the apex point $(0,\ldots,0)$ and is then uniquely defined by its finitely many extremal edges. Every edge is given by a set of the structure shown in \ref{qubicle} but with several `$=$'-signs instead of `$>$'-signs in the hierarchy, which then reduces the dimension from $d_A+d_B-2$ to $1$. These extremal edges can at least in principle be determined for every cubicle. But are they relevant for the solution of the QMP? Yes, since all eigenvalues enter linearly into the spectral inequalities (\ref{spectralineq4}), every inequality corresponding to some spectrum $(a,b)$ inside of the cubicle/polyhedral cone $Q^{(0)}$ depends linearly on finitely many inequalities corresponding to points (i.e.~test spectra) belonging to the extremal edges of this cubicle.
Finally we summarize these insights.
To determine the family of spectral inequalities (\ref{spectralineq4}) describing the solution of the QMP according Lem.~\ref{intersection2varieties}
it suffices to follow the strategy consisting of these finite processes:
\begin{enumerate}
\item Find all cubicles
\item Determine their edges
\item Check for each cubicle and for all possible triples $(\alpha,\beta,\gamma)$ of permutations whether the intersection property (\ref{intersectionprop4}) is fulfilled. If this is the case we find spectral inequalities (\ref{spectralineq4}) for every extremal edge belonging to the given cubicle.
\end{enumerate}
Note that the first two steps do neither depend on these permutations nor on the density operators. They are purely combinatorial and only depend on both dimensions $d_A$ and $d_B$. To illustrate the first two steps we present the example $d_A = d_B = 2$.
All test spectra (a,b) have the form
\begin{equation}
a= (a_0,-a_0) \qquad,\, b=(b_0,-b_0)\qquad,\, a_0,b_0 \in \mathbb{R}_0^+\,.
\end{equation}
Since $(a+b)$ has only four entries and since $a_1+b_1$ is always the largest and $a_2 + b_2$ is always the smallest, there are only two cubicles $Q_1$ and $Q_2$ in this setting
the first one (e.g.) is then characterized by $a_1+b_2 > a_2+b_1$ and the second one by $a_2+b_1 > a_1+b_2$.
%According to the visualization of cubicles by Young tableaux as explained above we find
%\begin{equation}
%Q_1 \,\,\widehat{=} \,\,\young(12,34) \qquad,\,\, Q_2 \,\,\widehat{=} \,\,\young(13,24)\qquad.
%\end{equation}
Moreover it is quite easy here to determine all extremal edges. Note that since edges are $1$-dimensional objects and all start at the apex $(0,\ldots,0)$ of the open polyhedral cone and go to infinity they are uniquely defined by one single point (vector) of this edge/line. We find
\begin{eqnarray}
\mbox{edge} \,E_1\,:\qquad (a,b) &=& ((1,-1),(0,0)) \nonumber \\
\mbox{edge} \,E_2\,:\qquad (a,b) &=& ((0,0),(1,-1)) \nonumber \\
\mbox{edge} \,E_3\,:\qquad (a,b) &=& ((1,-1),(1,-1)) \,.
\end{eqnarray}
After all, we observe that $Q_1$ has the extremal edges $E_1$ and $E_3$ and $Q_2$ has the extremal edges $E_2$ and $E_3$. This finishes the first two steps of the strategy. The third step is the most difficult one, namely the verification of the intersection property (\ref{intersectionprop4}) for all triples of permutations
\begin{equation}
(\alpha,\beta,\gamma) \in \mathcal{S}_2 \times \mathcal{S}_2 \times \mathcal{S}_4 \,.
\end{equation}
Analogously to the previous section on the solution of the QMP $\{A,AB\}$ we use the geometric (and slightly primitive) approach introduced in Sec.~\ref{GeometricApproach}.
Considering the intersection property (\ref{intersectionprop4}) we are dealing with subvarieties $\Phi_{a,b}(X_{\alpha}(\rho_A)),X_{\beta}(\rho_{B}))$ and $X_{\gamma \omega_0}(-\rho_{AB})$ of the flag variety $\mbox{\emph{Fl}}_4$. Since $\mbox{\emph{Fl}}_4$ has complex dimension $6$, we can only consider triples $(\alpha, \beta,\gamma)$ of permutations ensuring that $\Phi_{a,b}(X_{\alpha}(\rho_A)),X_{\beta}(\rho_{B}))$ and $X_{\gamma \omega_0}(-\rho_{AB})$ have complementary dimensions, i.e.~they add up to $6$. Only for these cases it might be possible that both varieties intersect transversally, which is necessary for the geometric deformation approach (see Sec.~\ref{GeometricApproach}).
Concretely speaking since $X_{\alpha}$ has dimension $l(\alpha)$, we restrict to the case $(\omega_0 = (4,3,2,1))$
\begin{equation}
6= l(\alpha)+l(\beta)+l(\gamma \omega_0) = l(\alpha)+l(\beta)+ 6- l(\gamma)\,,
\end{equation}
i.e.
\begin{equation}\label{lengthtransversal}
l(\alpha)+l(\beta)= l(\gamma)\,.
\end{equation}
Note, that we may miss some triples $(\alpha, \beta, \gamma)$ leading to nonempty intersection, but at the moment we ignore this aspect.
We consider four main cases namely
\begin{equation}
(\alpha,\beta)= ((1,2),(1,2)),\, ((1,2),(2,1)),\,((2,1),(1,2)),\,((2,1),(2,1)),\,
\end{equation}
but since the second and third one are equivalent, i.e.~all possible spectral inequalities for one of them can be obtained by taking all inequalities obtained for the other one by swapping all labels $A$ and $B$. We are left with three cases.
Let us now consider the first one, i.e.~$(\alpha,\beta)= ((1,2),(1,2))$. Condition (\ref{lengthtransversal}) then yields $\gamma =(1,2,3,4)$ and $\gamma \omega_0 =(4,3,2,1)$.
We denote the eigenvectors of $\rho_A, \rho_B$ and $-\rho_{AB}$ by $(f_1,f_2), (g_1,g_2)$ and $(h_1,h_2,h_3,h_4)$.
By referring to the CEF we find
\begin{equation}
X_{(1,2)}(\rho_A) = \left[ \begin{array}{cc} 1& 0 \\ 0&1 \end{array}\right]   \,\,\mbox{w.r.t}\, \{f_1,f_2\}
\end{equation}
and
\begin{equation}
X_{(1,2)}(\rho_B) = \left[ \begin{array}{cc} 1& 0 \\ 0&1 \end{array}\right]   \,\,\mbox{w.r.t}\, \{g_1,g_2\}\,,
\end{equation}
where the specification at the end of the equation refers the external basis set.
Then we find
\begin{equation}
\Phi_{Q_1}(X_{(1,2)}(\rho_A),X_{(1,2)}(\rho_B)) = \left[ \begin{array}{cccc} 1& 0 &0&0\\ 0&1&0&0 \\0& 0 &1&0\\ 0&0&0&1\end{array}\right]   \,\,\mbox{w.r.t}\, \{f_1\otimes g_1,f_1\otimes g_2,f_2\otimes g_1,f_2\otimes g_2\}\,
\end{equation}
\begin{equation}
\Phi_{Q_2}(X_{(1,2)}(\rho_A),X_{(1,2)}(\rho_B)) = \left[ \begin{array}{cccc} 1& 0 &0&0\\ 0&0&1&0 \\0& 1 &0&0\\ 0&0&0&1\end{array}\right]   \,\,\mbox{w.r.t}\, \{f_1\otimes g_1,f_1\otimes g_2,f_2\otimes g_1,f_2\otimes g_2\}
\end{equation}
and
\begin{equation}
X_{(4,3,2,1)}(\rho_{AB}) = \overline{\left[ \begin{array}{cccc} \ast & \ast  &\ast &1\\ \ast &\ast &1&0 \\\ast & 1 &0&0\\ 1&0&0&0\end{array}\right] }  \,\,\mbox{w.r.t}\, \{h_1,h_2,h_3,h_4\}= \mbox{\emph{Fl}}_4\,.
\end{equation}
Hence, the intersection property is trivial, i.e.~we find
\begin{eqnarray}\label{identityintersections}
\Phi_{Q_1}(X_{(1,2)}(\rho_A),X_{(1,2)}(\rho_B)) \cap X_{(4,3,2,1)}(\rho_{AB}) &=& \Phi_{Q_1}(X_{(1,2)}(\rho_A),X_{(1,2)}(\rho_B)) \nonumber \\
\Phi_{Q_2}(X_{(1,2)}(\rho_A),X_{(1,2)}(\rho_B)) \cap X_{(4,3,2,1)}(\rho_{AB}) &=& \Phi_{Q_2}(X_{(1,2)}(\rho_A),X_{(1,2)}(\rho_B)) \nonumber \\
&&
\end{eqnarray}
and both intersections contain exactly one point/flag. Hence, we are done without applying any deformation schema at all.
According to our strategy of determining spectral inequalities we can now write down an inequality for every extremal edge of each of both cubicles. This yields the following three spectral inequalities:
\begin{eqnarray}
\lambda_1^{(A)} -\lambda_2^{(A)} &\leq& \lambda_1^{(AB)} +\lambda_2^{(AB)}-\lambda_3^{(AB)}-\lambda_4^{(AB)} \nonumber \\
\lambda_1^{(B)} -\lambda_2^{(B)} &\leq& \lambda_1^{(AB)} +\lambda_2^{(AB)}-\lambda_3^{(AB)}-\lambda_4^{(AB)} \nonumber \\
\lambda_1^{(A)} -\lambda_2^{(A)}+\lambda_1^{(B)} -\lambda_2^{(B)} & \leq& 2 \lambda_1^{(AB)}-2 \lambda_4^{(AB)} \,.
\end{eqnarray}
They coincide with three of the four spectral inequalities found by S.Bravyi who has solved this version of the QMP in \cite{Brav}.
A.Klyachko called these three inequalities the basic ones because they are obtained by setting $\alpha,\beta,\gamma = id$. For all dimensions $d_A$ and $d_B$
one finds an intersection incidence for this specific triple of permutations! To obtain the remaining fourth spectral inequality for our setting $d_A= d_B= 2$ we have to consider the other two main cases. The problem there is that we have to deal with the closure of Schubert cells. We studied the Bruhat order in Sec.~\ref{flagvarieties} that allows us to determine the expansion of Schubert varieties in terms of Schubert cells (recall (\ref{partitionFlSchubert})). Nevertheless to study the intersection of two varieties is still very difficult in this geometric picture and therefore we will skip it here.
The remaining fourth inequality reads
\begin{equation}
|\lambda_1^{(A)}-\lambda_1^{(B)}| \,\,\leq\,\, \min(\lambda_1^{(AB)}-\lambda_3^{(AB)},\lambda_2^{(AB)}-\lambda_4^{(AB)})\,.
\end{equation}
Note that this single spectral inequality is obtained by combining several spectral inequalities of the form \ref{spectralineq4} (to get the minimum on the rhs and the absolute value on the lhs).

\chapter{Generalized Pauli Constraints and Concept of Pinning}\label{chap:Fermions}
\section{Motivation and summary}\label{sec:MotSumFerm}
The famous Pauli exclusion principle \cite{Pauli1925} states that no two electrons can occupy the same quantum state at the same time.
This restriction of fermionic occupation numbers plays an important role for the understanding of several quantum systems, from
few-electron systems as atoms up to macroscopic systems as solid materials. In 1926, already one year later, it was identified as a consequence
of another mathematical statement, the antisymmetry of the many-fermion wave function under particle exchange \cite{Dirac1926,Heis1926}.
Since the antisymmetry is much stronger than the exclusion principle, a natural question arises: Are there additional restrictions on fermionic
occupation numbers beyond Pauli's exclusion principle? This question can be reformulated as a (univariant) quantum marginal problem. It asks for the set of $1$-particle reduced density operators $\rho_1$, which can arise (via partial trace) from pure antisymmetric $N$-particle states $|\Psi_N\rangle$. Due to the unitary equivalence (see also Lem.~\ref{unitaryequiv}) only the natural occupation numbers (NONs) $\lambda_i$, the eigenvalues of $\rho_1$
play a role.

In this chapter we present the solution of that problem found by Klyachko \cite{Kly3, Kly2}, also shown qualitatively in Fig.~\ref{fig:Mot3}.
The NONs $\vec{\lambda}=(\lambda_1,\ldots,\lambda_d) \equiv \mbox{spec}(\rho_1)$ compatible to some
$N$-fermion state $|\Psi_N\rangle$ first of all fulfill
\begin{figure}[h!]
\includegraphics[width=10cm]{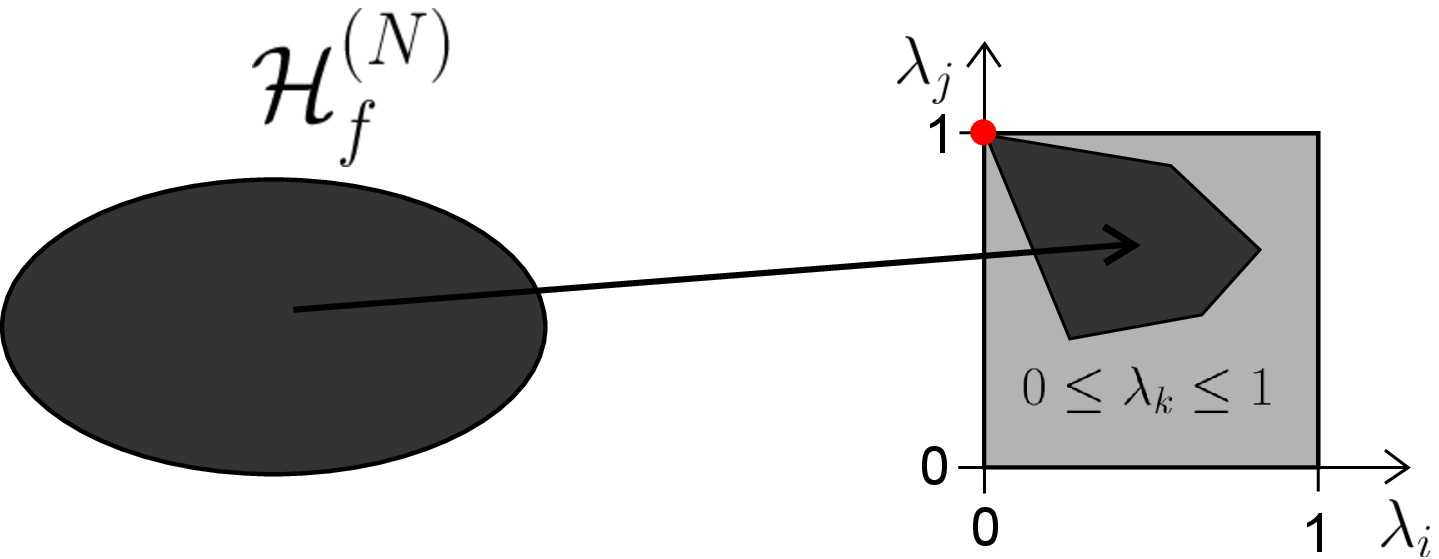}
\centering
\captionC{Schematic illustration of how the family of antisymmetric $N$-particle states $|\Psi_N\rangle$ maps to their vectors $\vec{\lambda}$ of natural occupation numbers forming a polytope (dark-gray), a proper subset of the light-gray hypercube describing Pauli's exclusion principle.
Single Slater determinants are mapped to the Hartree-Fock point $\vec{\lambda}_{HF}\equiv (1,\ldots,1,0,\ldots)$, a vertex (red dot) of that polytope.}
\label{fig:Mot3}
\end{figure}
Pauli's exclusion principle, which can be formulated as
\begin{equation}\label{PauliConstraint}
	0\leq \lambda_i \leq 1\qquad,\,\forall i\,.
\end{equation}
However, for each fixed particle number $N$ and dimension $d$ of the $1$-particle Hilbert space there are further restrictions, strengthening the exclusion principle. These so-called \emph{generalized Pauli constraints} take the form of affine inequalities,
\begin{equation}
D_j^{(N,d)}(\vec{\lambda}) = \kappa_0^{(j)} + \kappa_1^{(j)} \lambda_1+\ldots + \kappa_d^{(j)} \lambda_d \geq 0\qquad,\,j=1,2,\ldots,r^{(N,d)}
\end{equation}
with $\kappa_k^{(j)} \in \Z$ and give rise to a polytope $\mathcal{P}_{N,d}$, a proper subset of the ``Pauli hypercube'' (see Fig.~\ref{fig:Mot3}).

We address the question of their physical relevance. Do the generalized Pauli constraints limit the structure of fermionic states and can they restrict the behavior of fermionic systems? First evidence for some relevance for ground states was conjectured in form of the \emph{pinning-effect} by Klyachko \cite{Kly1}. He expected that NONs $\vec{\lambda}$ of (at least some) fermionic ground states lie \emph{exactly} on the boundary of the polytope.
From our viewpoint this effect is quite unnatural. Why should the NONs of some \emph{interacting} many-fermion system saturate exactly some of those $1$-particle constraints? We postulate the effect of \emph{quasi-pinning}, defined by NONs very close to the boundary of the polytope but not exactly on it. This will be justified by Chap.~\ref{chap:Physics} providing strong evidence that realistic fermionic ground states are indeed quasi-pinned.

Moreover, it turns out that (quasi-)pinning is highly physically relevant. Whenever the NONs $\vec{\lambda}$ are exactly pinned to the boundary of the polytope one can conclude that the corresponding many-fermion state has a very specific and simplified structure. As an example consider a three fermion state $|\Psi_3\rangle \in \wedge^3[\mathcal{H}_1^{(6)}]$, which is exactly pinned. Then the structural relation implies the form
\begin{equation}
|\Psi_3\rangle = \alpha \,|1,2,3\rangle + \beta \,|1,4,5\rangle + \gamma \,|2,4,6\rangle\,
\end{equation}
with some $1$-particle states $|i\rangle$ and $|i_1,i_2,i_3\rangle$ denotes a Slater determinant. The structure of that state is indeed much simpler then that of generic states, linear combination of $\binom{6}{3} = 20$ Slater determinants. This remarkable relation of pinning as an effect in the $1$-particle picture and the structure of $|\Psi_N\rangle$ as object in the $N$-particle picture seems to be also stable, i.e.~it still holds approximately for NONs close but not exactly on the polytope boundary. Therefore, quasi-pinning is also physically relevant.

Although Klyachko provides an algorithm for calculating the generalized Pauli constraints for each fixed $N$ and $d$ this is still quite involved and not feasible for settings with $d \gg 10$. So far the polytopes $\mathcal{P}_{N,d}$ are known only for settings with $d\leq 10$. They seem to be useless for realistic physical systems which are typically based on a very large or even infinite-dimensional $1$-particle Hilbert space. So how should one investigate possible (quasi-)pinning for given NONs $\vec{\lambda} \in \R^{d'}$, $10 <d' \in \N \cup\{\infty\}$? For concrete fermionic states one often observes NONs $\vec{\lambda}=(\lambda_i)_{i=1}^{d'}$, where most of them are very small, $\lambda_{d+1},\lambda_{d+2},\ldots \approx 0$.
We develop the concept of truncation, which allows us to consider only the first $d$ NONs $\vec{\lambda}^{tr}=(\lambda_1,\ldots,\lambda_{d})$ and
check their position inside of the polytope $\mathcal{P}_{N,D}$. The minimal distance $D^{tr}$ of $\vec{\lambda}^{tr}$ to the boundary
of $\mathcal{P}_{N,D}$ can then be related to the distance $D$ of the complete NONs $\vec{\lambda}$ to the boundary of the unknown polytope $\mathcal{P}_{N,d'}$,
\begin{equation}
D = D^{tr} + O(\lambda_{d+1})\,.
\end{equation}
Hence, we can investigate possible pinning whenever the NONs have main weight on a subspace $\R^d$ with $d\leq 10$. After all, the truncation error is given by the largest neglected eigenvalue $\lambda_{d+1}$.

\section{$N$-fermion concepts}\label{sec:fermionicconcepts}
In this section we introduce some basic concepts for the description of $N$-fermion quantum systems. This in particular contains the definition of antisymmetric states, $p$-particle reduced density operators ($p$-RDOs) and the role of symmetries.
Consider a particle described by states in a $1$-particle Hilbert space $\mathcal{H}_1^{(d)}$. If its dimension $d \in \N\cup\{\infty\}$ is not essential we will omit the superscript $d$. For most of the technical concepts the concrete form and possible substructure of $\mathcal{H}_1$ is not relevant. There it does not make a difference whether $\mathcal{H}_1$ corresponds e.g.~to orbital or spin degrees of freedom (or both). A system of $N$ such particles is described by states in the corresponding $N$-particle Hilbert space
\begin{equation}\label{HilbertN}
\mathcal{H}_N \equiv \left(\mathcal{H}_1\right)^{\otimes^N}\,.
\end{equation}
Below, we may also consider mixed states $\rho_N$, which are defined by
\begin{equation}
\rho_N\,:\, \mathcal{H}_N \overset{linear}{\longrightarrow} \mathcal{H}_N\,,\,\,\,\, \mbox{Tr}[\rho_N] = 1\,,\,\,\,\, \rho_N^\dagger= \rho_N\,,\,\,\,\,\rho_N\geq 0\,.
\end{equation}
If the $N$-particle system is described by a state $\rho_N$ the subsystem containing particles ${\bd{i}} = (i_1,\ldots,i_r)$ is then described by the corresponding reduced density operator
\begin{equation}\label{partialtrace}
\rho_{\bd{i}} = \mbox{Tr}_{\bd{i}^C}[\rho_N]\,.
\end{equation}
Here, $\bd{i}^C$ is the complement of $\bd{i}$ and $\mbox{Tr}_{\bd{j}}$ denotes the partial trace w.r.t.~the $1$-particle Hilbert spaces $j_k \in \bd{j}$.

Since $N$ identical particles cannot be distinguish from the quantum mechanical viewpoint, their mathematical description should reflect this. This leads to two particle types, fermions and bosons. Fermionic systems are described by states in the Hilbert space
\begin{equation}\label{HilbertNf}
\mathcal{H}_N^{(f)} \equiv  \wedge^N[\mathcal{H}_1]\leq \left(\mathcal{H}_1\right)^{\otimes^N}\,,
\end{equation}
which are antisymmetric under particle exchange.
Bosonic systems are identified with the Hilbert space of symmetric states,
\begin{equation}\label{HilbertNb}
\mathcal{H}_N^{(b)} \equiv  S^N[\mathcal{H}_1]\leq \left(\mathcal{H}_1\right)^{\otimes^N}\,.
\end{equation}

The exchange symmetry is at the heart of several interesting features of physical systems. Essentially fermionic quantum systems are often significantly dominated by their exchange symmetry. As a necessary tool to study this we consider permutations and their realization on $N$-particle Hilbert spaces.
Consider a \textbf{bijective} map
\begin{equation}
P:\,\,\{1,2,\ldots,N\} \rightarrow \{1,2,\ldots,N\}\,\,\,\,,\,\,\,k\mapsto P(k)\,.
\end{equation}
This map is a permutation of the elements $1,2,\ldots,N$. We interpret it actively: The $N$ numbers $1,2,\ldots,N$ as abstract ``objects'' are shuffled  according
\begin{equation}\label{perm}
(1,2,\ldots,N) \mapsto (P(1),P(2),\ldots,P(N))\,.
\end{equation}
This interpretation is equivalent saying that the element $k$ is mapped to the new position $P^{-1}(k)$. The collection of all permutations of $N$ elements defines a group,  denoted by $\mathcal{S}_N$. Instead of considering numbers $1,2,\ldots,N$ we can also consider $N$-particles which should be permuted. This defines a unitary representation $\Lambda$ of $\mathcal{S}_N$ on the $N$-particle Hilbert space $\mathcal{H}_N$ (\ref{HilbertN}), given by
\begin{equation}\label{permrep}
\Lambda(P) |\varphi_1\rangle\otimes\ldots\otimes |\varphi_N\rangle = |\varphi_{P^{-1}(1)}\rangle\otimes\ldots\otimes |\varphi_{P^{-1}(N)}\rangle\,,
\end{equation}
where $\{|\varphi_k\rangle\}_{k=1}^N$ are arbitrary $1$-particle states. (\ref{permrep}) can again be interpreted as an active permutation. The $N$ particles sitting at the beginning on the ordered ``seats'' $(\varphi_1,\ldots,\varphi_N)$ are shuffled according $P$. By assigning the $k$-th $1$-particle Hilbert space  in (\ref{HilbertN}) to the $k$-th particle we obtain (\ref{permrep}).
Moreover, note that we use the inverse of the permutation $P$ on the rhs in (\ref{permrep}) to fulfill $\Lambda(P_1) \Lambda(P_2) = \Lambda(P_1 P_2)$. Since it is convenient, we will skip in the following the symbol $\Lambda$ and denote the permutation and its representation on $\mathcal{H}_N$ by the same symbol.

Since $\mathcal{H}_N^{(f)}$ is a linear subspace of $\mathcal{H}_N$, we can define the orthogonal projection operator $\mathcal{A}_N$ projecting every state $|\Psi_N\rangle \in \mathcal{H}_N$ to $\mathcal{H}_N^{(f)}$. We find
\begin{equation}\label{antisymop}
\mathcal{A}_N = \frac{1}{\sqrt{N!}}\,\sum_{P \in \mathcal{S}_N}\, \mbox{sign}(P)\, P\,,
\end{equation}
where $\mbox{sign}(P)$ is the signature of the permutation $P$.
Due to its form the operator (\ref{antisymop}) is also called antisymmetrizing operator.
Notice that an $N$-particle density operators $\rho_N$ is antisymmetric whenever $\mathcal{A}_N \rho_N = \rho_N = \rho_N \mathcal{A}_N$. The condition $[\rho_N,\mathcal{A}_N]=0$ is necessary but not sufficient since e.g.~$\rho_N = |\Psi_N\rangle  \langle \Psi_N|$ with $|\Psi_N\rangle$ symmetric also fulfills $[\rho_N,\mathcal{A}_N]=0$. Since we will focus in the following on antisymmetric density operators, we observe
\begin{lem}\label{lem:permsign}
Recall (\ref{permrep}) and (\ref{antisymop}). Given a permutation $P \in \mathcal{S}_N$ and an antisymmetric $N$-particle density operator $\rho_N$. Then
\begin{equation}
P \mathcal{A}_N = \mbox{sign}(P) \mathcal{A}_N\,,
\end{equation}
which implies in particular
\begin{equation}
P \rho_N = \mbox{sign}(P) \rho_N = \mbox{sign}(P^\dagger) \rho_N = P^\dagger \rho_N  \,.
\end{equation}
\end{lem}
The proof is trivial since it immediately follows from the form of (\ref{antisymop}). To relate the concept of indistinguishability with that of exchange symmetry we state
\begin{lem}\label{lem:rrdounique}
Consider an antisymmetric $N$-particle density operator $\rho_N$. Then, for each fixed $r=1,\ldots,N-1$, all $r$-RDOs of $\rho_N$ are identical. The same also holds for bosonic density operators.
\end{lem}
\begin{proof}
First, we consider antisymmetric states. Let $r=1,\ldots,N-1$ be arbitrary, but fixed and choose two (different) subsets $\bd{i},\bd{j} \subset \{1,2,\ldots,N\}$ of length $r$. Let $P_{\bd{i},\bd{j}}$ be some permutation of the elements $1,2,\ldots,N$, which maps $\bd{j}$ to $\bd{i}$, $P_{\bd{i},\bd{j}}^{-1}(i_k)= j_k$, for all $k=1,2,\ldots,r$ (recall (\ref{perm}) and (\ref{permrep})). For the corresponding $r$-RDOs $\rho_{\bd{i}}$ and $\rho_{\bd{j}}$ we find by using $P_{\bd{i},\bd{j}} \rho_N = \mbox{sign}(P_{\bd{i},\bd{j}}) \rho_N$ and $\mbox{sign}(P_{\bd{i},\bd{j}}) = \mbox{sign}(P_{\bd{i},\bd{j}}^\dagger)$
\begin{eqnarray}
\rho_{\bd{i}} &=& \mbox{Tr}_{\bd{i}^C}[\rho_N]\nonumber \\
&=& \mbox{Tr}_{\bd{i}^C}[P_{\bd{i},\bd{j}} \rho_N P_{\bd{i},\bd{j}}^\dagger]\nonumber \\
&=& \mbox{Tr}_{\bd{j}^C}[ \rho_N ]\nonumber \\
&=& \rho_{\bd{j}}\,.
\end{eqnarray}
Here $\bd{i}^C$ stands for the complement of $\bd{i}$ w.r.t.~the total set $\{1,2,\ldots,N\}$.
The statement for symmetric density operators $\rho_N$ follows in a similar way.
\end{proof}
According to Lem.~\ref{lem:rrdounique} for each fixed $r=1,\ldots,N-1$  all $r$-RDO of $\rho_N$ are identical and we denote them by the symbol $\rho_r$.
Historically, Lem.~\ref{lem:rrdounique} led to the exchange symmetries. Since it is impossible to distinguish between identical particles, all their $1$-RDOs should be identical. This property is only compatible with $N$-particle states of specific exchange symmetry, which finally yields (\ref{HilbertNb}) and (\ref{HilbertNf}).

For the fermionic Hilbert space $\mathcal{H}_N^{(f)}$ bases of single Slater determinants are quite useful. In particular they allow us for a given $N$-particle state to determine its $r$-RDO just by symbolical calculation. First, we choose an orthonormal basis for $\mathcal{H}_1^{(d)}$ ($d$ may be infinite),
\begin{equation}\label{basis1}
\mathcal{B}_1 =\{|i\rangle\}_{i=1}^d\,.
\end{equation}
This basis gives rise to a natural $N$-fermion basis $\mathcal{B}_N$ spanned by all the Slater determinants
\begin{equation}
|\bd{i}\rangle \equiv \mathcal{A}_N[|i_1\rangle\otimes \ldots \otimes |i_N\rangle]\,,
\end{equation}
$1\leq i_1<\ldots <i_N \leq d$.
We still introduce the creation and annihilation operators $a_i^{\dagger}, a_i$, creating or annihilating a fermion in the state $|i\rangle$. They fulfill the standard anticommutator relations
\begin{equation}
\{a_i,a_j\}= 0\qquad\{a_i^{\dagger},a_j^{\dagger}\}= 0 \qquad\{a_i^{\dagger},a_j\}= \delta_{ij}\,\,\,.
\end{equation}
For a given pure $N$-fermion state $|\Psi_N\rangle$ normalized to $1$ the corresponding $r$-RDO can elegantly be expressed via
\begin{equation}\label{rrdomatrix}
\langle i_1,\ldots,i_r|  \rho_r |j_1,\ldots,j_r \rangle =\langle \Psi_N|a_{i_1}^{\dagger}\ldots a_{i_r}^{\dagger} a_{j_r}\ldots a_{j_1}|\Psi_N\rangle\,.
\end{equation}
Here and in the following we trace-normalize the $r$-RDO to $\binom{N}{r}$ (see Sec.~\ref{sec:Nrepresentability} for the motivation).
Given an expansion of $|\Psi_N\rangle$ in Slater determinants the corresponding $r$-RDO can easily be calculated via (\ref{rrdomatrix}). Notice that this can be done symbolically by just referring to the orthogonal character of the $1$-particle states (\ref{basis1}). Since the $1$-fermion picture plays the central role for Chap.~\ref{chap:Fermions} and Chap.~\ref{chap:Physics}, we define
\begin{defn}\label{defNOandNON}
Consider a system of $N$-identical fermions in a state $\rho_N$. The eigenvectors of the $1$-RDO $\rho_1$ are called \emph{natural orbitals} (NOs) and their occupancies, the eigenvalues of $\rho_1$ are called \emph{natural occupation numbers} (NONs).
\end{defn}

Now, we still explore the influence of $N$-particle symmetries on the $1$-RDO. First, we state
\begin{lem}\label{lem:globTolocSym}
Given an antisymmetric $N$-particle density operator $\rho_N$ and a unitary operator $U$ on the $1$-particle Hilbert space. Then,
\begin{equation}
[\rho_N,U^{\otimes^N}] =0 \qquad \Rightarrow \qquad [\rho_1,U] =0\,.
\end{equation}
\end{lem}
\begin{proof}
By using $\mathds{1}= U^\dagger U$, the cyclicity of the trace and $[\rho_N,U^{\otimes^N}] =0$ we observe
\begin{eqnarray}
U \rho_1 U^\dagger &=& N\,Tr_{2 \ldots N}[U \otimes \mathds{1}^{\otimes^{N-1}} \rho_N U^\dagger \otimes \mathds{1}^{\otimes^{N-1}} ] \nonumber \\
&=& N\,Tr_{2 \ldots N}[U^{\otimes^N} \rho_N \left(U^\dagger\right)^{\otimes^N} ] \nonumber \\
&=& N\,Tr_{2 \ldots N}[\rho_N  ] \nonumber \\
&=& \rho_1\,,
\end{eqnarray}
which immediately proves Lem.~\ref{lem:globTolocSym}.
\end{proof}
Lem.~\ref{lem:globTolocSym} is quite important for concrete physical applications. Consider an operation on single particles, as e.g.~the flipping of the electron spin w.r.t.~the $3$-axis or the translation of a particle on a lattice with periodic boundary conditions. Its representation on $\mathcal{H}_1$ is given by a unitary operator $U$. In a natural way, we can apply the same transformation to each of $N$ fermions simultaneously, e.g.~we flip the spin of all $N$ electrons simultaneously or translate the total $N$-fermion system on the lattice. This transformation is then represented on $\mathcal{H}_N^{(f)}$ by $U^{\otimes^N}$. So-called $U^{\otimes^N}$-\emph{symmetry-adapted} states $|\Psi_N\rangle$, defined as eigenstates of $U^{\otimes^N}$,
\begin{equation}\label{symadapted}
U^{\otimes^N} |\Psi_N\rangle = e^{i \varphi} |\Psi_N\rangle \,\,\,,\, \varphi \in \R
\end{equation}
automatically fulfill the assumption of Lem.~\ref{lem:globTolocSym}, $[\rho_N, U^{\otimes^N}]=0$. As a consequence their NOs are always adapted to the symmetry $U$. This becomes quite important by considering a Hamiltonian $H$ with a symmetry,
\begin{equation}
[H,U^{\otimes^N}] = 0\,.
\end{equation}
The energy eigenstates $|\Psi_N\rangle$ can always be chosen as eigenstates of the unitary operator $U^{\otimes^N}$ and are thus symmetry-adapted and their NOs are eigenstates of $U$, as well.

\begin{ex}\label{ex:translationBloch}
Consider a finite $1$-dimensional lattice $\mathcal{L}_X$ with $d$ sites and a translationally invariant Hamiltonian for $N$-fermions,
\begin{equation}
[H,T^{\otimes^N}] = 0\,,
\end{equation}
where $T$ is the translation of a fermion from one site to the next one.
Hence, the energy eigenstates $|\Psi_N\rangle$ are (or can be chosen as) $T^{\otimes^N}$-symmetry adapted states, with a $T^{\otimes^N}$-eigenvalue of the form $e^{i \frac{2\pi K}{d}}$, $K=0,1,\ldots,d-1$. The NOs of $|\Psi_N\rangle$ are $T$-symmetry adapted states, i.e.~Bloch states $|k\rangle_Q$, $k=0,1,\ldots,d-1$.
%Moreover, by denoting the corresponding NONs by $\lambda_k$ one can show in addition a relation of the $N$-particle wavenumber $K$ and the $1$-particle wavenumbers,
%\begin{equation}
%K = \sum_{k=0}^{d-1} \lambda_k \,k\,.
%\end{equation}
%This relation leads to an interesting conclusion on $|\Psi_N\rangle$.
In addition, we can make a related statement on the structure of $|\Psi_N\rangle$. By expanding $|\Psi_N\rangle$ w.r.t.~Slater determinants $|k_1,\ldots,k_N\rangle_Q$ built up from $1$-particle Bloch states $|k_i\rangle_Q$ only those Slater determinants can show up, which respect the ``conservation law'' $k_1+\ldots+k_N = K$ $\mod{d}$.
\end{ex}
\begin{ex}\label{ex:Spincomponent}
Consider $N$ electrons in an eigenstate $|\Psi_N\rangle$ of the $N$-particle operator $S_m^{(N)} \equiv \hat{m}\cdot \vec{S}^{(N)}$, the projection of the spin vector to the axis $m$ defined by the normalized vector $\hat{m}$. How is this related to symmetry adaption? Well, being an eigenstate of $S_m^{(N)}$ is equivalent to be an eigenstate of the $N$-electron spin rotation operator $D_{\hat{m}}(\varphi)^{\otimes^N}$, where
\begin{equation}
D_{m}(\varphi) = e^{\frac{i}{\hbar} \varphi \hat{m}\cdot \vec{S}^{(1)}} = \cos{\big(\frac{\varphi}{2}\big)} \mathds{1} +  2 i \sin{\big(\frac{\varphi}{2}\big)}\,\hat{m}\cdot \frac{\vec{S}^{(1)}}{\hbar}\,
\end{equation}
is the $1$-electron spin rotation by an angle $\varphi$ w.r.t.~the $m$-axis.
Therefore, $|\Psi_N\rangle$ is $D_{\hat{m}}(\varphi)^{\otimes^N}$-symmetry adapted and the corresponding NOs are $D_{m}(\varphi)$-symmetry adapted. After all, this also implies that the NOs are $1$-particle eigenstates of $S_{m}^{(1)}$, as well.
\end{ex}

\section{$N$-representability problem}\label{sec:Nrepresentability}
The present section will introduce the central mathematical problem on whose solution this thesis is built up. This will also connect Chap.~\ref{chap:Math} with Chap.~\ref{chap:Fermions} and Chap.~\ref{chap:Physics}.
Let us consider a system of $N$ identical fermions described by a Hamiltonian $H_N$. The indistinguishability means (recall Sec.~\ref{sec:fermionicconcepts})
\begin{equation}
[H_N,P] = 0\,
\end{equation}
for any permutation of the $N$ fermions. Moreover, we assume that $H_N$ has a fermionic ground state and contains only $1$- and $2$-particle interaction terms,
\begin{equation}\label{Hamiltoniangeneral12}
H_N = \sum_{i=1}^N h^{(i)}_1 +  \sum_{1\leq i < j \leq N} h^{(i,j)}_2\,.
\end{equation}
Here, $h^{(i)}_1$ acts on the $1$-particle Hilbert space of the $i$-th particle and $h^{(i,j)}_2$ acts on the two $1$-particle Hilbert spaces of particles $i$ and $j$. Moreover, we skip several identity operators in Eq.~(\ref{Hamiltoniangeneral12}) for the other $1$-particle Hilbert spaces $k \neq i,j$.
The ground state $|\Psi_0\rangle$ of (\ref{Hamiltoniangeneral12}) and its energy $E_0$ can be obtained by a minimization,
\begin{eqnarray}\label{energyGSminim1}
E_0 &=& \underset{\Psi_N \in \mathcal{H}_N^{(f)}}{\min}(\langle \Psi_N|H_N|\Psi_N\rangle)
\nonumber\\
&=& \underset{\Psi_N \in \mathcal{H}_N^{(f)}}{\min}( \mbox{Tr}[|\Psi_N\rangle \langle \Psi_N|H_N]) \nonumber\\
&=& \underset{\rho_N \in \mathcal{B}(\mathcal{H}_N^{(f)})}{\min}( \mbox{Tr}[\rho_N H_N])\,.
\end{eqnarray}
In the last step of (\ref{energyGSminim1}) we used the fact that the ground state energy does not reduce by increasing the optimization area to fermionic mixed states. Due to the structure of the Hamiltonian $H_N$ (recall (\ref{Hamiltoniangeneral12}) contains only $1$-body and $2$-body interactions) we can simplify (\ref{energyGSminim1}). By using the definition of the $1$-RDO $\rho_1$ and $2$-RDO $\rho_2$ (recall Lem.~\ref{lem:rrdounique} and remark thereafter) we find
\begin{equation}\label{energyGSminim2}
E_0 = \underset{\rho_2 \in D_2^{(e)}}{\min}( \mbox{Tr}[\rho_1 h_1] + \mbox{Tr}[\rho_2 h_2])\,.
\end{equation}
Here $\rho_1$ is the $1$-RDO following from $\rho_2$, $\rho_1 = \frac{2}{N-1}\,\mbox{Tr}_1[\rho_2]$ and $D_r^{(e)}$ with $r=2$ denotes the set of all antisymmetric density operators, which are \emph{ensemble-$N$-representable}. This means that for each of its elements $\rho_r$ there exists a mixed $N$-fermion state $\rho_N$, which leads via partial trace to $\rho_r$. We also introduce the set $D_r^{(p)}$ of \emph{pure-$N$-representable} $r$-RDOs by restricting to $\rho_N$ pure. $D_2^{(e)}$ is (obviously) convex in contrast to $D_2^{(p)}$. The computational advantages following from this property motivated the relaxing of the optimization in Eq.~(\ref{energyGSminim1}) to all mixed $N$-fermion states. Although the minimization (\ref{energyGSminim2}) over just $2$-particle density operators seem to be much easier than that over $N$-fermion states there is one fundamental problem: One does not know the set $D_2^{(p/e)}$. This gives rise to the following definition.
\begin{defn}\label{def:Nrepresentability}
Consider $N$ fermions. For each fixed $r=1,2,\ldots,N-1$ the problem of determining the set $D_r^{(p/e)}$ of possible $r$-RDOs $\rho_r$ arising via partial trace from a corresponding pure/ensemble $N$-fermion state is known as the $r$-body pure/ensemble-$N$-representability problem.
\end{defn}
\begin{rem}\label{rem:NrepQMP}
The $r$-body $N$-representability problem (Def.~\ref{def:Nrepresentability}) is a quantum marginal problem. Although it looks univariant for all $r=1,2,\ldots,N-1$ this is the case just for $r=1$. For $r>1$ it is not possible to treat the system of $r$ of the $N$ fermions as one single quantum system since their reduced state has still a substructure (antisymmetry) that one needs to take into account (see also Sec.~\ref{sec:QMPdef}, in particular Sec.\ref{sec:QMPpureuni}).
\end{rem}
Due to its strong importance for physics and quantum chemistry the $2$-body $N$-representability problem was studied intensively (see e.g.~\cite{Col2,Dav,Col,MC}). However, it was shown in \cite{Liu} that this problem is QMA-complete: By using a quantum computer it is not possible to even verify of falsify in polynomial time a suggested description of $D_2^{(p)}$. Practically, this means that one cannot find $D_2^{(p)}$ for arbitrary particle numbers $N$ and dimensions $d$ of the $1$-particle Hilbert space. However, the form (\ref{energyGSminim2}) can still be used as a significant simplification of the minimization (\ref{energyGSminim1}). Using the positivity of $\rho_N$ and that the partial trace is positivity preserving allows to derive some elementary outer approximation of $D_2^{(e)}$ (defined by so-called $P,Q,G, T_1,T_2$ conditions; see also \cite{PQGTZhao}). Optimizing (\ref{energyGSminim2}) over this larger set leads to a lower bound for the ground state energy, which is often remarkably close to the correct one \cite{PQGTZhao}.

Although the $r$-body pure/ensemble $N$-representability problem seems to be not solvable for $r>1$ this is not true for the case $r=1$.
We will study the case $r=1$ for ensembles right now. The description of $D_1^{(e)}$ is well-known and easy to derive in contrast to that of $D_1^{(p)}$ found recently by Klyachko et al.~\cite{Kly3,Altun} and presented in the next section, Sec.~\ref{sec:generalizedPC}.

First, note that according to Remark \ref{rem:NrepQMP} and the unitary equivalence Lem.~\ref{unitaryequiv} the set $D_1^{(e)}$ (as well as $D_1^{(p)}$) depends only on the eigenvalues $\vec{\lambda}=(\lambda_1,\ldots,\lambda_d)$ of the $1$-RDO (NONs). Moreover we state
\begin{lem}\label{lem:Pauli}
Recall Def.~\ref{def:Nrepresentability}. The set $D_1^{(e)}$ is given by all $1$-particle density operators with eigenvalues $\lambda_i$ satisfying
Pauli's exclusion principle \cite{Pauli1925}
\begin{equation}\label{PauliexclusionP}
0\leq \lambda_i \leq 1\qquad,\,\,\forall i=1,2,\ldots,d\,.
\end{equation}
\end{lem}
\begin{proof}
The proof that Pauli's exclusion principle is a necessary condition is quite elementary. We choose some arbitrary orthonormal basis $\{|i\rangle\}_{i=1}^d$ for the $1$-particle Hilbert space and introduce the corresponding creation and annihilation operators $a_i^\dagger, a_j$. By using the anticommutation relation we find for the particle number operator $\hat{n}_i \equiv a_i^\dagger a_i$ of the state $|i\rangle$ the relation $0\leq a_i^\dagger a_i = \mathds{1} - a_i a_i^\dagger$\,. Since $a_i a_i^\dagger \geq 0$, this finally yields
\begin{equation}
0\leq \hat{n}_i \leq \mathds{1}\,.
\end{equation}
This implies that every occupation number (e.g.~any NON) is bounded by Eq.~(\ref{PauliexclusionP}). To show that these constraints are also sufficient for
the compatibility of NONs to a mixed $N$-fermion state consider $\vec{\lambda}$ fulfilling (\ref{PauliexclusionP}).
It can be shown (e.g.~as a corollary of the Birkhoff-von Neumann theorem) that $\vec{\lambda}$ is a convex sum $\vec{\lambda} = \sum_i q_i \vec{v}_i$ of vectors $\vec{v}=(v_1,\ldots,v_d) \in \{0,1\}^d$ with $\|\vec{v}\|_1=N$. Then $\vec{\lambda}$ is compatible to the mixed state
\begin{equation}
\rho_N = \sum_i q_i \,|\vec{v}_i\rangle \langle \vec{v}_i|\,.
\end{equation}
Here $|\vec{v}_i\rangle$ denotes a Slater determinant in occupation number representation.
\end{proof}

\section{Generalized Pauli constraints}\label{sec:generalizedPC}
In contrast to the $1$-body ensemble $N$-representability problem (see previous section and Lem.~\ref{lem:Pauli}) determining the set $D_1^{(p)}$ of pure $N$-representable $1$-RDOs is highly non-trivial. Although the set $D_1^{(p)}\subset D_1^{(e)}$ was found for some smaller settings as e.g.~$\wedge^3[\mathcal{H}_1^{(6)}]$ already some decades ago \cite{Borl1972} it was not clear at all how to determine $D_1^{(p)}$ for general fermion numbers $N$ and dimensions $d$. In particular, for generic settings $\wedge^N[\mathcal{H}_1^{(d)}]$ it was even not known whether there are further restrictions on $D_1^{(p)}$, beyond Pauli's exclusion principle. Just recently, in the framework of the univariant quantum marginal problem (recall Sec.~\ref{sec:QMPpureuni}) A.Klyachko has solved the $1$-body pure $N$-representability problem \cite{Kly3}. His solution is illustrated in Fig.~\ref{fig:QMPfermSolut}.
\begin{figure}[h!]
\includegraphics[width=10.0cm]{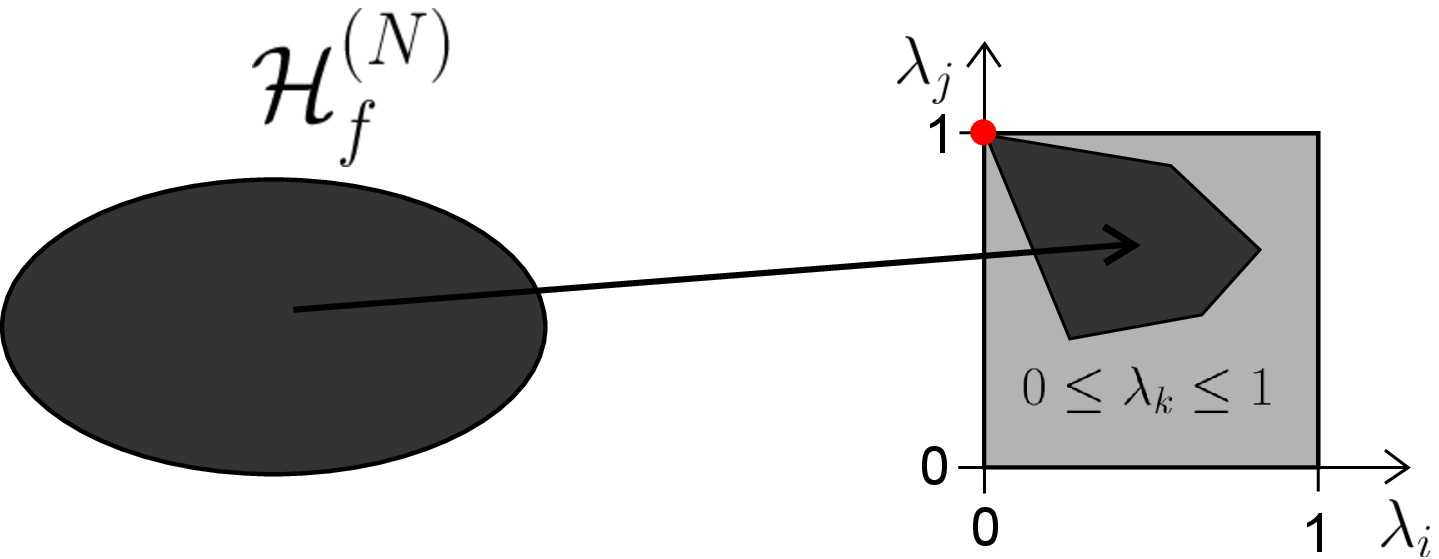}
\centering
\captionC{Schematic illustration of how the family of antisymmetric $N$-particle states $|\Psi_N\rangle$ maps to their vectors $\vec{\lambda}$ of natural occupation numbers forming a polytope (dark-gray), a proper subset of the light-gray hypercube describing Pauli's exclusion principle.
Single Slater determinants are mapped to the Hartree-Fock point $\vec{\lambda}_{HF}\equiv (1,\ldots,1,0,\ldots)$, a vertex (red dot) of that polytope.}
\label{fig:QMPfermSolut}
\end{figure}
The NONs $\vec{\lambda}\equiv (\lambda_1,\ldots,\lambda_d)$ compatible to antisymmetric pure states do not only lie inside the high-dimensional ``Pauli hypercube'', but are further restricted to a polytope $\mathcal{P}_{N,d}$. This polytope is defined by a \emph{finite} family of so-called generalized Pauli constraints, conditions of the form,
\begin{equation}\label{generalizedPC}
D_j^{(N,d)}(\vec{\lambda}) \equiv \kappa_0^{(j)} + \kappa_1^{(j)}  \lambda_1+\ldots + \kappa_d^{(j)}  \lambda_d \geq 0\qquad,\,j=1,\ldots,r^{(N,d)}\,,
\end{equation}
where $\kappa_i^{(j)} \in \Z$.
\begin{rem}\label{rem:genPCpolytope}
We always order the NONs decreasingly and normalize their sum to the particle number $N$. Consequently, the polytope $\mathcal{P}_{N,d}\subset \R^d$ is
given by all those NONs $\vec{\lambda}$ fulfilling all generalized Pauli constraints (\ref{generalizedPC}) and
\begin{eqnarray}
&&\lambda_1\geq \lambda_2\geq \ldots \geq \lambda_d \nonumber \\
&& \lambda_1+\ldots +\lambda_d = N\,.
\end{eqnarray}
The boundary of the polytope is given by all those $\vec{\lambda}$ saturating either a generalized Pauli constraint or an ordering constraint. However, in the following we mean by boundary $\partial\mathcal{P}_{N,d}$ of the polytope just the part corresponding to saturation of a generalized Pauli constraints.
\end{rem}
Summarizing the geometric picture shown in Fig.~\ref{fig:QMPfermSolut}, a vector $\vec{\lambda}$ of NONs is compatible to some antisymmetric pure $N$-fermion quantum state $|\Psi_N\rangle \in \wedge^N[\mathcal{H}_1^{(d)}]$ \emph{if and only if} $\vec{\lambda}$ fulfills the constraints $D_j^{(N,d)}(\vec{\lambda})\geq 0$ \emph{for all} $j = 1,2,\ldots,r^{(N,d)}$.

For each fixed $N$ and $d$ the generalized Pauli constraints (\ref{generalizedPC}) can in principle be calculated following Klyachko's algorithm \cite{Kly3}. However, since this is based on involved concepts as e.g.~cohomology theory of flag varieties, applying the algorithm is still very challenging. In particular it seems  unfeasible for those who are not familiar with algebraic topology. In general, the mathematics behind the generalized Pauli constraints is the same as that presented in Chap.~\ref{chap:Math}. Hence, to calculate the generalized Pauli constraints one may be tempted to circumvent cohomology theory by using the same elementary brute force method developed in Chap.~\ref{chap:Math}. There we succeeded in determining some of the constraints for the quantum marginal problem $\{A, AB\}$ (see Sec.~\ref{sec:QMPAvsAB}). Unfortunately, this turns out to be not systematic enough and the more profound approach based on cohomology theory is necessary. We skip this in the present thesis and refer to \cite{Kly3}.

In the following we present the generalized Pauli constraints for the smallest few settings $\wedge^N[\mathcal{H}_1^{(d)}]$.
As a first result we state
\begin{lem}\label{lem:partholegenPC}
Given the finite family $\{D_j^{(N,d)}(\cdot)\geq 0\}_{j=1}^{r^{(N,d)}}$ of generalized Pauli constraints for the setting $\wedge^N[\mathcal{H}_1^{(d)}]$. Then the constraints on the NONs $\vec{\lambda}$ for the setting $\wedge^{d-N}[\mathcal{H}_1^{(d)}]$
are given by $\{D_j^{(d-N,d)}(\cdot)\geq 0 \}_{j=1}^{r^{(d-N,d)}}$ with $r^{(d-N,d)} = r^{(N,d)}$ and $\forall j=1,\ldots, r^{(d-N,d)}$
\begin{equation}
D_j^{(d-N,d)}(\lambda_1, \ldots,\lambda_d) = D_j^{(N,d)}(1-\lambda_d,\ldots,1-\lambda_1)\,.
\end{equation}
\end{lem}
\noindent The proof is based on a particle-hole transformation, a natural isomorphism for $\wedge^N[\mathcal{H}_1^{(d)}] \cong \wedge^{d-N}[\mathcal{H}_1^{(d)}]$. Since this is straightforward, we skip it.
Due to Lem.~\ref{lem:partholegenPC} we will restrict in the following on the settings with $d\geq 2 N$.

The settings with $N=2$ particles are all elementary and solved by
\begin{lem}\label{lem:setting2}
The generalized Pauli constraints on the ordered NONs $\vec{\lambda}$ for the setting $\wedge^2[\mathcal{H}_1^{(d)}]$ are given by
\begin{equation}
\lambda_i = \lambda_{i+1}\,,
\end{equation}
for all $i$ odd, i.e.~all non-zero NONs are evenly many times degenerate.
\end{lem}
\begin{proof}
The proof is elementary \cite{DGpriv}. Consider for a given possible $\vec{\lambda} \in \R^d$, $d$ finite, the corresponding $2$-particle state $|\Psi_2\rangle$. According to the Schmidt decomposition there exist two orthonormal bases $\{|i\rangle_A\}_{i=1}^d$ and $\{|j\rangle_B\}_{j=1}^d$ both for the same $1$-particle Hilbert space $\mathcal{H}_1^{(d)}$ such that
\begin{equation}
|\Psi_2\rangle = \sum_{i=1}^d\,\sqrt{\lambda_i}\,|i\rangle_A\otimes |i\rangle_B\,.
\end{equation}
For $\lambda_i >0$ both $|i\rangle_A$ and $|i\rangle_B$ are corresponding eigenvectors of $\rho_1$. On the other hand we also find by using the antisymmetry of $|\Psi_2\rangle$
\begin{equation}
{}_A\langle i|i\rangle_B = \frac{1}{\sqrt{\lambda_i}} ({}_A\langle i|\otimes{}_A\langle i|) |\Psi_2\rangle =0.
\end{equation}
\end{proof}
Moreover, since each equality $A=B$ can be written as two conditions $A \geq B $ and $B \geq A$, the constraints $\lambda_i = \lambda_{i+1}$ for the settings with $N=2$ have indeed the general form of (\ref{generalizedPC}).

The so-called Borland-Dennis setting $\wedge^3[\mathcal{H}_1^{(6)}]$ is the smallest setting where some of the generalized Pauli constraints do not effectively take the form of equalities. The constraints were already found in 1972 \cite{Borl1972}.
They read
\begin{eqnarray}\label{set36}
\lambda_1+\lambda_6=\lambda_2+\lambda_5=\lambda_3+\lambda_4=1\nonumber \\
D^{(3,6)}(\vec{\lambda}) = 2-(\lambda_1+\lambda_2+\lambda_4) \geq 0\,.
\end{eqnarray}
Note that the inequality in (\ref{set36}) strengthens Pauli's exclusion principle, which here states $2-(\lambda_1+\lambda_2)\geq 0$. Moreover, it can also be written as $\lambda_5+\lambda_6-\lambda_4 \geq 0$.
For the setting $\wedge^3[\mathcal{H}_1^{(7)}]$ the generalized Pauli constraints are given by
\begin{eqnarray}\label{set37}
D_1^{(3,7)}(\vec{\lambda}) &=& 2-(\lambda_1+\lambda_2+\lambda_5+\lambda_6)\geq 0\nonumber \\
D_2^{(3,7)}(\vec{\lambda}) &=& 2-(\lambda_1+\lambda_3+\lambda_4+\lambda_6) \geq 0\nonumber \\
D_3^{(3,7)}(\vec{\lambda}) &=& 2-(\lambda_2+\lambda_3+\lambda_4+\lambda_5)\geq 0\nonumber \\
D_4^{(3,7)}(\vec{\lambda}) &=& 2-(\lambda_1+\lambda_2+\lambda_4+\lambda_7) \geq 0\,.
\end{eqnarray}
The 31 constraints for the larger setting $\wedge^3[\mathcal{H}_1^{(8)}]$ can be found in the Appendix \ref{app:GPClargsettings}.
It is important to notice that the number $r^{(N,d)}$ of constraints significantly increases if $N$ and $d$ increases. This is also the reason
why Altunbulak and Klyachko determined those constraints only for the settings with $N\leq 5$ and $d\leq 10$.

\section{Concept of (quasi-)pinning}\label{sec:conceptPinning}
In the previous section we have learned that the antisymmetry of $N$-fermion quantum states leads to so-called generalized Pauli constraints. These constraints on NONs $\vec{\lambda}$ are significantly stronger than the famous Pauli exclusion principle (see e.g.~Fig.~\ref{fig:QMPfermSolut}).
\begin{figure}[h!]
\centering
\includegraphics[width=2.6cm]{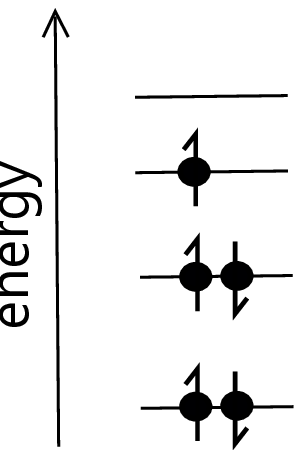}\hspace{2.0cm}
\includegraphics[width=3.9cm]{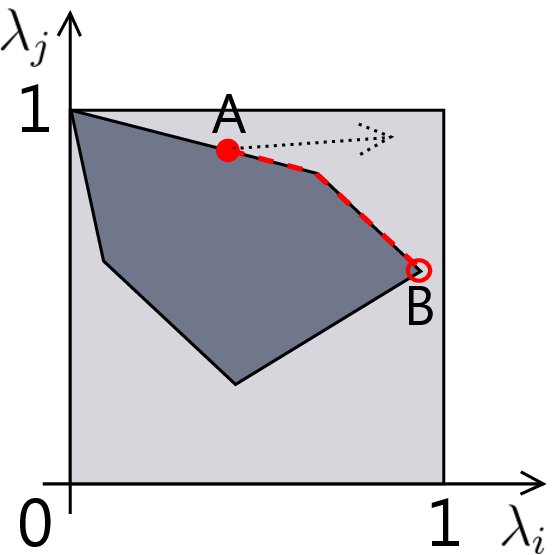}
\centering
\captionC{Right: Initial NONs on the boundary (A) of the polytope. Any time evolution that would like to drive them out of the polytope (according dashed arrow) is dominated by the geometry of the polytope rather than by the Hamiltonian. This kinematical influence by the generalized Pauli constraints generalizes that by the Pauli exclusion principle (on the left). There e.g.~the electron in the highest shell cannot decay to a lower one.}
\label{fig:PCpinned}
\end{figure}
The underlying $1$-particle picture is often the starting point for developing an understanding for given physical systems. Therefore, we may expect that
the knowledge of the generalized Pauli constraints could lead to new conceptual insights.
Nevertheless, except that some vectors $\vec{\lambda}$ are mathematically impossible it is not clear how to use the information about the polytope for the
understanding of physical systems. A first physical relevance is suggested by the following definition due to Klyachko \cite{Kly1}
\begin{defn}\label{def:pinning}
Recall Remark \ref{rem:genPCpolytope} and consider an $N$-fermion quantum state $|\Psi_N\rangle \in \wedge^N[\mathcal{H}_1^{(d)}]$ with NONs $\vec{\lambda}$ on the boundary of the polytope $\mathcal{P}_{N,d}$. This incidence as well as the possible mechanism behind it is called \emph{pinning}. Moreover, we then say that the NONs are \emph{pinned by} the generalized Pauli constraints \emph{to the boundary} $\partial\mathcal{P}_{N,d}$ of the polytope $\mathcal{P}_{N,d}$.
\end{defn}
There are two reasons why the effect of pinning would be interesting. The first one was mentioned by Klyachko in \cite{Kly1}.
\begin{enumerate}
\item Given an $N$-fermion Hamiltonian. Its ground state $|\Psi_N\rangle$ can be obtained via a minimization of the energy expectation value (recall (\ref{energyGSminim1})). If the ground state turns out to be pinned to the boundary of the polytope a generalized Pauli constraint (\ref{generalizedPC}) is active for the minimization in the sense that any further minimization of the energy would violate it. In that case the corresponding constraint would  have a strong influence on the ground state and its energy.
\item Consider an $N$-fermion system initially prepared in a quantum state $|\Psi_N\rangle$ with NONs $\vec{\lambda}$ pinned to the boundary $\partial\mathcal{P}_{N,d}$ of the polytope. By switching on a unitary time evolution for that system the NONs $\vec{\lambda}(t)$ restricted to the polytope will begin to move. For some very specific time evolutions $\vec{\lambda}(t)$ may like to leave the polytope. Such a scenario is illustrated on the right side of Fig.~\ref{fig:PCpinned}. Then, since $\vec{\lambda}(t)$ cannot leave the polytope, it will move along the boundary from the initial point $A$ to the final point $B$. In that case the time evolution is dominated by the geometry of the polytope rather than by the Hamiltonian. This kinematical effect on time evolutions is a generalization of a similar more elementary effect shown on the left side of Fig.~\ref{fig:PCpinned}. There we can see some (non-interacting) electrons occupying low-lying energy shells. Coupling this systems to photons may in principle lead to a decaying of the electrons in the higher energy shells to the lowest one. However, this is impossible due to Pauli's exclusion principle. In the same way the physics of solid bodies at low temperatures is dominated by the electrons close to the Fermi level.
\end{enumerate}
In the spirit of the mechanism explained in point 1., Klyachko expected that typical ground states may be pinned. In \cite{Kly1} he claimed first numerical evidence for that.
Analyzing the NONs obtained by Nakata et al.~in a variational calculation for some (artificial) Beryllium state and presented in \cite{Nakata2001} with six digits accuracy Klyachko found pinning. However, by reproducing the Nakata results and taking more than just six digits into account we found that the distance to the polytope boundary is indeed very small (order $10^{-7}$) but finite \cite{CS2013QChem}. This leads us to the definition of a new physical phenomenon called \emph{quasi-pinning}.
\begin{defn}\label{def:quasipinning}
Consider an $N$-fermion quantum state with NONs $\vec{\lambda}$ in a distance $\mbox{dist}(\vec{\lambda},\partial \mathcal{P}_{N,d})$ to the boundary of the polytope $\mathcal{P}_{N,d}$. If this distance is surprisingly small compared to the natural size $1$ of the polytope,
\begin{equation}
\mbox{dist}(\vec{\lambda},\partial \mathcal{P}_{N,d}) \ll 1
\end{equation}
we say that the NONs are \emph{quasi-pinned} to the boundary of the polytope. Moreover $\mbox{dist}(\vec{\lambda},\partial \mathcal{P}_{N,d})$ is a measure for the \emph{quasi-pinning}.
\end{defn}
In Chap.~\ref{chap:Physics} we will find strong evidence that ground states of few fermions confined to some trap exhibit strong quasi-pinning. Moreover, the following section on the physical relevance of (quasi-)pinning strongly emphasizes that exact pinning is unnatural and should not appear for realistic physical systems.

Why may quasi-pinning be physically relevant and why may it show up? The two statements above (1. and 2.) providing an answer for the phenomenon of \emph{exact} pinning do not apply for quasi-pinning. There it was essential that the NONs are \emph{exactly on} the boundary of the polytope. Hence, quasi-pinning has so far no physical relevance at all. Fortunately, this will change in the next section.

As a caveat for Chap.~\ref{chap:Physics}, where we will investigate possible pinning for ground states of concrete systems we observe.

\section{Physical relevance of (quasi-)pinning}\label{sec:physRelPinning}
In this section we will learn that pinning of NONs $\vec{\lambda}$ as effect in the $1$-particle picture corresponds to a very specific and simplified structure of the corresponding $N$-fermion quantum state $|\Psi_N\rangle$. Moreover, we find strong evidence that this statement is stable under small deviations of the NONs. Hence, quasi-pinning is also highly physically relevant since any quasi-pinned $|\Psi_N\rangle$ has approximately this simplified structure.

To introduce this relation of pinning and structure of $|\Psi_N\rangle$ we start with a special case. Consider for the setting with $N$ fermions and a $d$-dimensional $1$-particle Hilbert space NONs
\begin{equation}\label{NONsHF}
\vec{\lambda}_{HF} \equiv (\underbrace{1,\ldots,1}_N,0,\ldots,0)\,.
\end{equation}
This vector is called Hartree-Fock point. The following Lemma and its proof will justify this name.
\begin{lem}\label{lem:HFNONtoSlater}
Given a state $|\Psi_N\rangle \in \wedge^N[\mathcal{H}_1^{(d)}]$ with NOs $\{|i\rangle\}_{i=1}^d$ and NONs $\vec{\lambda} \in \mathcal{P}_{N,d}$ close to the Hartree-Fock point (\ref{NONsHF})
\begin{equation}\label{HFdistance}
(\lambda_{HF,1}-\lambda_1)+\ldots+(\lambda_{HF,N}-\lambda_N) = N-(\lambda_1+\ldots+\lambda_N) =:\delta \ll 1\,.
\end{equation}
Then we find
\begin{equation}\label{HFstability}
1-\frac{1}{N}\delta \geq |\langle 1,2,\ldots,N|\Psi_N\rangle |^2 \geq 1-\delta\,.
\end{equation}
Moreover, by optimizing over $1$-particle states $|\tilde{i}\rangle$ the overlap $|\langle \tilde{1},\tilde{2},\ldots,\tilde{N}|\Psi_N\rangle |^2$
can exceed the upper bound in (\ref{HFstability}) only on the scale $\delta^2$.
\end{lem}
\begin{proof}
Let $|\Psi_N\rangle$ be a state with NONs $\vec{\lambda}$ and its NOs give rise to a basis $\mathcal{B}_1 =\{|i\rangle\}_{i=1}^d$ of $\mathcal{H}_1^{(d)}$. We denote the corresponding creation and annihilation operators by $a_i^\dagger, a_j$ (recall Sec.~\ref{sec:fermionicconcepts})
and define an operator (depends implicitly on $|\Psi_N\rangle$)
\begin{equation}\label{Shatoperator}
\hat{S} = N \mathds{1}-\left(a_1^{\dagger} a_1+\ldots + a_N^{\dagger} a_N\right)\,.
\end{equation}
Since all operators $a_i^{\dagger} a_i, i=1,\ldots, d$ commute, $\hat{S}$ has the spectrum $\{0,1,\ldots,N\}$. We denote the orthogonal projection operator on the $k$-eigenspace of $\hat{S}$ by $P_k$.
%
%we can expand $|\Psi_N\rangle$ in Slater determinants constructed from those $1$-particle states,
%\begin{equation}\label{expansionstate1}
%|\Psi\rangle = \sum_{\textbf{\emph{i}}}\,c_{\textbf{\emph{i}}}\,|\textbf{\emph{i}}\rangle \,.
%\end{equation}
%Here we sum over all $N$-tuples $\textbf{\emph{i}}=(i_1,\ldots,i_N)$, $1\leq i_1 <i_2<\ldots i_N\leq d$.
%Consider now the operator (implicitly depending on $|\Psi_N\rangle$)
%\begin{equation}
%\hat{S} = N \mathds{1}-\left(a_1^{\dagger} a_1+\ldots + a_N^{\dagger} a_N\right)\,.
%\end{equation}
%Since all operators $a_i^{\dagger} a_i, i=1,\ldots, d$ commute $\hat{S}$ has the spectrum $\{0,1,\ldots,N\}$ with eigenstates $|\textbf{\emph{i}}\rangle$.
%The eigenvalue of $|\textbf{\emph{i}}\rangle$ is given by the number of indices in the set $\textbf{\emph{i}}$, which are not larger than $N$. We denote the set of indices leading to the eigenvalue $k$ by $J_k$. Using this yields
For fixed $|\Psi_N\rangle$ the $\hat{S}$-operator was chosen such that the lhs in (\ref{HFdistance}) can be expressed as the expectation value
$\langle\hat{S}\rangle_{\Psi_N} \equiv\langle \Psi_N|\hat{S} |\Psi_N\rangle = N- (\lambda_1+\ldots+\lambda_d )$. The idea is now to use the gapped structure of $\hat{S}$ to show that whenever $\langle\hat{S}\rangle_{\Psi_N} \approx 0$, $|\Psi_N\rangle$ has the main weight in the eigenspace of the minimal $\hat{S}$-eigenvalue $0$.
Indeed, we easily find
\begin{eqnarray}\label{HFstabilityEstimtion}
\langle \Psi_N|\hat{S} |\Psi_N\rangle &=& \sum_{k=0}^N k\,\langle\Psi_N|P_k |\Psi_N\rangle \nonumber \\
&\geq& \sum_{k=1}^N \,\langle\Psi_N|P_k |\Psi_N\rangle \nonumber \\
&=& 1- \|P_0 \Psi_N\|_{L^2}^2\,.
\end{eqnarray}
Since the $0$-eigenspace is $1$-dimensional and spanned by the Slater determinant \\$|1,2,\ldots,N\rangle$, the lower bound for $|\langle 1,2,\ldots,N|\Psi_N\rangle |^2$ in Lem.~\ref{lem:HFNONtoSlater} follows. To derive the upper bound we repeat the calculation (\ref{HFstabilityEstimtion}) but in the second line we estimate $\sum_{k=0}^N k\,\langle\Psi_N|P_k |\Psi_N\rangle \leq \sum_{k=1}^N \,N \langle\Psi_N|P_k |\Psi_N\rangle$.
The proof that the best overlap of $|\Psi_N\rangle$ with a Slater determinant is given by $1-\frac{1}{N}\delta +O(\delta^2)$ is straightforward and we skip it here.
\end{proof}
\begin{rem}\label{rem:HFpoint}
Lem.~\ref{lem:HFNONtoSlater} in particular states that the NONs $\vec{\lambda}_{HF}$ (\ref{NONsHF}) do arise exactly from those states which can be written as one single Slater determinant. Since this relation of NONs and structure of the corresponding $N$-fermion quantum state is stable under small deviations of the NONs, the Hartree-Fock optimization is meaningful and leads to excellent approximations of ground states whenever their NONs are close to the Hartree-Fock point.
\end{rem}
\begin{rem}\label{rem:weakcorrel}
Continuing on Remark \ref{rem:HFpoint} we have learned that the neighborhood of $\vec{\lambda}_{HF}$ (see e.g.~red dot in Fig.~\ref{fig:QMPfermSolut}) describes the regime of weak fermion correlations. Moreover, the ground state NONs of non-interacting fermions are exactly pinned to the boundary of the polytope and coincide with the Hartree-Fock point $\vec{\lambda}_{HF}$. A natural question arises as whether they are still pinned to the boundary after turning on some interaction between the fermions. This question is investigated in Chap.~\ref{chap:Physics}.
\end{rem}

Can we generalize Lem.~\ref{lem:HFNONtoSlater} to arbitrary points on or close to the polytope boundary $\partial \mathcal{P}_{N,d}$? As a first step we should find an appropriate generalization of the $\hat{S}$-operator (\ref{Shatoperator}). For this consider a generalized Pauli constraint
\begin{equation}\label{gPC1}
D(\lambda)=\kappa_0  + \kappa_1 \lambda_1+\ldots+\kappa_{d} \lambda_{d} \geq 0
\end{equation}
and a state $|\Psi_N\rangle$ with NONs $\vec{\lambda}$ and NONs $\{|i\rangle\}_{i=1}^d$. Then by defining
\begin{equation}\label{Dhatoprator}
\hat{D}:= \kappa_0 \mathds{1} + \kappa_1 a_1^{\dagger} a_1+\ldots+\kappa_{d} a_{d}^{\dagger}a_{d}
\end{equation}
we find indeed
\begin{equation}\label{DtoDhat}
D(\vec{\lambda}) = \langle\Psi_N|\hat{D}|\Psi_N\rangle\,.
\end{equation}
Note that since $\kappa_i \in \Z$ we have $\mbox{spec}(\hat{D}) \subset \Z$. In contrast to $\hat{S}$ (\ref{Shatoperator}), $\hat{D}$ is not positive semi-definite anymore and takes the general form $\hat{D} = P_-\oplus P_0 \oplus P_+$\,. $P_\pm$ denotes the projection operator to the positive/negative eigenspace and $P_0$ that onto the $0$-eigenspace.

The hope is now to find a lemma of the following type. Whenever $D(\vec{\lambda}) \approx 0$, $|\Psi_N\rangle$ has main weight in the $0$-eigenspace of $\hat{D}$. Unfortunately, such a statement cannot easily be proven anymore. In contrast to the $\hat{S}$-operator studied above the interesting $0$-eigenvalue is not extremal in $\mbox{spec}(\hat{D})$. Therefore, it seems possible that $|\Psi_N\rangle$ has weights in the spaces
corresponding to $P_-$ and $P_+$, which then cancel out in Eq.~(\ref{DtoDhat}).

However, with more effort one can prove a statement for exact pinning.
\begin{lem}\label{lem:exPintoStruct}
Consider an $N$-fermion state $|\Psi_N\rangle$ with NONs $\vec{\lambda}$ exactly saturating a generalized Pauli constraint (\ref{gPC1}).
Recall the corresponding $\hat{D}$-operator (\ref{Dhatoprator}) referring  to the NOs of $|\Psi_N\rangle$. Then,
\begin{equation}
\hat{D}|\Psi_N\rangle := \left(\kappa_0 \mathds{1}+\kappa_1 a_1^{\dagger} a_1+\ldots+\kappa_{d} a_{d}^{\dagger}a_{d}\right)|\Psi_N\rangle = 0\,.
\end{equation}
\end{lem}
\noindent For the proof we refer to \cite{DGpriv}.

From Lem.~\ref{lem:exPintoStruct} we can immediately conclude that whenever NONs are pinned to some facet of the polytope the corresponding $|\Psi_N\rangle$ has weight only in the $0$-eigenspace of the corresponding $\hat{D}$-operator (\ref{Dhatoprator}). This can be formulated as a selection rule for Slater determinants:
\begin{cor}\label{cor:selrulSlaters}
Consider an $N$-fermion state $|\Psi_N\rangle$ with NOs $\{|i\rangle\}_{i=1}^d$ and NONs $\vec{\lambda}$ exactly saturating a generalized Pauli constraint (\ref{gPC1}). By expanding $|\Psi_N\rangle$ in Slater determinants,
\begin{equation}
|\Psi_N\rangle = \sum_{\bd{i}}\,c_{\bd{i}}\,|\bd{i}\rangle,
\end{equation}
we find the following selection rule (recall (\ref{Dhatoprator}))
\begin{equation}
\hat{D} |\bd{i}\rangle \neq 0 \qquad \Rightarrow \qquad c_{\bd{i}} = 0\,.
\end{equation}
\end{cor}

To emphasize the importance of Lem.~\ref{lem:exPintoStruct} we apply the selection rule to an example.
\begin{ex}\label{ex:PintoStructBD}
Consider a state $|\Psi_3\rangle \in \wedge^3[\mathcal{H}_1^{(6)}]$ with NONs $\vec{\lambda}$. The generalized Pauli constraints are given by (\ref{set36}). The first three constraints take independent of $\vec{\lambda}$ the form of equalities. According to Cor.~\ref{cor:selrulSlaters} this leads to universal structural implications for any arbitrary $|\Psi_3\rangle \in \wedge^3[\mathcal{H}_1^{(6)}]$.
In the expansion
\begin{equation}
|\Psi_3\rangle = \sum_{1\leq i_1 <i_2 < i_3\leq 6} c_{i_1,i_2,i_3}\,|i_1,i_2,i_3\rangle
\end{equation}
only those Slater determinants $|i_1,i_2,i_3\rangle$ can show up which have exactly one index $i_k$ in each of the three sets $\{1,6\}$, $\{2,5\}$ and $\{3,4\}$. There are only $2^3=8$ such Slater determinants: $|1,2,3\rangle$, $|1,2,4\rangle$, $|1,3,5\rangle$, $|1,4,5\rangle$, $|2,3,6\rangle$, $|2,4,6\rangle$, $|3,5,6\rangle$ and $|4,5,6\rangle$. This universal statement is not in contradiction to the dimension $\binom{6}{3}=20$ of the $3$-fermion Hilbert space $\wedge^3[\mathcal{H}_1^{(6)}]$: For each given state $|\Psi_3\rangle$ we can find an appropriate $1$-particle basis (the NOs) such that $|\Psi_3\rangle$ can be written as a linear combination of those eight Slater determinants.

Assume now, that the corresponding NONs $\vec{\lambda}$ are pinned to the facet described by saturation of $D^{(3,6)}(\cdot) = 0$ (recall (\ref{set36})). The Cor.~\ref{cor:selrulSlaters} implies that
\begin{equation}\label{DBPinningstructure}
|\Psi_3\rangle = \alpha |1,2,3\rangle +\beta |1,4,5\rangle + \gamma |2,4,6\rangle\,.
\end{equation}
The three coefficients are free but should be chosen such that $\lambda_1 \geq \lambda_2 \geq \ldots\geq \lambda_6$.
\end{ex}

Example \ref{ex:PintoStructBD} impressively emphasizes the importance of Lem.~\ref{lem:exPintoStruct} and Cor.~\ref{cor:selrulSlaters}. We summarize these insights by
\begin{rem}\label{rem:Pinningisrelev}
Pinning corresponds to specific and simplified structures of the corresponding $N$-fermion quantum state $|\Psi_N\rangle$. In that sense pinning is highly physically relevant. It is also remarkable that pinning as phenomenon in the simple $1$-particle picture allows to reconstruct the structure of $|\Psi_N\rangle$ as object in the important $N$-particle picture.
\end{rem}

Due to Remark \ref{rem:Pinningisrelev} and the selection rule Cor.~\ref{cor:selrulSlaters} we can confirm our intuitive expectation that exact pinning is unnatural. Indeed, for interacting fermions confined to some potential trap it seems quite unlikely that for instance their ground state can be expanded by just a few Slater determinants. More realistically \emph{all} Slater determinants will show up. If their weights have a clear decaying hierarchy at least quasi-pinning is still possible. Therefore, the central question is whether the powerful Lem.~\ref{lem:exPintoStruct} and its Cor.~\ref{cor:selrulSlaters} are stable. Does a NONs-vector $\vec{\lambda}$ in the vicinity of some polytope facet implies that the corresponding $|\Psi_N\rangle$ has approximately the specific and simplified structure stated in Cor.~\ref{cor:selrulSlaters}?
Unfortunately, this is not the case for the whole vicinity of the polytope boundary $\partial\mathcal{P}_{N,d}$. A counterexample for the Borland-Dennis setting (\ref{set36}) is given by the state
\begin{equation}\label{DBcounterex}
|\Psi_3\rangle = \alpha |1,3,5\rangle +\sqrt{|\alpha|^2+|\gamma|^2-\delta} \,|1,2,4\rangle + \gamma |2,3,6\rangle\,.
\end{equation}
with $|\alpha| > |\gamma|$ and $\delta>0$ arbitrarily small.
This state leads to properly ordered NONs, exhibits strong quasi-pinning $D^{(3,6)} = \delta$ but violates maximally the structural implications by the selection rule in Cor.~\ref{cor:selrulSlaters}. We also notice that for this counterexample strong quasi-pinning, $\delta \ll 1$, is equivalent to an approximate saturation of some ordering constraints, $D^{(3,6)}(\vec{\lambda})= \lambda_3-\lambda_4$.

Can we use the latter insight to state a generalized form of Lem.~\ref{lem:exPintoStruct} and Cor.~\ref{cor:selrulSlaters} for some specific region in the polytope $\mathcal{P}_{N,d}$? Yes, at least for the Borland-Dennis setting we can provide an analytical result.
\begin{thm}\label{thm:BDSelRulestable}
Given a state $|\Psi_3\rangle \in \wedge^3[\mathcal{H}_1^{(6)}]$ with NONs $\vec{\lambda}\equiv (\lambda_k)_{k=1}^6$.
Let $P$ be the projection operator onto the subspace spanned by the states $|1,2,3\rangle, |1,4,5\rangle, |2,4,6\rangle$, which corresponds to exact pinning
of $D^{(3,6)}(\vec{\lambda})=\lambda_5+\lambda_6-\lambda_4 \geq 0$ (recall (\ref{set36}) and Cor.~\ref{cor:selrulSlaters}). Then as long as
\begin{equation}
\delta\equiv 3-\lambda_1-\lambda_2-\lambda_3 \leq \frac{1}{4}
\end{equation}
(which means being not to far away from the Hartree-Fock point) we find
\begin{equation}
1- \chi_{\delta} D^{(3,6)}(\vec{\lambda}) \,\,\leq\,\, \|P \Psi_3\|_{L^2}^2 \,\, \leq\,\, 1- \frac{1}{2} D^{(3,6)}(\vec{\lambda})\,,
\end{equation}
with
\begin{equation}
\chi_{\delta} \equiv \frac{1+2\delta}{1 - 4\delta} \,.
\end{equation}
\end{thm}
\noindent The proof is presented in Appendix \ref{app:BDSelRulestable}.

Unfortunately, it is impossible to generalize the proof of Thm.~\ref{thm:BDSelRulestable} to arbitrary $N$ and $d$. However, by using Monte Carlo simulation we found evidence that Lem.~\ref{lem:exPintoStruct} and Cor.~\ref{cor:selrulSlaters} are stable under small deviations of the NONs whenever the saturation of some generalized Pauli constraint is much smaller than the saturation of the ordering constraints, $\lambda_i-\lambda_{i+1} \geq 0$.
Therefore, we conjecture
\begin{conj}\label{conj:SelRulestable}
Given an $N$-fermion state $|\Psi_N\rangle \in \wedge^N[\mathcal{H}_1^{(d)}]$ with NOs $\{|i\rangle\}_{i=1}^d$ and NONs $\vec{\lambda}$ approximately saturating a generalized Pauli constraint, $D^{(N,d)}(\vec{\lambda}) = \delta \approx 0$.
Moreover, we introduce $P$ as the projection operator onto the subspace spanned by all those Slater determinants $|\bd{i}\rangle$ fulfilling $\hat{D}|\bd{i}\rangle = 0$  (recall (\ref{Dhatoprator})). Then, whenever $\min_{i}(\lambda_i-\lambda_{i+1}) \equiv \varepsilon \gg \delta$ the selection rule Cor.~\ref{cor:selrulSlaters} is stable,
\begin{equation}
\|P \Psi_N\|_{L^2}^2 = 1-O(\delta)\,.
\end{equation}
\end{conj}

\section{Pinning analysis and concept of truncation}\label{sec:truncation}
In this section we explain in detail how to explore possible pinning of concrete NONs.
At first sight comments on how to perform a pinning analysis seem to be superfluous. Is it not just about plugging in the given NONs $\vec{\lambda}$ into the corresponding generalized Pauli constraints and verifying whether some of them are exactly or approximately saturated?
Yes, that is indeed that simple whenever the constraints for the underlying setting $\wedge^N[\mathcal{H}_1^{(d)}]$ are known.
Unfortunately, this is only the case for settings with $d\leq 10$ and there is no hope that one can determine the constraints for settings with $d\gg 10$. Since most physical systems have an \emph{infinite-dimensional} underlying $1$-particle Hilbert space, performing a pinning analysis for their states seems to be impossible. However, e.g.~for ground states of trapped fermions one typically observes NONs with main weight in some low-dimensional subspace $\R^d$, i.e.~almost all NONs except the first few ($d$) are very close to $0$. In that case we expect that one can omit those close to $0$ and investigate possible pinning of the truncated NONs $\vec{\lambda}^{tr}=(\lambda_1,\ldots,\lambda_d)$ w.r.t.~the setting $\wedge^N[\mathcal{H}_1^{(d)}]$. In the following we will verify that such a truncation is indeed possible. Moreover, we will explain how the results found in the pinning analysis in a truncated setting are related to possible pinning in the correct setting.

The strategy for this consists of the following three aspects.
\begin{enumerate}
\item Relation of the generalized Pauli constraints of two different settings $\wedge^N[\mathcal{H}_1^{(d)}]$ and $\wedge^N[\mathcal{H}_1^{(d')}]$
\item Definition of a measure for quasi-pinning
\item Concept of truncation and pinning analysis
\end{enumerate}

\subsection{Relation of generalized Pauli constraints for different settings}\label{sec:genPCtwosettings}
We start by introducing some notation.
By recalling in particular Sec.~\ref{sec:fermionicconcepts} we can expand every $|\Psi_N\rangle$ w.r.t. to the Slater determinants $\mathcal{B}_N = \{|\bd{i}\rangle\}$ built up from its NOs,
\begin{equation}\label{expansion}
|\Psi_N\rangle = \sum_{\bd{i}}\,c_{\bd{i}}\,|\bd{i}\rangle\,.
\end{equation}
The NONs then follow as
\begin{equation}\label{noncoef}
\lambda_k = \sum_{\bd{i},\,k \in \bd{i}}\,|c_{\bd{i}}|^2 \,.
\end{equation}
To compare settings of different dimensions, $d, d'$ with $d < d'\leq \infty$ we embed $\mathcal{H}_1^{(d)}$ into $\mathcal{H}_1^{(d')}$,
\begin{equation}
\mbox{span}\{|i\rangle\}_{i=1}^{d}\equiv\mathcal{H}_1^{(d)} \leq \mathcal{H}_1^{(d')} \equiv \overline{\mbox{span}\{|i\rangle\}_{i=1}^{d'}}\,,
\end{equation}
where the closure is only relevant for the case $d'$ infinite.
In the same way,
\begin{equation}
\wedge^N[\mathcal{H}_{1}^{(d)}] \leq \wedge^N[\mathcal{H}_{1}^{(d')}]\,.
\end{equation}
Indeed, according to (\ref{expansion}), we find that every state
\begin{equation}\label{expansionsmall}
|\Psi_N\rangle =  \sum_{1\leq i_1 <\ldots<i_N \leq d} c_{\bd{i}} \,|\bd{i}\rangle \,\,\,\,\in \wedge^N[\mathcal{H}_1^{(d)}]
\end{equation}
can be embedded into $\wedge^N[\mathcal{H}_1^{(d')}]$ by
\begin{equation}\label{expansionlarge}
|\Psi'_N\rangle =  \sum_{1\leq i_1 <\ldots<i_N \leq d} c_{\bd{i}} \,|\bd{i}\rangle \,\,\,\,\in \wedge^N[\mathcal{H}_1^{(d')}]\,,
\end{equation}
and all the other coefficients $c_{\bd{i}}$ in (\ref{expansionlarge}) with $i_N>d$ vanish. Although they look the same we used different symbols for both states $|\Psi_N\rangle$ and $|\Psi'_N\rangle$ to distinguish between the two different spaces $\wedge^N[\mathcal{H}_1^{(d)}]$ and $\wedge^N[\mathcal{H}_1^{(d')}]$. This subtle difference is becoming relevant if we determine the NONs $\vec{\lambda}'$ of $|\Psi'_N\rangle$ (recall (\ref{noncoef})), \begin{equation}\label{nonimbedded}
\vec{\lambda}' = (\lambda_1,\ldots,\lambda_d,\underbrace{0,\ldots,0}_{d'-d})
\end{equation}
differing from $\vec{\lambda}=(\lambda_1,\ldots\lambda_d)$ by additional zeros. In the following to simplify the notation we use the same symbols for mathematical objects and their embeddings into larger spaces.

As a first step for developing the concept of truncation we state
\begin{lem}\label{lemzeros}
Consider for $N, N', d, d' \in \N$, $N\leq N'$, $d < d'$  the two settings $\wedge^N[\mathcal{H}_1^{(d)}]$ and $\wedge^{N'}[\mathcal{H}_{1}^{(d')}]$, and let $\vec{\lambda} = (\lambda_1,\ldots,\lambda_{d})$ be some NONs. Then, by introducing $N'-N\equiv r$ and $d'-d -r \equiv s$ we find
\begin{eqnarray}
(\lambda_1,\ldots,\lambda_d) &\mbox{compatible w.r.t.}& \wedge^N[\mathcal{H}_1^{(d)}]\qquad \qquad\qquad \nonumber \\
&\Leftrightarrow& \nonumber\\
(\underbrace{1,\ldots,1}_r,\lambda_1,\ldots,\lambda_d,\underbrace{0,\ldots,0}_s) &\mbox{compatible w.r.t.} & \wedge^{N'}[\mathcal{H}_{1}^{(d')}]\,\,.
\end{eqnarray}
Rephrasing this geometrically, we have
\begin{equation}
\mathcal{P}_{N,d} = \mathcal{P}_{N',d'}|_{\begin{array}{l}\lambda_{1},\ldots,\lambda_{r}=1 \\
                                          \lambda_{d+r+1},\ldots,\lambda_{d'}=0 \end{array}}\,,
\end{equation}
i.e.~the polytope $\mathcal{P}_{N',d'}$ intersected with the hyperplane
defined by $\lambda_{1},\ldots,\lambda_{r}=1$, $\lambda_{d+r+1},\ldots,\lambda_{d'}=0$ coincides with
$\mathcal{P}_{N,d}$.
\end{lem}
\begin{proof}
We prove Lem.~\ref{lemzeros} only for $r=0$. For $r>0$ the proof is similar but slightly lengthy.
The direction ``$\Rightarrow$'' was already explained above at the beginning of the present section.
To prove ``$\Leftarrow$'' we show that a state $|\Psi'_N\rangle$ expanded according to (\ref{expansion}),
\begin{equation}\label{expansionlarge}
|\Psi'_N\rangle =  \sum_{1\leq i_1 <\ldots<i_N \leq d'} c_{\bd{i}} \,|\bd{i}\rangle \,\,\,\,\in \wedge^N[\mathcal{H}_{1}^{(d')}]\,,
\end{equation}
with natural occupation numbers $(\lambda_1,\ldots,\lambda_d,\underbrace{0,\ldots,0}_s)$ contains only Slater determinants $|\bd{i}\rangle$, with $i_1,\ldots,i_N\leq d$. But this is clear
due to (\ref{noncoef}), which then yields
\begin{equation}
\forall \,k\,>d\,:\,\,\, 0 \stackrel{!}{=}\lambda_k = \sum_{\bd{i},\,k \in \bd{i}}\,|c_{\bd{i}}|^2 \,.
\end{equation}
Hence, $c_{\bd{i}} = 0$ if $i_N>d$.
\end{proof}
What does Lem.~\ref{lemzeros} imply for the relation between the families of generalized Pauli constraints of two settings?
Let us consider two settings with $d,d'$ finite, $d<d'$ and we restrict ourself just for simplicity to $N'=N$. Every constraint $D'_j$ for the setting $\wedge^N[\mathcal{H}_{1}^{(d')}]$
is linear and hence its restriction
\begin{equation}\label{constrestricted}
\hat{D}_j'(\lambda_1,\ldots,\lambda_d) \equiv D_j'(\lambda_1,\ldots,\lambda_d,0,\ldots)  \geq 0
\end{equation}
to the hyperplane defined by $0=\lambda_{d+1},\lambda_{d+2},\ldots$ is also a linear constraint in the remaining coordinates $\lambda_1,\ldots,\lambda_d$. How is the half space $V_j\subset \mathbb{R}^d$ corresponding to (\ref{constrestricted}) related to the polytope $\mathcal{P}_{N,d}$? Lem.~\ref{lemzeros} states that
\begin{equation}
\mathcal{P}_{N,d} \subset V_j
\end{equation}
and
\begin{equation}
\mathcal{P}_{N,d} = \cap_j V_j|_{\ast} \,,
\end{equation}
where the star $\ast$ denotes here the restriction to spectra, i.e.~ordered and normalized vectors.
There are two possible relations between $V_j$ (or $V_j|_{\ast}$) and $\mathcal{P}_{N,d}$. They are illustrated in Fig.~\ref{facetproj} in form of a simplified $2-$dimensional picture:
\begin{figure}[!h]
\centering
\includegraphics[width=8cm]{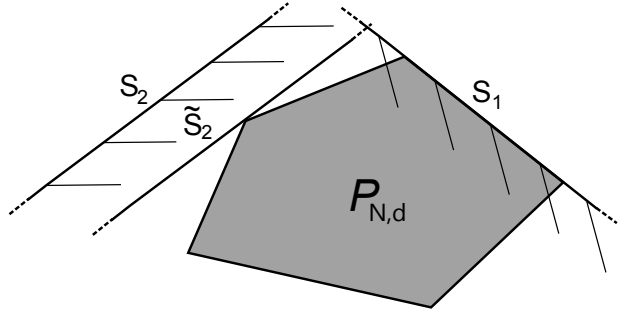}
\centering
\caption{Polytope $\mathcal{P}_{N,d}$ and projected constraints of larger settings.}
\label{facetproj}
\end{figure}
There, we consider two half spaces $V_1$ and $V_2$ corresponding to the ``restricted'' constraints $\hat{D}'_1\geq 0$ and $\hat{D}'_2\geq 0$ with boundaries $S_1$ and $S_2$ and orientation indicated by stripes. Such hyperplanes can either contain a facet of maximal (example $S_1$) or lower dimension of $\mathcal{P}_{N,d}$ or they lie outside of $\mathcal{P}_{N,d}$ (example $S_2$). The third case of a proper intersection is not possible due to Lem.~\ref{lemzeros}. Every constraint $D'$ with boundary $S$ of its restriction $\hat{D}'$ lying outside of $\mathcal{P}_{N,d}$ has the form
\begin{equation}
D'(\vec{\lambda}) = c + \tilde{D}(\vec{\lambda}^{tr}) + O(\lambda_{d+1})\,,
\end{equation}
where $\tilde{D}(\vec{\lambda}^{tr})\geq 0$ is a constraint of the setting $\wedge^N[\mathcal{H}_1^{(d)}]$ with a boundary shown in Fig.~\ref{facetproj} as hyperplane $\tilde{S}_2$ and $c >0$ is some offset.
Hence, if the NONs $\vec{\lambda}$ are sufficiently fast decaying, constraint $D'$ is not saturated at all due to the off set $c$ and thus irrelevant.
Moreover, for every facet of $\mathcal{P}_{N,d}$ corresponding to some constraint $D>0$ Lem.~\ref{lemzeros} guarantees the existence of a constraint $D'>0$ in the larger setting whose projection $\hat{D'}$ coincides with $D$.
All these insights imply the following lemma.
\begin{lem}\label{lemdistancemodif}
Given two settings $\wedge^N[\mathcal{H}_1^{(d)}]$ and $\wedge^{N'}[\mathcal{H}_{1}^{(d')}]$ with $N\leq N' \in \N$, $d < d' \in \mathbb{N}$ and decreasingly ordered NONs $\vec{\lambda}$ with $ \max(1-\lambda_r,\lambda_{d+r+1}) \equiv \varepsilon \ll 1$, $r\equiv N'-N$.
Every constraint $D'\geq 0$ of the setting $\wedge^{N'}[\mathcal{H}_{1}^{(d')}]$ relevant for the pinning analysis is then given by a linear modification of some constraint $D\geq 0$ of the setting $\wedge^N[\mathcal{H}_1^{(d)}]$,
\begin{equation}\label{distancemodif}
D'(\vec{\lambda}) = D(\vec{\lambda}^{tr}) + O(\varepsilon)\,.
\end{equation}
Here $\vec{\lambda}^{tr} \equiv (\lambda_{r+1},\ldots,\lambda_{d+r})$.
\end{lem}
Finally, we remark that for the important case $d'$ infinite effectively the same results hold but one has to deal with one subtlety. Since $\mathcal{P}_{N,\infty}$ is described by infinitely many constraints, Lem.~\ref{lemzeros} guarantees for every constraint $D\geq 0$ of the setting $\wedge^N[\mathcal{H}_1^{(d)}]$, only the existence of a sequence of constraints $D'_j\geq 0$ whose restrictions $\hat{D}'_j\geq 0$ converge to the constraint $D\geq 0$. This means that condition (\ref{distancemodif}) in Lem.~\ref{lemdistancemodif} holds up to a small error $\varepsilon$,
\begin{equation}
D_{\varepsilon}'(\vec{\lambda}) = \varepsilon + D(\vec{\lambda}^{tr}) + O(\lambda_{d+1})\,,
\end{equation}
which can be made arbitrarily small by choosing appropriate constraints $D'_{\varepsilon}$. Hence, to minimize the technical effort and not to confuse the reader, we assume that Lem.~\ref{lemdistancemodif} holds in its original form also for the case $d'$ infinite.\\

\subsection{Measures for quasi-pinning}\label{sec:PinningMeasures}
To explore and quantify possible quasi-pinning of given NONs $\vec{\lambda}\in \R^{d'}$ we need to find an appropriate measure. Consider an arbitrary facet $F_D$ of the polytope $\mathcal{P}_{N,d}$ defined by exact saturation of the constraint $D(\vec{\lambda})\equiv \kappa_0 + \vec{\kappa}\cdot \vec{\lambda}\geq 0$, $\vec{\kappa}\equiv (\kappa_1,\ldots,\kappa_{d'})$,
\begin{equation}\label{facet}
F_D = \{\vec{\lambda} \in \mathcal{P}_{N,d}\,|\,D(\vec{\lambda}) = 0 \}\,.
\end{equation}
From our viewpoint there are three natural measures for quantifying quasi-pinning of NONs $\vec{\lambda}$ by constraint $D(\cdot)\geq 0$. Those are given by
\begin{enumerate}
\item The generalized Pauli constraint adapted measure
      \begin{equation}\label{PinMeasD}
      d_D(\vec{\lambda}) \equiv D(\vec{\lambda})\,.
      \end{equation}
      Since the constraint $D(\cdot)\geq 0$ is defined only up to a positive factor, this constraint
      highly depends on the ``normalization'' of $D$, i.e.~on the choice of the coefficients $\kappa_j$, $j=0,1,\ldots,d'$.
      In the following we fix this normalization by choosing the minimal absolute values for the \emph{\emph{integer}} coefficients $\kappa_j$.
\item The euclidian, i.e.~$l^2$-distance of $\vec{\lambda}$ to the facet $F_D$,
       \begin{equation}\label{PinMeas2}
      d_2(\vec{\lambda}) \equiv \frac{D(\vec{\lambda})}{\|\vec{\kappa}\|_2}\,.
      \end{equation}
\item The $l^1$-distance of $\vec{\lambda}$ to the facet $F_D$,
       \begin{equation}\label{PinMeas1}
      d_1(\vec{\lambda}) \equiv \frac{D(\vec{\lambda})}{\|\vec{\kappa}\|_{\infty}}\,,
      \end{equation}
      where $\|\vec{\kappa}\|_{\infty} \equiv \max_{1\leq i\leq d'}(|\kappa_i|)$.
\end{enumerate}
Obviously, all three measures $d_k, k=D,1,2$ are equivalent in the sense $d_k(\vec{\lambda}) < d_k(\vec{\mu})$ $\Rightarrow d_l(\vec{\lambda}) < d_l(\vec{\mu})$
for any other measure $d_l, l=1,2,D$. In addition, they depend linearly on each other.

Deciding whether one of those three measures is the most natural one or whether there is even a more appropriate fourth one is not possible. The reason for this is the following. An optimal measure should be maximally adapted to the physical relevance of quasi-pinning. The candidate for the physical relevance is the expected relation between quasi-pinning and the simplified structure of the corresponding $N$-fermion state $|\Psi_N\rangle$ (recall Thm.~\ref{thm:BDSelRulestable} and Conj.~\ref{conj:SelRulestable}): Whenever $\vec{\lambda}$ is quasi-pinned some structural overlap $\|P \Psi_N\|_{L^2}$ (recall e.g.~Conj.~\ref{conj:SelRulestable}) is expected to be close to one. However, we cannot yet prove this and in particular we do not know the optimal upper bound for $1-\|P \Psi_N\|_{L^2}$. But this bound is nothing else but the optimal pinning measure.

In the following we choose the first measure $d_D$ and define the quasi-pinning of $\vec{\lambda}$ in the setting $\wedge^N[\mathcal{H}_1^{(d)}]$ by
\begin{equation}\label{PinMeasure}
d_{N,d}(\vec{\lambda}):=\min\left(\{d_D(\vec{\lambda})\,|\,D(\cdot)\geq 0\,\,\mbox{a gen.~Pauli constraint of} \wedge^N[\mathcal{H}_1^{(d)}]\}\right)\,.
\end{equation}

\subsection{Concept of truncation and pinning analysis}\label{sec:truncatedPinanalysis}
In this section we develop the concept of a truncated pinning analysis by combining the results and concepts from the previous sections Sec.~\ref{sec:genPCtwosettings} and Sec.~\ref{sec:PinningMeasures}.

As a preliminary step, we observe that there are important structural insights following from pinning by the Pauli exclusion principle (\ref{PauliexclusionP}). For this consider NONs $\vec{\lambda}$ of a state $|\Psi_{N'}\rangle \in \wedge^{N'}[\mathcal{H}_1^{(d')}]$, where some of them are pinned by Pauli's exclusion principle to $1$ or $0$, $\lambda_1=\ldots=\lambda_{r} = 1$, $\lambda_{d+r+1}=\ldots=\lambda_{d'} = 0$. From Eq.~(\ref{noncoef}) it follows immediately that whenever $\lambda_{j}=0$ the corresponding natural orbital $|j\rangle$ does never show up in the expansion (\ref{expansion}) of $|\Psi_{N'}\rangle$ in Slater determinants and we can omit it. In a similar way, whenever $\lambda_i = 1$ the natural orbital $|i\rangle$ is contained in \emph{every} Slater determinant. As a consequence we can restrict our pinning analysis to the reduced NONs $(\lambda_{r+1},\ldots,\lambda_{d+r})$ arising from the state
\begin{equation}
|\Psi_N\rangle = a_r\cdot\ldots\cdot a_1\,|\Psi_{N'}\rangle \,\,\in \,\wedge^{N}[\mathcal{H}_1^{(d)}]\,,
\end{equation}
where $a_i$ is the annihilation operator w.r.t.~$|i\rangle$.

From now on we assume that such a reduction of the setting was already performed and the resulting NONs $\vec{\lambda}$ are not pinned by Pauli's exclusion principle, i.e.~$0 < \lambda_i < 1\,,\,\forall i= 1,\ldots,d'$. However, we would like to assume that some NONs are quasi-pinned by the exclusion principle,
\begin{equation}
1\gg \varepsilon \equiv \max \left(\{1-\lambda_1, \ldots, 1 -\lambda_r, \lambda_{d+r+1}, \ldots, \lambda_{d'}\}\right)\,.
\end{equation}
To investigate possible (quasi-)pinning we need to check the saturation of each generalized Pauli constraint of the setting $\wedge^{N'}[\mathcal{H}_1^{(d')}]$. Let $D'(\cdot) \geq 0$ be such a constraint. According to Lem.~\ref{lemdistancemodif} there exists a constraint $D(\cdot)\geq 0$ for the truncated setting $\wedge^N[\mathcal{H}_1^{(d)}]$ with
\begin{equation}\label{distancemodif2}
D'(\vec{\lambda}) = D(\vec{\lambda}^{tr}) + O(\varepsilon)\,,
\end{equation}
where $\vec{\lambda}^{tr} \equiv (\lambda_{r+1},\ldots,\lambda_{d+r})$ and the term $O(\varepsilon)$ is given by $ \sum_{i=1}^{r}\,\kappa_i (1-\lambda_i) + \sum_{j=d+r+1}^{d'}\,\kappa_j \lambda_j$. From that relation it is clear how to perform a truncated pinning analysis.
Given NONs $\vec{\lambda}'$ for which we would like to investigate possible quasi-pinning. The strategy for that is given by the following consecutive steps.
\begin{enumerate}
\item Omit all NONs, which are pinned by Pauli's exclusion principle, i.e the $1$'s and $0$'s. Let us denote the reduced set of NONs obtained in that way by $\vec{\lambda}$\,.
\item Choose appropriate $r,s \geq 0$, such that
\begin{equation}
\varepsilon \equiv \max(1-\lambda_1,\ldots,1-\lambda_r, \lambda_{d'-s+1},\ldots,\lambda_{d'})
\end{equation}
is sufficiently small. Whether one can find such $r,s>0$ does not only depend on $\vec{\lambda}'$, but also on the concrete purpose of the pinning analysis, which defines the scale of ``sufficiently small''.
\item Investigate possible pinning of the truncated vector $\vec{\lambda}^{tr} \equiv (\lambda_{r+1},\ldots,\lambda_{d+r})$ by the generalized Pauli constraints of the truncated setting $\wedge^N[\mathcal{H}_1^{(d)}]$, where $N= N'-r$ and $d=d'-r-s$.
\item The quasi-pinning of strength $d_{N,d}(\vec{\lambda}^{tr})$ (recall (\ref{PinMeasure})) found in the truncated analysis translates to the quasi-pinning in the correct untruncated setting $\wedge^{N'}[\mathcal{H}_1^{(d')}]$ by
    \begin{equation}
    d_{N',d'}(\vec{\lambda}) = d_{N,d}(\vec{\lambda}^{tr}) + O(\varepsilon)\,.
    \end{equation}
\end{enumerate}

\begin{rem}\label{rem:PinAnalyrem1}
Exact pinning found in a truncated setting can vanish in the correct setting due to any arbitrarily small truncation error $\varepsilon$. Hence, exact pinning can only be excluded, but never confirmed in a truncated analysis.
\end{rem}

\begin{ex} As an example for Remark \ref{rem:PinAnalyrem1} consider a spectrum $\vec{\lambda} = (\lambda_i)_{i=1}^9$ and assume that the truncated spectrum
$\vec{\lambda}^{tr} = (\lambda_i)_{i=1}^8$ saturates the constraint (\ref{D8.21}),
\begin{equation}
D^{(3,8)}_{21}(\vec{\lambda}^{tr}) = -\lambda_1+\lambda_3+\lambda_4 +\lambda_5-\lambda_8=0
\end{equation}
and $\lambda_9= \varepsilon >0$.
In the correct setting $\wedge^3[\mathcal{H}_1^{(9)}]$ one natural extension of constraint (\ref{D8.21}) is given by (see \cite{Altun})
\begin{equation}
D^{(3,8)}_{21}(\vec{\lambda}) = -\lambda_1+\lambda_3+\lambda_4 +\lambda_5-\lambda_8+2\lambda_9\geq0.
\end{equation}
Since the truncation error $\varepsilon$ is finite $D^{(3,8)}_{21}(\vec{\lambda})$ is not zero.
\end{ex}

\begin{rem}\label{rem:PinAnalyrem2}
A statement of the form (\ref{distancemodif2}) is very subtle. In principle even for very small $\varepsilon$ one may find $O(\varepsilon) = O(1)$. However since this requires that the corresponding coefficients $\kappa_j$, $j=1,\ldots,r,d+r+1,\ldots,d'$ are much larger than all the other ones this is not plausible at all. Here we can see that a better understanding of the derivation of generalized Pauli constraints is preferable and necessary to rule out with absolute certainty such a strange behavior of the coefficients $\kappa_i$ (\ref{generalizedPC}).
\end{rem}

\chapter{Pinning Analysis for Specific Physical Systems}\label{chap:Physics}
\section{Motivation and summary}\label{sec:MotSumPhys}
 In the previous chapters we have learned that the antisymmetry of $N$-fermion quantum states under particle exchange implies not only Pauli's famous exclusion principle but also leads to additional, even stronger restrictions on the natural occupation numbers (NONs) $\vec{\lambda}$, the eigenvalues of the $1$-particle reduced density operator ($1$-RDO). The mapping of $N$-fermion quantum states $|\Psi_N\rangle$ to their NONs is illustrated in Fig.~\ref{fig:mapnon}.
\begin{figure}[h!]
\includegraphics[width=10cm]{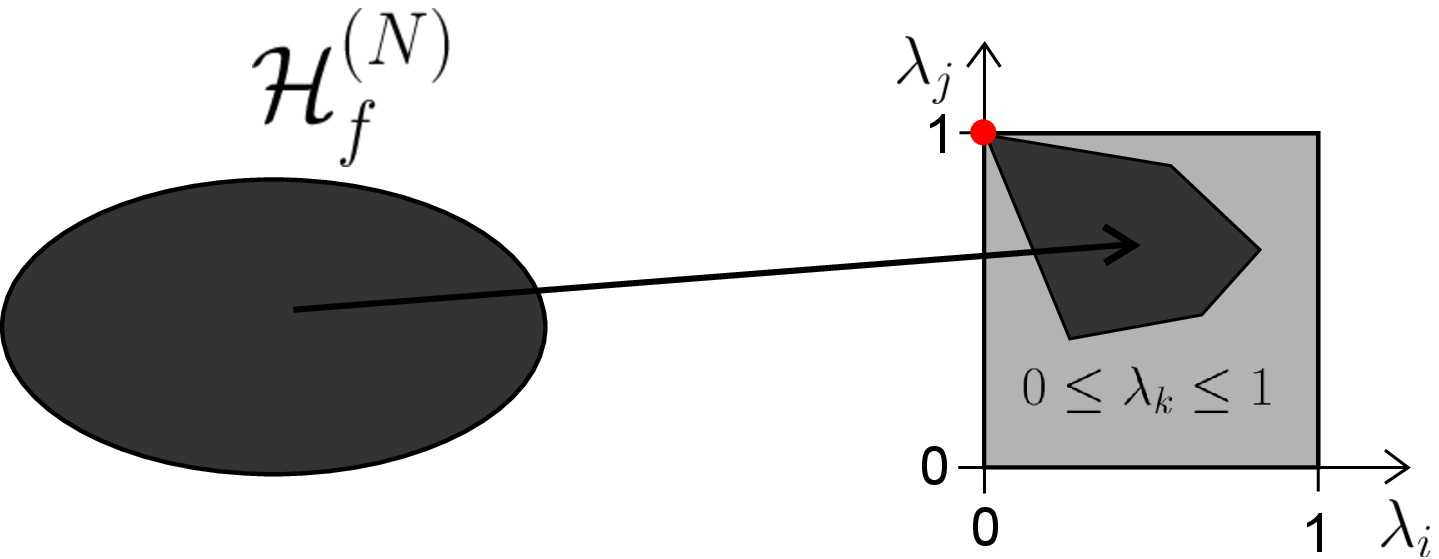}
\centering
\captionC{Schematic illustration of how the family of antisymmetric $N$-particle states $|\Psi_N\rangle$ maps to their vectors $\vec{\lambda}$ of natural occupation numbers forming a polytope (dark-gray), a proper subset of the light-gray hypercube describing Pauli's exclusion principle.
Single Slater determinants are mapped to the Hartree-Fock point $\vec{\lambda}_{HF}\equiv (1,\ldots,1,0,\ldots)$, a vertex (red dot) of that polytope.}
\label{fig:mapnon}
\end{figure}
According to Pauli's exclusion principle, which can be formulated as
\begin{equation}\label{PauliConstraint}
	0\leq \lambda_i \leq1\,,
	\forall i
\end{equation}
only natural occupation numbers $\vec{\lambda}\equiv (\lambda_1,\ldots,\lambda_d)$ inside of the high-dimensional ``Pauli hypercube'' are possible. However, there are stronger restrictions, so-called generalized Pauli constraints, conditions of the form,
\begin{equation}
\kappa_0 + \kappa_1 \lambda_1+\ldots + \kappa_d \lambda_d \geq 0\qquad,\,\kappa_j \in \Z\,,
\end{equation}
giving rise to a proper subset, the dark-gray polytope. Only those $\vec{\lambda}$ inside of that polytope are mathematically possible. In this chapter we explore the relevance of those generalized Pauli constraints for concrete physical systems.

As a central step we investigate how fermionic quantum systems behave from this new viewpoint. Where do the NONs of specific states, as e.g.~ground states of concrete physical systems lie? For fermionic systems with zero interaction the ground state can always be written as a single Slater determinant and the corresponding NONs $\vec{\lambda}=(1,\ldots,0,\ldots)$ lies at a vertex of the polytope, the so-called Hartree-Fock point (red dot in Fig.~\ref{fig:mapnon}). What happens if we switch on an interaction between the fermions by increasing a coupling parameter $\delta$? Does $\vec{\lambda}(\delta)$, which moves away from the Hartree-Fock point, still lie on the boundary of the polytope or does it move to the middle? Gaining some insights into these questions analytically is very difficult because we do not only need to calculate e.g.~the ground state of an interacting $N$-fermion model but also integrate out $N-1$ fermions to get the $1$-RDO and diagonalize it.

For a model of few harmonically coupled spinless fermions in one dimension confined to a harmonic trap all those three steps are feasible. For finite interaction $\delta$ we find that the NONs are not on the boundary of the polytope anymore. However, it is quite surprising that the distance to the boundary is very small: In the regime of weak interaction we find for that distance the algebraic dependence $\delta^8$ and even for medium interaction strengths the distance is of order $10^{-6}-10^{-9}$. This is a new phenomenon, which we call \emph{quasi-pinning}. It is highly physical relevant because it allows to reconstruct the structure of the corresponding $N$-particle state. Moreover by investigating the first few excited states of that model we identify this as an effect of the lower lying energy eigenstates. Hence, quasi-pinning may originate from the strong conflict of the antisymmetry (responsible for the existence of the generalized Pauli constraints) and the energy minimization in the sense that skipping the antisymmetry e.g.~for the energy minimization for the ground state would lead to much lower energies.

As a second model we study the $1$-band Hubbard model in one dimension with a few electrons on a few lattice sites. For the setting of three electrons on three sites we find exact pinning for a certain regime of the on-site interaction $u$. This means that the system by changing $u$ undergoes a kind of ``phase transition'' from exact pinning to no pinning. Although we find the same behavior for two additional settings with four sites, the generalized Pauli constraints seem not to play any role for Hubbard systems with more than four sites. However, we learn that the unexpected effect of \emph{exact pinning} is possible whenever the physical model is rather simple and exhibits sufficiently many symmetries as e.g.~the translational invariance of the Hubbard model.
%
%We finalize by starting to explore the generic character of strong ground state quasi-pinning of trapped fermionic systems in the regime of not too strong interacting by considering more general interaction forms than the harmonic one. By discretizing such models and resorting to numerical methods we find first evidence that strong quasi-pinning is indeed generic and does not depend on the concrete interaction form.

\section{N-Harmonium}\label{sec:NHarm}
\subsection{Model and $N$-particle energy states}\label{sec:model}
Exploring possible pinning analytically for concrete physical systems is very difficult because we do not only need to calculate the ground state of an interacting $N$-fermion system but also integrate out $N-1$ fermions to get the $1$-RDO and diagonalize it. However, there exists at least one model
with a continuous configuration space for which these three steps are possible: The $N$-Harmonium\footnote{It is also known as Hook or Moshinsky atom and plays an important role e.g.~in quantum chemistry. Since it can completely be solved analytically, several abstract theoretical concepts and optimization methods can be probed with that model \cite{Nagy2011a, Nagy2011b, Nagy2012, Nagy2013}}, a model of harmonically coupled spinless particles confined to a harmonic trap. The corresponding Hamiltonian of its one-dimensional version in spatial representation is given by
\begin{equation}\label{HamHarmonium}
H_N^{(X)} = \sum_{i=1}^N\,\left(\,\frac{p_i^2}{2m}+\frac{1}{2}m\omega^2 x_i^2\,\right) + \frac{1}{2} D\,\sum_{i,j=1}^N \, (x_i - x_j)^2 \,.
\end{equation}
Since the concept of generalized Pauli constraints exists only for fermions, it is necessary to restrict this model of $N$ identical spinless particles to the Hilbert space $\mathcal{H}_N^{(f)}$ of fermions, the $N$-particle states, which are antisymmetric under particle exchange. However, it will prove instructive to solve first the energy eigenvalue problem of (\ref{HamHarmonium}) on the larger $N$-particle Hilbert space $\mathcal{H}_N \equiv L^2(\mathbb{R})^{\otimes^N}$, without any symmetry constraints. Here $L^2(\mathbb{R})$ denotes the $1$-particle Hilbert space of square-integrable functions on $\mathbb{R}$. With such a solution at hand we can construct all fermionic eigenstates. This is true since $H_N^{(X)}$ commutes with any permutation $P$ of the $N$ identical particles,
\begin{equation}
[H_N^{(X)},P]=0\,.
\end{equation}
In particular, since the antisymmetrizing operator $\mathcal{A}_N$ is built up by a linear combination of permutation operators (recall Eq.~\ref{antisymop}), the same also holds for $\mathcal{A}_N$,
\begin{equation}\label{comantisym}
[H_N^{(X)},\mathcal{A}_N]=0\,.
\end{equation}
Property (\ref{comantisym}) allows us to simultaneously diagonalize $H_N^{(X)}$ and $\mathcal{A}_N$ on the total Hilbert space $\mathcal{H}_N$ and thus the family of all fermionic eigenstates can be obtained by applying $\mathcal{A}_N$ to each $N$-particle eigenstate of $H_N^{(X)}$. In the following we will adapt this strategy and start by solving the $N$-particle eigenvalue problem of $H_N^{(X)}$ on $\mathcal{H}_N$.

First, to assure the existence of bound states we choose $D>-\frac{m \omega^2}{N}$.
%Here and in the following we use the indices and superscripts $x_i,y_j, X,Y $ to label the spatial coordinates of the system and differential operators referring to them and the symbol $T$ for the orthogonal conjugate of a matrix or vector.
 With the short hand notation $\vec{x} \equiv (x_1,\ldots,x_N)^{T}$, the kinetic energy operator $T^{(X)}$ and the coupling matrix
\begin{equation}
\mathcal{D}= (m {\omega}^2 + N D) \mathds{1}_N - \begin{pmatrix} D & \ldots &D\\ \vdots && \vdots\\ D & \ldots & D  \end{pmatrix}
\end{equation}
the Hamiltonian (\ref{HamHarmonium}) becomes
\begin{equation}
H_N^{(X)} = T^{(X)} + \frac{1}{2} \vec{x}^{T} \mathcal{D} \vec{x}\,.
\end{equation}
The eigenvalue problem
\begin{equation}\label{eigenvalueproblem}
\left(H_N^{(X)} \Phi_N\right)(\vec{x}) = E \Phi_N(\vec{x})
\end{equation}
can easily be solved by diagonalizing $\mathcal{D}$.
As eigenvectors we find
\begin{eqnarray}\label{eigenvecD}
\vec{e}_1 &=& \frac{1}{\sqrt{N}}(1,\ldots,1)^T \\
%e_2 & := & \frac{1}{\sqrt{2}}(1,-1,0,\ldots,0)^{T} \\ \nonumber
%e_3 &:=& \frac{1}{\sqrt{6}}(1,1,-2,0,\ldots,0)^{T}  \\ \nonumber
%\vdots &&\vdots  \\ \nonumber
\vec{e}_k  &= & \frac{1}{\sqrt{k(k-1)}}(\underbrace{1,\ldots,1}_{k-1},-(k-1),0,\ldots)^{T} \,\,\,\,,k=2,\ldots,N. \nonumber
\end{eqnarray}
The first one, $\vec{e}_1$ describes the center of mass motion. Consequently, the corresponding eigenvalue is $d_-:= m {\omega}^2$. The remaining eigenvectors $\vec{e}_k$, $k=2,\ldots,N$ describe the relative motion between the $N$ particles and their eigenvalues are degenerate and given by $d_+:= m {\omega}^2 + N D$. We also introduce the corresponding frequencies $\omega_{\pm} \equiv  \sqrt{\frac{d_{\pm}}{m}}$. The real and orthogonal matrix
\begin{equation}\label{Smatrix}
S := (\vec{e}_1,\vec{e}_2,\ldots,\vec{e}_N)
\end{equation}
then diagonalizes $\mathcal{D}$,
\begin{equation}
S \mathcal{D} S^{T} = \mbox{diag}(d_-,d_+,\ldots,d_+) =:\mathcal{D}_{0} \,.
\end{equation}
Thus, by introducing $\vec{y}\equiv (y_1,\ldots,y_N)^{T}:= S \vec{x}$ and
\begin{equation}\label{HamHarmoniumY}
H_N^{(Y)}:= T_N^{(Y)}+\frac{1}{2} \vec{y}^{T} \mathcal{D}_0 \vec{y} \,,
\end{equation}
eigenvalue problem (\ref{eigenvalueproblem}) is equivalent to (use $\vec{x}^T \mathcal{D} \vec{x} = (S \vec{x})^T \mathcal{D}_0 (S \vec{x})= \vec{y}^T \mathcal{D}_0 \vec{y} $ and $T_N^{(X)}=T_N^{(Y)}$)
\begin{equation}\label{eigenvalueproblem2}
\left(H_N^{(Y)} \Psi_N\right)(\vec{y}) = E \Psi_N(\vec{y}) \,,
\end{equation}
with $\Psi_N(\vec{y}(\vec{x})) = \Phi_N(\vec{x})$.
An alternative diagonalization of $\mathcal{D}$ and decoupling of the $N$ coordinates $x_i$, respectively, can be obtained by introducing the Jacobian coordinates as done in \cite{harmOsc2012}.

Since (\ref{HamHarmoniumY}) describes $N$ decoupled one-dimensional harmonic oscillators, the spectrum of (\ref{eigenvalueproblem2}) is well-known as $\forall \bd \nu \equiv (\nu_1,\ldots,\nu_N) \in \NN^N$
\begin{equation}\label{energies}
E_{\bd \nu} =  \hbar \omega_- \,(\nu_1 + \frac{1}{2}) + \hbar \omega_+\,\sum_{i=2}^{N} (\nu_i+ \frac{1}{2}) \,,
\end{equation}
with corresponding eigenstate
\begin{equation}\label{eigenstatesY}
\Phi_{\bd \nu}(\vec{y}) = \varphi_{\nu_1}^{(l_-)}(y_1) \prod_{i=2}^N \varphi_{\nu_i}^{(l_+)}(y_i),
\end{equation}
where $\varphi_{\nu}^{(l)}(y) $ is the $\nu$-th Hermite function, an eigenfunction of a single harmonic oscillator with natural length scale $l$, i.e.
\begin{equation}\label{Hermitefunc}
\varphi_{\nu}^{(l)}(y) = \pi^{-\frac{1}{4}} l^{-\frac{1}{2}} (2^{\nu} {\nu}!)^{-\frac{1}{2}} H_{\nu}(\frac{y}{l}) \mbox{e}^{-\frac{y^2}{2 l^2}}
\end{equation}
and $H_{\nu}$ is the $\nu$-th Hermite polynomial. The natural length scales for the center of mass motion and the relative motion, $l_-$ and $l_+$, respectively, are given by $l_{\pm}= \sqrt{\hbar/(m \omega_{\pm})}$. They are related to the coupling constants $D$ and $m \omega^2$ by
\begin{equation}\label{LvsHook}
\kappa\equiv \frac{N D}{m \omega^2} = \left(\frac{l_-}{l_+}\right)^4-1\,,
\end{equation}
where the dimensionless quantity $\kappa$ is a measure for the interaction strength.
For macroscopic particle numbers one should rescale $D$ by N, i.e.\ $D \rightarrow D/N$, in order that the energy per particle is of order one in $N$.

According to (\ref{energies}) and (\ref{eigenstatesY}) the ground state of (\ref{HamHarmonium}) is characterized by  $\nu_{i}=0$ for $i=1,2,...,N$.
Moreover, by using the orthogonal character of the transformation matrix $S$ (\ref{Smatrix}) and by reintroducing the physical coordinates $x_i$ we find
\begin{eqnarray}\label{gsbosons}
\Phi_{0}(\vec{x})& =& \Phi_{0,\ldots,0}(\vec{y}(\vec{x})) \nonumber \\
&=&\mathcal{N}  \exp{(-\frac{1}{2 l_-^{\,2}} y_1(\vec{x})^2-\frac{1}{2 l_+^{\,2}} \sum_{k=2}^N y_k(\vec{x})^2)} \nonumber \\
&=& \mathcal{N}  \exp{(-\frac{1}{2 N}\left(\frac{1}{l_-^{\,2}} - \frac{1}{l_+^{\,2}}\right) (x_1+\ldots+x_N)^2 - \frac{1}{2 l_+^{\,2}} \vec{x}^2)} \nonumber \\
&=:& \mathcal{N} e^{-A\vec{x}^2 + B_N (x_1+\ldots+x_N)^2} \,
\end{eqnarray}
where $\mathcal{N}$ is the normalization factor and
\begin{equation}\label{parameterAB}
A \equiv \frac{1}{2 l_+^{\,2}}    \,\,\,,\,B_N\equiv \frac{1}{2N}\left(\frac{1}{l_+^{\,2}} - \frac{1}{l_-^{\,2}}\right)\,.
\end{equation}
Note that for zero interaction $B_N$ vanishes since $l_-=l_+$.
Since (\ref{gsbosons}) is symmetric in the variables $x_i$, $\Phi_0(\vec{y}(\cdot))$ is also the ground state in the bosonic Hilbert space $\mathcal{H}_N^{(b)}$.

In the following we will construct the fermionic ground state $\Psi_0$ of (\ref{HamHarmonium}). In principle, as explained at the beginning of this section this can be done by projecting the eigenstates (\ref{eigenstatesY}) to the fermionic subspace $\mathcal{H}_N^{(f)}$. However, since it will turn out that the first few $N$-particle eigenstates have zero weight in $\mathcal{H}_N^{(f)}$, a more constructive approach is necessary.

As a first step notice that the $N$-particle eigenstates (\ref{eigenstatesY}) form a complete orthonormal basis of $\mathcal{H}_N$. Moreover, according to the remarks at the beginning of this section, in particular recall (\ref{comantisym}), the fermionic ground state $\Psi_{0}$ can be expanded as a linear combination of finitely many states $\Phi_{\bd \nu}(\vec{y}(\vec{x}))$, all belonging to the \textbf{same} energy eigenspace. According to (\ref{energies}) these finite-dimensional subspaces $S_{n_-,n_+}$ are described by just two integers $n_-,n_+\in \mathbb{N}_0$, the number of center of mass excitations $n_-$ and the number of relative excitations $n_+$. By defining the following index set
\begin{equation}\label{indsetNU}
J_{n_-,n_+} := \{\bd{\nu}\in \NN^N\,|\,\nu_1=n_-\,,\,\nu_2+\ldots+\nu_N=n_+ \}
\end{equation}
we have
\begin{equation}\label{eigenspaces}
S_{n_-,n_+} = \mbox{span}(\{\Phi_{\bd \nu}(\vec{y}(\cdot)) \,| \,\bd \nu \in J_{n_-,n_+}\})\,.
\end{equation}
Let $(n_-^{(f)},n_+^{(f)})$ be the pair of excitations such that the fermionic ground state $\Psi_0$ lies in $S_{n_-^{(f)},n_+^{(f)}}$.
Then, by using (\ref{eigenstatesY}) and (\ref{Hermitefunc}) we find
\begin{eqnarray}\label{gsfermionicexpan}
\Psi_0(\vec{x}) &=& \sum_{\bd \nu \in J_{n_-,n_+}} c_{\bd \nu} \, \Phi_{\bd \nu}(\vec{y}(\vec{x}))\\
&=&  \sum_{\bd \nu \in J_{n_-,n_+}} d_{\bd \nu} \, \left(H_{\nu_1}\big(\frac{y_1(\vec{x})}{l_-}\big)  \prod_{j=2}^N H_{\nu_j}\big(\frac{y_j(\vec{x})}{l_+}\big)\right)\, e^{-A\vec{x}^2 + B_N (x_1+\ldots+x_N)^2} \nonumber\,.
\end{eqnarray}
where the coefficients $\{c_{\bd \nu}\}$ and $\{d_{\bd \nu}\}$, respectively, are such that $\mathcal{A}_N \Psi_0 = \Psi_0$. Since the exponential function in the last line is symmetric, the polynomial in front of it needs to be antisymmetric. What is the minimal degree of a polynomial, which is antisymmetric in $N$ variables? The answer follows from
\begin{lem}\label{lempolylowest}
Applying the antisymmetrizing operator $\mathcal{A}_N$ to a monomial $m_{\bd \nu}(\vec{x})\equiv \prod_{j=1}^N\, x_j^{\nu_j}$ yields zero whenever its degree fulfills $\nu_1+\ldots+\nu_N < \binom{N}{2}$. Moreover there exist an antisymmetric polynomial of degree $\binom{N}{2}$, unique up to a global factor, the Vandermonde determinant
\begin{equation}\label{Vandermonde}
\prod_{1\leq i<j\leq N}(x_i-x_j) = \left|\begin{array}{lll}\,1&\ldots&\,1\\x_1&\ldots&x_N\\ \,\vdots& &\,\vdots\\ x_1^{N-1}&\ldots& x_N^{N-1} \end{array}\right|\,.
\end{equation}
\end{lem}
\begin{proof}
Given a monomial $m_{\bd \nu}$ with degree of $\bd \nu$ smaller $\binom{N}{2}$. Then, there exists $i\neq j \in \{1,\ldots,N\}$ such that $\nu_i = \nu_j$. Hence, swapping the $i$-th and $j$-th variable in $m_{\bd \nu}(\vec{x})$ does not change the monomial, $P_{ij}m_{\bd \nu} = m_{\bd \nu}$ and we get $\mathcal{A}_N m_{\bd \nu} = \mathcal{A}_N P_{ij}m_{\bd \nu} = - \mathcal{A}_N m_{\bd \nu}$, since $\mbox{sign}(P_{ij})=-1$ and thus $\mathcal{A}_N m_{\bd \nu} = 0$. The Vandermonde determinant is indeed antisymmetric and non-zero. Moreover, we can easily verify that every antisymmetric polynomial can be written as Vandermonde determinant multiplied by a symmetric polynomial. Since there is only one such symmetric polynomial of degree zero, the constant, (\ref{Vandermonde}) is the only antisymmetric polynomial of degree $\binom{N}{2}$.
\end{proof}
In the spirit of Lem.~\ref{lempolylowest} it is clear that the fermionic ground state lies in the space with $n_-=0$, since $y_1(\vec{x})\sim (x_1+\ldots+x_N)$
treats all $N$ physical variables $x_i$ equally and cannot ``contribute'' to the antisymmetry. Since the Hermite polynomial $H_{\nu}$ has degree $\nu$ and since $\vec{y}$ depends linearly on $\vec{x}$, we find that the polynomial in (\ref{gsfermionicexpan}) in front of the exponential function has degree $n_+$ in the physical coordinates $x_i$. Lem.~\ref{lempolylowest} then implies that the fermionic ground states lies in the space $S_{0,n_+}$, with $n_+ = \binom{N}{2}$. It is unique up to a global factor and can be expressed as
\begin{equation}\label{gsfermions}
\Psi_0(\vec{x}) = \mathcal{N}\,\left(\prod_{1\leq i<j\leq N}(x_i-x_j)\right)\, e^{-A\vec{x}^2 + B_N (x_1+\ldots+x_N)^2}\,.
\end{equation}
\begin{rem}
The only difference between the fermionic (\ref{gsfermions}) and bosonic (\ref{gsbosons}) ground state of the $N$-Harmonium model (\ref{HamHarmonium}) is the Vandermonde determinant. Consequently, all differences between the bosonic and fermionic physics of that model are hidden in that polynomial. Moreover, (\ref{gsfermions}) has a strong similarity to the Laughlin wave functions of the fractional Quantum Hall effect \cite{Laughlin1983}.
\end{rem}

As a next step we would like to calculate the $1$-RDO $\rho_1^{(f)}(x,y)$ of (\ref{gsfermions}). This means to integrate out $N-1$ fermion coordinates. Although the structure of (\ref{gsfermions}), a product of a polynomial and an exponential function with a quadratic exponent is quite simple, this turns out to be quite complicated and does not lead to an elementary analytical expression for arbitrary particle number $N$. There are two structural features of (\ref{gsfermions}), which together are responsible for that: The mixing of different particle coordinates $x_i$ in the exponent and the polynomial in front of the exponential function. Indeed, skipping the coordinate mixing in the exponent would allow to expand (\ref{gsfermions}) in a finite linear combination of Slater determinants built up by Hermite functions and the integration could be done symbolically by referring to the orthogonality of the Hermite functions.
On the other hand, skipping the polynomial would make the integration trivial after a diagonalization of the quadratic form. The latter scenario means nothing else but to deal with the bosonic ground state (\ref{gsbosons}). Since the corresponding calculation will be instructive for later considerations of the fermionic ground state, we will first calculate the bosonic $1$-RDO $\rho_1^{(b)}(x,y)$ and its NONs and NOs.

\subsection{Bosonic $1$-particle reduced density operator $\rho_1^{(b)}(x,y)$}\label{sec:1RDObosons}
The calculation of the $1$-RDO $\rho_1^{(b)}(x,y)$ for the \emph{bosonic} ground state is straightforward for arbitrary particle number $N$ (see Appendix \ref{app:bosonic}). We get
\begin{equation}\label{1RDOb}
\rho_1^{(b)}(x,y) = c_N\, \exp{\left[-a_N (x^2+y^2) +b_N x y\right]}
\end{equation}
with (recall Eq.~(\ref{parameterAB}))
\begin{eqnarray}\label{parameterab}
b_N&=&\frac{(N-1)B_N^2}{A-(N-1)B_N}\,\,,\,a_N=(A-B_N)-\frac{1}{2}b_N \nonumber \\
c_N&=&N \sqrt{\frac{2 a_N-b_N}{\pi}}\,.
\end{eqnarray}
Note that $\rho_1^{(b)}$ is normalized to the particle number $N$, i.e.\ $\int \!\mathrm{d}x \,\rho_1^{(b)}(x,x)=N$\,.
This result resembles those in Refs.~\cite{Rob,Dav,Peschel1999}. The difference to Ref.~\cite{Peschel1999} is that the coefficients $a_N, b_N$ of the bilinear exponent can be expressed explicitly by both length scales $l_-, l_+$ for \emph{all} $N$ (cf.~Eq.~(\ref{parameterAB}) and (\ref{parameterab})).

The $1$-RDO $\rho_1^{(b)}$ for bosons, Eq.~(\ref{1RDOb}), has the form of a Gibbs state in coordinate representation \cite{feyn}
\begin{eqnarray}\label{Gibbs}
\rho_1^{(b)}(x,y)&=& \frac{1}{Z_{eff}}\, \langle x|\exp{[-\beta_N H_{eff}]}|y\rangle\nonumber \\
&=& N \sqrt{\frac{1}{\pi L_N^2}\tanh{(\beta_N \hbar \Omega_N/2)}} \\
&&\cdot\exp{\left(-\frac{1}{2 L_N^2 \sinh{(\beta_N\hbar \Omega_N)}}\left[(x^2+y^2)\cosh{(\beta_N\hbar \Omega_N)}- 2 x y\right]\right)}\nonumber
\end{eqnarray}
where  $H_{eff}$ is the effective Hamiltonian for a single harmonic oscillator with mass $M_N$, frequency $\Omega_N$ and length scale $L_N = \sqrt{\frac{\hbar}{M_N  \Omega_N}}$:
\begin{equation}\label{Hamiltonianeff}
H_{eff} = \frac{1}{2} \hbar \Omega_N \left[-L_N^2\frac{\mathrm{d}^2}{\mathrm{d}x^2}+\frac{1}{L_N^2} x^2\right]\,.
\end{equation}
From Eqs.~(\ref{1RDOb}), (\ref{Gibbs}) and $\rho_1^{(b)}(x,y) = \langle x| \rho_1^{(b)}|y\rangle$ we obtain
\begin{equation}\label{GibbsOp}
\rho_1^{(b)} = \frac{1}{Z_{eff}}\,\exp{[-\beta_N H_{eff}]}\,,
\end{equation}
with
\begin{eqnarray}\label{effectiveoscill}
L_N &=& (4 a_N^2-b_N^2)^{-\frac{1}{4}} \nonumber \\
\beta_N \hbar \Omega_N &=& \arcsinh \left(\frac{1}{L_N^2 b_N}\right) \nonumber \\
%Z_{eff}&=& \left[2 N \sinh{\left(\frac{\beta_N \hbar \Omega_N}{2}\right)}\right]^{-1}\,. \nonumber
Z_{eff}&=& \frac{N}{2}\,\left[ \sinh{\left(\frac{\beta_N \hbar \Omega_N}{2}\right)}\right]^{-1}\,.
\end{eqnarray}

\noindent These quantities can also be expressed by the original parameters $l_-$ and $l_+$, only:
\begin{eqnarray}\label{parameterbhm}
\beta_N \hbar \Omega_N &=& \arcsinh{ \left[\frac{2 l_+ l_- \sqrt{[(N-1)l_+^{\,2}+l_-^{\,2}][l_+^{\,2}+(N-1)l_-^{\,2}]}}{(1-1/N)(l_+^{\,2}-l_-^{\,2})^2}\right]}\nonumber \\
L_N &=& \sqrt{l_- l_+}\,\left[\frac{(N-1)l_+^{\,2} +l_-^{\,2}}{l_+^{\,2} +(N-1)l_-^{\,2}}\right]^{\frac{1}{4}} \,.
\end{eqnarray}
Note that $L_N\rightarrow (N-1)^{-\frac{1}{4}}\,l_-\,\left(l_+/l_-\right)^{\frac{1}{2}}$, $\beta_N \hbar \Omega_N\rightarrow \left(2N/\sqrt{N-1}\right) \,l_+/l_-$ for $l_+/l_-\rightarrow 0$ corresponding to $D\rightarrow \infty$, and $L_N\rightarrow l_-$, $\beta_N \hbar \Omega_N\rightarrow  \infty$ for $l_+/l_-\rightarrow 1$, i.e.\ $D \rightarrow 0$. The result (\ref{GibbsOp}) demonstrates that the $1$-RDO can exactly be represented by the Gibbs state of an effective harmonic oscillator at a ``temperature'' $T_N = 1/(k_B \beta_N)$. That $\rho_1^{(b)}$ is a Gibbs state for an effective harmonic oscillator has already been shown in \cite{Peschel1999} for a harmonic chain with nearest neighbor interactions. Due to the permutation invariance of the harmonic potential of our model, the parameters of the effective Hamiltonian can be calculated explicitly as functions of $l_-$ and $l_+$ (see Eqs.~(\ref{parameterAB}), (\ref{parameterab}), (\ref{Hamiltonianeff}) and (\ref{effectiveoscill})). For $D=0$, i.e.\ non-interacting bosons, it follows $l_-=l_+$. For that case, the ``temperature'' is zero.

Due to the elementary form of $\rho_1^{(b)}$, see Eq.~(\ref{GibbsOp}), it is quite easy to determine the bosonic NOs $\chi_k^{(b)}(x)$ and the corresponding occupation numbers $\lambda_k^{(b)}$, which obey the eigenvalue equation
\begin{equation}
\rho_1^{(b)} \chi_k^{(b)} = \lambda_k^{(b)} \chi_k^{(b)}\,.
\end{equation}
By recalling the Hermite functions $\varphi_k^{(l)}(x)$ (see Eq.~(\ref{Hermitefunc})) we find $\chi_k^{(b)}(x) = \varphi_k^{(L_N)}(x)$.
Moreover, the NONs obey the Boltzmann law
\begin{equation}\label{NONb}
\lambda_k^{(b)} = N \left[1-\exp{(-\beta_N \hbar \Omega_N)}\right]\, e^{-(\beta_N \hbar \Omega_N)k}\,,k=0,1,\ldots
\end{equation}
It is obvious that  $\lambda_k^{(b)}$ fulfill the standard normalization $\sum_{k=0}^{\infty} \lambda_k^{(b)} = N$.

To finish the study of the bosonic ground state we analyze the bosonic NONs as function of the interaction strength.
For this we define the quantity
\begin{equation}\label{paramterqN}
q_N \equiv e^{-\beta_N \hbar \Omega_N}\,,
\end{equation}
which depends via $\beta_N \hbar \Omega_N$ (recall (\ref{parameterbhm})) on the relative fermion interaction strength $\kappa$ or alternatively (c.f.~Eq.~(\ref{LvsHook})) on the quantity $\frac{l_+}{l_-}$.

Since $(\beta_N \hbar \Omega_N)$ is invariant under swapping the length scales $l_-$ and $l_+$, we find that the NONs (\ref{NONb}) fulfill
the duality,
\begin{equation}\label{dualityb}
\lambda_k^{(b)}\big(\frac{l_+}{l_-}\big) = \lambda_k^{(b)}\big(\frac{l_-}{l_+}\big)\,,
\end{equation}
which relates the attractive ($l_- > l_+$) and the repulsive ($l_- < l_+$) interaction regime. In the asymptotic regimes of weak ($\frac{l_+}{l_-}\rightarrow 1$) and strong interaction ($\frac{l_+}{l_-}\rightarrow 0$) we find
\begin{eqnarray}
q_N &\sim &  \frac{(N-1) }{N^2}\epsilon^2 +\frac{(N-1) }{N^2} \epsilon^3   \qquad, \,\epsilon = 1- \frac{l_+}{l_-}\rightarrow 0^+\nonumber \\
q_N &\sim &  1 -\frac{2 N }{\sqrt{N-1}}\delta + \frac{2 N^2}{N-1}\delta^2   \qquad, \,\delta = \frac{l_+}{l_-}\rightarrow 0^+
\end{eqnarray}
Due to its physical relevance we still consider the relative entropy,
\begin{eqnarray}
S[\rho_1^{(b)}]&:=& \mbox{Tr}[\rho_1 \ln{\rho_1^{(b)}}] \\
&=&- \sum_{k=0}^{\infty} (1-q_N)q_N^k \left( \ln{(1-q_N)}+k \ln{q_N}\right)\\
&=& -\ln{(1-q_N)} - q_N (1-q_N) \ln{q_N} \sum_{k=0}^{\infty} q_N^{k-1}  k   \\
&=&  -\ln{(1-q_N)}- \frac{q_N}{1-q_N} \ln{q_N} \,.
\end{eqnarray}
The behavior of $q_N$ and $S[q_N]$ is visualized in Fig.~\ref{fig:qSdecay} for particle numbers $N=2,10,100$.
\\
\\
\begin{figure}[h]
\centering
\includegraphics[scale=0.72]{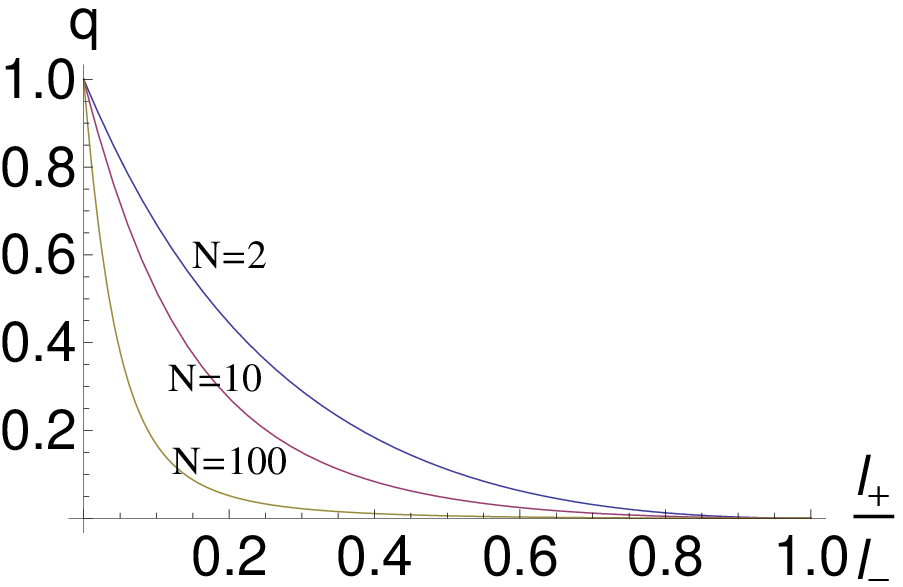}
\hspace{0.3cm}
\includegraphics[scale=0.63]{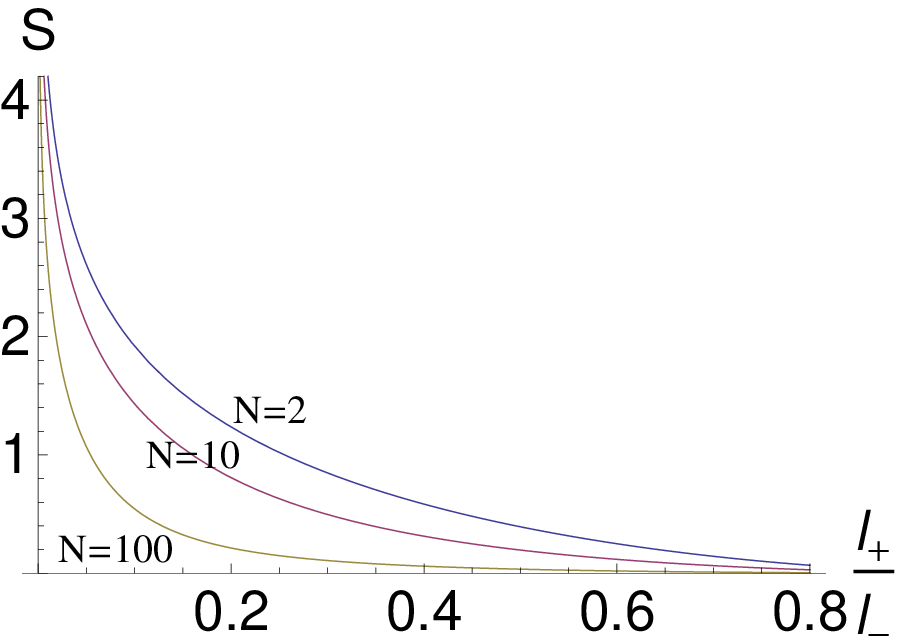}
\captionC{Parameter $q_N$ (left) and relative $1$-particle entropy $S[q_N]$ (right) as function of $\frac{l_+}{l_-}$ for the three particle numbers $N=2,10,100$.}
\label{fig:qSdecay}
\end{figure}
\\
\\
The largest relative occupation number $\frac{\lambda_0}{N}$ as function of $\frac{l_+}{l_-}$ and the decay of the NONs for the strong interaction value $\frac{l_+}{l_-}=\frac{3}{100}$ is presented in Fig.~\ref{fig:occbosons} for three particle numbers.
%Remarkable is that the largest occupancy for large particle numbers does not significantly deviates from the maximum $1$ for not too strong interactions.
\\
\\
\begin{figure}[h]
\centering
\includegraphics[scale=0.62]{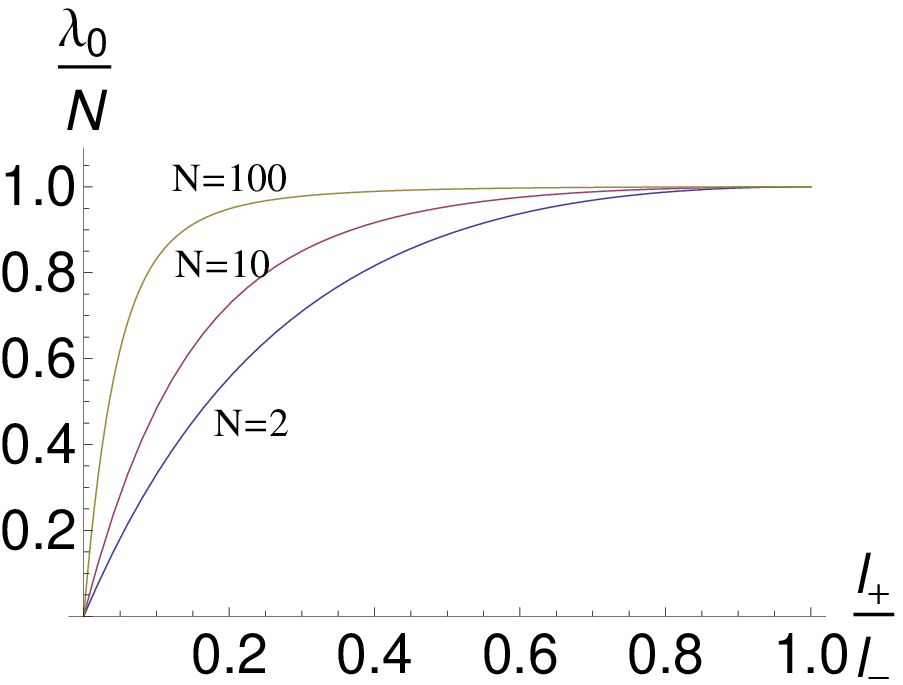}
\hspace{0.3cm}
\includegraphics[scale=0.60]{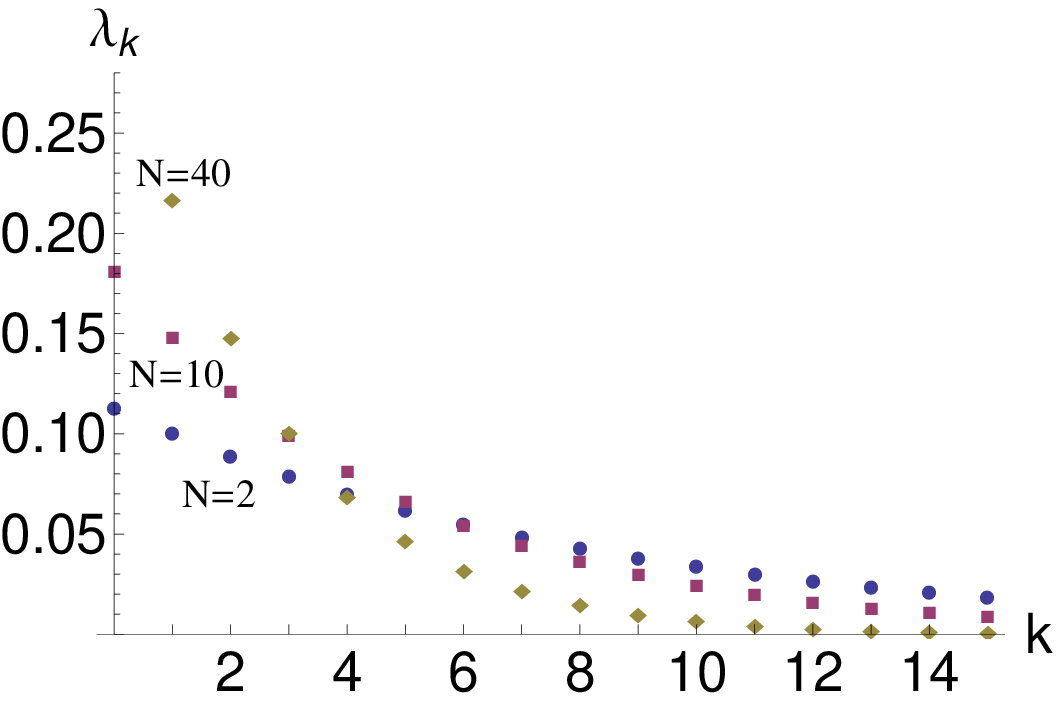}
\captionC{left: relative occupancy $\frac{\lambda_0}{N}$ as function of $\frac{l_+}{l_-}$ for the three particle numbers $N=2,10,100$; right: decay behavior of the NONs $\lambda_k$ for the interaction defined by $\frac{l_+}{l_-}=\frac{3}{100}$ for the three particle numbers $N=2,10,40$.}
\label{fig:occbosons}
\end{figure}

\subsection{Fermionic $1$-particle reduced density operator $\rho_1^{(f)}$}\label{sec:1RDOfermions}
In this section we first determine the fermionic $1$-RDO $\rho_1^{(f)}(x,y)$ in spatial representation. In a second step
we represent $\rho_1^{(f)}$ as matrix w.r.t.~to a given $1$-particle reference basis. Since an analytic diagonalization of
$\rho_1^{(f)}$ turns out to be not possible anymore, such a matrix form is the starting point for numerical and perturbation theoretical
approaches used in the following sections to determine the corresponding NONs.

As already pointed out in Sec.~\ref{sec:model} it is impossible to derive an elementary analytical expression for $\rho_1^{(f)}(x,y)$. In the Appendix \ref{app:fermionic} the integral
\begin{equation}
\rho_1^{(f)}(x,y) = \int\mathrm{d}x_2\ldots\mathrm{d}x_N\,\Psi_0(x,x_2,\ldots,x_N)\Psi_0(y,x_2,\ldots,x_N)
\end{equation}
is simplified for arbitrary $N$ to a finite sum of single integrals.
Again, as for the $N$-particle ground states, the exponential part of the fermionic $1$-RDO coincides with the bosonic one. The Vandermonde determinant in front of the exponential term in Eq.~(\ref{gsfermions}) leads to an additional symmetric polynomial $F_N(x,y)$ of degree $2(N-1)$ and with only even order monomials (see Appendix \ref{app:NOsfermionic}):
\begin{equation}\label{1RDOfpoly}
F_N(x,y) = \sum_{\nu=0}^{N-1} \sum_{\mu=0}^{2 \nu} \,c_{\nu,\mu}\, x^{2\nu-\mu} y^{\mu}\,.
\end{equation}
The coefficients $c_{\nu,\mu}$ depend on the model parameters and fulfill $c_{\nu,\mu} = c_{\nu,2\nu-\mu}$. Accordingly, we have
\begin{equation}\label{1RDOf}
\rho_1^{(f)}(x,y) = F_N(x,y)\, \exp{\left[-a_N (x^2+y^2) +b_N x y\right]},
\end{equation}
which is again normalized to $N$. The expression for the coefficients $c_{\nu,\mu}$  is rather cumbersome (see Eq.~(\ref{1RDOfUint})). The number of terms contributing to $c_{\nu,\mu}$ increases with increasing $N$.
As an example we present the explicit result for $N=3$:
\begin{eqnarray}
F_3(x,y) &=& d_3 \big[ C_1 (x^4+y^4)+C_2 (x^3 y+x y^3)+C_3 x^2 y^2  \nonumber \\
&&+ C_4 (x^2+y^2) +C_5 x y + C_6\big]
\end{eqnarray}
with
\begin{eqnarray}
C_1&=&\frac{1}{24} \big(96 A^4 B_N^2-480 A^3 B_N^3+600 A^2 B_N^4\big) \nonumber \\
C_2&=&\frac{1}{6} \big(-96 A^5 B_N+720 A^4 B_N^2-1824 A^3 B_N^3 +1560 A^2 B_N^4\big)\nonumber \\
C_3&=&\frac{1}{4} \big(64 A^6-640 A^5 B_N+2464 A^4 B_N^2-4320 A^3 B_N^3 +2904 A^2 B_N^4\big)\nonumber \\
C_4&=&\frac{1}{2} \big(-8 A^5+72 A^4 B_N-264 A^3 B_N^2+460 A^2 B_N^3 -312 A B_N^4\big)\nonumber \\
C_5&=&8 A^5-48 A^4 B_N+72 A^3 B_N^2+44 A^2 B_N^3-120 A B_N^4\nonumber \\
C_6&=&3 A^4-24 A^3 B_N+75 A^2 B_N^2-108 A B_N^3+60 B_N^4\nonumber \\
d_3&=& \frac{\sqrt{A^2-3 A B_N}}{\sqrt{2 \pi } \left(A-2 B_N\right){}^{9/2}}\,.
\end{eqnarray}

Due to the involved form of (\ref{1RDOf}) we cannot analytically diagonalize the $1$-RDO for fermions anymore,
even not for $N=2,3$. As a preparation for numerical and perturbation theoretical approaches we represent
$\rho_1^{(f)}$ as an infinite matrix. Since the exponential factor of (\ref{1RDOf}) coincides with that for the
bosons, we choose the bosonic NOs, Hermite functions with natural length scale $L_N$ ((\ref{effectiveoscill})),
found analytically in the previous section, Sec.~\ref{sec:1RDObosons}, as reference basis. The matrix elements are then given by
\begin{eqnarray}\label{1RDOfmat1}
\left(\rho_1^{(f)}\right)_{nm}  &\equiv&  \langle \varphi_{n}^{(L_N)}, \rho_1^{(f)} \varphi_{m}^{(L_N)} \rangle \nonumber \\
&=& \int \mathrm{d}x \mathrm{d}y\, \varphi_{n}^{(L_N)}(x) \rho_1^{(f)}(x,y) \varphi_{n}^{(L_N)}(y)\,.
\end{eqnarray}
Surprisingly, it turns out that for not too strong interaction this matrix is approximately diagonal, i.e.~the choice of
$1$-particle reference states was an excellent guess for the fermionic NOs. This will be investigated in more detail
in Sec.~\ref{sec:NO}.

Since the function $F_N$ in $\rho_1^{(f)}(x,y)$ has finite degree and due to the explicit diagonalization of the exponential
factor (recall (\ref{1RDOb}) and (\ref{GibbsOp})) we can find an analytic expression for $\left(\rho_1^{(f)}\right)_{nm}$ for
arbitrary $n,m \in \NN$. We explain the main steps here and present details in the Appendix \ref{app:fermionicmatrix}.
By rescaling the coefficients of the polynomial $F_N$ from Eq.~(\ref{1RDOfpoly}) and using (\ref{GibbsOp}) we find
\begin{equation}\label{1RDOfMehler}
\rho_1^{(f)}(x,y) = \tilde{F}_N\left(\frac{x}{L_N},\frac{y}{L_N}\right)\, \sum_{k=0}^\infty q_N^k \varphi_k^{(L_N)}(x)\varphi_k^{(L_N)}(y)\,,
\end{equation}
where the parameter $q_N$ is given by Eq.~(\ref{paramterqN}).
As a next step we introduce the standard ladder operators $a_x^{(\dagger)},a_y^{(\dagger)}$ for the harmonic oscillator acting here w.r.t.~the coordinates $\frac{x}{L_N}$ and $\frac{y}{L_N}$, respectively. Plugging in this in Eq.~(\ref{1RDOfmat1}) and using the orthonormality of the Hermite functions $\varphi_k^{(L_N)}$ allows to analytically calculate the matrix elements $\left(\rho_1^{(f)}\right)_{nm}$. Since $F_N$ contains only even order monomial, $\left(\rho_1^{(f)}\right)_{nm}$ vanishes whenever $n+m$ is odd. Since $F_N$ has degree $2(N-1)$, the matrix $\rho_1^{(f)}$ is in addition sparse and has only non-zero entries when $|n-m| \leq 2(N-1)$. The analytical expressions for $\left(\rho_1^{(f)}\right)_{nm}$ are highly complex and are present in the Appendix \ref{app:fermionicmatrix} for the case of three particles.

Moreover, since $\left(\rho_1^{(f)}\right)_{nm} \sim q_N^n$ for $n\approx m$ and $q_N\ll1$ for weak interaction, the matrix elements strongly decay with increasing $n$ and $m$. The numerical approach we use to determine the NONs is based on that fact. For specific interaction strengths $\delta$  we
truncate $\rho_1^{(f)}$ at $n = m = d$ and then diagonalize it numerically. This is then fully justified since the corresponding error is of the order $q_N^{d}$.

\subsection{Fermionic natural occupation numbers}\label{sec:NONsfermionic}
For $N=3$ and several interaction strengths we calculate numerically the NONs $\vec{\lambda}^{(f)}$ by following the strategy explained in the previous section. Intriguingly, we find a duality,
\begin{equation}\label{duality}
\vec{\lambda}^{(f)}\big(\frac{l_-}{l_+}\big) = \vec{\lambda}^{(f)}\big(\frac{l_+}{l_-}\big),
\end{equation}
which connects the regime of attractive interaction ($l_+<l_-$) with that of repulsive interaction ($l_+ >l_-$). (\ref{duality}) suggests the definition\footnote{The function $\log(\cdot)$ stand always for the natural logarithm. If we change the base from $e$ to $b$ we write $\log_b(\cdot)$.}
\begin{equation}\label{delta}
\delta := -\log{\big(\frac{l_+}{l_-}\big)}\,.
\end{equation}
The NONs $\lambda^{(f)}_i$ as function of the interaction strength $\delta$ are then symmetric functions. In the present and remaining three subsections the superscript $f$ will be suppressed.
In Tab.~\ref{tab:NONs3Harm} numerical results for the first nine digits are presented for the five coupling strength $\delta=0.2,0.4,0.6,0.8,1.0$, which corresponds to $3D/m\omega^2 = 1.23,3.95,10.0,23.5,53.6$.
\begin{table}[h]
\centering
$\begin{array}{|c|r|r|r|r|r|}
\hline
\delta & 0.2 \qquad& 0.4 \qquad& 0.6 \qquad& 0.8 \qquad& 1.0 \qquad \\
\hline
\lambda_1 & 0.99999655 & 0.99979195 & 0.99788745 & 0.99008995 & 0.97039917 \\
\lambda_2 & 0.99966393 & 0.99541159 & 0.98161011 & 0.95575302 & 0.91814283 \\
\lambda_3 & 0.99966062 & 0.99523526 & 0.98014300 & 0.95053687 & 0.90742159 \\
\lambda_4 & 0.00033932 & 0.00475175 & 0.01963124 & 0.04811132 & 0.08809330 \\
\lambda_5 & 0.00033608 & 0.00458920 & 0.01837054 & 0.04380671 & 0.07883271 \\
\lambda_6 & 3.416\cdot 10^{-6} & 0.00020069 & 0.00196597 & 0.00886269 & 0.02537522 \\
\lambda_7 & 8.8\cdot 10^{-8} & 0.00001857 & 0.00034683 & 0.00225847 & 0.00807419 \\
\lambda_8 & 1.1\cdot 10^{-9} & 9.36\cdot 10^{-7}& 0.00004026 & 0.00048034 & 0.00273121 \\
\lambda_9 & 0 & 4.4\cdot 10^{-8} & 4.133\cdot 10^{-6}& 0.00008335 & 0.00068859 \\
\lambda_{10} &0&0.2\cdot 10^{-8} & 4.12\cdot 10^{-7} & 0.00001447 & 0.00018185 \\
\hline
\end{array}
$
\captionC{First ten NONs of the $3$-Harmonium ground state for the five different coupling strengths $\delta=0.2, 0.4, 0.6, 0.8, 1.0$.}
\label{tab:NONs3Harm}
\end{table}
Still for very strong interaction $\delta =1.0$ the NONs deviate from $1$ and $0$ only slightly.
For the two interaction strengths $\delta=0.2, 0.4$ the NONs are sufficiently fast decaying and a truncation (recall Sec.~\ref{sec:truncation})
to the setting with dimension $d=7$ is possible. The truncation errors are given by $\lambda_8 = 1.1\cdot 10^{-9}$ and $
9.36\cdot 10^{-7}$, respectively. The polytope distances we find (with $D^{(3,7)}_i$ according (\ref{set37})) are listed in Tab.~\ref{tab:Ds}.
\begin{table}[h]
\centering
\renewcommand{\arraystretch}{1.25}
$\begin{array}{|c|r|r|r|r|}
\hline
\delta & D^{(3,7)}_1 \,\,\,\,\,\,\,& D^{(3,7)}_2 \,\,\,\,\,\,\,& D^{(3,7)}_3 \,\,\,\,\,\,\,& D^{(3,7)}_4 \,\,\,\,\,\,\,\\
\hline
0.2 &2.4\cdot 10^{-8} & 9.8\cdot 10^{-8} & 5.6\cdot 10^{-8} & 1.19\cdot 10^{-7}\\
\hline
0.4 & 6.565\cdot10^{-6}& 0.00002034 & 0.00001220 &
0.00002613 \\ \hline
\end{array}
$
\captionC{Saturations of the generalized Pauli constraints for the coupling strengths $\delta=0.2, 0.4$ in the setting $\wedge^3[\mathcal{H}_1^{(7)}]$.}
\label{tab:Ds}
\end{table}
All those distances are larger than the truncation error. It is remarkable that the distance to the boundary of the polytope is
significantly smaller than the distance to the Hartree-Fock point, which is $6\cdot 10^{-4}$ and $9\cdot 10^{-4}$, respectively.
This effect of non-trivial quasi-pinning is becoming even more extreme if we reduce the coupling $\delta$ and consider the weak interaction
regime, $\delta\ll 1$. Studying this regime also allows to determine the algebraic behavior of $\lambda_i$ and the polytope distances as function
of $\delta$. To explain this, consider a function $f(\delta)$ with the following analytical behavior
\begin{equation}
f(\delta) = a^{(r)} \delta^r + \mbox{O}(\delta ^{r+1}) \qquad,\, \delta \rightarrow 0^+ \,\,\,,\, r\in \N\,.
\end{equation}
This leads to
\begin{equation}
-\log_{10}{f(\delta)}= -\log_{10}{a^{(r)}}+\mbox{O}(\delta)+ r (-\log_{10}{\delta})     \qquad,  \,\delta \rightarrow 0^+ \,.
\end{equation}
By introducing $\tilde{y} := -\log_{10}{f(\delta)}$, $\tilde{x} := -\log_{10}{\delta}$, $b := -\log_{10}{a^{(r)}}$ and observing that the regime $\delta \gtrapprox 0$  corresponds to $\tilde{x} \gg 1$ we find
\begin{equation}\label{numlinearreg}
\tilde{y} \approx b + r \tilde{x}   \qquad \mbox{for} \, \tilde{x} \gg 1\,.
\end{equation}
By plotting such an approximately linear graph for $\tilde{y}(\tilde{x})$ we can determine the parameters $a^{(r)}$ and $r$. This is done for the first few NONs by analyzing the functions
\begin{eqnarray}
1-\lambda_i(\delta) &=& a_i^{(r_i)}\, \delta ^{r_i} + \mbox{O}(\delta ^{r_i+2}) \qquad \mbox{for} \, i=1,2,3 \nonumber \\
\lambda_j(\delta) &=& a_j^{(r_j)}\, \delta ^{r_j} + \mbox{O}(\delta ^{r_j+2}) \qquad \mbox{for} \, j\geq 4\qquad.
\end{eqnarray}
and also for the polytope distances $D_i^{(3,7)}$ (\ref{set37}).
The corresponding double-logarithmic plots are presented in Fig.~\ref{fig:nons} and \ref{fig:Ds}. There we chose the small interaction strengths $\delta= 10^{-1},10^{-2},10^{-3}$. In all those plots a linearly fitted curve describes the three points quite well.
For the plot describing the leading corrections of $\lambda_1$ the slope is $6$, i.e.~ $1-\lambda_1(\delta) \sim \delta^6$. For $\lambda_4$ we find the slope $4$.
\begin{figure}
\includegraphics[scale=0.50]{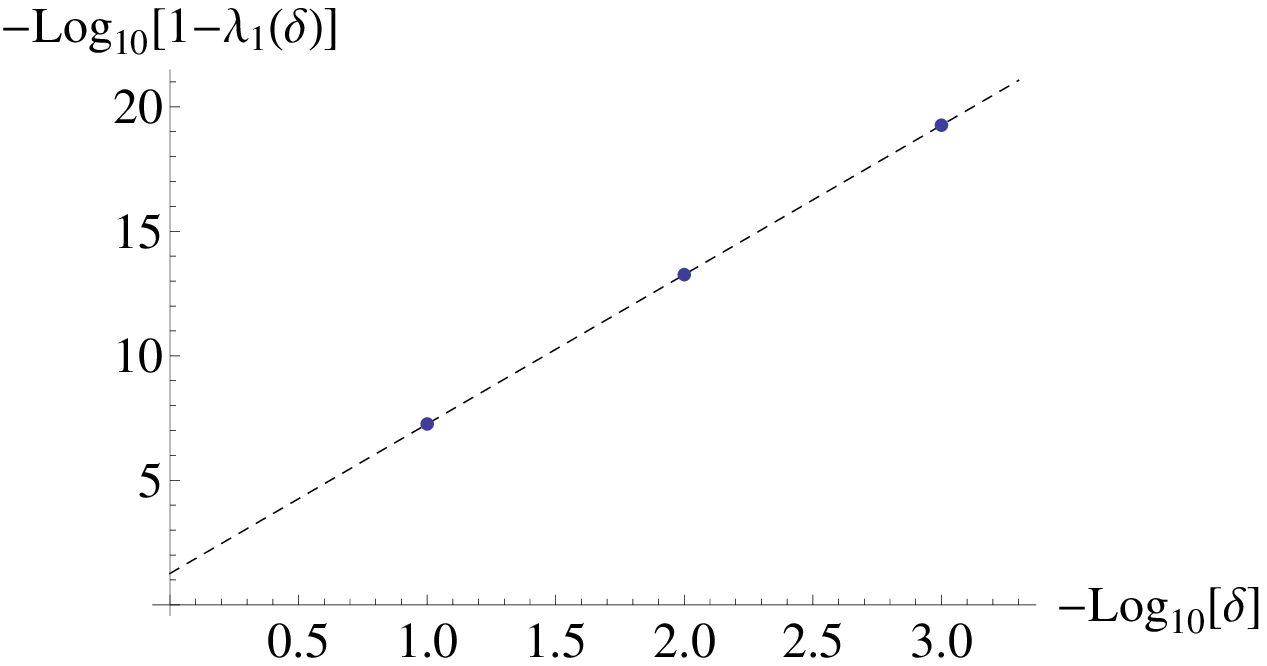}
\hspace{0.1cm}
\includegraphics[scale=0.50]{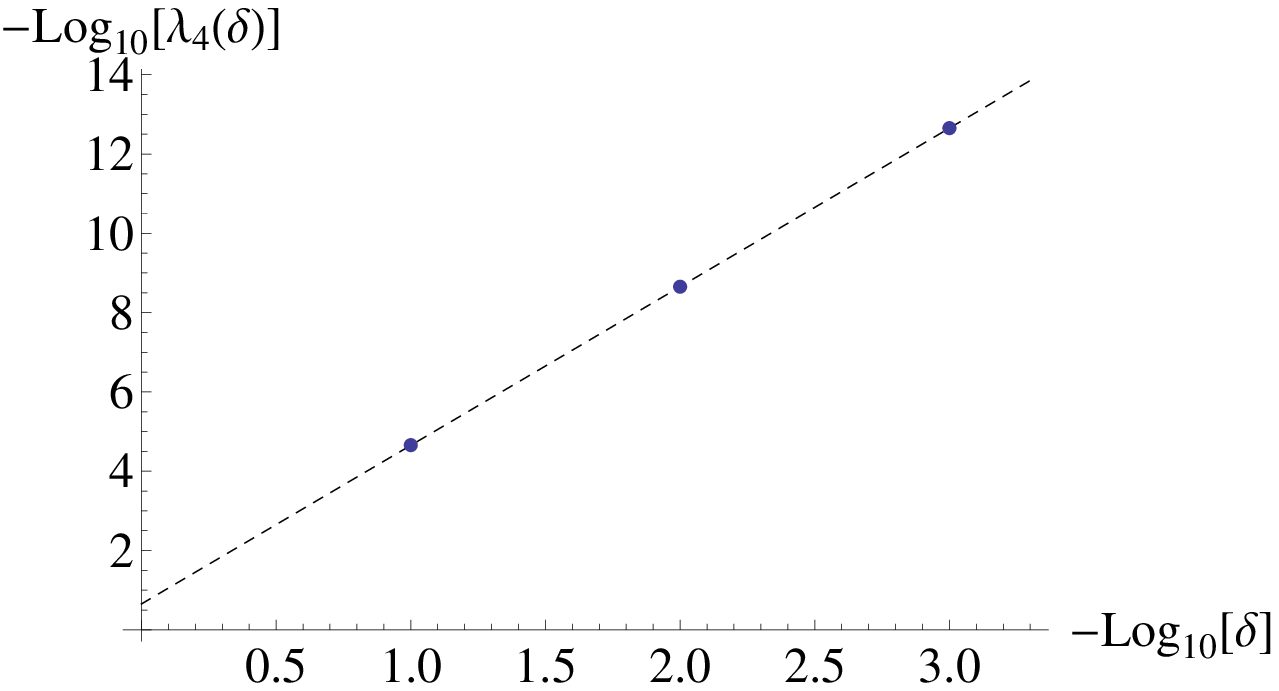}
\captionC{Double-logarithmic plots of the quantities $1-\lambda_1,\lambda_4$ for the interaction strengths $\delta= 10^{-1}, 10^{-2}, 10^{-3}$ for the $3$-Harmonium ground state. Moreover, a linear curve fitting the three points is shown.}
\label{fig:nons}
\end{figure}
Intriguingly, the saturations $D_1^{(3,7)}$ and $D_4^{(3,7)}$ presented in Fig.~\ref{fig:Ds} have slope $8$. This saturation behavior $D_i^{(3,7)}\sim \delta^8$ is surprising and was not expected. Since $\vec{\lambda}$ has a distance $\delta^4$ to the Hartree-Fock point, this quasi-pinning of $\delta^8$ is non-trivial.
\begin{figure}
\includegraphics[scale=0.50]{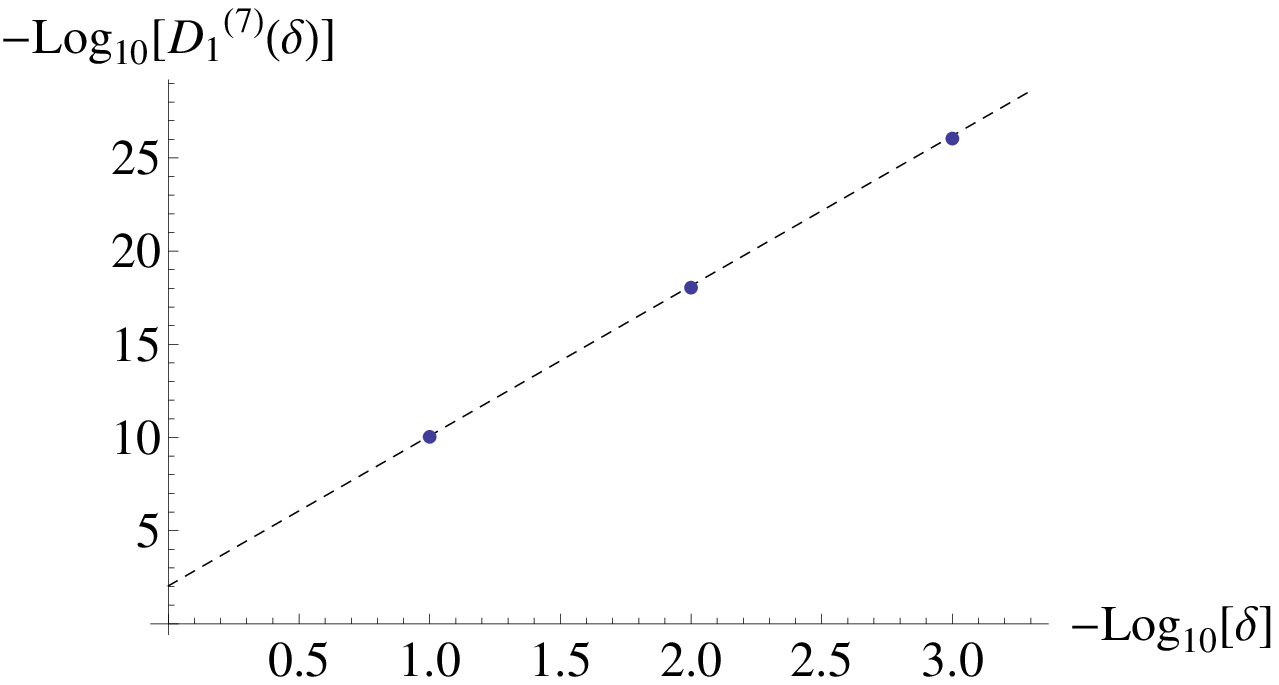}
\hspace{0.1cm}
\includegraphics[scale=0.50]{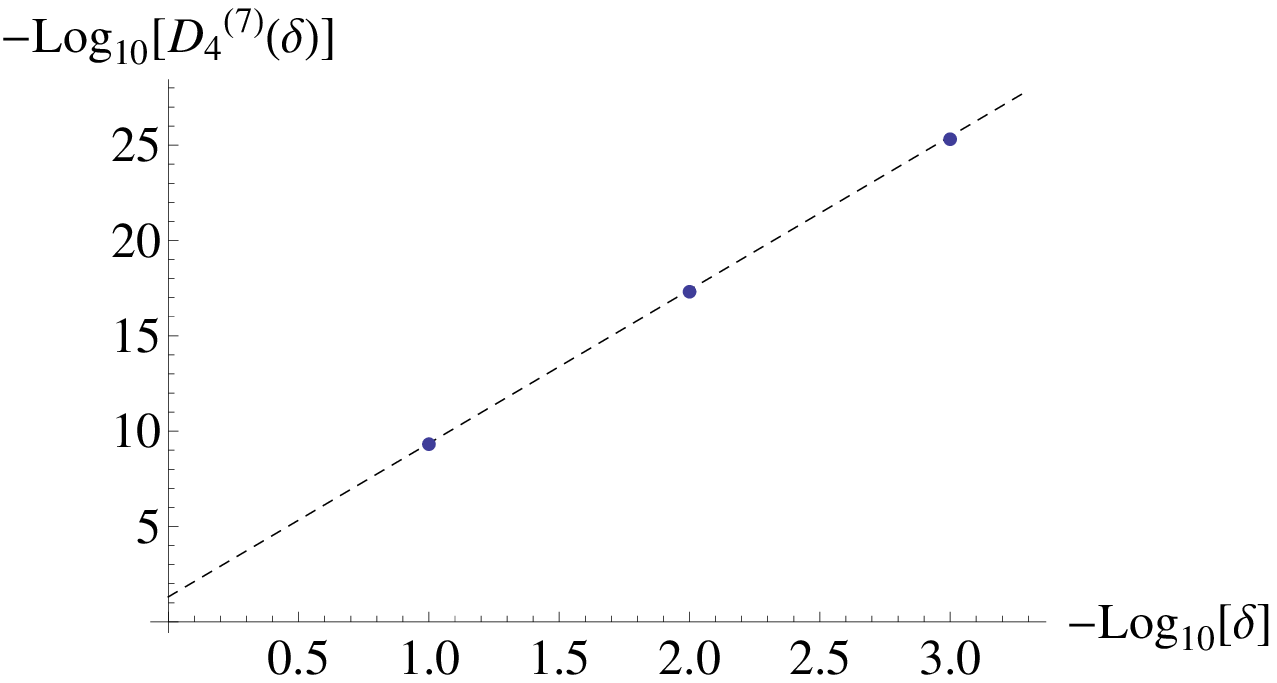}
\captionC{Double-logarithmic plots of the saturations $D_1^{(3,7)}, D_4^{(3,7)}$ for the interaction strengths $\delta= 10^{-1}, 10^{-2}, 10^{-3}$ for the $3$-Harmonium ground state. Moreover, a linear curve fitting the three points is shown.}
\label{fig:Ds}
\end{figure}

We summarize the numerical results in Tab.~\ref{tab:NONs3Harmnanalyt} and Tab.~\ref{tab:Ds3Harmnanalyt}. Calculating $\lambda_k$ for smaller and smaller interaction strengths we reached the linear regime (\ref{numlinearreg}) arbitrarily well and it was possible to predict
the coefficients of the leading order terms presented in Tabs.~\ref{tab:NONs3Harmnanalyt}, \ref{tab:Ds3Harmnanalyt}.
\begin{table}[h]
\centering
$
\setlength{\arraycolsep}{0.065cm}
\renewcommand{\arraystretch}{1.2}
\begin{array}{c|c|c|c|c|c|c|c|c}
& 1-\lambda_1 & 1-\lambda_2 & 1-\lambda_3 & \lambda_4 &\lambda_5 &\lambda_6 & \lambda_7 & \lambda_8\nonumber \\
\hline
r_i & 6 &4&4&4&4&6&8 &10\nonumber \\
\hline
a_i^{(r_i)}&  0.055&  0.222 &  0.222 &  0.222 &  0.222 &  0.055 &  0.037 &  0.011  \nonumber \\
& \approx \frac{40}{729}&  \approx\frac{2}{9} &  \approx\frac{2}{9} &  \approx\frac{2}{9} &  \approx\frac{2}{9} &  \approx\frac{40}{729} & \approx\frac{80}{2187} &  \approx\frac{224}{19683}
\end{array}
$
\captionC{Leading order of the largest eight NONs of the $3$-Harmonium ground state in the regime of weak interaction, $\delta \ll 1$.}
\label{tab:NONs3Harmnanalyt}
\end{table}

\begin{table}[h]
\centering
$
\renewcommand{\arraystretch}{1.2}
\begin{array}{c|c|c|c|c|c}
&D^{(3,6)}_1&D^{(3,7)}_1&D^{(3,7)}_2&D^{(3,7)}_3&  D^{(3,7)}_4 \nonumber \\ \hline
r_i & 8&8&8&8&8 \nonumber \\ \hline
a_i^{(r_i)} & 0.076&  0.009 &  0.041 & 0.023 &  0.049 \nonumber \\
 & \approx \frac{4510}{59049} & \approx \frac{20 }{2187} & \approx \frac{10 }{243}& \approx \frac{50 }{2187} & \approx \frac{2890 }{59049}
\end{array}
$
\captionC{Leading order of the saturations for the settings $\wedge^3[\mathcal{H}_1^{(6)}]$ and $\wedge^3[\mathcal{H}_1^{(7)}]$ of the $3$-Harmonium ground state in the regime of weak interaction, $\delta \ll 1$.}
\label{tab:Ds3Harmnanalyt}
\end{table}

From a numerical point of view it is quite difficult to determine the higher-order coefficients $a_i^{(r_i+2)}, a_i^{(r_i+4)}, \ldots$. Therefore, to understand and also to verify the results for the NONs and the saturations in the regime of weak interaction we apply degenerate Rayleigh-Sch\"odinger perturbation to the $1$-RDO represented as matrix w.r.t.~the bosonic NOs, i.e.~(\ref{1RDOfmat1}).

Notice that due to the duality (\ref{duality}) the expansion of $\lambda_k(\delta)$ as function of $\delta$ contains only even order terms, which makes the perturbation theory simpler. As a first step, we expand the matrix $\rho_1$ up to $\delta^{10}$
\begin{equation}
\rho_1 = \sum_{k=0}^{10} \rho_{1,k}\,\delta^k + O(\delta^{11})\,.
\end{equation}
The unperturbed matrix $\rho_{1,0}$ is highly degenerate. The eigenvalue $1$ is three times degenerate and the other one, $0$, infinitely many
times. The degenerate perturbation theory requires as first step to decouple these two eigenblocks by applying an appropriate unitary transformation $U_0$,
\begin{equation}\label{ptunitary}
\tilde{\rho}_1 = U_0^{\dagger}\rho_1 U_0\,.
\end{equation}
Then afterwards we can consider both blocks separately. Both steps are done in the Appendix \ref{app:fermionicPT} and we just present the results here.
We find for the NONs
\begin{eqnarray}\label{spectrum}
1-\lambda_1 &= & \frac{40}{729} {\delta}^6 - \frac{1390}{59049} {\delta}^8 + O(\delta^{10}) \nonumber \\
1-\lambda_2 &= & \frac{2}{9} {\delta}^4 - \frac{232}{729}{\delta}^6 + \frac{3926}{10935} {\delta}^8 +O(\delta^{10}) \nonumber  \\
1-\lambda_3 &=& \frac{2}{9}{\delta}^4 - \frac{64}{243}{\delta}^6 + \frac{81902}{295245}{\delta}^8 +O(\delta^{10}) \nonumber \\
\lambda_4 &= & \frac{2}{9}{\delta}^4 - \frac{64}{243}{\delta}^6 + \frac{73802}{295245}{\delta}^8 + O(\delta^{10}) \nonumber \\
\lambda_5 &= &\frac{2}{9} {\delta}^4 - \frac{232}{729} {\delta}^6 + \frac{3976}{10935} {\delta}^8 +O(\delta^{10}) \nonumber \\
\lambda_6 &= & \frac{40}{729} {\delta}^6 - \frac{2200}{59049} {\delta}^8 + O(\delta^{10}) \nonumber \\
\lambda_7 &= & \frac{80}{2187} {\delta}^8 + O(\delta^{10}) \nonumber \\
\lambda_8 &=& O(\delta^{10}) \nonumber \\
\lambda_9 &= & O(\delta^{12})\,.
\end{eqnarray}
Moreover, there is strong numerical evidence for the hierarchy $\lambda_k \sim \delta^{2k-6}$ for $k>9$. The analytical result (\ref{spectrum}) confirm the numerical results presented in Tab.~\ref{tab:NONs3Harmnanalyt} and Tab.~\ref{tab:Ds3Harmnanalyt}.

\subsection{Pinning analysis for weak interaction}\label{sec:pinninganalysis}
In this section we use the analytical results for the NONs (\ref{spectrum}) and perform a detailed pinning analysis (recall Sec.~\ref{sec:pinninganalysis}) for the regime of weak interaction. We first start by considering only terms on the scale $\delta^4$, then afterwards step by step include further orders in $\delta$.
\begin{itemize}
\item \emph{scale $\delta^4$}: On that scale all NONs are pinned by the Pauli exclusion principle to either $1$ or $0$, except $\lambda_2,\lambda_3,\lambda_4,\lambda_5$, which have first corrections of order $\delta^4$.  This means that the distance to the Hartree-Fock point
    behaves as $\delta^4$. Moreover, since we can truncate to the setting $\wedge^2[\mathcal{H}_1^{(4)}]$, there are only strict equalities (see Lem.~\ref{lem:setting2}). As consistency check, we observe indeed that
    \begin{equation}
    \lambda_2 = \lambda_3 + O(\delta^6)\qquad,\,\lambda_4 = \lambda_5 + O(\delta^6)\,.
    \end{equation}
\item \emph{scale $\delta^6$}: On this finer scale, also the eigenvalues $\lambda_1,\lambda_6$ have contributions. Therefore, we deal here with the setting $\wedge^3[\mathcal{H}_1^{(6)}]$, where the constraints are given by (\ref{set36}). Three of them are equalities. As a consistency check we notice that
    \begin{eqnarray}
    1-(\lambda_1+\lambda_6)&=& \frac{10}{729}\delta^8 + O(\delta^{10}) \nonumber \\
    1-(\lambda_2+\lambda_5)&=& \frac{10}{2187}\delta^8 + O(\delta^{10}) \nonumber \\
    1-(\lambda_3+\lambda_4)&=& \frac{20}{729}\delta^8 + O(\delta^{10}) \,,
    \end{eqnarray}
    i.e.~the equalities are indeed satisfied on the scale $\delta^6$. The only saturation that we can investigate is that of constraint $D^{(3,6)}(\vec{\lambda})\geq 0$. Geometrically, this means to determine the distance of the truncated NONs to the boundary of the polytope $\mathcal{P}_{3,6}$. Due to the three equalities for the Borland-Dennis setting (\ref{set36}) $\mathcal{P}_{3,6}$ is just three dimensional and we can visualize it. By choosing $\vec{v}\equiv (\lambda_4,\lambda_5,\lambda_6)$ as the vector of independent variables we present $\mathcal{P}_{3,6}$ qualitatively in Fig.~\ref{fig:traj} in the neighborhood of the Hartree-Fock point $\vec{v}^{(a)} \equiv (0,0,0)$.
    \begin{figure}[!h]
    \includegraphics[width=11.5cm]{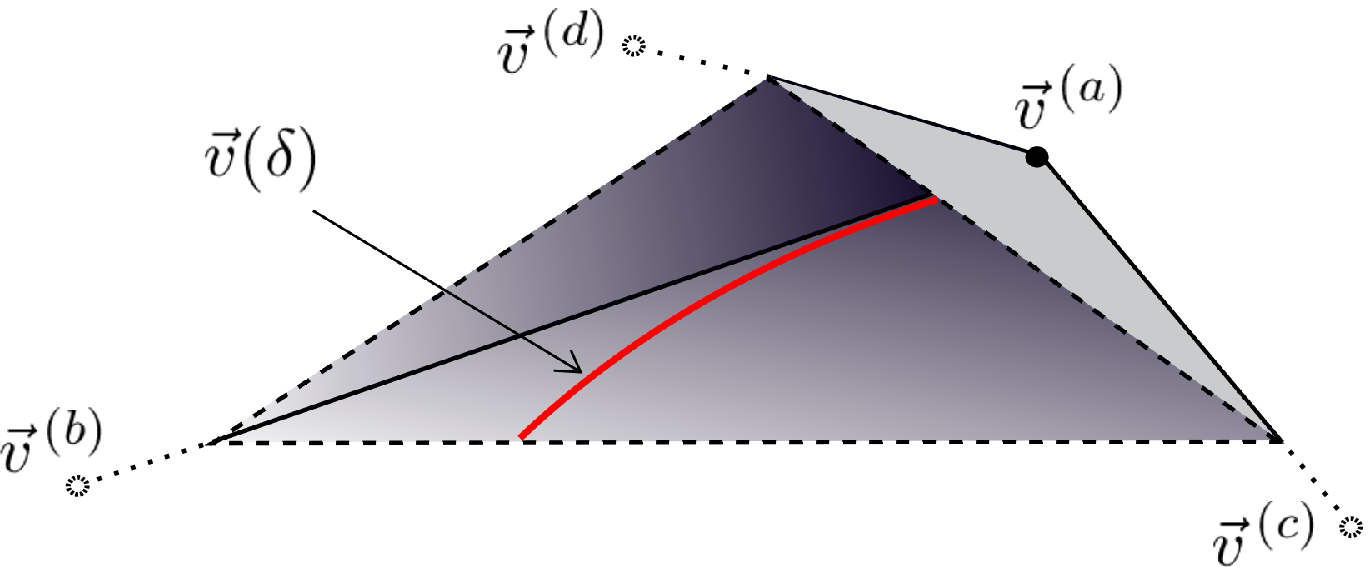}
    \centering
    \captionC{Polytope $\mathcal{P}_{3,6}$ (qualitatively) in the neighborhood of the Hartree-Fock point $\vec{v}^{(a)}$ and the spectral trajectory $\vec{v}(\delta)$ for the $3$-Harmonium ground state.}
    \label{fig:traj}
    \end{figure}
    The other three vertices are indicated. They are given by
    \begin{equation}\label{verticesBD}
    \vec{v}^{(b)} \equiv (\frac{1}{2},\frac{1}{2},0)\,,\,\,    \vec{v}^{(c)} \equiv (\frac{1}{2},\frac{1}{4},\frac{1}{4})\,,\,\,    \vec{v}^{(d)} \equiv (\frac{1}{2},\frac{1}{2},\frac{1}{2})\,.
    \end{equation}
    The spectral ``trajectory'' $\vec{v}(\delta)$ is drawn as red curve in Fig.~\ref{fig:traj}. It starts at the Hartree-Fock point vertex
    $\vec{v}^{(a)}$ which corresponds  to the non-interacting situation $\delta=0$. When increasing the fermion-fermion interaction, $\vec{v}(\delta)$ leaves the vertex $\vec{v}^{(a)}$ and moves along the edge $(\vec{v}^{(a)}, \vec{v}^{(b)})$, the distance to $\vec{v}^{(a)}$ growing as $\delta^4$.
    On the finer scale $\delta^6$, $\vec{v}(\delta)$ also moves away from the edge but is still \emph{pinned to the boundary}
    of the polytope, lying on the 2-facet spanned by $\vec{v}^{(a)}, \vec{v}^{(b)}$ and $\vec{v}^{(c)}$. This is the bottom area in Fig.~\ref{fig:traj}, corresponding to saturation of constraint $D^{(3,6)}\geq 0$ (\ref{set36}).

    From (\ref{spectrum}), we can infer that the distance to the 2-facet $(\vec{v}^{(a)}, \vec{v}^{(b)}, \vec{v}^{(c)})$ increases as $\delta^8$, \begin{equation}\label{pinning}
    D^{(3,6)}(\delta) = \zeta^{(6)} \,\delta^8 + O(\delta^{10})
    \end{equation} with $\zeta^{(6)}= \frac{4510}{59049}$. This means we find exact pinning on a scale $\delta^6$. To conclude that we have only quasi-pinning on the absolute scale is not possible, as the distance to the boundary is of the same order
    -- $\delta^8$ -- as the truncation error.
\item \emph{scale $\delta^8$}: To improve the accuracy of our pinning analysis we take the seventh-largest NON, $\lambda_7$, into account.
    For the four generalized Pauli constraints (\ref{set37}) we find the saturations
    \begin{equation}
    D_i^{(3,7)}  = \zeta_i^{(7)}{\delta}^8+ O({\delta}^{10}) \,\,\,,\label{pinning2}
    \end{equation}
    with $\zeta_1^{(7)}=\frac{20 }{2187}$, $\zeta_2^{(7)}=\frac{10 }{243}$,  $\zeta_3^{(7)}=\frac{50 }{2187}$, $\zeta_4^{(7)}=\frac{2890 }{59049}$.
    Here in the $\wedge^3[\mathcal{H}_1^{(7)}]$-analysis, the new result is that all four distances $D^{(3,7)}_i$ are non-zero to a smaller order
    ($\delta^8$) than the error of spectral truncation ($\delta^{10}$) of the NONs.
\end{itemize}
According to the concept of truncation explained in Sec.~\ref{sec:truncation}, this shows that the absence of pinned spectra is genuine, rather than an artifact of
the truncation. Given this, the quasi-pinning found here is surprisingly strong. In
particular it exceeds by four additional orders the
(quasi-)pinning by Pauli's exclusion principle constraints
(\ref{PauliConstraint}),
\begin{equation}\label{PauliConstraint2}
0\leq 1-\lambda_2(\delta), 1-\lambda_3(\delta), \lambda_4(\delta),
\lambda_5(\delta) = \frac{2}{9}\,\delta^4 + O(\delta^{6}).
\end{equation}

To finish the pinning analysis for the $3$-Harmonium ground state we would like to emphasize again that quasi-pinning does not only occur for weak interaction but even for medium one (see numerical results in Tab.~\ref{tab:Ds}).

%For the interaction $m\omega^2 = 3D$ we find:
%\\
%\begin{minipage}{0.3\textwidth}
%$
%\begin{array}[t]{cc}
%\lambda_1 & 0.999998534\\
%\hline
%\lambda_2 & 0.999807956\\
%\hline
%\lambda_3 & 0.999806535\\
%\hline
%\lambda_4 & 0.000193444\\
%\hline
%\lambda_5 & 0.000192047\\
%\hline
%\lambda_6 & 0.000001455\\
%\hline
%\lambda_7 & 0.000000028\\
%\hline
%\lambda_8 & 0.000000000\\
%\hline
%\lambda_9 & 0.000000000\\
%\end{array}
%$
%\end{minipage}
%\begin{minipage}{0.1\textwidth}
%$\Rightarrow$
%\end{minipage}
%\begin{minipage}{0.3\textwidth}
%$
%\begin{array}{c|c}
%\lambda_1+\lambda_2+\lambda_5+\lambda_6 & 1.999999992\\
%\hline
%\lambda_1+\lambda_3+\lambda_4+\lambda_6 & 1.999999968\\
%\hline
%\lambda_2+\lambda_3+\lambda_4+\lambda_5 & 1.999999982\\
%\hline
%\lambda_1+\lambda_2+\lambda_4+\lambda_7 & 1.999999961\\
%\hline
% D_1^{(6)}& 0.000000059 \\
%\hline
% D_1^{(7)}& 0.000000008 \\
%\hline
% D_2^{(7)}& 0.000000032 \\
%\hline
% D_3^{(7)}& 0.000000012 \\
%\hline
% D_4^{(7)}& 0.000000038
%\end{array}
%$
%\end{minipage}

\subsection{Numerical results for larger particle numbers}\label{sec:4fermions}
The result on strong quasi-pinning found for the $3$-Harmonium ground state in the previous section, Sec.~\ref{sec:pinninganalysis}, is quite surprising. We are wondering whether this effect also shows up for larger particle numbers $N>3$. Do we find again quasi-pinning of order $\delta^8$ or might it be even stronger? To which reduced setting can we truncate our pinning analysis? For $N>3$, do more than six NONs contribute to the pinning effect or do we always fall back to the Borland-Dennis setting $\wedge^3[\mathcal{H}_1^{(6)}]$? The latter would emphasize the relevance of the active space picture: For large particle number $N$ and weak interaction some of the fermions are frozen in the lowest $1$-particle states of the harmonic confinement potential and all the highly excited eigenstates of the trap also do not contribute to the physical behavior of the system. In that sense we could restrict to the active space given by the $1$-particle states around the Fermi level and study the physics of the remaining non-frozen fermions within that finite dimensional subspace. Such approximations are quite popular in condensed matter physics and are also used for variational optimizations in quantum chemistry, manifested in the so-called complete active space self-consistent field methods (CASSCF).
%A second motivation for studying possible ground state pinning for the Harmonium model (\ref{HamHarmonium}) for $N>3$ is that the original of quasi-pinning for $N=3$ might be less

The strategy for determining the NONs of the $N$-Harmonium ground state for $N>3$ is the same as that for $N=3$ presented in the previous sections.
The ground state is given by (\ref{gsfermions}) and the $1$-RDO $\rho_1(x,y)$ in spatial representation can be calculated for each fixed $N$ and takes the form (\ref{1RDOf}). Finally by choosing the analytically known bosonic NOs as reference basis we can express $\rho_1$ as a matrix. By using the same analytical and numerical methods as in Sec.~\ref{sec:NONsfermionic}, D.Ebler determined the NONs for $N=4$ \cite{Ebler}. We have extended these results by numerical brute force up to $N= 8$. From those results we conjecture the leading order behavior in the regime of weak interaction for arbitrary $N$ shown in Tab.~\ref{tab:NONsNHarmnanalyt}.
\begin{table}[h]
\centering
\renewcommand{\arraystretch}{1.2}
$\begin{array}{|c|c|c|c||c|c|c|c|}
\hline
 1-\lambda_{N-3} & 1-\lambda_{N-2} & 1-\lambda_{N-1} &1-\lambda_N& \lambda_{N+1} & \lambda_{N+2} & \lambda_{N+3} & \lambda_{N+4} \\
\hline
\delta^8 & \delta^6 & \delta^4 & \delta^4 & \delta^4 & \delta^4 & \delta^6 & \delta^8 \\
\hline
\end{array}
$
\captionC{NONs of the $N$-Harmonium ground state close to the ``Fermi level'' in the regime of weak interaction, $\delta \ll 1$, up to corrections of order $\delta^{10}$.}
\label{tab:NONsNHarmnanalyt}
\end{table}

For the other indices $k$ the NONs follow the hierarchy
\begin{eqnarray}
1-\lambda_{N+1-k} &\sim& \delta^{2 k}\qquad,k=5,6,\ldots,N \nonumber \\
\lambda_{N+k} &\sim & \delta^{2 k} \qquad,k=5,6,7,\ldots \,.
\end{eqnarray}
We clearly see the existence of an active space. By considering the NONs up to corrections to $1$ or $0$ of order $\delta^{2r}$ only $r$ fermions are active in a $2r$-dimensional $1$-particle Hilbert space.

Given the algebraic behavior of the NONs we can perform a systematic pinning analysis by considering different scales $\delta^{2r}$ starting with $\delta^4$ and continuing with $\delta^6$ and $\delta^8$ as it was done for $N=3$ in Sec.~\ref{sec:pinninganalysis}. On the scale $\delta^6$
the active space is given by the Borland-Dennis setting $\wedge^3[\mathcal{H}_1^{(6)}]$. As saturation for $D^{(3,6)}$ (recall (\ref{set36})) we find again
\begin{equation}
D^{(3,6)} \sim \delta^8\,,
\end{equation}
which is of the same order as the truncation error, given by $\lambda_{N+4}, 1-\lambda_{N-3}\sim \delta^8$. On the finer scale $\delta^8$ we need to consider the NONs $\lambda_{N-3},\ldots,\lambda_{N+4}$, i.e.~we have to consider the setting $\wedge^4[\mathcal{H}_1^{(8)}]$. All generalized Pauli constraints turn out to be saturated up to corrections of order $\delta^8$. This means the $N$-Harmonium ground state is pinned for all $N$ up to corrections of order $\delta^8$. Moreover, the pinning is reducible to the Borland-Dennis setting defined by three fermions in the six modes around the Fermi level. We are wondering whether there exists a physical model for few fermions confined to some trap with ground state quasi-pinning not reducible to just three fermions. This question seems to be very challenging and we expect that one needs to understand the mechanism behind pinning to answer it.

\subsection{Excited states}\label{sec:excited}
In this chapter we investigate possible pinning of the first few excited $3$-Harmonium states. This will give some first insight into the mechanism behind the effect of ground state quasi-pinning found in Sec.~\ref{sec:NONsfermionic} and Sec.~\ref{sec:pinninganalysis}.

As a first step we calculate the first few excited $N$-fermion states. For this recall (\ref{Smatrix}), (\ref{eigenspaces}), expansions of the form (\ref{gsfermionicexpan}) and that $y_1(\vec{x})$ depends symmetrically on the physical coordinates $x_1,\ldots,x_N$. Then, in the same way as explained in Sec.~\ref{sec:model} for the ground state, we can determine all fermionic excitations by studying the $N$-particle subspaces $S_{n_-,n_+}$, separately. Moreover, according to the structure of $y_1(\vec{x})$ we can restrict ourself to the spaces $S_{0,n_+}$, i.e.~states with zero center of mass excitations. The fermionic excited states lying in $S_{n_-,n_+}$ with $n_- >0$ can be obtained by multiplying those in $S_{0,n_+}$ with $H_{n_-}(y_1(\vec{x}))$ (see also the structure in Eq.~(\ref{gsfermionicexpan})). By either using symbolic tools like Mathematica or using the results from \cite{harmOsc2012} we can easily determine the first few excited fermionic states. Since they are quite lengthy, they are not presented explicitly.

\subsubsection{Zero interaction}\label{sec:zeroInt}
First, we study pinning of the lowest few excited states of our harmonic model (\ref{HamHarmonium}) for the case of zero fermion-fermion interaction, i.e.~$\delta=0$. Notice, that this task is trivial for most $1$-dimensional confinement potentials $V(x)$: For zero interaction, the $N$-particle Schr\"odinger equation effectively reduces to a $1$-particle equation, the problem of diagonalizing the $1$-particle Hamiltonian $\frac{p^2}{2m} +V(x)$.
Given the corresponding $1$-particle solution, energies $\{\epsilon_i\}$ and eigenstates $\{|i\rangle\}$, we find as $N$-particle eigenstates single Slater determinants $|i_1,\ldots,i_N\rangle$ with energies $E_{\emph{\textbf{i}}} = \epsilon_{i_1}+\ldots + \epsilon_{i_N}$.
For generic $V(x)$ these energies $E_{\emph{\textbf{i}}}$ are not degenerate and the corresponding eigenstates are unique.
Consequently, every fermionic eigenstate is then pinned to the Hartree-Fock point on the boundary of the polytope and it makes sense to investigate its possible pinning after switching on a small $2$-particle interaction driving the NONs away from the Hartree-Fock point. In contrast to such generic models, our harmonic trap model has degenerate $N$-particle eigenenergies for zero interaction since the energy of the center of mass excitation, $\hbar \omega_-$ and relative excitation energy $\hbar \omega_+$ are identical. As a consequence, the $N$-particle states are not uniquely defined
and from Sec.~\ref{sec:model} we can infer that the degeneracy is $n_+-n_+^{(f)}-1$. Nevertheless, since the degeneracies do vanish for finite fermion-fermion interaction $\delta>0$, we can define the zero interaction states naturally via
\begin{equation}
\Psi_{n_-,n_+}(\delta=0) := \lim_{\delta\rightarrow 0}\,\Psi_{n_-,n_+}(\delta)\,.
\end{equation}
\begin{table}[h]
\centering
$
\setlength{\arraycolsep}{0.04cm}
\renewcommand{\arraystretch}{1.3}
\begin{array}{c|cc|cl|cc|c}
 \#&n_- & \Delta n_+ &  & \mbox{NONs} & d_{eff} & N_{eff} &  \\
 \hline
 0 & 0 & 0 &  & \{1,1,1,0,0,0,0,0,0,0,0\} & 0 & 0 & \text{p} \\
 \hline
 1 & 1 & 0 &  & \{1,1,1,0,0,0,0,0,0,0,0\} & 0 & 0 & \text{p} \\
 \hline
 2 & 2 & 0 &  & \left\{1,\frac{2}{3},\frac{2}{3},\frac{1}{3},\frac{1}{3},0,0,0,0,0,0\right\} & 4 & 2 & \text{p} \\
 3 & 0 & 2 &  & \left\{1,\frac{2}{3},\frac{2}{3},\frac{1}{3},\frac{1}{3},0,0,0,0,0,0\right\} & 4 & 2 & \text{p} \\
 \hline
 4 & 3 & 0 &  & \left\{\frac{26}{27},\frac{17}{27},\frac{16}{27},\frac{11}{27},\frac{10}{27},\frac{1}{27},0,0,0,0,0\right\} & 6 & 3 & \text{p} \\
 5 & 1 & 2 &  & \left\{\frac{7}{9},\frac{7}{9},\frac{5}{9},\frac{4}{9},\frac{2}{9},\frac{2}{9},0,0,0,0,0\right\} & 6 & 3 & \text{p} \\
 6 & 0 & 3 &  & \left\{\frac{25}{27},\frac{22}{27},\frac{20}{27},\frac{7}{27},\frac{5}{27},\frac{2}{27},0,0,0,0,0\right\} & 6 & 3 & \text{p} \\
 \hline
 7 & 4 & 0 &  & \left\{\frac{8}{9},\frac{2}{3},\frac{5}{9},\frac{8}{27},\frac{7}{27},\frac{5}{27},\frac{4}{27},0,0,0,0\right\} & 7 & 3 & \text{p} \\
 8 & 2 & 2 &  & \left\{\frac{8}{9},\frac{2}{3},\frac{5}{9},\frac{4}{9},\frac{1}{3},\frac{1}{9},0,0,0,0,0\right\} & 6 & 3 & \text{p} \\
 9 & 1 & 3 &  & \left\{\frac{20}{27},\frac{19}{27},\frac{2}{3},\frac{4}{9},\frac{5}{27},\frac{4}{27},\frac{1}{9},0,0,0,0\right\} & 7 & 3 & \text{p} \\
 10 & 0 & 4 &  & \left\{1,\frac{5}{9},\frac{5}{9},\frac{1}{3},\frac{1}{3},\frac{1}{9},\frac{1}{9},0,0,0,0\right\} & 6 & 2 & \text{p} \\
 \hline
 11 & 5 & 0 &  & \left\{\frac{64}{81},\frac{44}{81},\frac{37}{81},\frac{32}{81},\frac{10}{27},\frac{8}{27},\frac{7}{81},\frac{5}{81},0,0,0\right\} & 8 & 3 & \text{p} \\
 12 & 3 & 2 &  & \left\{\frac{22}{27},\frac{50}{81},\frac{35}{81},\frac{25}{81},\frac{7}{27},\frac{20}{81},\frac{16}{81},\frac{10}{81},0,0,0\right\} & 8 & 3 & \text{p} \\
 13 & 2 & 3 &  & \left\{\frac{56}{81},\frac{16}{27},\frac{47}{81},\frac{14}{27},\frac{25}{81},\frac{14}{81},\frac{10}{81},\frac{1}{81},0,0,0\right\} & 8 & 3 & \text{p} \\
 14 & 1 & 4 &  & \left\{\frac{19}{27},\frac{5}{9},\frac{11}{27},\frac{1}{3},\frac{8}{27},\frac{7}{27},\frac{7}{27},\frac{5}{27},0,0,0\right\} & 8 & 3 & \text{np} \\
 15 & 0 & 5 &  & \left\{\frac{20}{27},\frac{14}{27},\frac{35}{81},\frac{35}{81},\frac{28}{81},\frac{25}{81},\frac{14}{81},\frac{4}{81},0,0,0\right\} & 8 & 3 & \text{np}
   %\\
%   \hline
% 16 & 6 & 0 &  &
%   \left\{\frac{496}{729},\frac{376}{729},\frac{343}{729},\frac{280}{729},\frac{256}{729},\frac{178}{729},\frac{175}{729},\frac{55}{729},\frac{28}{729},0,0\right\} & 9 & 3 &
%   \text{np} \\
% 17 & 4 & 2 &  & \left\{\frac{160}{243},\frac{145}{243},\frac{88}{243},\frac{88}{243},\frac{85}{243},\frac{70}{243},\frac{70}{243},\frac{13}{243},\frac{10}{243},0,0\right\} & 9
%   & 3 & \text{np} \\
% 18 & 3 & 3 &  &
%   \left\{\frac{490}{729},\frac{364}{729},\frac{361}{729},\frac{331}{729},\frac{226}{729},\frac{148}{729},\frac{148}{729},\frac{112}{729},\frac{7}{729},0,0\right\} & 9 & 3 &
%   \text{np} \\
% 19 & 2 & 4 &  & \left\{\frac{61}{81},\frac{43}{81},\frac{37}{81},\frac{28}{81},\frac{25}{81},\frac{16}{81},\frac{13}{81},\frac{13}{81},\frac{7}{81},0,0\right\} & 9 & 3 &
%   \text{np} \\
% 20 & 1 & 5 &  & \left\{\frac{176}{243},\frac{110}{243},\frac{98}{243},\frac{95}{243},\frac{77}{243},\frac{56}{243},\frac{50}{243},\frac{35}{243},\frac{32}{243},0,0\right\} & 9
%   & 3 & \text{np} \\
% 21 & 0 & 6 &  &
%  \{\frac{19015}{22599},\frac{9940}{22599},\frac{9646}{22599},\frac{8302}{22599},\frac{7252}{22599},\frac{6532}{22599},\frac{124}{729}, & 9 & 3 & \text{np} \\
%   &&&& \frac{2305}{22599},\frac{31}{729},0,0\} &&&
\end{array}
$
\captionC{NONs of the fifteen lowest $3$-Harmonium eigenstates $\Psi_{n_-, n_+^{(f)}+\Delta n_+}$ for zero interaction. In addition, indication of pinning (p) or no pinning (np) in the effective setting $\wedge^{N_{eff}}[\mathcal{H}_1^{(d_{eff})}]$.}
\label{tab:NONs3Harmexcitations0}
\end{table}

For all those zero-interaction energy states $\Psi_{n_-,n_+^{(f)}+\Delta n_+}(0)$ up to eight additional excitations, $n_- + \Delta n_+ \leq 8$ relative to the fermionic ground state $\Psi_0 \in S_{0,n_+^{(f)}}$, we calculate the NONs and present them (up to five additional excitations) in Tab.~\ref{tab:NONs3Harmexcitations0}. There, we can see that the lowest thirteen excited states converge for $\delta \rightarrow 0^+$ to the boundary of the polytope. However, this is not true anymore for the further, higher excited states presented in Tab.~\ref{tab:NONs3Harmexcitations0}.
The same also holds for the excited states with $n_- + \Delta n_+ = 6,7$, not shown in Fig.~\ref{tab:NONs3Harmexcitations0}. For those with $n_- + \Delta n_+ = 8$ such a statement on zero-interaction pinning is not possible anymore since their effective settings are $\wedge^3[\mathcal{H}_1^{(d)}]$ with $d\geq 11$ and the corresponding polytopes are not known yet. Nevertheless, we expect that they follow the trend seen for the states with $n_- + \Delta n_+ \leq 7$, i.e.~they are not pinned.

We shortly present all generalized Pauli constraints, which are saturated for the states discussed in Tab.~\ref{tab:NONs3Harmexcitations0}. First, note that for $\#=0,1,2,3,10$ the pinning is trivial, it immediately follows from the small value for the effective particle number and the effective dimension of the $1$-particle Hilbert space. The cases $\#=4,5,6,8$ belong to the Borland-Dennis setting, which has only one generalized Pauli constraint, $D^{(3,6)}\geq 0$ (cf.~(\ref{set36})). Now, we still present for all the other cases with pinning, $\#=7,9,11,12,13$, the constraints which are saturated:
\begin{itemize}
\item $\#7: $ We find saturation of
\begin{eqnarray}
D^{(3,7)}_1 &\equiv& 2-(\lambda_1+\lambda_2+\lambda_5+\lambda_6) \geq 0  \nonumber \\
D^{(3,7)}_4 &\equiv& 2-(\lambda_1+\lambda_2+\lambda_4+\lambda_7) \geq 0
\end{eqnarray}
\item $\#9: $ Besides the saturation of $D^{(3,7)}_4\geq 0$ we also find saturation of
\begin{eqnarray}
D^{(3,7)}_2 &\equiv& 2-(\lambda_1+\lambda_3+\lambda_4+\lambda_6) \geq 0  \nonumber \\
D^{(3,7)}_3 &\equiv& 2-(\lambda_2+\lambda_3+\lambda_4+\lambda_5) \geq 0
\end{eqnarray}
\item $\#11: $ We find saturation of
\begin{eqnarray}
D^{(3,8)}_1 &\equiv& 2-(\lambda_1+\lambda_2+\lambda_5+\lambda_6) \geq 0  \nonumber \\
D^{(3,8)}_7 &\equiv& 1-\lambda_1-\lambda_6+\lambda_7  \geq 0
\end{eqnarray}
\item $\#12: $ We find saturation of
\begin{eqnarray}
D^{(3,8)}_5 &\equiv& 1-\lambda_1-\lambda_2+\lambda_3 \geq 0  \nonumber \\
D^{(3,8)}_{19} &\equiv& -\lambda_1-\lambda_2+2\lambda_3  + \lambda_4 + \lambda_5  \geq 0
\end{eqnarray}
\item $\#13: $ We find saturation of
\begin{equation}
D^{(3,8)}_3 \equiv 2-(\lambda_2+\lambda_3+\lambda_4+\lambda_5) \geq 0
\end{equation}
\end{itemize}

\subsubsection{Finite interaction}\label{sec:finiteInt}
Using Mathematica we investigate the lowest six fermionic excited energy states of (\ref{HamHarmonium}), $1\leq n_- + \Delta n_+\leq 3$, for the weak interaction regime, $\delta \ll 1$ and determine the algebraic behavior of the distance of $\vec{\lambda}(\delta)$ to the polytope boundary as function of the fermion-fermion coupling strength $\delta$.
\begin{table}[h]
\centering
$
\setlength{\arraycolsep}{0.04cm}
\renewcommand{\arraystretch}{1.1}
\begin{array}{c|cc|cc|l|c}
\# & n_- & n_+ & N_{red} & d_{red} & \mbox{order of }|\lambda_i(\delta)-\lambda_i(0)| & \mbox{dist}(\vec{\lambda}(\delta),\partial P) \\
\hline
0&0&0&3&6& (6,4,4,4,4,6,8,10,12,\ldots) & \sim \delta^8 \nonumber \\
\hline
1&1&0&3&6&(4,2,2,2,2,4,8,8,\ldots)& \sim \delta^8 \nonumber \\
\hline
2&2&0&3&6&(4,2,2,2,2,4,4,8,8,\ldots)&\sim \delta^6 \nonumber \\
3&0&2&3&6&(4,2,2,2,2,4,4,8,8,\ldots)&\sim \delta^6 \nonumber \\
\hline
4&3&0&3&8&(2,2,2,2,2,2,4,4,8,\ldots)&\sim \delta^6 \nonumber \\
5&1&2&3&8&(2,2,2,2,2,2,4,4,\ldots)&\sim \delta^6 \nonumber \\
6&0&3&3&8&(2,2,2,2,2,2,4,4,8,8,\ldots)&\sim \delta^6
\end{array}
$
\captionC{Leading order behavior of the NONs for the lowest seven eigenstates of $3$-Harmonium and their corresponding distances to the polytope boundary   in the regime of weak interaction, $\delta \ll 1$. The quasi-pinning is reducible to the setting $\wedge^{N_{red}}[\mathcal{H}_1^{(d_{red})}]$.}
\label{tab:D3Harmex}
\end{table}

Results are presented in Tab.~\ref{tab:D3Harmex}. We observe that each of these first six excited states exhibits strong quasi-pinning, at least for weak interaction. Besides the ground state, also the first excitation is pinned up to corrections of order $\delta^8$. Remarkably, in contrast to the ground state behavior this is now six orders in $\delta$ stronger than its pinning to the Hartree-Fock point (order $\delta^2$). The other five exited states are pinned up to corrections of $\delta^6$. Unfortunately, the computational power did not allow us to study the pinning behavior of the next excitations. Although such quasi-pinning will not show up for the fourteenth and higher excitations (even for zero interaction they are not pinned to the boundary anymore) it would be instructive to investigate the intermediate regime, the excitations $\#7 -\#13$. We are wondering whether the pinning order succinctly reduces from $\delta^6$ to $\delta^4$ to $\delta^2$ and eventually reaches no pinning.

\subsection{Natural orbitals and decay behavior of their occupancies}\label{sec:NO}
The ultimate purpose of the present chapter, Chap.~\ref{chap:Physics}, is to study specific physical systems from the viewpoint of generalized Pauli constraints. So far, by studying the $N$-Harmonium model (\ref{HamHarmonium}), we have found the remarkable effect of quasi-pinning providing strong evidence that the generalized Pauli constraints have some influence on fermionic ground states. Moreover, from Sec.~\ref{sec:truncation} we already have learned that (quasi-)pinning, as effect in the $1$-particle picture, is physically relevant in the sense that it corresponds to very specific and simplified structures of the corresponding $N$-fermion quantum state, expanded w.r.t.~Slater determinants built up from the NOs. Although such strong structural insights are quite spectacular there is no chance to deduce any statement from possible pinning on the NOs. All these structural implications do not make a difference e.g.~between spatially localized or completely delocalized $1$-particle states. They are invariant under $1$-particle basis transformations, i.e.~unitary transformations of the form $U^{\otimes^N}$ acting on $\mathcal{H}_N^{(f)}$. However, for physical applications the concrete spatial form of $1$-particle quantum states is most relevant and at the heart of several physical effect. To shed some light on the NOs, we determine them numerically for the $N$-Harmonium ground state (\ref{gsfermions}). After all, we investigate the decay behavior of the fermionic NONs, in particular for the regime of strong interaction. This regime was not suitable at all for the pinning analysis since no truncation of the NONs to a sufficiently small setting was possible.

The strategy for this section is first to study the eigenvalue equation
\begin{equation}
\rho_1^{(f)} |\chi^{(f)}\rangle = \lambda^{(f)} |\chi^{(f)}\rangle,
\end{equation}
for the fermionic $1$-RDO represented as matrix w.r.t.~the bosonic NOs $\chi_m^{(b)}(x)$, the Hermite functions $\varphi_m^{(L_N)}(x) \equiv \langle x | m\rangle$. The analytic results we will find for the decay behavior of the NONs and the NOs are afterwards confirmed numerically.

We start by expressing the fermionic NOs w.r.t.~the bosonic ones.
\begin{equation}
|\chi^{(f)}\rangle = \sum_{m=0}^{\infty} \zeta_m |m\rangle\,.
\end{equation}
The eigenvalue equation for $\rho_1^{(f)}(x,y)$ reduces to a discrete equation for the expansion coefficients $\{\zeta_m\}$,
\begin{equation}\label{eigenprobmatrixf}
\sum_{n=0}^{\infty} \langle m |\rho_1^{(f)} | n \rangle \zeta_n = \lambda^{(f)} \zeta_m\,.
\end{equation}
In the following we choose the particle number $N$ arbitrary, but fixed. Using Eq.~(\ref{coefrelNOf}) from the Appendix \ref{app:NOsfermionic}, Eq.~(\ref{eigenprobmatrixf}) for sufficiently large $m$ reduces to
\begin{equation}\label{coefrelation}
m^{N-1} e^{-\beta_N \hbar  \Omega_N (m+\frac{1}{2})}\,\sum_{r=-(N-1)}^{N-1}\,h_{m,m-2r}\,\zeta_{m-2r} \simeq \lambda^{(f)} \zeta_m\,.
\end{equation}
Here we used, e.g.\ $\sqrt{m+r} \simeq \sqrt{m}$  for $m \gg 1$  and $r=O(1)$.
For illustration, we discuss this equation for  $N=2$ (for larger $N$ one can proceed similarly), i.e.
\begin{equation}\label{coefrelation2}
m e^{-\beta_2 \hbar  \Omega_2 (m+\frac{1}{2})}\left[h_- \zeta_{m-2} + h_0 \zeta_m +h_+ \zeta_{m+2}\right] \simeq \lambda^{(f)} \zeta_m,
\end{equation}
where $h_0 \equiv h_{m,m}, h_{\pm} \equiv h_{m,m\pm 2}$ do not depend on $m$. For vanishing interaction the eigenfunctions of $\rho_1^{(f)}(x,y)$ are the Hermite functions $\varphi_k^{(L_N)}$. Accordingly, we can label the eigenfunctions by $k$ and find for that case $|\chi_k^{(f)}\rangle = |k \rangle$ and thus $\zeta_m^{(k)} = \delta_{k,m}$. Turning on the interaction we expect the main contributions to $|\chi_k^{(f)}\rangle$ coming still from $|k\rangle$. Since $\zeta_{k\pm 2}$ is at most of the same order as $\zeta_k$, we conclude from Eq.~(\ref{coefrelation2}) with $m=k$ that
\begin{equation}\label{NONf}
\lambda^{(f)} \rightarrow \lambda^{(f)}_k \sim  k \, e^{-\beta_2 \hbar  \Omega_2 (k+\frac{1}{2})}\,, k\gg 1\,.
\end{equation}
For each $k$, $\zeta_m^{(k)}$ for $m \rightarrow \infty$ decays to zero, due to the normalization of  $\chi_k^{(f)}$. Therefore, as a consistency ansatz, let us assume that
$|\zeta_m^{(k)}/\zeta_{m-2}^{(k)}| \ll 1$ for $m\gg k$. This together with (\ref{coefrelation2}) and (\ref{NONf}) leads to
\begin{equation}\label{coefdecayr}
\frac{\zeta_m^{(k)}}{\zeta_{m-2}^{(k)}} \sim \frac{m}{k}\, e^{-\frac{1}{4}\beta_2 \hbar  \Omega_2 (m-k)}\,,
\end{equation}
which is indeed consistent with our assumption $|\zeta_m^{(k)}/\zeta_{m-2}^{(k)}| \ll 1$. Moreover, from (\ref{coefdecayr}) we obtain the Gaussian decay behavior
\begin{equation}
\zeta_m^{(k)} \sim  e^{-\frac{1}{4}\beta_2 \hbar  \Omega_2 (m-k)^2}
\end{equation}
for $k\gg1$ and $m\gg k$.
For the opposite regime, $1\ll m \ll k$, and taking $h_0 = O(m^0)$ into account we have
\begin{eqnarray}
|m e^{-\beta_2 \hbar  \Omega_2 (m+\frac{1}{2})} h_0 \zeta_m^{(k)}| & \gg & |\lambda_k \zeta_m^{(k)}| \nonumber \\
&\propto & |k  e^{-\beta_2 \hbar  \Omega_2 (k+\frac{1}{2})} \zeta_m^{(k)}|\,.
\end{eqnarray}
Eq.~(\ref{coefrelation2}) then implies
\begin{equation}\label{coefrelationl}
\zeta_m^{(k)} + \frac{h_-}{h_0} \,\zeta_{m-2}^{(k)} + \frac{h_+}{h_0} \,\zeta_{m+2}^{(k)}\simeq 0\,,
\end{equation}
which is solved by an exponential
\begin{equation}\label{coefdecayl}
\zeta_m^{(k)} \sim  e^{\alpha_2 (m-k)}\,,
\end{equation}
where $\alpha_2$ depends on $h_0, h_{\pm}$, but not on the orbital index $k$.  $\alpha_2$ can be determined by plugging the ansatz (\ref{coefdecayl}) into Eq.~(\ref{coefrelationl}) and solving the emerging quadratic equation for $e^{2\alpha_2}$. Since $\zeta_m^{(k)}$ for $1 \ll m \ll k$ should decay with decreasing $m$, the root with $Re(\alpha_2) > 0$ should be taken.

By just repeating all these steps for Eq.~(\ref{coefrelation}) we find for arbitrary $N$
\begin{equation}\label{NONfN}
\lambda^{(f)} \rightarrow \lambda^{(f)}_k \sim  k^{N-1}  e^{-\beta_N \hbar  \Omega_N (k+\frac{1}{2})}\,,k\gg 1\,.
\end{equation}
The decay behavior of $\zeta_m^{(k)}$  for $m\gg k$ is again Gaussian,
\begin{equation}\label{coefdecayrN}
\zeta_m^{(k)} \sim  e^{-\frac{1}{4(N-1)}\beta_N \hbar  \Omega_N (m-k)^2}\,,
\end{equation}
and exponential for $1\ll m \ll k$ ,
\begin{equation}\label{coefdecaylN}
\zeta_m^{(k)} \sim  e^{\alpha_N (m-k)}\,,
\end{equation}
where $\alpha_N$ depends on $h_0, h_{\pm r} \equiv h_{m,m\pm 2 r}, r=1,\ldots,N-1$, but not on the orbital index $k$. $\alpha_N$  is the root of a polynomial of degree $2(N-1)$ for which $Re(\alpha_N) > 0$.

In order to check all the analytical predictions made in this section for the NONs $\lambda_k^{(f)}$ and the NOs $\chi_k^{(f)}$ we have solved Eq.~(\ref{eigenprobmatrixf}) numerically for $N=3$ and $N=5$, by representing $\rho_1^{(f)}$ and the states $|\chi_k^{(f)}\rangle$ again w.r.t to the bosonic NOs $|\chi_m^{(b)}\rangle$ and then truncating the corresponding matrix $((\rho_1^{(f)})_{n,m})$ and vectors $(\zeta_m^{(k)}), \zeta_m^{(k)} \equiv \langle \chi_m^{(b)} |\chi_k^{(f)}\rangle $ at $m_{max}$. All the results presented here are obtained with $m_{max}=500$. As dimensionless interaction strengths we choose $l_+/l_- = 4/5, 1/2$ and $1/3$, which corresponds (according to Eq.~(\ref{LvsHook})) to $ND/(m \omega^2) = 369/256 \simeq 1.44, 15$ and $80$.
The numerical calculations in particular allow us to investigate $\lambda_k$ for the regime $k = O(1)$.
\begin{figure}[!h]
\includegraphics[width=8cm]{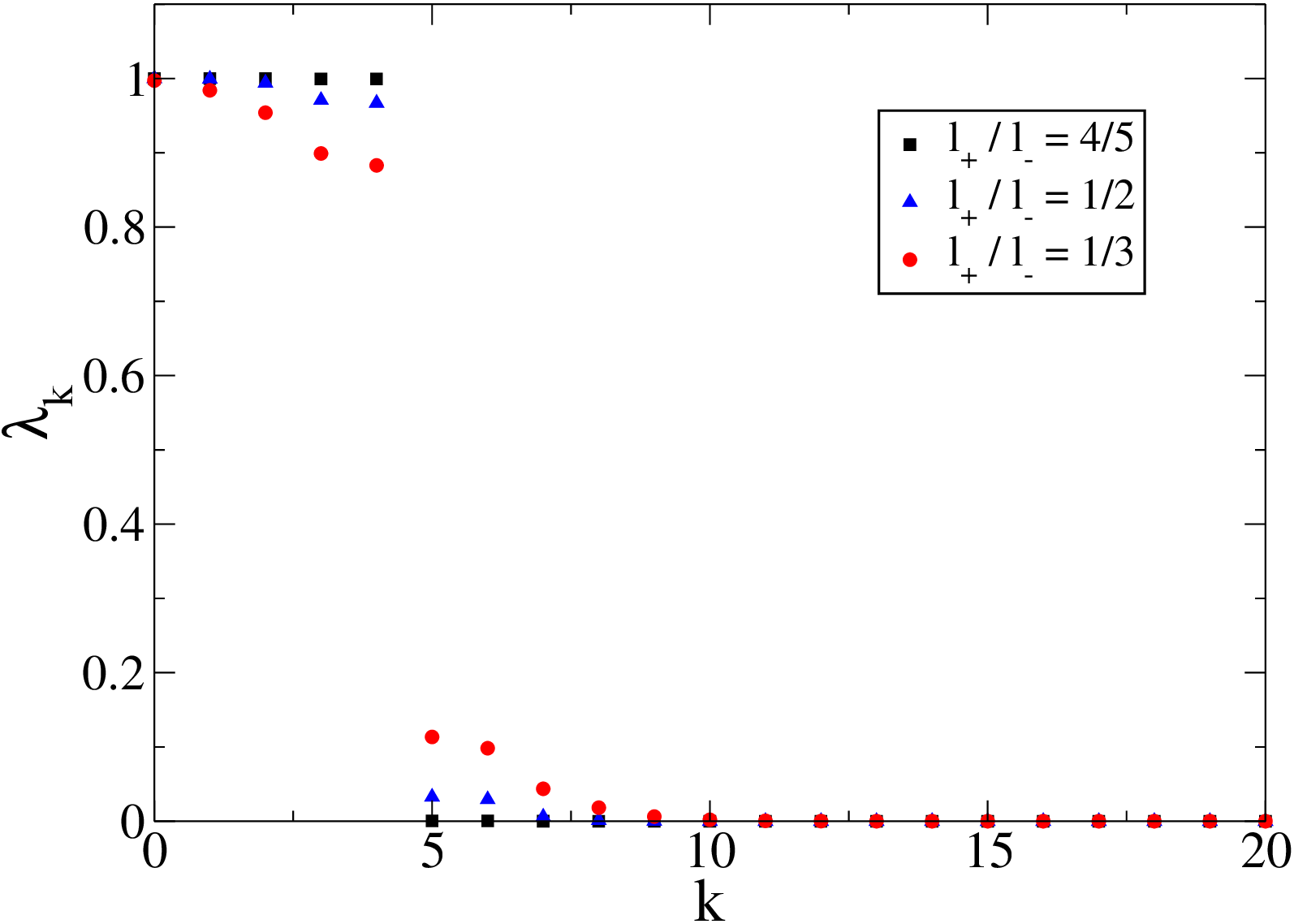}
\centering
\captionC{NONs $\lambda_k^{(f)}$ for $N=5$ and three different interaction strengths (see legend).}
\label{fig:non}
\end{figure}

Fig.~\ref{fig:non} depicts the NONs $\lambda_k^{(f)}$ for three different coupling strengths and  $N=5$. Note, even  for quite strong interaction the ``gap'' at the ``Fermi level'' is still well pronounced.
\begin{figure}[!h]
\includegraphics[width=8cm]{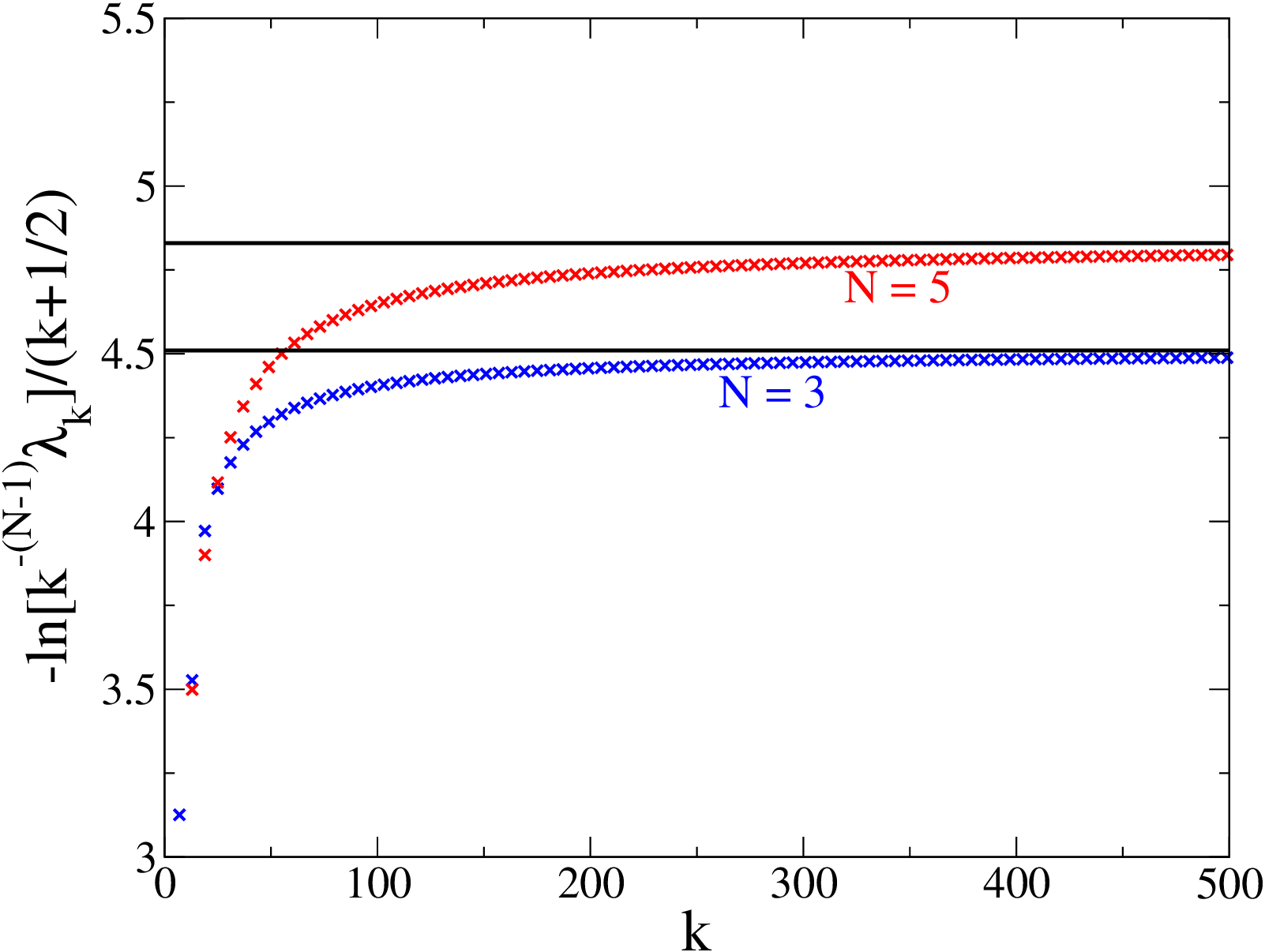}
\centering
\captionC{$k$-dependence of $-\ln{(\lambda_k k^{-(N-1)})}/(k+\frac{1}{2})$ for $N=3,5$ and interaction $l_+/l_- = 4/5$. The horizontal lines represent the asymptotic values $\beta_N \hbar \Omega_N$.}
\label{fig:nonbeta}
\end{figure}
Although the $\lambda_k^{(\alpha)}$'s for $k\gg N$ behave very similar for bosons and fermions this is not true  anymore for the regime $k = O(N)$ or smaller. Whereas $\lambda_k^{(b)}$ exhibit a purely exponential decay for all $k$, $\lambda_k^{(f)}$ has a ``gap'' at the ``Fermi level'' $k_F =N$. For zero interaction, it is
\begin{equation}\label{FermiDirac0}
\lambda_k^{(f)} = \begin{cases}
  1,  & k < k_F\\
  0, & k  \geq k_F\,.
\end{cases}
\end{equation}
With increasing interaction Fig. \ref{fig:non} demonstrates that $\lambda_k^{(f)}$ deviates from one for $k<k_F$ and from zero for $k \geq k_F$. The gap at $k_F$ becomes smaller but remains significantly large even for rather strong interactions. This behavior resembles the Fermi-Dirac distribution function. At zero ``temperature'', which corresponds to zero interaction, this distribution function is identical with the behavior in Eq.\ (\ref{FermiDirac0}). The ``softening'' of the $k$-dependence of $\lambda_k^{(f)}$ with increasing interaction strength corresponds to the softening of the Fermi-Dirac distribution for increasing temperature. It would be interesting to study $\lambda_k^{(f)}$ in the ``thermodynamic'' limit, i.e.\ $\lambda^{(f)}(\tilde{k}) = \lim_{N\rightarrow \infty} \lambda^{(f)}_{N \tilde{k}}$ and to investigate the dependence of the gap in $\lambda^{(f)}(\tilde{k})$ on the coupling constant $D$ at the ``Fermi level'' $\tilde{k}_F=1$, provided the gap survives the limit $N\rightarrow \infty$.

In Fig.~\ref{fig:nonbeta} we verify the dominant Boltzmann-like behavior found in Eq.~(\ref{NONfN}) for interaction strength $l_+/l_- = 4/5$ by plotting the $k$-dependence of $-\ln{(\lambda_k k^{-(N-1)})}/(k+\frac{1}{2})$,
which should converge for $k\rightarrow \infty$ to the constant $\beta_N \hbar \Omega_N$. Indeed, this happens since the curves are approaching the values $\beta_3 \hbar \Omega_3 \simeq 4.51$ and $\beta_5 \hbar \Omega_5 \simeq 4.83$ quite well.

One of the most remarkable results of our analysis is shown in Fig.~\ref{fig:coef}. Even for $l_+/l_- = 1/3$, which for $N=5$ corresponds to a very large coupling ratio $D/(m \omega^2) = 16$, the fermionic NOs $\chi_k^{(f)}$ are very well approximated by a superposition of very few bosonic orbitals $\chi_m^{(b)}$, with $m\approx k$, and the dominant weight comes from $m=k$. This also holds for smaller orbital indices $k$, than those shown in Fig.~\ref{fig:coef}, except for $k=2,3,\ldots,8$, those closer to the ``Fermi level''. Nevertheless, for comparable couplings, $D/(m\omega^2) \approx 1$ or weak interaction also the fermionic orbitals in the vicinity of the ``Fermi level'' are well approximated by the bosonic counterparts.
\begin{figure}[!h]
\includegraphics[width=8cm]{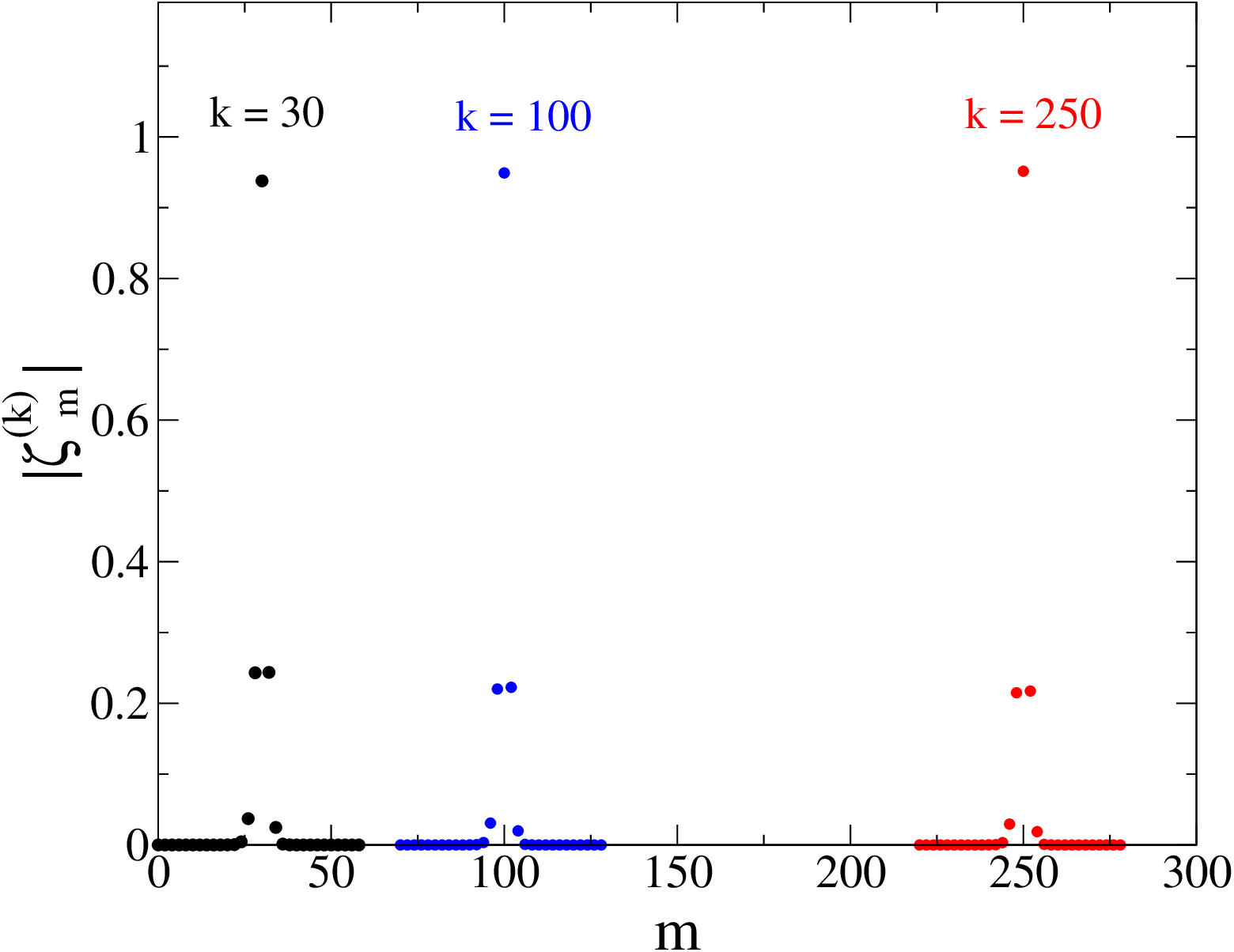}
\centering
\captionC{Expansion coefficients $\zeta_m^{(k)} \equiv \langle \chi_{m}^{(b)} | \chi_{k}^{(f)} \rangle$ for the NOs $\chi_{k}^{(f)}$, $k=30,100$ and $250$ for $N=5$ and $l_+/l_- = 1/3$ for the relevant regime $m \approx k$.}
\label{fig:coef}
\end{figure}
To verify the Gaussian decay, Eq.~(\ref{coefdecayrN}), for $m\gg k$ and for $m \ll k$ the exponential decay, Eq.~(\ref{coefdecaylN}),  of the expansion coefficients $\zeta_m^{(k)} \equiv \langle \chi_m^{(b)} |\chi_k^{(f)}\rangle$   we plot  $-\ln{(|\zeta_m^{(k)}|)}/(m-k)^2$ and $-\ln{(|\zeta_m^{(k)}|)}/(k-m)^2$, respectively, as a function of $m-k$. From Fig.~\ref{fig:coefr}, one can infer that $\zeta_m^{(k)}$  indeed decays Gaussian-like, and the decay constants are as predicted in Eq.~(\ref{coefdecayrN}), i.e.\ $\frac{1}{8}\beta_3 \hbar \Omega_3 \simeq 0.56$ and $\frac{1}{16}\beta_5 \hbar \Omega_5 \simeq 0.30$.
\begin{figure}[!h]
\includegraphics[width=6.5cm]{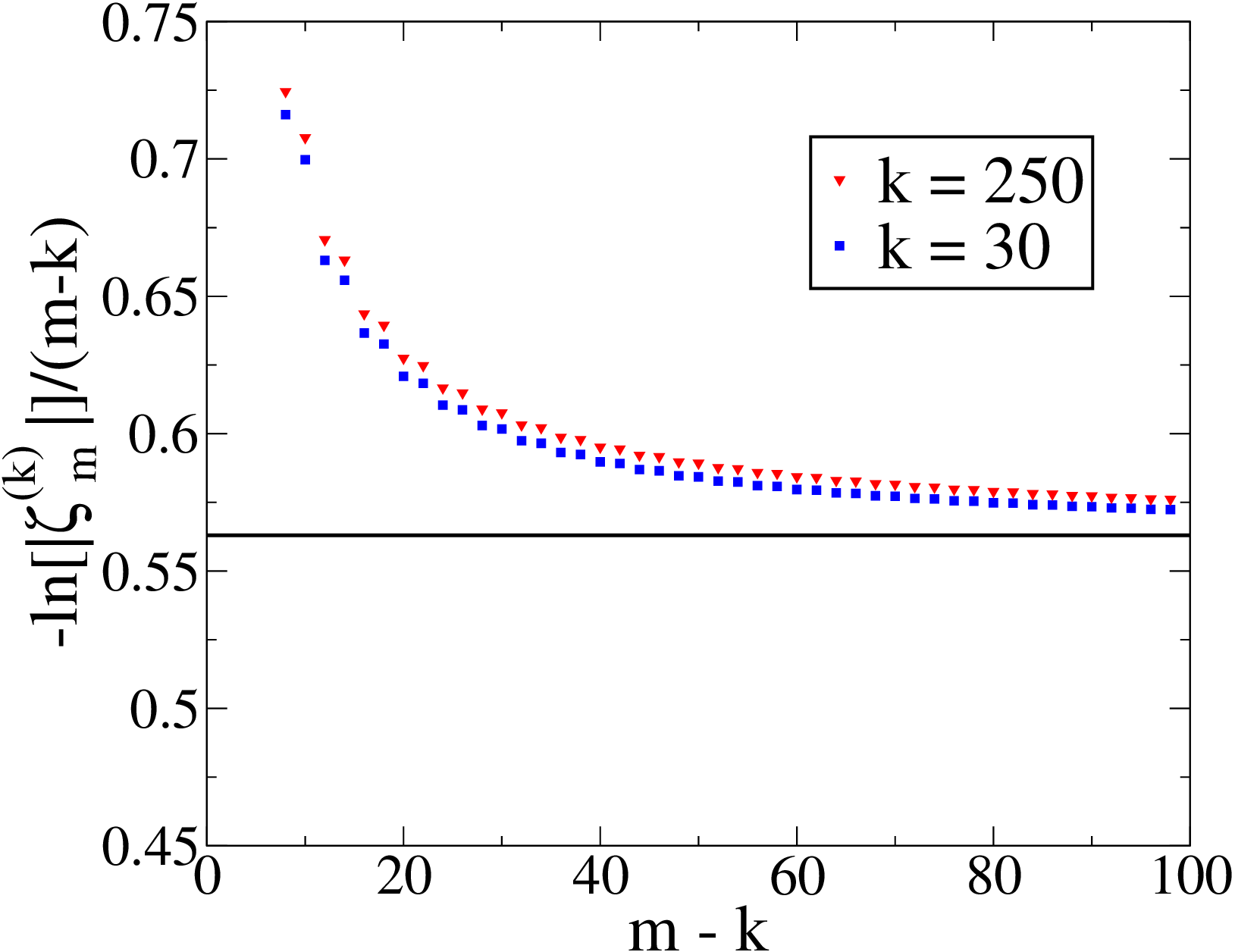}\hspace{0.4cm}
\includegraphics[width=6.5cm]{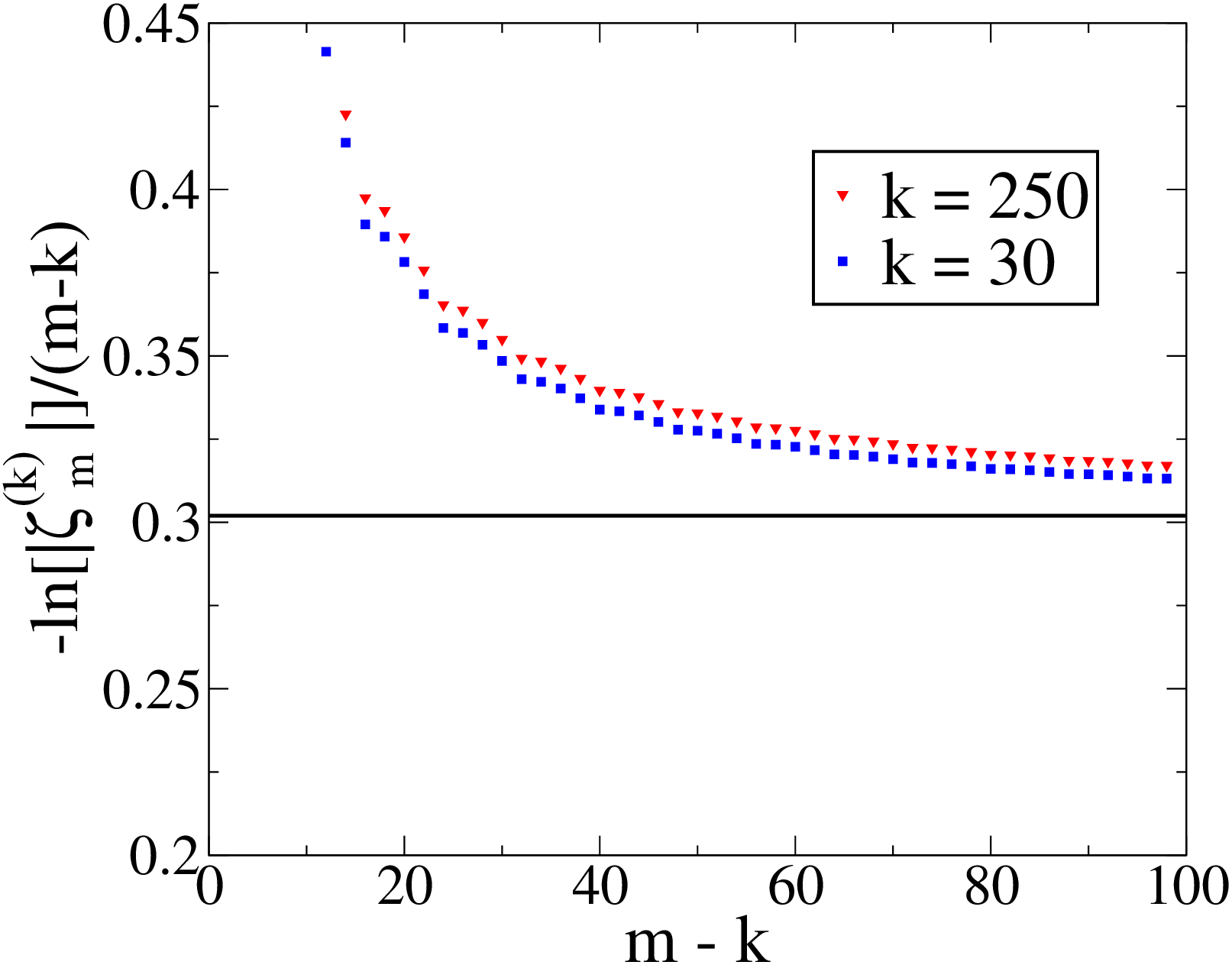}
\centering
\captionC{$-\ln{(|\zeta_m^{(k)}|)}/(m-k)^2$ as function of $m-k$ for the orbital indices $k= 30,100,250$ and  $l_+/l_- = 4/5$. Left:$N=3$ and  right: $N=5$. The horizontal lines represent $\beta_N \hbar \Omega_N/(4 (N-1))$.}
\label{fig:coefr}
\end{figure}
Fig.~\ref{fig:coefl} confirms the average exponential decay for the regime $1\ll m\ll k$.
\begin{figure}[!h]
\includegraphics[width=6.5cm]{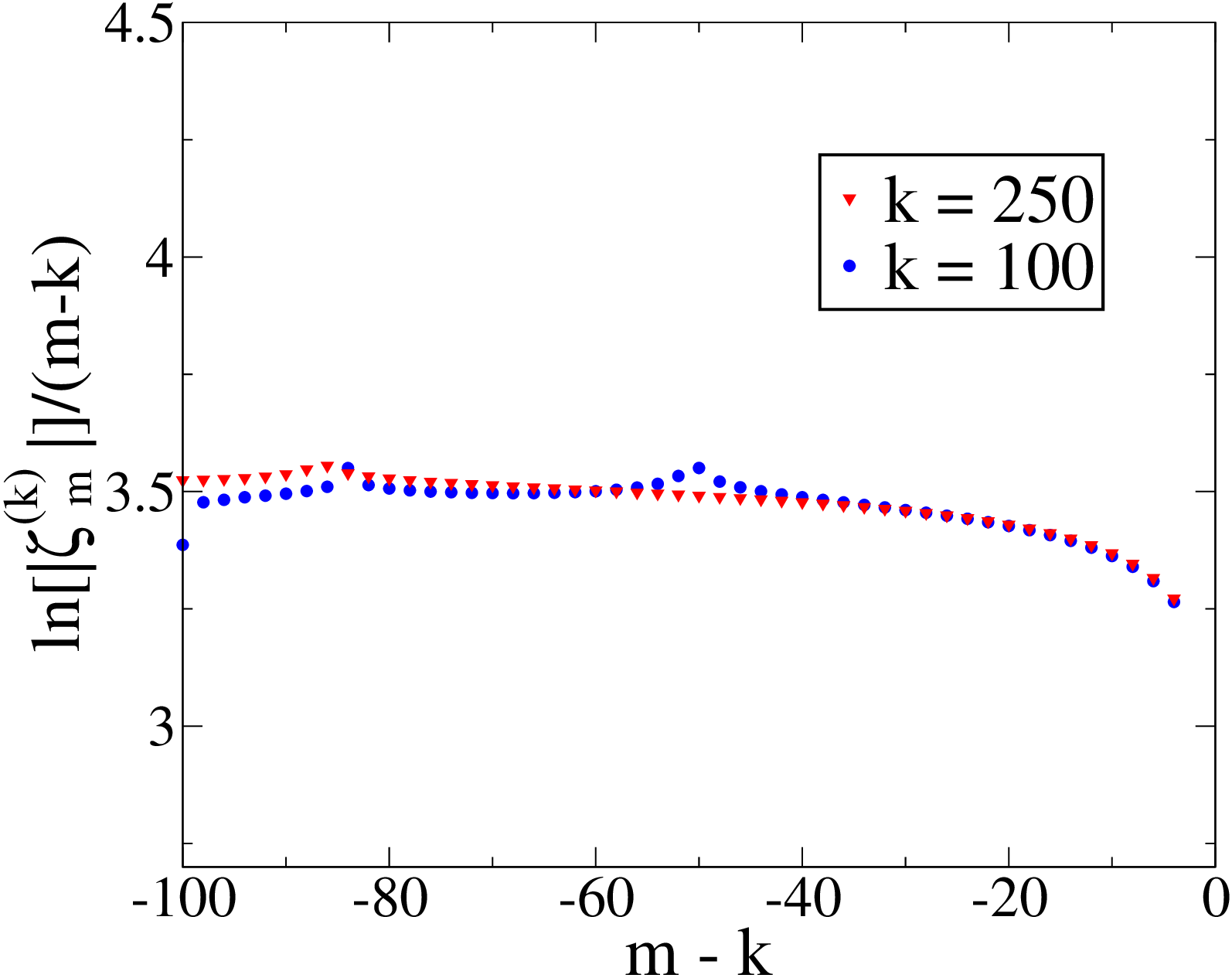}\hspace{0.4cm}
\includegraphics[width=6.5cm]{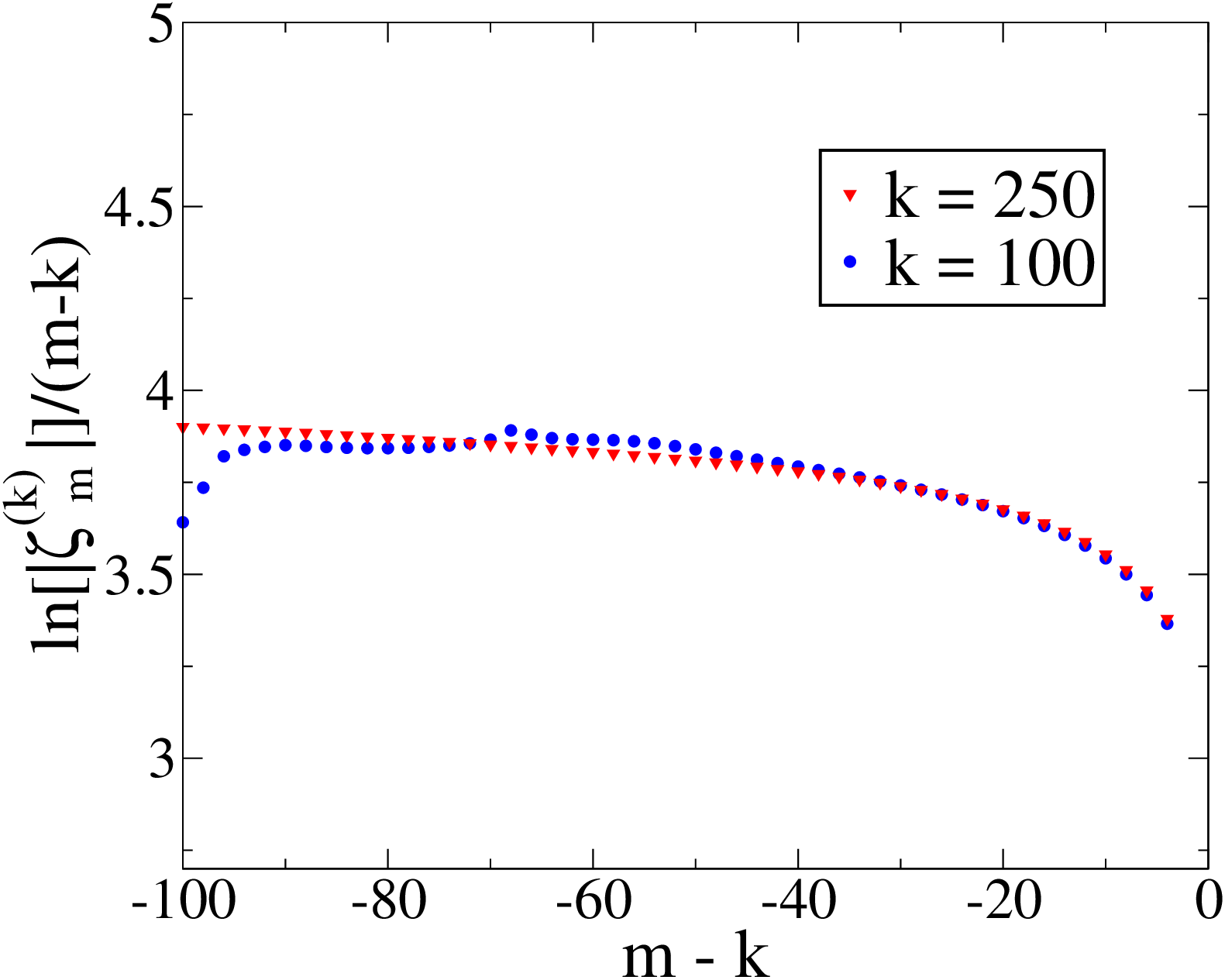}
\centering
\captionC{$-\ln{(|\zeta_m^{(k)}|)}/(k-m)^2$ as function of $k-m$ for the orbital indices $k= 100,250$ and  $l_+/l_- = 4/5$. Left: $N=3$ and  right: $N=5$. }
\label{fig:coefl}
\end{figure}

We close this discussion of NOs by presenting one general idea connecting the result on pinning with that on the strong similarity between the fermionic and bosonic NOs.
For this we consider a system of $N$ identical interacting particles for which one would like to find a good approximation of the fermionic ground state $|\Psi_N\rangle$. Assume that the following two statements hold
\begin{enumerate}
\item The unknown fermionic ground state is (quasi)-pinned
\item The bosonic and fermionic NOs are very similar
\end{enumerate}
Then, if we can determine the bosonic ground state we can make an excellent guess for the fermionic ground state: Since pinning corresponds to specific structures of the corresponding $N$-fermion state (see Sec.~\ref{sec:truncation}), statement $1$ already fixes the structure of $|\Psi_N\rangle$, expanded in specific Slater determinants, built up from the unknown fermionic NOs. Moreover, due to point $2$ we can replace them by the known bosonic NOs.
This strategy fixes $|\Psi_N\rangle$ up two a few unknown expansion coefficients, which follow from the concrete position of the fermionic NONs on (close to) the boundary of the polytope. If this information is not feasible one can determine those coefficients variationally.

\section{Hubbard model with a few lattice sites}\label{sec:hubbard}
\subsection{Introduction}\label{sec:hubbardIntro}
In this section we investigate possible pinning for a second physical model, the Hubbard model. Although it is one of the most elementary models it describes quite well some of the most important features in condensed matter physics, as e.g.~the existence of Mott-insulators and their phase transition. However, since the generalized Pauli constraints are just known for settings with at most ten-dimensional $1$-particle Hilbert spaces we cannot study the macroscopic regime. Therefore, we restrict to microscopic situations, which also play some role, at least in quantum optics. By loading optical lattices with ultracold atoms effects of condensed matter physics can be reproduced \cite{Bloch05, Esslinger2008}.
In particular, as a new direction the transition of few-fermion physics to macroscopic physics is studied by considering just a few optical traps with ultracold fermionic atoms. Recent techniques already allow to control those fermionic atoms, tune their interaction and measure occupancies in the $1$-particle picture \cite{Jochim2013}.

The model, we consider is the $1$-band Hubbard model for electrons on a $1$-dimensional lattice $\mathcal{L}_X$ with $d$ sites and we choose periodic boundary conditions. The corresponding Hamiltonian in second quantization reads
\begin{equation}\label{HamHubbard}
H = -t\,\sum_{i=0}^{d-1} \sum_{\sigma}\,\left(c_{i+1\, \sigma}^\dagger\,c_{i \sigma}\,+\mbox{h.c.}\right)\,\,+\,\,u\,\sum_{i=0}^{d-1}\,n_{i\uparrow}\,n_{i\downarrow}
\end{equation}
where $c_{i\sigma}^\dagger$ and $c_{i\sigma}$ are the creation and annihilation operators for an electron on the $i$-th lattice site with spin $\sigma = \uparrow,\downarrow$ w.r.t.~the $3$-axis. They fulfill the standard anticommutation relations
\begin{equation}
\{c_{i\sigma},c_{j\mu}\}\,=\,\{c_{i\sigma}^\dagger,c_{j\mu}^\dagger\}\,=\,0\,\,,\,\,\,\,\{c_{i\sigma},c_{j\mu}^\dagger\}\,=\,\delta_{ij}\,\delta_{\sigma\mu}\,.
\end{equation}
Moreover, $n_{i \sigma}\,\equiv c_{i\sigma}^\dagger\,c_{i\sigma}$ is the particle number operator for the state on site $i$ with spin $\sigma$ and w.l.o.g. we set $t=1$ in (\ref{HamHubbard}).

\subsection{Symmetries}\label{sec:symmetriesLattice}
To simplify calculations below we first study the symmetries of the Hubbard Hamiltonian (\ref{HamHubbard}) (see also \cite{Fradkin}). Let us denote the $3$-vector of Pauli matrices by $\vec{\boldsymbol \sigma}\,\equiv\,(\boldsymbol{\sigma}_1,\boldsymbol{\sigma}_2,\boldsymbol{\sigma}_3)$. Then, the total spin operator $\vec{S}$ has the compact form (we set $\hbar \equiv 1$)
\begin{equation}
S_k\,\equiv \frac{1}{2}\,\sum_{i=0}^{d-1}\,\sum_{\mu,\tau}\,c_{i\mu}^\dagger\,(\boldsymbol{\sigma}_k)_{\mu \tau}\,c_{i \tau}\,,\,\,k=1,2,3\,.
\end{equation}
Since the total spin turns out to be also a good quantum number, we recall the form
\begin{equation}\label{spinsquared}
\vec{S}^2 = \sum_{i,j=0}^{d-1}\,S_3^{(i)}\,S_3^{(j)}\,+\,\frac{1}{2}\,\left(S_+^{(i)}\,S_-^{(j)}\,+\,S_-^{(i)}\,S_+^{(j)}\right)\,,
\end{equation}
where $S_{\pm}^{(j)}\equiv S_{1}^{(j)}\pm i S_{2}^{(j)}$ are the ladder operators for the $j$-th lattice site's spin space.
As elementary exercise one can verify the following commutation relations
\begin{equation}\label{HamHubbardSpin}
[H,\vec{S}^{\,2}]\,=\,0\,\,,\,\,\,[H,S_3]\,=\,0\,\,,\,\,\,[H,n]\,=\,0,
\end{equation}
where $n\,\equiv\,\sum_{i\,\mu}\,n_{i\mu}$ is the total electron number operator. Due to the particle number conservation we restrict to the sector of $N$ electrons. There are further symmetries referring to the orbital (spatial) degrees of freedom. Let $T_{k}$ be the translation of the (real space) lattice $\mathcal{L}_X$ by $k$ sites. Then,
\begin{equation}\label{HamHubbardTransl}
[H,T_{k}]\,=\,0\qquad,\,\forall\,k\,=\,\mathbb{Z}\,.
\end{equation}
and $T_d\,=\,(T_1)^{d}\,\equiv\,T^{\,d}\,=\,\mbox{Id}$. Moreover, let $P_X$ be the ``inversion'' operator that inverts the $i$-th lattice site to the $(-i)$-th one, i.e.~to the $(d-i)$-th one. We find
\begin{equation}\label{HamHubbardPX}
[H,P_X]\,=\,0
\end{equation}
and the orbital symmetry group of $H$ is generated by the two operators $T$ and $P_X$. Since the antisymmetrizing operator $\mathcal{A}_N$ commutes with $H, T, P_X,\vec{S}^2, S_3$, the $N$-electron energy eigenstates are (or can be chosen as) eigenstates also of \footnote{Note that we omit $P_X$ since $[T,P_X]\neq 0$} $T$, $\vec{S}^2$, $S_3$ and we can split the corresponding $N$-electron Hilbert space
\begin{equation}
\wedge^N[\mathcal{H}_1] \equiv \wedge^N[\mathcal{H}_1^{(l)}\otimes \mathcal{H}_1^{(s)}]\,=\,\bigoplus_{E,S,M,K}\,\mathcal{H}_{E,S,M,K}
\end{equation}
where $\mathcal{H}_1\equiv \mathcal{H}_1^{(l)}\otimes \mathcal{H}_1^{(s)} \cong \mathbb{C}^d\otimes \mathbb{C}^2$ is the $1$-particle Hilbert space, $E$ is the energy, $S$ the total spin, $M$ the total magnetic quantum number (w.r.t.~the $3$-axis) and $K=0,1,\ldots,d-1$ the $1$-dimensional wave number ``vector'' up to a dimensional factor. In detail, by first using Weyl's decomposition in orbital and spin part \cite{Weyl} we find
\begin{equation}\label{Weyl}
\wedge^N[\mathcal{H}_1^{(l)}\otimes \mathcal{H}_1^{(s)}]\,=\,\bigoplus_{\nu,\,|\nu|=N}\,\mathcal{H}_{\nu}^{(l)}\otimes \mathcal{H}_{\nu^T}^{(s)}\,.
\end{equation}
Here, $\mathcal{H}_{\mu}^{(l)}, \mathcal{H}_{\mu}^{(s)}$ denote the irreducible representations of the group $SU(\mathcal{H}_1^{(l)})$ and $SU(\mathcal{H}_1^{(s)})$, respectively, labeled by a Young diagram $\mu$. In Eq.~(\ref{Weyl}) we sum over all Young diagrams $\nu$ with $N$ boxes and (due to the spin-$\frac{1}{2}$ of the electron) at most 2 columns. Since $\mathcal{H}_1^{(s)}$ has dimension $2$, only transposed diagrams $\nu^T$ with at most two rows can show up and one finally can show
\begin{equation}
S = \frac{1}{2}\,\left[(\nu^T)_1- (\nu^T)_2\right]\,.
\end{equation}
We can express the rhs of (\ref{Weyl}) in a simpler form by labeling the Young diagram $\nu$ by the corresponding spin $S$,
\begin{equation}
\nu(S) = \left(\frac{N}{2}+S,\frac{N}{2}-S\right)\,
\end{equation}
and find
\begin{equation}
\wedge^N[\mathcal{H}_1^{(l)}\otimes \mathcal{H}_1^{(s)}]\,=\,\bigoplus_{S=S_-}^{S_+}\,\mathcal{H}_{\nu(S)}^{(l)}\otimes \mathcal{H}_{\nu(S)^T}^{(s)}\,,
\end{equation}
where $S_-$ and $S_+$ are the minimal and maximal spin (for the coupling of $N$ spin-$\frac{1}{2}$'s), respectively.
Finally,
\begin{eqnarray}
\wedge^N[\mathcal{H}_1^{(l)}\otimes \mathcal{H}_1^{(s)}]\, &=&\bigoplus_{S=S_-}^{S_+}\,
\bigoplus_{M=-S}^S\,\bigoplus_{K=0}^{d-1}\,
\mathcal{H}_{E,S,M,K}\,, \nonumber \\
\mathcal{H}_{E,S,M,K} &\equiv& \pi_E\, \bigoplus_{S=S_-}^{S_+}\,\mathcal{H}_{\nu(S)}^{(l,K)}\otimes \mathcal{H}_{\nu(S)^T}^{(s,M)}\,
\end{eqnarray}
and $\pi_E$ is the projection operator onto the subspace corresponding to the energy $E$.

\subsection{Three electrons on three sites: analytic results}\label{sec:hubbardAnalyt}
In this section we study the $1$-dimensional Hubbard model (\ref{HamHubbard}) for three electrons on three lattice sites and periodic boundary conditions.  Consequently, we deal with the Borland-Dennis setting, $\mathcal{H}_3 = \wedge^3[\mathcal{H}_1^{(6)}]$. In Sec.~\ref{sec:hubbardLarger} we will also consider larger settings, with more sites and more electrons. Note that the finiteness of the Hilbert space underlying the Hubbard model is due to a truncation of the physical Hilbert space $\wedge^N[L^2(\R^3)\otimes \C^2]$, which was applied in the derivation of the Hubbard model. This in particular means that in contrast to the Harmonium model (\ref{HamHarmonium}) studied in Sec.~\ref{sec:NHarm} we do not need to truncate the NONs for the pinning analysis after solving the model and calculating the $1$-RDO. This will also shed some light on the `exact' pinning found by Klyachko for the Beryllium state obtained by variational methods based on finite dimensional spaces (see also \cite{Kly1,CS2013QChem}).

First, we calculate analytically the eigenenergies and determine the ground state. Since the total Hilbert space has already dimension $\binom{6}{3} =20$,
we use the simplifications following from the symmetries introduced in Sec.~\ref{sec:symmetriesLattice}. It turns out that most of the results we find do not directly depend on the Hamiltonian (\ref{HamHubbard}), but just follow from its symmetries. In that sense the results derived in the present section have a more universal meaning.

To succinctly split the $3$-electron Hilbert space $\mathcal{H}_3$ according to the symmetries of $H$ we introduce a basis of Slater determinants. As corresponding $1$-particle basis we choose the six states
$|k,\sigma\rangle_Q \equiv |k\rangle_Q\otimes |\sigma\rangle $, $k =0,1,2$, $\sigma = \uparrow,\downarrow$. Here $|k\rangle_Q$, $k=0,1,2$ are the $1$-particle Bloch states. This choice of the basis is optimal for the calculation of the NONs of $H$-eigenstates: According to the examples Ex.~\ref{ex:Spincomponent} and Ex.~\ref{ex:translationBloch} the $1$-RDO $\rho_1$ of every $N$-electron $S_3$-spin eigenstate $|\Psi_N\rangle$ adapted to the translational symmetry is diagonal w.r.t.~$|k,\sigma\rangle_Q$ and the NONs are therefore just given by its diagonal entries.

To visualize the symmetries we introduce an occupation number picture and describe every $3$-electron Slater determinant by a symbol of the form $ \z{1}{0}{1}{0}{0}{1}$. In such a $2 \times 3$ array every dot represents one electron occupying a state $|k,\sigma\rangle_Q$, where the two rows stand for $\uparrow$ and $\downarrow$ and the three columns (from left to right) represent $|0\rangle_Q$, $|1\rangle_Q$ and $|2\rangle_Q$.

As a first split of the $3$-electron Hilbert space we consider the four possible magnetic quantum numbers $M=-\frac{3}{2}, -\frac{1}{2},\frac{1}{2},\frac{3}{2}$. The subspaces corresponding to $M=-\frac{3}{2}$ and $\frac{3}{2}$ are both $1$-dimensional and are spanned by the states represented by
\begin{equation}
\z{0}{1}{0}{1}{0}{1} \qquad \mbox{and}\qquad \z{1}{0}{1}{0}{1}{0}\,,
 \end{equation}
respectively. Note that we have $ \z{1}{0}{1}{0}{1}{0}
= \zX{1}{0}{1}{0}{1}{0}$ and $ \z{0}{1}{0}{1}{0}{1}
= \zX{0}{1}{0}{1}{0}{1}$, where $\zX{0}{0}{0}{0}{0}{0}$ is again a symbolical notation, now w.r.t.~the real space lattice $\mathcal{L}_X$.
The two subspaces corresponding to $M=-\frac{1}{2},\frac{1}{2}$ are also isomorphic (isomorphism is given by the flipping of all three spins) and thus we consider only the case $M = \frac{1}{2}$. The corresponding subspace is $9$-dimensional and a basis is given by
$\z{2}{1}{1}{2}{1}{2}$, $\z{1}{2}{2}{1}{1}{2}$, $\z{1}{2}{1}{2}{2}{1}$, $\z{1}{1}{1}{2}{2}{2}$, $\z{2}{2}{1}{1}{1}{2}$, $\z{1}{2}{2}{2}{1}{1}$, $\z{1}{1}{2}{2}{1}{2}$, $\z{1}{2}{1}{1}{2}{2}$, $\z{2}{2}{1}{2}{1}{1}$.
Those $3$-electron states are already $T$-eigenstates, but unfortunately not all are eigenstates of $\vec{S}^2$. To make them eigenstates of $\vec{S}^2$ we introduce three different $K=0$ states,
\begin{eqnarray}
|K=0\rangle_a &\equiv& \frac{1}{3}\,\left[\z{2}{1}{1}{2}{1}{2}\,+\,
\z{1}{2}{2}{1}{1}{2}\,+\,
\z{1}{2}{1}{2}{2}{1}\right] \nonumber \\
|K=0\rangle_b &\equiv& \frac{1}{3}\,\left[\z{2}{1}{1}{2}{1}{2}\,+\,e^{\frac{2 \pi}{3} i}\,
\z{1}{2}{2}{1}{1}{2}\,+\,e^{\frac{4 \pi}{3} i}\,
\z{1}{2}{1}{2}{2}{1}\right] \nonumber \\
|K=0\rangle_c &\equiv& \frac{1}{3}\,\left[\z{2}{1}{1}{2}{1}{2}\,+\,e^{\frac{4 \pi}{3} i}\,
\z{1}{2}{2}{1}{1}{2}\,+\,e^{\frac{2 \pi}{3} i}\,
\z{1}{2}{1}{2}{2}{1}\right]\,.
\end{eqnarray}
The first has eigenvalue $S=\frac{3}{2}$. The other two and also the remaining six states
\begin{eqnarray}
|K=1\rangle_a \equiv \z{1}{1}{1}{2}{2}{2}\,,\,
|K=1\rangle_b \equiv \z{2}{2}{1}{1}{1}{2}\,,\,
|K=1\rangle_c \equiv \z{1}{2}{2}{2}{1}{1} &&\nonumber \\
|K=2\rangle_a \equiv \z{1}{1}{2}{2}{1}{2}\,,\,
|K=2\rangle_b \equiv \z{1}{2}{1}{1}{2}{2}\,,\,
|K=2\rangle_c \equiv \z{2}{2}{1}{2}{1}{1} &&,
\end{eqnarray}
are $S=\frac{1}{2}$-eigenstates. Moreover, notice that by introducing the inversion/reflection operator $P_Q$ for the wavenumber space we have
\begin{eqnarray}\label{BrillReflex}
P_Q^{(3)} \,\z{1}{1}{1}{2}{2}{2} &=&
\z{1}{1}{2}{2}{1}{2}  \nonumber \\
P_Q^{(3)} \,\z{2}{2}{1}{1}{1}{2} &=&
\z{2}{2}{1}{2}{1}{1}      \nonumber \\
P_Q^{(3)} \,\z{1}{2}{2}{2}{1}{1} &=&
\z{1}{2}{1}{1}{2}{2}
\end{eqnarray}
and
\begin{equation}
[H,P_Q]=0\,.
\end{equation}
This completes our construction of a symmetry adapted $3$-electron basis and the total Hilbert space splits according to
\begin{eqnarray}\label{Hilbertsplited}
\mathcal{H}_3 &=& \bigoplus_{S=\frac{1}{2}, \frac{3}{2}} \bigoplus_{M=-S}^S \bigoplus_{K=0, \pm 1} \mathcal{H}_{S,M,K}  \\
&=&\big(\mathcal{H}_{\frac{3}{2},\frac{3}{2},0} \oplus \mathcal{H}_{\frac{3}{2},\frac{1}{2},0} \oplus \mathcal{H}_{\frac{1}{2},\frac{1}{2},0} \oplus \mathcal{H}_{\frac{1}{2},\frac{1}{2},1}\oplus P_Q \mathcal{H}_{\frac{1}{2},\frac{1}{2},1}\big)\oplus I_s \left(\cdot\right) \,,\nonumber \\
&& \hspace{0.39cm}(\hspace{0.10cm} 1 \hspace{0.40cm} + \hspace{0.45cm} 1\hspace{0.41cm} + \hspace{0.44cm} 2 \hspace{0.41cm} + \hspace{0.44cm} 3 \hspace{0.41cm} + \hspace{0.65cm} 3 \hspace{0.10cm}) \hspace{0.59cm} \times \hspace{0.28cm} 2 \nonumber
\end{eqnarray}
where $I_s$ is the spin flip operator. The last line presents the dimensions of the corresponding subspaces and $(\cdot)$ stands for the whole direct sum of subspaces presented in the same line to its left.

Diagonalizing $H$ is now much easier since it is block-diagonal w.r.t.~the subspaces in (\ref{Hilbertsplited}) and we can consider them separately. Since the first two spaces $\mathcal{H}_{\frac{3}{2},\frac{3}{2},0}$,  $\mathcal{H}_{\frac{3}{2},\frac{1}{2},0}$ are just $1$-dimensional, $H$ restricted to them is already diagonal, which also holds for the $2$-dimensional space $\mathcal{H}_{\frac{1}{2},\frac{1}{2},0}$. It remains to diagonalize $H$ in the subspace $\mathcal{H}_{\frac{1}{2},\frac{1}{2},1}$. The corresponding eigenvalue problem
\begin{equation}\label{eigenvalueQ}
\left(
\begin{array}{ccc}
 \frac{2 u}{3}-3 & -\frac{u}{3} & -\frac{u}{3} \\
 -\frac{u}{3} & \frac{2 u}{3}+3 & -\frac{u}{3} \\
 -\frac{u}{3} & -\frac{u}{3} & \frac{2 u}{3}
\end{array}
\right)\vec{v} = E \, \vec{v}\,
\end{equation}
is solved in Appendix \ref{app:cubicEigenProb}. The complete spectrum of the Hamiltonian (\ref{HamHubbard}) with three sites and three electrons is shown in Fig.~\ref{fig:energyplots}.
\begin{figure}[h]
\centering
	\includegraphics[width=7.7cm]{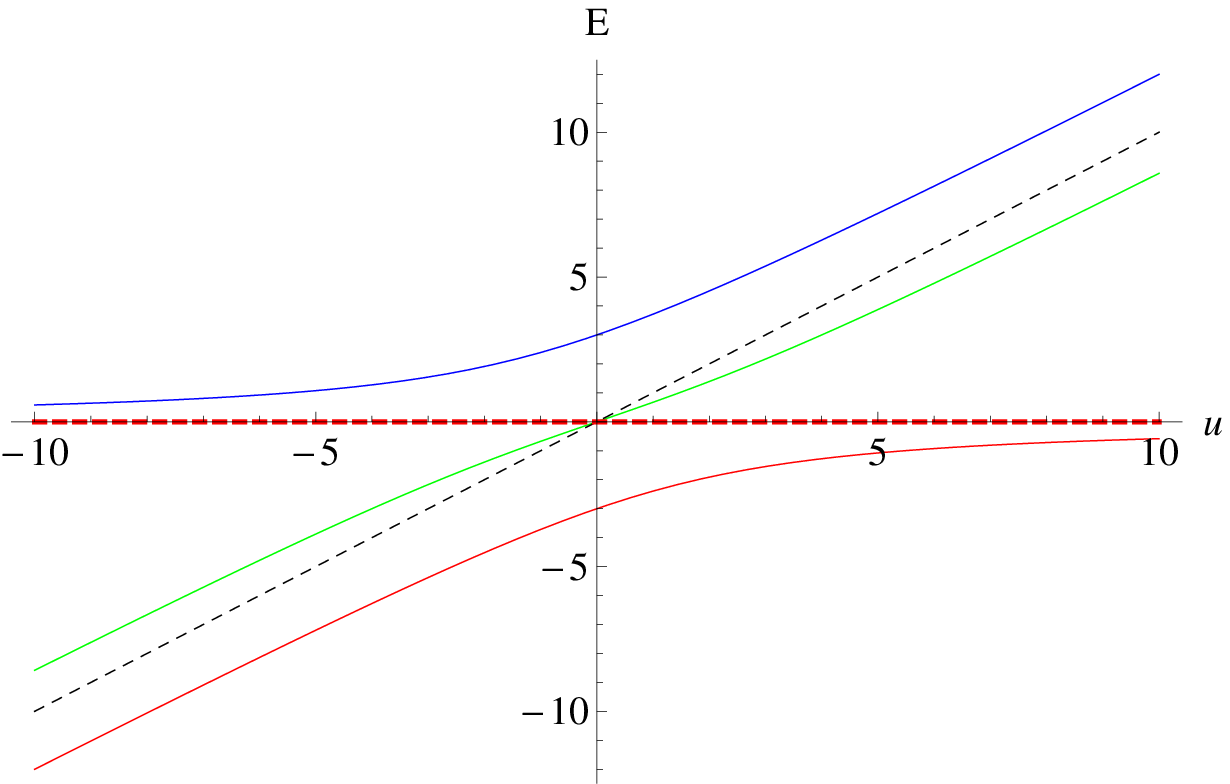}
\captionC{All energies (each with multiplicity four) of the 3-site Hubbard model for 3 electrons. The solution of the nontrivial subspace problem (\ref{eigenvalueQ}) is given by the red, blue and green curve. The eigenvalues $E(u)=0$ for $S=\frac{3}{2}$ and $E(u)= u$ for $K=\pm 1$, $S=\frac{1}{2}$
are shown as thick (red) and thin (black) dashed line, respectively.}
	\label{fig:energyplots}
\end{figure}
The red, blue and green curve describe the energies $E(u)$ arising from the non-trivial subspace problem (\ref{eigenvalueQ}) and according to the symmetries they have multiplicity four (quantum numbers $K=\pm 1$, $S = \frac{1}{2}$ and $M = \pm \frac{1}{2}$). For the four eigenstates with $S=\frac{3}{2}$ we find the eigenvalue $E(u) \equiv 0$ and for those with $K=0$, $S=\frac{1}{2}$ and $M= \pm \frac{1}{2}$, $E(u)=u$ also with multiplicity four. The behavior of the red, blue and green curve for $u\rightarrow \pm \infty$ and $u\rightarrow 0$ is given in Appendix \ref{app:cubicEigenProb}. In Fig,~\ref{fig:energyplots} we can also see that the red curve $E_{gs}(u)$ is below the other curves and does not intersect them. Therefore, $E_{gs}(u)$ is the ground state energy for the whole $u$-regime $\R$.

Now, we focus on the ground state and would like to investigate possible pinning. Since the ground state space has dimension four, this is not that easy. In particular, whenever the ground state is highly degenerate a pinning analysis can becoming meaningless. This can easily be seen for the extreme case of a Hamiltonian $H=\mathds{1}$. In that case every fermionic state $|\Psi_N\rangle$ is a ground state. Hence, one can find for each $\vec{\lambda} \in \mathcal{P}_{N,d}$ a corresponding ground state with NONs $\vec{\lambda}$. To circumvent that we switch on an external magnetic field along the $3$-axis, the Zeeman term then breaks the spin flip symmetry and the ground state space reduces its dimensionality from four ($K=\pm 1$, $M=\pm \frac{1}{2}$) to two ($K=\pm 1$, $M= +\frac{1}{2}$). Although, it is not clear to us whether this is experimentally feasible, we first also break the $P_Q$-symmetry and consider a unique ground state, labeled by the quantum numbers $K= 1$ and $M= \frac{1}{2}$.

The corresponding $3$-electron energy eigenstate $|\Psi_1\rangle$ that is the solution of the subspace problem (\ref{eigenvalueQ}) has then the form
\begin{equation}\label{gsHubbard1}
|\Psi\rangle = \alpha\, \z{1}{1}{1}{2}{2}{2} \,+
\beta\, \z{2}{2}{1}{1}{1}{2} \,+
\gamma\, \z{1}{2}{2}{2}{1}{1} \,.
\end{equation}
The three coefficients $\alpha(u), \beta(u)$ and $\gamma(u)$ are calculated in the Appendix \ref{app:cubicEigenProb}
and are presented in Fig.~\ref{fig:HubGScoef}.
\begin{figure}[h]
\centering
\includegraphics[scale=0.70]{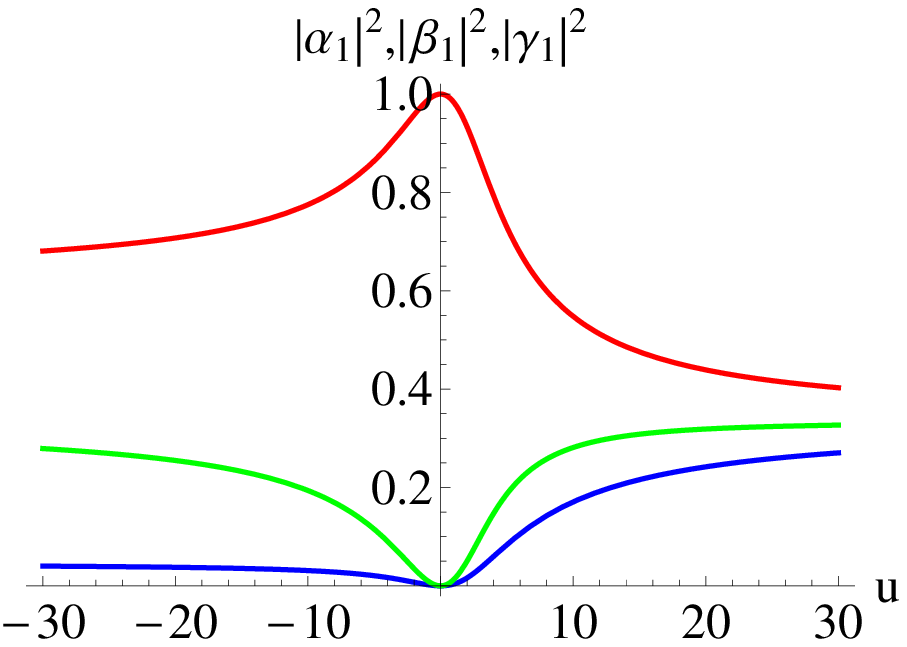}
\captionC{The coefficients $\alpha(u)$ (red), $\beta(u)$ (blue) and $\gamma(u)$ (green) of the Hubbard ground state (\ref{gsHubbard1}).}
\label{fig:HubGScoef}
\end{figure}
From their asymptotic behavior presented in Appendix \ref{app:cubicEigenProb} we can derive the asymptotic behavior of the ground state:\\
\\
$
\begin{array}{llll}
|\Psi_1(u)\rangle &\sim& \frac{1}{\sqrt{3}}\,\left[\z{1}{1}{1}{2}{2}{2}+\z{2}{2}{1}{1}{1}{2}+\z{1}{2}{2}{2}{1}{1} \right] & \nonumber \\
& & +\frac{\sqrt{3}}{u}\, \left[\z{1}{1}{1}{2}{2}{2}-\z{2}{2}{1}{1}{1}{2} \right]+O\left(\left(\frac{1}{u}\right)^2\right)
&,\,u\rightarrow +\infty \nonumber \\
|\Psi_1(u)\rangle &\sim& \frac{3+\sqrt{3}}{6}\,\z{1}{1}{1}{2}{2}{2}
-\frac{3-\sqrt{3}}{6}\,
\z{2}{2}{1}{1}{1}{2} -
\frac{1}{\sqrt{3}}\,
\z{1}{2}{2}{2}{1}{1} & \nonumber \\
& & +\frac{1}{u}\, \left[-\frac{3+\sqrt{3}}{4}\,\z{1}{1}{1}{2}{2}{2}-\frac{3-\sqrt{3}}{4}\,\z{2}{2}{1}{1}{1}{2} - \frac{3}{2}\, \z{1}{2}{2}{2}{1}{1}\right]  &\nonumber \\
&&+O\left(\left(\frac{1}{u}\right)^2\right)  &,\,u\rightarrow -\infty \nonumber \\
|\Psi_1(u)\rangle&\sim& \z{1}{1}{1}{2}{2}{2}
 +\frac{u}{18}\, \left[\z{2}{2}{1}{1}{1}{2} + 2\,  \z{1}{2}{2}{2}{1}{1}\right]+O\left(u^2\right)&,\, u\rightarrow 0\,.
\end{array}
$
\\
\\
According to the previous remarks on symmetry adaption, determining the $1$-RDO for any state of the form (\ref{gsHubbard1}) should be trivial. Indeed, we find
\begin{equation}\label{NONsHubbard1}
\left({}_Q\langle k,\sigma | \rho_1 | l,\mu \rangle_Q\right) = \mbox{diag}(|\alpha|^2+|\gamma|^2,|\alpha|^2 ,|\alpha|^2+|\beta|^2,|\beta|^2,|\beta|^2+|\gamma|^2,|\gamma|^2)\,.
\end{equation}
It is important to notice that this simplified form of the $1$-RDO arises just from the symmetries of the Hamiltonian. Its concrete form plays a role just for the coefficients $\alpha, \beta, \gamma$.

To finally study possible pinning we still need to find the hierarchy of those six NONs. Although this hierarchy seems to solely depend on the concrete form of the Hamiltonian we can still make a quite universal statement on possible pinning of the NONs (\ref{NONsHubbard1}), arising from a state of the form (\ref{gsHubbard1}) without knowing the three coefficients. For this, w.l.o.g.~assume $|\alpha|^2 > |\gamma|^2 > |\beta|^2$ (otherwise just relabel the three coefficients in Eq.~(\ref{gsHubbard1})). For the decreasingly ordered NONs it follows immediately that $\lambda_1 = |\alpha|^2+ |\gamma|^2$,  $\lambda_2 = |\alpha|^2+ |\beta|^2$,  $\lambda_5 = |\gamma|^2$
and  $\lambda_6 = |\beta|^2$. The remaining two NONs are given by
\begin{equation}
\lambda_4 \,=\, \min(|\alpha|^2, |\beta|^2+|\gamma|^2)\,\,\,\,,\,\,\,\lambda_3 = 1- \lambda_4\,.
\end{equation}
For the saturation $D^{(3,6)}$ (recall (\ref{set36})) this yields
\begin{equation}
D^{(3,6)}=\begin{cases}
  0,  & \mbox{whenever}\,|\alpha|^2\geq |\beta|^2+|\gamma|^2 \\
  |\beta|^2+|\gamma|^2-|\alpha|^2>0, & \mbox{otherwise}
\end{cases}\,.
\end{equation}
Intriguingly, this means that exact pinning is possible. Moreover the high symmetry of the $3$-electron Hamiltonian leads independent of its concrete form to two manifest different scenarios, exact pinning and no-pinning. They are $1-1$ related to the hierarchy of the two terms $|\alpha|^2$ and $|\beta|^2+|\gamma|^2$, which is fixed by the concrete form of the Hamiltonian.

To check possible pinning for the Hubbard Hamiltonian we just need to check the ordering of the coefficients. As it can be seen in Fig \ref{fig:HubGScoef} we have
\begin{equation}
|\alpha(u)|^2 > |\gamma(u)|^2 \geq |\beta(u)|^2 \qquad,\,\,\forall u \in \mathbb{R}\,.
\end{equation}
Exact pinning is then given whenever $|\alpha(u)|^2 - |\gamma(u)|^2 - |\beta(u)|^2 \geq 0$. The $u$-regime, where this is fulfilled can be derived from Fig.~\ref{fig:saturationBD}.
\begin{figure}[h]
\centering
\includegraphics[scale=0.65]{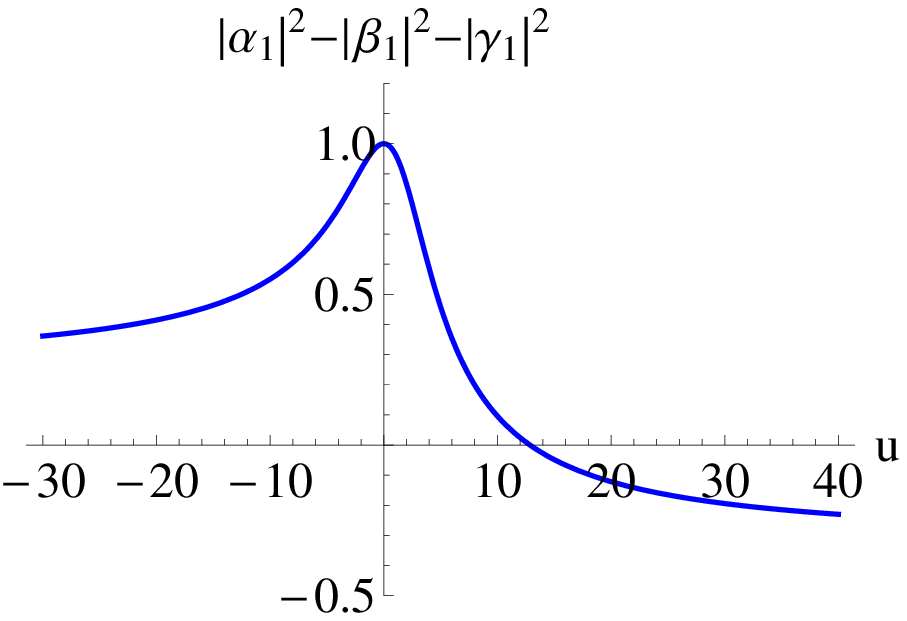}\hspace{0.4cm}
\includegraphics[scale=0.65]{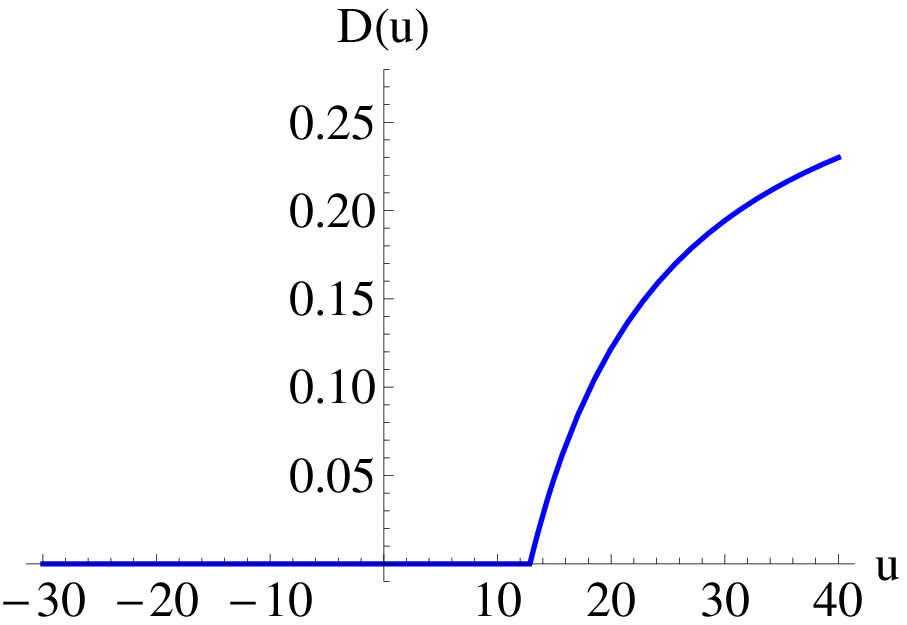}
\captionC{Saturation $D(u)$ (right) for the Hubbard ground state $|\Psi_1(u)\rangle$ of setting with $3$ electrons on $3$ sites. Exact pinning is given whenever the term $|\alpha(u)|^2-|\beta(u)|^2-|\gamma(u)|^2$ (left) is non-negative.}
\label{fig:saturationBD}
\end{figure}
The regime of exact pinning is the $u$-interval $\mathcal{I}_p\equiv\, ]-\infty,u_p]$ and the crossing point $u_p$ is numerically found as
\begin{equation}
u_p \approx 12.86\,.
\end{equation}
This kind of ``microscopic phase transition'' is also shown on the right side of Fig.~\ref{fig:saturationBD}

Numerically, we find the asymptotic behavior of the saturation $D^{(3,6)}$ for $u$ close to $u_p$ and analytically that for $u\rightarrow \infty$:
\begin{eqnarray}
D^{(3,6)}(u) & \approx & 0.02641\, (u-u_p) -0.00199 \,(u-u_p)^2+O\left((u-u_p)^3\right)\,\,\,\,\,,\mbox{for}\,u\geq u_p \nonumber \\
D^{(3,6)}(u) &\sim &\frac{1}{3}-\frac{4}{u}-\frac{6}{u^2}+O\left(\frac{1}{u^3}\right) \qquad,\, u\rightarrow \infty\,.
\end{eqnarray}

Do these remarkable results in Fig.~\ref{fig:saturationBD} mean that the Hubbard ground state shows pinning for all $u\leq u_p$, which vanishes for $u>u_p$? No, the ground state is not unique and the results we have found so far just apply for the ground state uniquely found after fixing $K=\pm 1$ and $M=\pm \frac{1}{2}$. In the following we still break the spin flip symmetry, restrict to $M=+\frac{1}{2}$, but allow an arbitrary ground state $|\Psi\rangle$ w.r.t.~the wavenumber $K =\pm 1$,i.e.~we remove the breaking of the $P_Q$-symmetry . Then the ground state has the general form
\begin{equation}\label{gsHubbard2}
|\Psi\rangle = \zeta |\Psi_1\rangle +\xi |\Psi_2\rangle\,,
\end{equation}
where $|\Psi_2\rangle = P_Q |\Psi_1\rangle $, $|\zeta|^2+|\xi|^2=1$ and $|\Psi_K\rangle$ are the unique ground states after fixing the wavenumber to $K=\pm 1$. Since (\ref{gsHubbard2}) is not adapted to the translational symmetry anymore, the corresponding $1$-RDO represented w.r.t.~the states $\{|k,\sigma\rangle_Q\}$ is unfortunately not diagonal anymore. However, since (\ref{gsHubbard2}) is still an $S_3$-eigenstate, the new NOs are still spin eigenstates w.r.t.~$3$-axis (see also Ex.~\ref{ex:Spincomponent}). For the $1$-RDO this means
\begin{equation}\label{rho1DirectSum}
\rho_1 = \rho_1^{\uparrow} \oplus \rho_1^{\downarrow}\,.
\end{equation}
We find
\begin{equation}\label{rho1up}
 \rho_1^{\uparrow} = \left(\begin{array}{ccc}|\alpha|^2+|\gamma|^2&\zeta \xi^\ast\gamma \beta^\ast& \zeta^\ast \xi \beta^\ast \gamma\\ \zeta^\ast \xi\gamma^\ast \beta&|\zeta|^2 |\alpha|^2+|\xi|^2 |\gamma|^2+|\beta|^2& \zeta \xi^\ast |\alpha|^2\\ \zeta \xi^\ast \beta \gamma^\ast&\zeta^\ast \xi |\alpha|^2& |\zeta|^2 |\gamma|^2+|\xi|^2 |\alpha|^2+|\beta|^2 \end{array}\right)
\end{equation}
and
\begin{equation}\label{rho1down}
 \rho_1^{\downarrow} = \left(\begin{array}{ccc}|\alpha|^2& -\zeta \xi^\ast\alpha \gamma^\ast& -\zeta^\ast \xi \gamma^\ast \alpha\\ -\zeta^\ast \xi\alpha^\ast \gamma&|\zeta|^2 |\beta|^2+|\xi|^2 |\gamma|^2& -\zeta \xi^\ast |\beta|^2\\ -\zeta \xi^\ast \gamma \alpha^\ast& -\zeta^\ast \xi |\beta|^2& |\zeta|^2 |\gamma|^2+|\xi|^2 |\beta|^2 \end{array}\right)\,.
\end{equation}
By solving again cubic equations we will find analytic expressions for the six NONs. Before doing so we first briefly discuss some general structure.
Note that
\begin{equation}\label{exchange}
\vec{\lambda}(u;\zeta,\xi) = \vec{\lambda}(u;\xi,\zeta)\,,
\end{equation}
which follows directly from (\ref{BrillReflex}) and (\ref{gsHubbard2}).
The two blocks $\rho_1^{\uparrow/\downarrow}(u;\zeta,\xi)$ both obviously fulfill
\begin{equation}
\rho_1^{\uparrow/\downarrow}(u;\zeta,\xi e^{\frac{2 \pi}{3}i}) = U^\dagger\, \rho_1^{\uparrow/\downarrow}(u;\zeta,\xi) \,U
\end{equation}
with the unitary matrix $U = \mbox{diag}(1,e^{-\frac{2 \pi}{3}i},e^{\frac{2 \pi}{3}i})$. Thus,
\begin{equation}
\vec{\lambda}(u;\zeta,\xi e^{\frac{2 \pi}{3}i}) = \vec{\lambda}(u;\zeta,\xi)\,.
\end{equation}
Now, we calculate the six NONs $\lambda_1\geq \ldots \geq \lambda_6\geq0$. One can show that it suffices to determine the three eigenvalues of the spin down block $\rho_1^{\downarrow}$. Since $\rho_1^{\downarrow}$ is normalized to $1$ and $\rho_1^{\uparrow}$ to $2$, the three generalized Pauli constraints
\begin{equation}
\lambda_1+\lambda_6 = \lambda_2+\lambda_5 = \lambda_3+\lambda_4 = 1\,,
\end{equation}
for the Borland-Dennis setting imply that each eigenvalue $n$ of $\rho_1^{\downarrow}$ has a partner eigenvalue $m$ of $\rho_1^{\uparrow}$ such that
$1 = n+ m$.

Again as for the pinning analysis for the states of the form (\ref{gsHubbard1}) we can draw some general conclusions on pinning for states of the form (\ref{gsHubbard2}). The only property we need is that (\ref{gsHubbard2}) is a $S_3=\frac{1}{2}$ eigenstate in $\wedge^3[\mathcal{H}_1^{(6)}]$. This implies the structure (\ref{rho1DirectSum}) for its $1$-RDO with normalizations $\mbox{Tr}[\rho_1^{\uparrow}]=m_1+m_2+m_3= 2$ and $\mbox{Tr}[\rho_1^{\downarrow}]=n_1+n_2+n_3= 1$. Due to these normalizations we can conclude that the two smallest NONs $\lambda_5$, $\lambda_6$ are arising as eigenvalues of $\rho_1^{\downarrow}$ and the two largest, $\lambda_1,\lambda_2$ from $\rho_1^{\uparrow}$. W.l.o.g.~we order the $n_i$ and $m_j$ such that $\lambda_1= m_1$, $\lambda_2 = m_2$, $\lambda_5= n_2$ and $\lambda_6 = n_3$. The remaining relations are given by $\lambda_3 = \max{(m_3,n_1)}$ and $\lambda_4 = \min{(m_3,n_1)}$.
Hence, there are two different cases and it turns out that we can make concrete statements on possible pinning for both of them:
\begin{enumerate}
\item $m_3 > n_1:$\, In this case we have $\lambda_3 = m_3$ and $\lambda_4 = n_1$. For the saturation $D^{(3,6)}$ (recall Eq.~(\ref{set36})) we find
     \begin{equation}
     D^{(3,6)} =\lambda_5+\lambda_6-\lambda_4 = n_2 + n_3 - n_1 = 1-2 n_1 >0\,,
     \end{equation}
     where we have used in the last step that $m_3 + n_1 = 1$ and $m_3 > n_1$. Hence, there is no pinning.
\item $m_3 < n_1:$\, In that case we have $\lambda_3 = n_1$ and $\lambda_4 = m_3$. Then, by using $m_3 + n_1 = 1$ and $n_1+n_2+n_3 = 1$ we obtain
     \begin{equation}
     D^{(3,6)} =\lambda_5+\lambda_6-\lambda_4 = n_2 + n_3 - m_3 =n_2+n_3 -(1-n_1) = 0\,.
     \end{equation}
     We found exact pinning.
\end{enumerate}

Whether a concrete state belongs to the first or second case depends on the specific form of the Hamiltonian and the two coefficients $\zeta, \xi$ in (\ref{gsHubbard2}). To explore possible pinning for the Hubbard ground state (\ref{gsHubbard2}) we present the plots for $D^{(3,6)} = \lambda_5+\lambda_6-\lambda_4$ for the two values $u=100$ and $u=20$ as function of $\zeta$ and $\xi$ in Fig~ \ref{fig:saturationBDsuper}. W.l.o.g.~we choose $\zeta= |\zeta| = \sqrt{1-|\xi|^2}$.
\begin{figure}[h]
\centering
\includegraphics[width=0.47\textwidth]{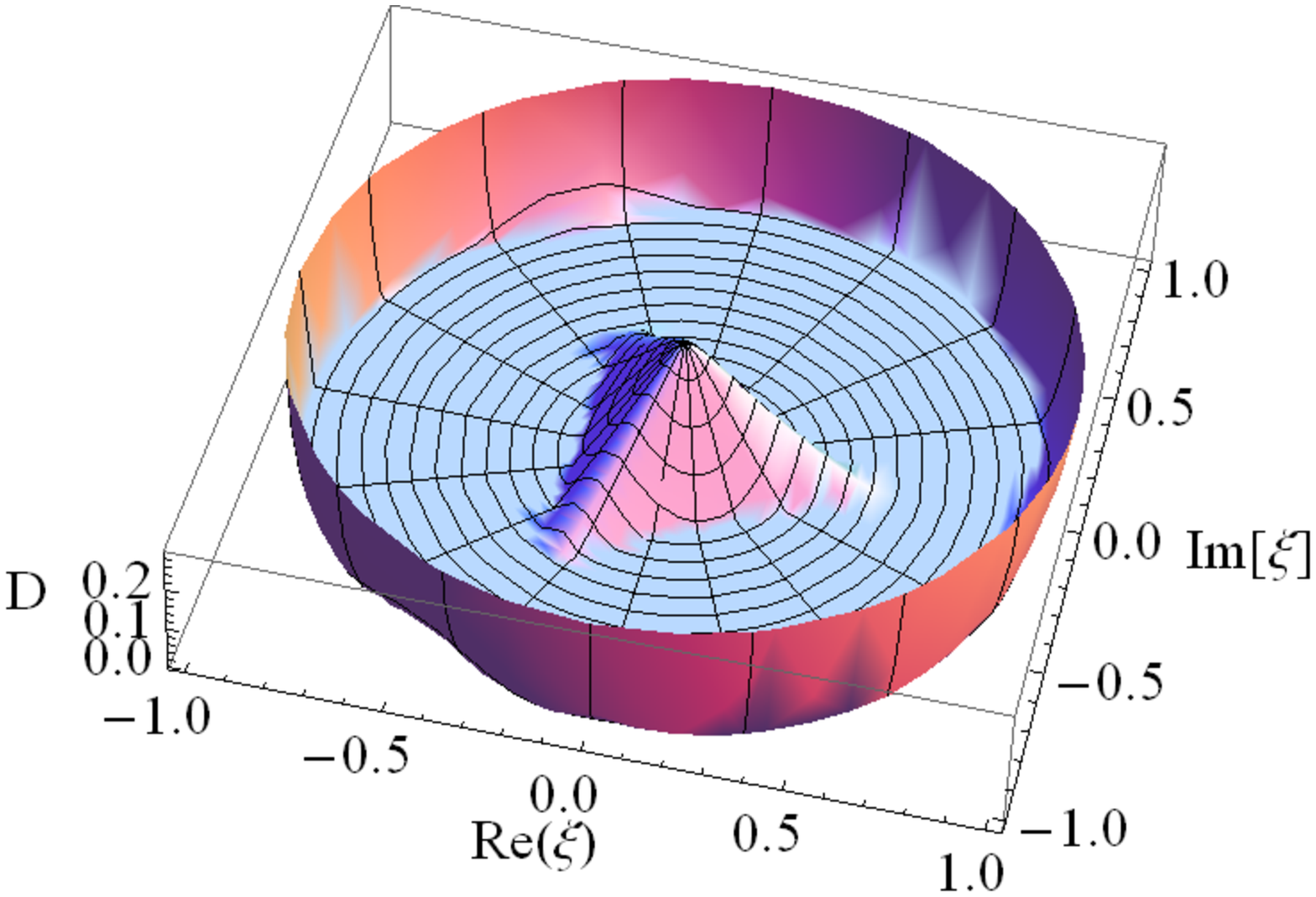}
\includegraphics[width=0.52\textwidth]{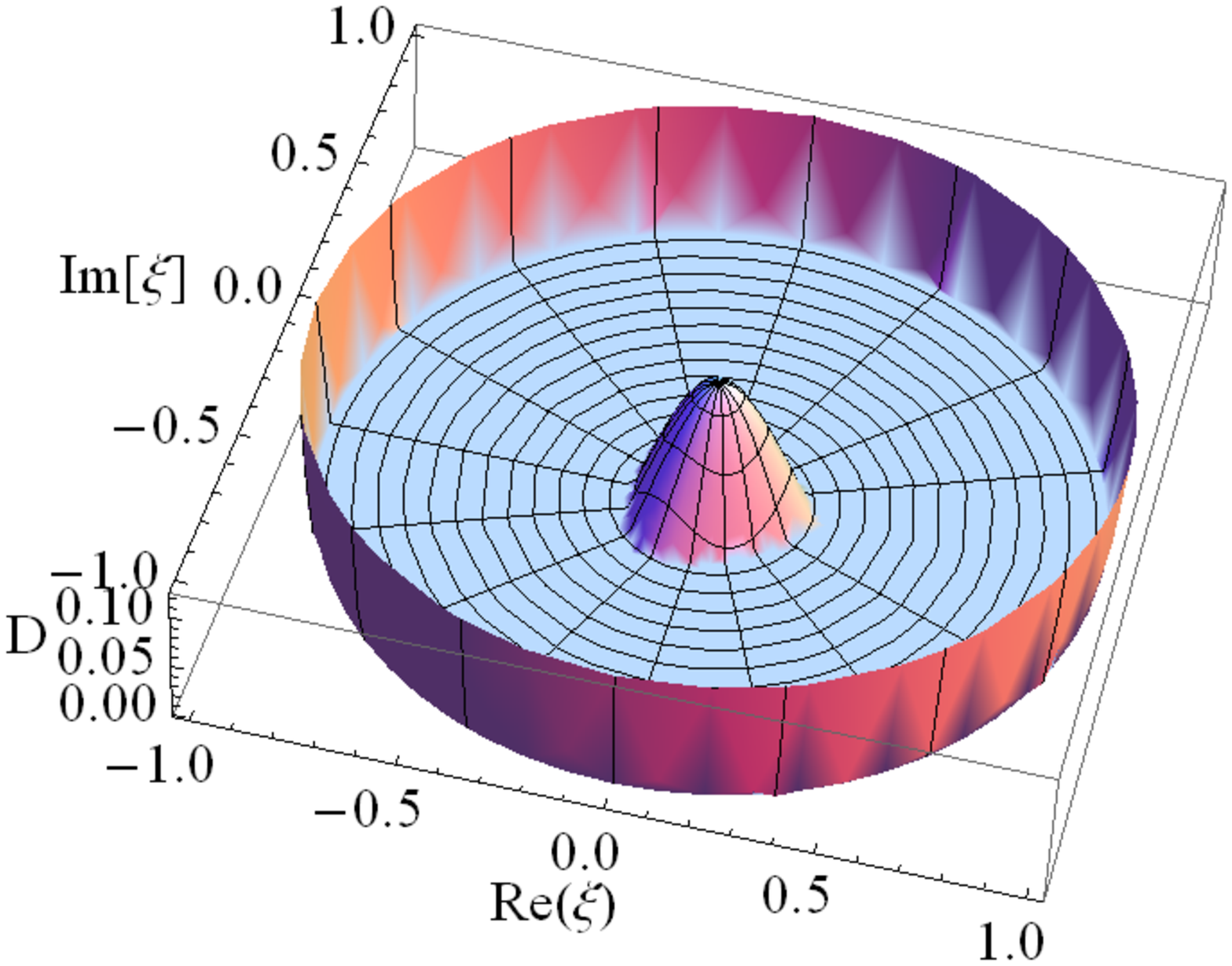}
\captionC{Saturation $D$ for $u=100$ (left) and $u=20$ (right) for the ground state $|\Psi(u)\rangle = \zeta |\Psi_1(u)\rangle +\xi |\Psi_2(u)\rangle$ of the Hubbard model with $3$ sites and $3$ electrons with $\zeta = |\zeta| = \sqrt{1-|\xi|^2}$.}
\label{fig:saturationBDsuper}
\end{figure}
From Fig.~\ref{fig:saturationBDsuper} we learn that there is exact pinning if we superpose $|\Psi_1\rangle$ and $|\Psi_2\rangle$ with similar weights, $|\zeta| \approx |\xi|$. Moreover, a relative phase of $\frac{\pi}{3}$ between $\zeta$ and $\xi$ favors pinning best. In the regime of $u\rightarrow \infty$ the $(\zeta,\xi)$-regime of exact pinning is becoming smaller and smaller.  To understand this better we
illustrate the whole $u$-regime for the two extremal relative phases $\phi =0, \frac{\pi}{3}$ between $\xi$ and $\zeta$ in Fig.~\ref{fig:saturationBDsuperU}.
\begin{figure}[h]
\centering
\includegraphics[scale=0.34]{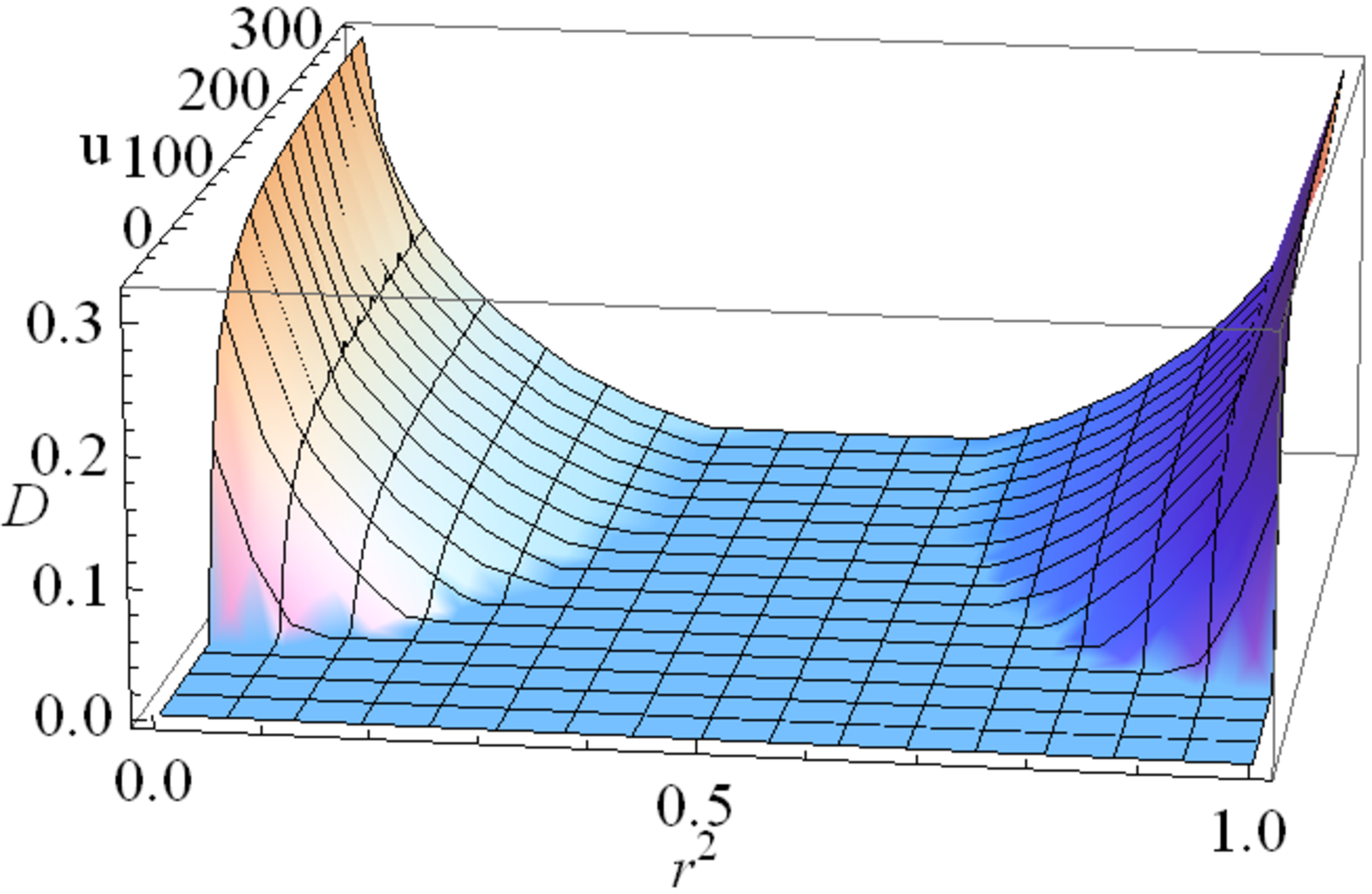}\hspace{0.4cm}
\includegraphics[scale=0.35]{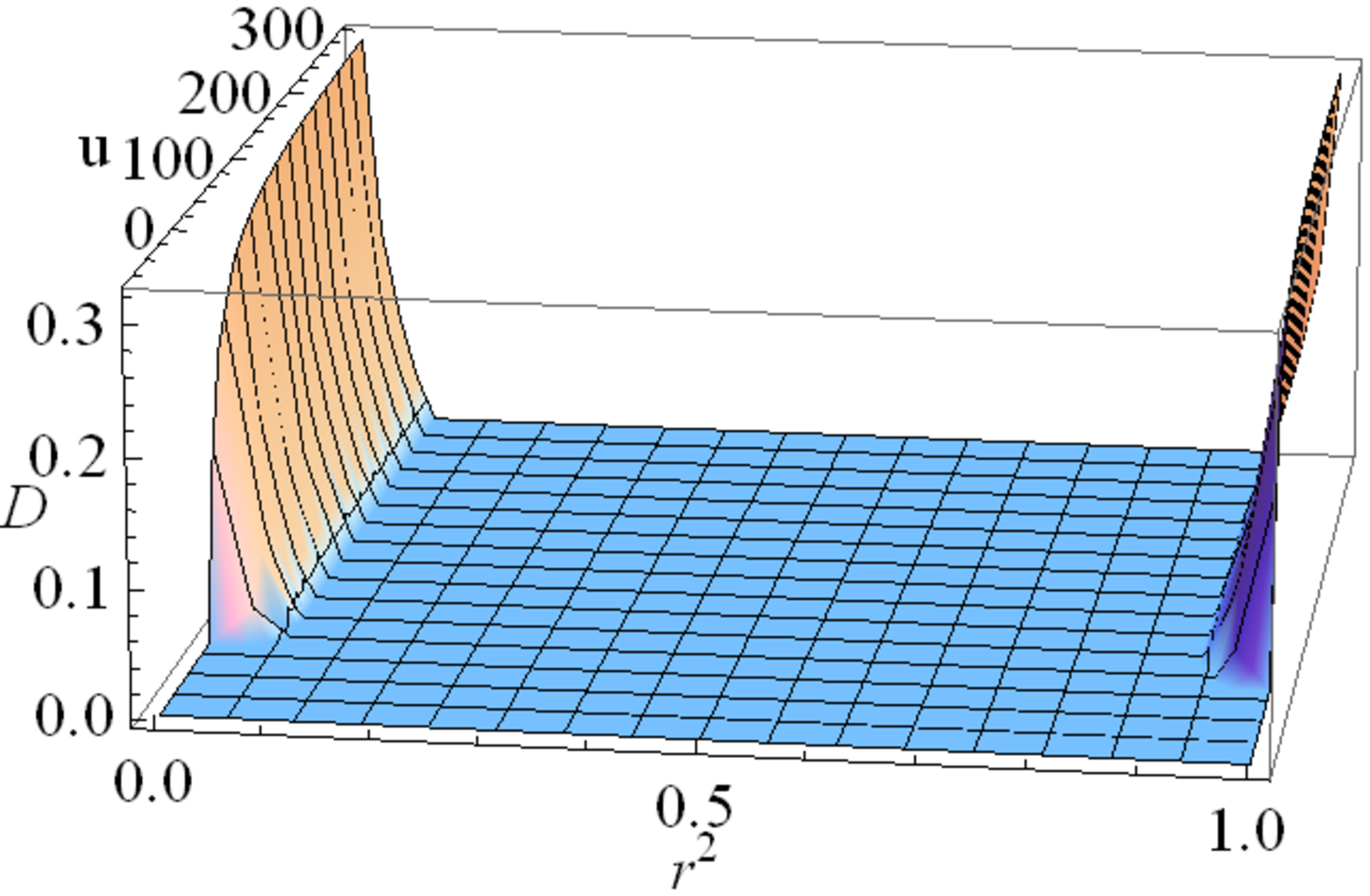}
\captionC{Saturation $D(u)$ for the ground state $|\Psi(u)\rangle = \sqrt{1-r^2} |\Psi_1(u)\rangle + r\,e^{i \varphi}\, |\Psi_2(u)\rangle$ of the Hubbard model with $3$ sites and $3$ electrons for the two extremal cases $\varphi= 0$ and $\varphi= \frac{\pi}{3}$.}
\label{fig:saturationBDsuperU}
\end{figure}
We can see that the exact pinning is uniform in $u$  whenever $|\zeta| = |\xi| = \frac{1}{\sqrt{2}}$. For all the other superpositions it vanishes for some critical value $u_p(\zeta, \xi)$. Moreover, since this critical value $u_p(\zeta, \xi)$ is minimal for either $\zeta=0$ or $\xi =0$, we can conclude that superposing $|\Psi_1\rangle$ and $|\Psi_2\rangle$ (recall (\ref{gsHubbard2})) enhances the $u$-regime of exact pinning from $(-\infty,12.86]$ to some larger interval
$(-\infty, u_p(\zeta, \xi)]$, $u_p(\zeta, \xi) \geq 12.86$.

\subsection{Larger settings}\label{sec:hubbardLarger}
In this section we explore possible pinning for the Hubbard ground state of settings with more than three lattice sites. Since we just know the generalized Pauli constraints for settings with a $1$-particle Hilbert space of dimension $d\leq 10$, we need to restrict ourself to lattices with at most five sites. For each such lattice we consider particle numbers $N=3,4$ and $5$. For each such setting we diagonalize the Hamiltonian (\ref{HamHubbard}) numerically.
As a first difference to the setting with three sites and three electrons studied analytically in Sec.~\ref{sec:hubbardAnalyt} the energy ``curves''  $E_j(u)$ as function of the on-site interaction $u$ do cross each other (recall Fig.~\ref{fig:energyplots}). Therefore, to explore ground state pinning we need to consider every $u$-interval between two consecutive energy crossing points separately. In addition, the ground state Hilbert space is again invariant under spin flips and inversions $P_Q$ of the reciprocal lattice, which means that the ground state is typically not unique. We circumvent that problem by breaking these symmetries (e.g.~by switching on an external magnetic field) and consider in the following just symmetry-adapted ground states, i.e.~eigenstates of $\vec{S}^2$, $S_3$ and the translation operator $T$.

The results of the pinning analysis for larger settings strongly suggest that the generalized Pauli constraints do not play a role for Hubbard models with significantly more than $3$ sites. Indeed, for all settings with $5$ lattice sites we find no pinning at all, independent of the electron number $N=3,4,5$ and the value for $u \in \R$. However, for the settings with four sites and the particle numbers $N=3$ and $N=5$ we find again a phase transition between no-pinning and exact pinning by changing $u$ as already found for the setting with three electrons on three sites (see previous section, Sec.~\ref{sec:hubbardAnalyt}). We briefly present those results. By applying a particle-hole transformation to the Hubbard Hamiltonian (\ref{HamHubbard}) it can be shown that the ground states for $N=3$ and $N=5$ are identical from a structural viewpoint. This implies that the pinning behavior of both states is identical and we consider just $N=3$. By studying the energy curves $E_j(u)$ we find three intersections of the lowest energy curve with the next highest, at the positions $u_1=-18.6$, $u_2=0$ and $u_3 = 18.6$. This requires to study each $u$-regime $]-\infty, u_1]$, $[u_1,u_2]$,  $[u_2, u_3]$ and $]u_3,+\infty[$ separately. For the first two regimes and also the last one we do not find pinning, in contrast to the third one. There, the ground state space is four dimensional and described by the quantum numbers $(S,M,K)=(\frac{1}{2},\pm\frac{1}{2},\pm 1)$.
\begin{figure}[h]
    \centering
\includegraphics[scale=0.37]{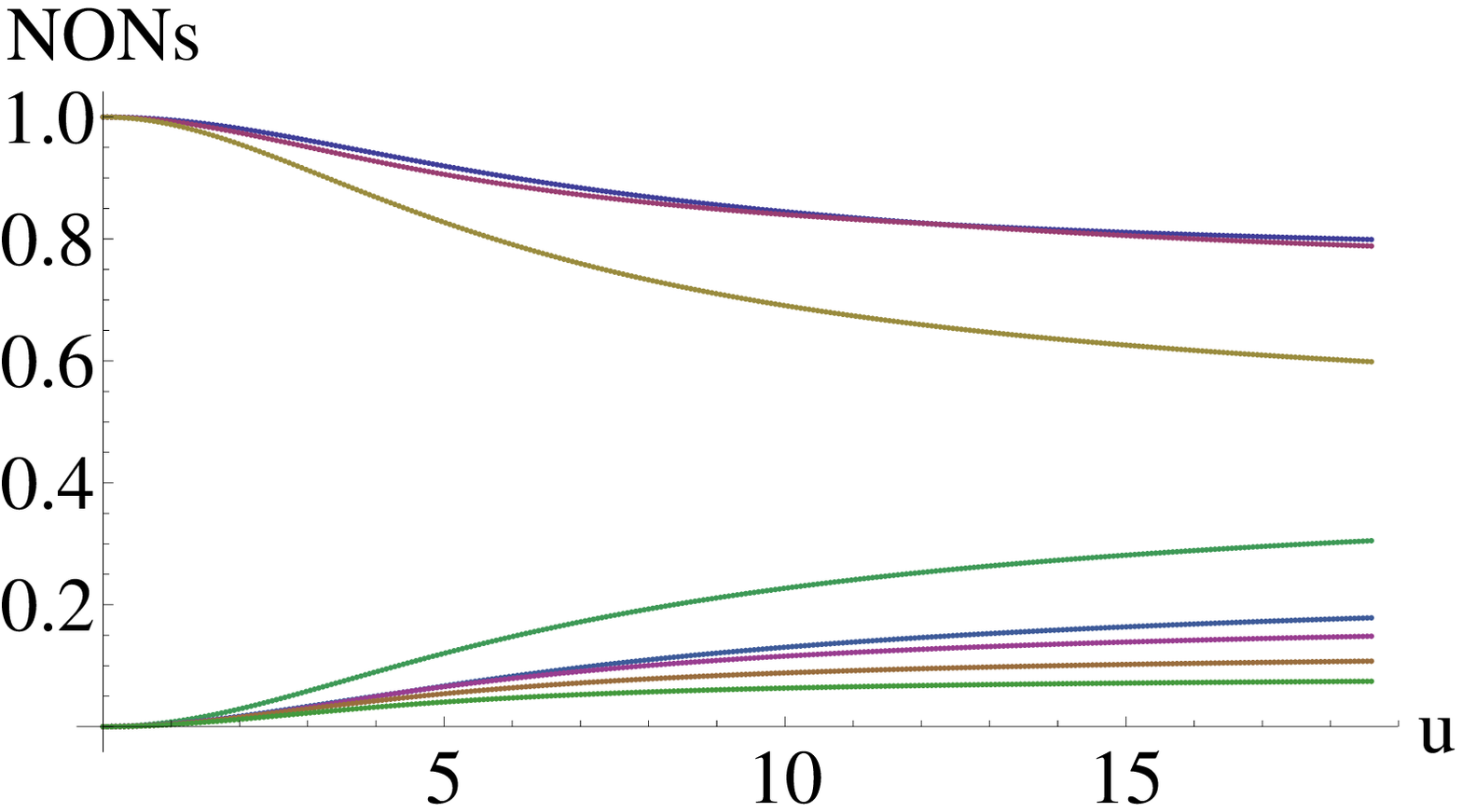}\hspace{0.3cm}
\includegraphics[scale=0.37]{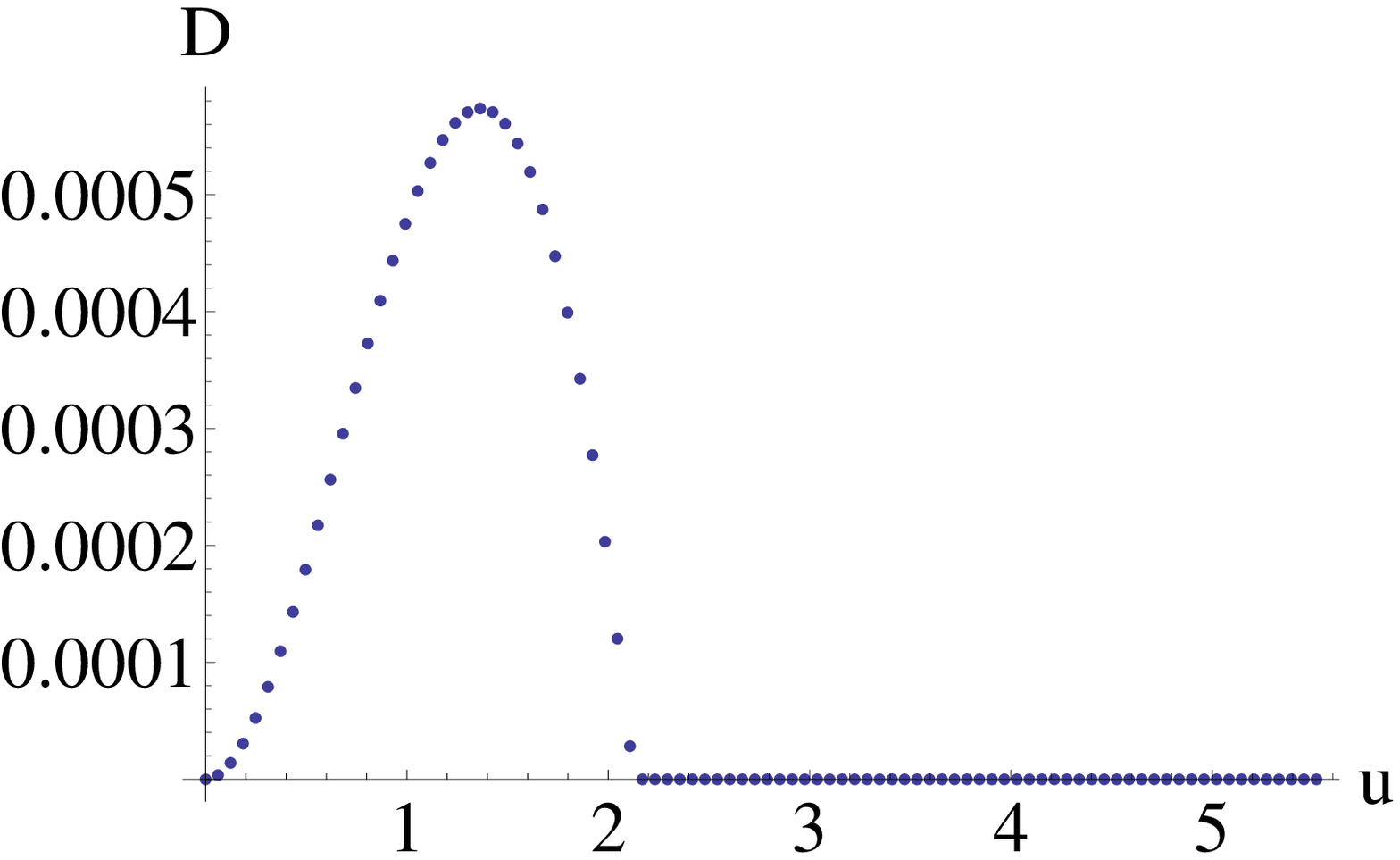}
\captionC{NONs for the Hubbard model ground state for $3$ electrons on $4$ sites as function $u$ in the pinning-relevant interval $[0,18.6]$ (left). The pinning is shown on the right side.}
\label{fig:non@34}
\end{figure}
We choose the state $|\Psi_{\frac{1}{2},\frac{1}{2},1}\rangle$ and the other three are of course identical from a structural viewpoint since they arise from $|\Psi_{\frac{1}{2},\frac{1}{2},1}\rangle$ just by flipping all spins and/or inverting the state in the reciprocal lattice. The NONs are presented on the left side of Fig.~\ref{fig:non@34}. On the right side the distance to the polytope boundary is shown. We find exact pinning for $u$ above the crossing point, $u_p\geq 2.3$. Among the 31 generalized Pauli constraints for the setting $\wedge^3[\mathcal{H}_1^{(8)}]$ (shown in Appendix \ref{app:GPClargsettings}) there is exactly one, which is saturated, the constraint
\begin{equation}
    D_2^{(3,8)} = 2-(\lambda_1+\lambda_2+\lambda_4+\lambda_7) \geq 0\,.
\end{equation}
The mathematical origin of the phase transition at $u_p=2.3$ is the same as that for the transition in the setting with three electrons on three sites: At $u_p$ some NONs change their hierarchy. All those crossings between NONs are illustrated in Fig.~\ref{fig:crossing@34}.
\begin{figure}[h]
\centering
\includegraphics[scale=0.27]{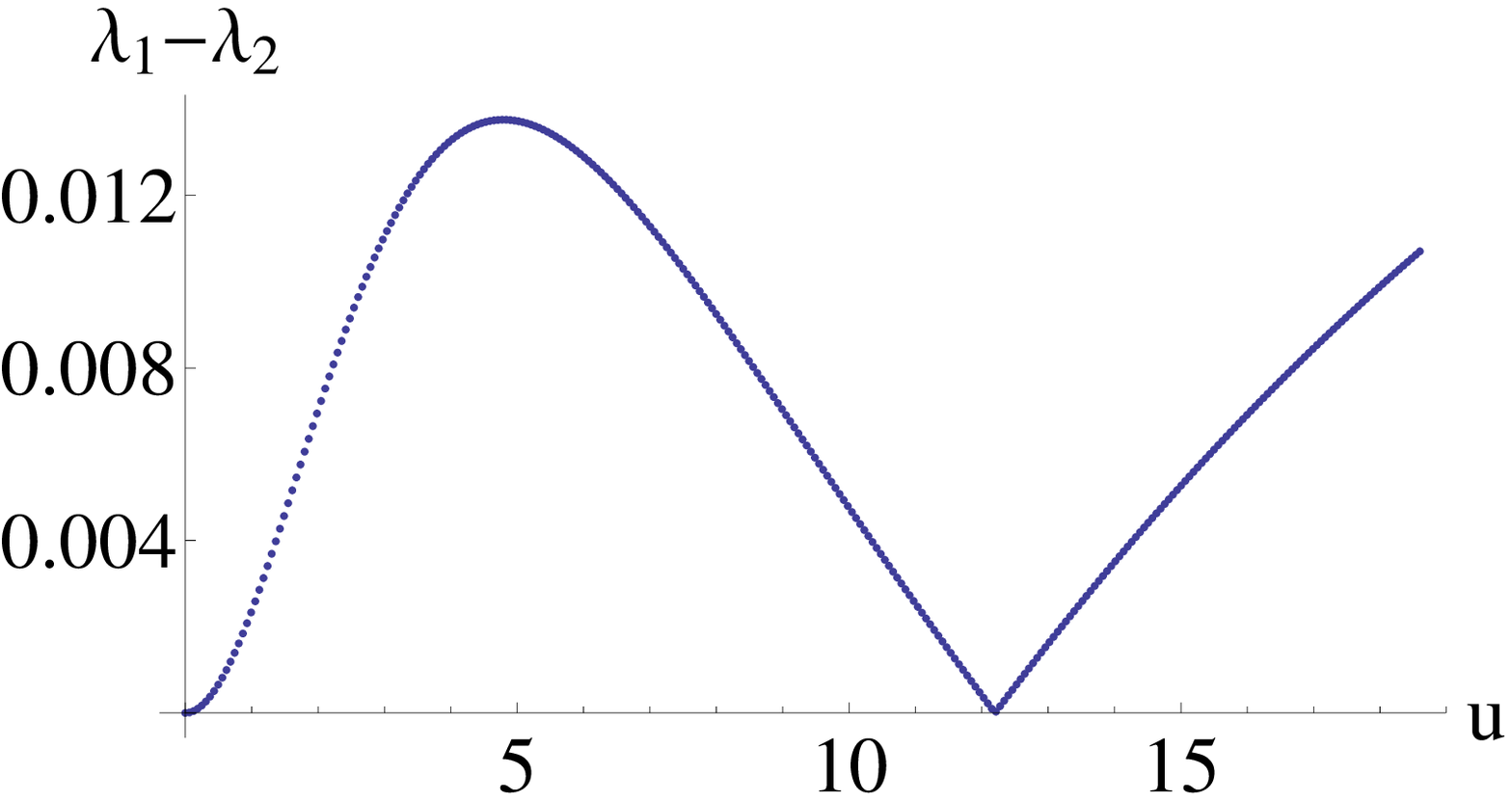}\hspace{0.1cm}
\includegraphics[scale=0.27]{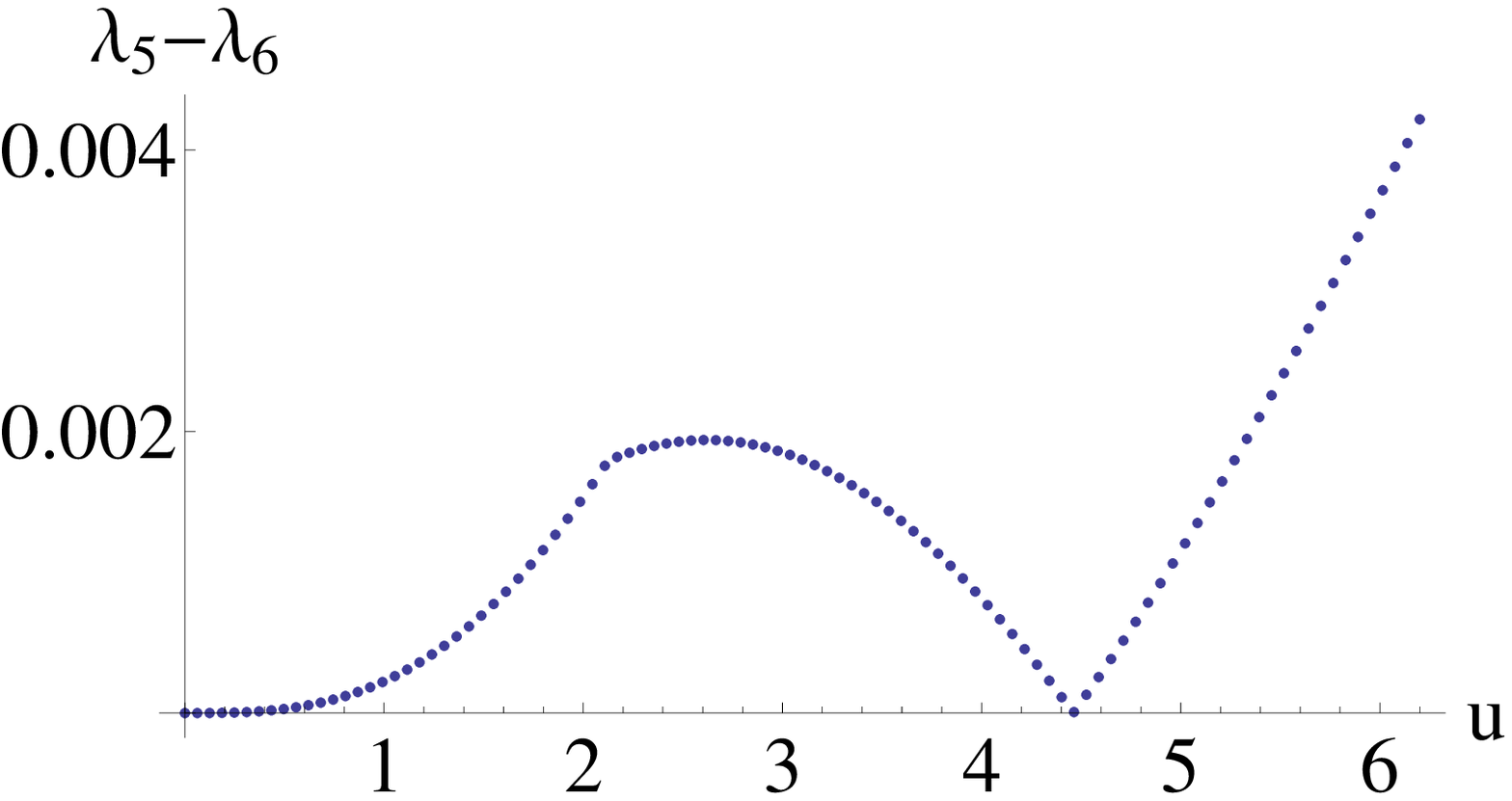}\hspace{0.1cm}
\includegraphics[scale=0.27]{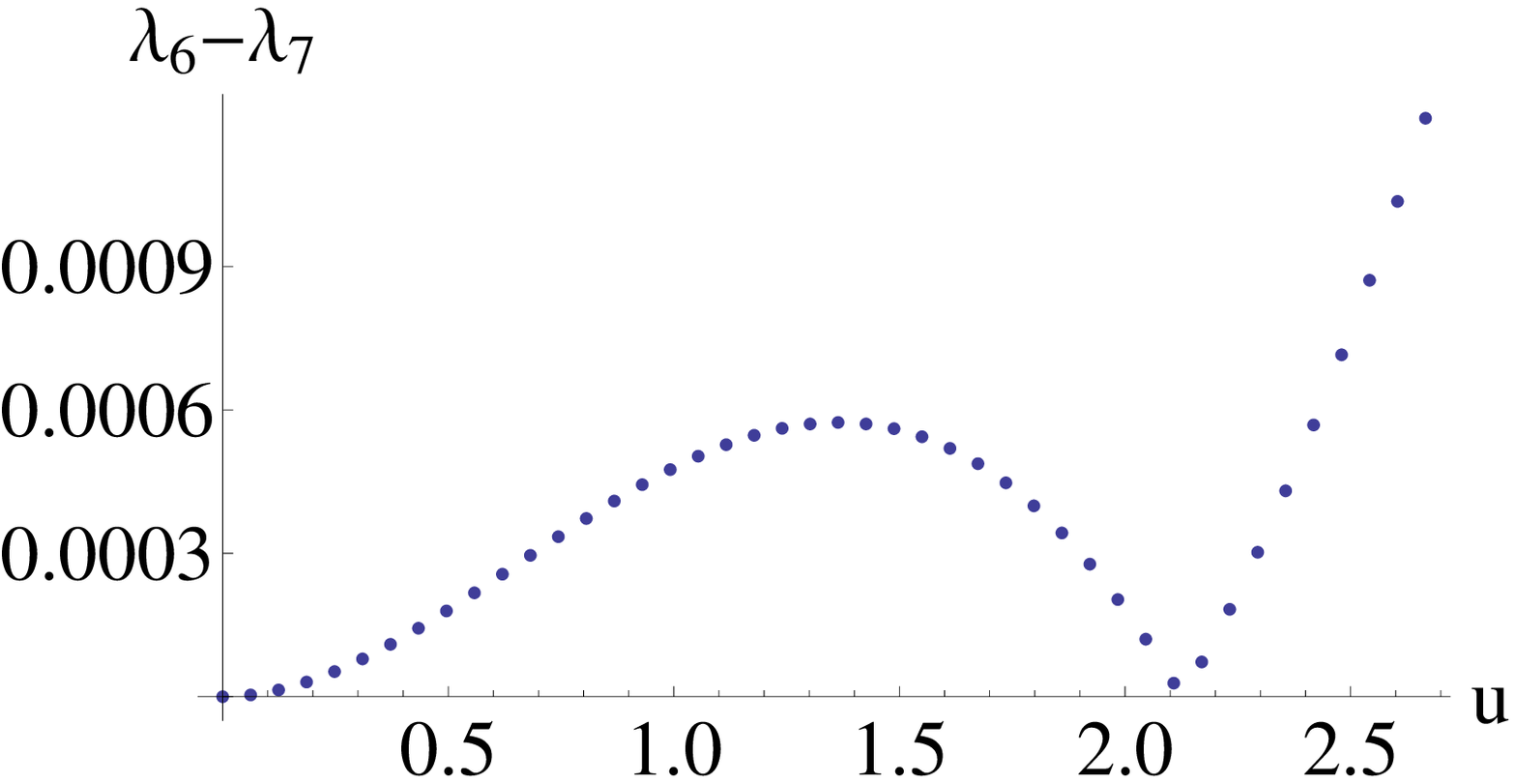}
\captionC{The only crossing points for the NONs in the relevant $u$-interval $[0,18.6]$.}
\label{fig:crossing@34}
\end{figure}
The change of the hierarchy between $\lambda_1$ and $\lambda_2$ (on the left) does not affect $D_2^{(3,8)}$ since both NONs contribute equally. The same also holds for the crossing between $\lambda_5$ and $\lambda_6$, which both do not contribute to $D_2^{(3,8)}$ at all. In contrast to that, the change of hierarchy between $\lambda_6$ and $\lambda_7$ (see right side in Fig.~\ref{fig:crossing@34}) matters, because just $\lambda_7$ shows up in $D_2^{(3,8)}$.

Although we did not find an analytical expression for the ground state $|\Psi_{\frac{1}{2},\frac{1}{2},1}\rangle$ we would like to give the reader an impression how it looks. For this we present its form for three exemplary values of $u$, $u=2 < u_p$ and $u = 3, 12 > u_p$:
\begin{eqnarray}
 |\Psi_{\frac{1}{2},\frac{1}{2},1}(2)\rangle &=&    0.007 \zz{0}{1}{0}{0}{1}{0}{1}{0}+0.124 \zz{0}{0}{1}{1}{0}{0}{1}{0}-0.058 \zz{0}{0}{1}{0}{1}{1}{0}{0}\nonumber \\&&-0.117 \zz{1}{0}{0}{0}{0}{1}{1}{0}+0.110 \zz{1}{0}{0}{0}{1}{0}{0}{1}+0.977 \zz{1}{1}{1}{0}{0}{0}{0}{0}\nonumber\\
 |\Psi_{\frac{1}{2},\frac{1}{2},1}(3)\rangle&=& 0.0139 \zz{0}{1}{0}{0}{1}{0}{1}{0}+0.177 \zz{0}{0}{1}{1}{0}{0}{1}{0}-0.0811 \zz{0}{0}{1}{0}{1}{1}{0}{0}\nonumber \\
&&-0.163 \zz{1}{0}{0}{0}{0}{1}{1}{0}+0.150 \zz{1}{0}{0}{0}{1}{0}{0}{1} +0.955 \zz{1}{1}{1}{0}{0}{0}{0}{0} \nonumber \\
|\Psi_{\frac{1}{2},\frac{1}{2},1}(12)\rangle&=&
0.0613 \zz{0}{1}{0}{0}{1}{0}{1}{0} +0.382  \zz{0}{0}{1}{1}{0}{0}{1}{0} -0.155 \zz{0}{0}{1}{0}{1}{1}{0}{0}\nonumber \\
&&-0.321 \zz{1}{0}{0}{0}{0}{1}{1}{0}+0.260 \zz{1}{0}{0}{0}{1}{0}{0}{1}+0.810 \zz{1}{1}{1}{0}{0}{0}{0}{0}\,.\nonumber\\
&&
\end{eqnarray}
The corresponding NONs read
    \begin{eqnarray}
    \vec{\lambda}(2) &=& (0.981,0.974,0.955,0.029,0.017,0.016,0.015,0.012) \nonumber \\
    \vec{\lambda}(3) &=& (0.962, 0.951, 0.913, 0.0583, 0.0333, 0.0314, 0.0291, 0.0224) \nonumber \\
    \vec{\lambda}(12) &=& (0.826,0.826,0.659,0.253,0.146,0.127,0.095,0.067)\,.
\end{eqnarray}

\chapter{Summary and Conclusion}\label{chap:concloutl}
In this chapter we summarize our main results, draw some conclusions and present an outlook of further research questions inspired by the present thesis.

In Chap.~\ref{chap:Math} we studied the univariant quantum marginal problem (QMP). We learned that a unitary equivalence holds and as a consequence only the spectra of the reduced density operators are relevant. The solution of the univariant QMP found by Klyachko \cite{Kly4} is given by so-called marginal constraints, linear inequalities which define a high-dimensional polytope of possible spectra. We succeeded in breaking down
Klyachko's very abstract derivation to a more elementary level. We found a variational principle expressing arbitrary sums of eigenvalues of a hermitian matrix, which stands at the origin of the marginal constraint of the prototype QMP $\{A,B,AB\}$. The problem of finding all constraints amounts to a study of a very specific intersection property for so-called Schubert varieties, which arise naturally in the variational principle. The remaining steps of the derivation were straightforward but still quite involved. We confirmed that the Schubert varieties have indeed the structure of projective subvarieties of the flag variety and calculated at least for the Grassmannian flag variety the cohomology ring. We resigned from studying the important intersection property from this algebraic viewpoint. Instead, we developed the geometric picture behind it, which allowed us to study the intersection property by brute force. In that way we succeeded in determining marginal constraints for some elementary QMPs. However, for larger settings our brute force method seems to be useless and one needs to use Klyachko's abstract algorithm based on algebraic topology. Due to the same reason we also did not succeed in deriving the generalized Pauli constraints, the marginal constraints for the $1$-body pure $N$-representability problem. As an additional handicap we notice that even for small dimensional Hilbert spaces the number of marginal constraints is already very large. Due to the great importance of the generalized Pauli constraints it would be very important to get a more intuitive understanding of Klyachko's abstract approach using algebraic topology.

Although the structure of the family of generalized Pauli constraints is already quite involved for small particle numbers it is also important to study the generalized Pauli principle for macroscopic settings $\wedge^N[\mathcal{H}_1^{(d)}]$. This is at least very challenging but there is also hope that for large $N$ and $d$ some approximate statements can be derived. From the viewpoint of condensed matter physics these possible restrictions on occupation numbers $n_{\vec{k}}$ may be very interesting. It could be possible that the generalized Pauli constraints rule out certain decay behaviors of occupation numbers $n_{\vec{k}}$ in $\vec{k}$-space or restrict possible non-smoothness of $n_{\vec{k}}$ around the Fermi level.

In the second part, presented in Chap.~\ref{chap:Fermions}, we learned that the antisymmetry of $N$-fermion states leads to so-called generalized Pauli constraints which significantly strengthen Pauli's exclusion principle. These constraints on natural occupation numbers (NONs)
are again linear inequalities giving rise to a polytope $\mathcal{P}_{N,d}$ of possible NONs $\vec{\lambda}$, a proper subset of the ``Pauli-hypercube'' given by $0 \leq \lambda_i\leq 1$. To explore the physical relevance of those remarkable constraints we elaborated the pinning effect, suggested by Klyachko \cite{Kly1}. There, the NONs $\vec{\lambda}$ are lying exactly on the boundary of the polytope and one might expect that the saturated generalized Pauli constraints lead to strong restrictions for the corresponding quantum state. However, this exact pinning contradicts intuition. Why should NONs of interacting fermions exactly saturate $1$-particle constraints? Instead, we suggested the new effect of \emph{quasi-pinning}. There, one observes NONs $\vec{\lambda}$ very close but not exactly on the polytope boundary. Although (quasi-)pinning seems to be an interesting effect its physical relevance is not obvious. In \cite{Kly1} it was indicated and in the presented thesis it has been elaborated that pinning corresponds to a very specific and simplified structure of the corresponding $N$-fermion state $|\Psi_N\rangle$. As an example we have seen that pinning of NONs $\vec{\lambda}$ belonging to the Borland-Dennis setting $\wedge^3[\mathcal{H}_1^{(6)}]$
implies the structure
\begin{equation}\label{BDPinStruc5}
|\Psi_3\rangle = \alpha \,|1,2,3\rangle + \beta \,|1,4,5\rangle + \gamma \,|2,4,6\rangle\,
\end{equation}
with some $1$-particle states $|i\rangle$, and $|i_1,i_2,i_3\rangle$ denotes a Slater determinant. Indeed, (\ref{BDPinStruc5}) is much simpler than generic states $|\Psi_3\rangle \in \wedge^3[\mathcal{H}_1^{(6)}]$, linear combination of $\binom{6}{3}=20$ Slater determinants.
We provided strong evidence that this important relation of pinning and structure of the $N$-fermion quantum state is stable under small deviations.
On the one hand we verified this analytically in the neighborhood of the Hartree-Fock point $\vec{\lambda}_{HF} \equiv (1,1,1,0,0,0)$ for the specific setting $\wedge^3[\mathcal{H}_1^{(6)}]$ and on the other hand claimed numerical evidence for larger settings. Hence, quasi-pinning seems to correspond  approximately to the simplified structure of $|\Psi_N\rangle$ and would therefore be physically relevant. It is one of the most important tasks for the future to confirm this rigorously.

The expected structural implication of quasi-pinning for $|\Psi_N\rangle$  (see e.g.~\ref{BDPinStruc5}) leads to a first application, a generalization of the Hartree-Fock method. For most Hamiltonians describing interacting fermions it is impossible to analytically determine the ground state. Therefore, one often resorts to approximations. One of the most well-known ones is the Hartree-Fock approximation, where one minimizes the energy expectation value only for single Slater determinants, $|\Psi_N\rangle = |i_1,\ldots,i_N\rangle$. Whenever the correct, but unknown ground state is weakly correlated (this is e.g.~the case for atoms) this method will work very well. We suggest to generalize the ansatz of a single Slater determinant by the more general state structure corresponding to exact pinning (see e.g.~(\ref{BDPinStruc5})). The method of linearly superposing several Slater determinants to improve the ground state approximation is of course already well-established in quantum chemistry as multi-configurational self-consistent field  methods (MCSCF). However, our findings show that we need to superpose only a few, but carefully chosen, Slater determinants. Compared to the standard MCSCF method this would lead to a significant computational speed up, which can be used to increase the accuracy by increasing the dimension of the truncated $1$-particle Hilbert space.

The number of generalized Pauli constraints significantly increases as one increases the dimensions $d$ of the underlying $1$-particle Hilbert space. Moreover, the family of constraints has only been determined for settings with $d\leq 10$ and without a more intuitive understanding of the algebraic topological approach there is not much hope that one can calculate these constraints with Klyachko's algorithm for $d\gg10$. Since most physical models are based on a large or even infinite-dimensional $1$-particle Hilbert space it is not clear at all how one can investigate possible (quasi-)pinning of specific physical states. To make this possible we developed the concept of a truncated pinning analysis. We learned that NONs which are exactly equal to $1$ or $0$ can be omitted for the pinning analysis. Moreover, we quantitatively verified that all NONs very close to $1$ or $0$ can also be omitted. In that sense one can perform a pinning analysis in a truncated setting, where a possible result on quasi-pinning translates one-to-one to the correct setting, but with a small truncation error, given by the distance of the neglected NONs to $1$ and $0$, respectively. The systematic application of this pinning analysis for systems with a few electrons, like atoms or ions in their ground state or in an excited state would allow to explore the existence of possible quasi-pinning.

In Chap.~\ref{chap:Physics} we investigated fermionic states from the new viewpoint of generalized Pauli constraints. The central question was whether we can find pinning or quasi-pinning. As a first model we studied in one dimension few harmonically coupled spinless fermions confined to a harmonic trap. For the case of three fermions, we did not only succeed in analytically determining the ground state but also found the $1$-particle reduced density operator $\rho_1$. We obtained the corresponding NONs numerically and afterwards confirmed them for the regime of weak interacting by applying high-order degenerate perturbation theory to $\rho_1$.
Using the truncated pinning analysis we observed strong quasi-pinning. In the regime of a small coupling $\delta\ll 1$ we found a distance of the ground state NONs $\vec{\lambda}(\delta)$ to the polytope boundary of order $\delta^8$ (see also \cite{CS2013}). This quasi-pinning is also present for medium interaction and vanishes only for very strong couplings. The same results were confirmed numerically for particle numbers $N=4,\ldots,8$. Moreover, by also investigating the first few excited states for $N=3$ we identified quasi-pinning as an effect of the lowest few energy eigenstates. This provides first insights into the mechanism behind ground state quasi-pinning. The energy minimization, which yields the ground state is in strong conflict to the antisymmetry (responsible for the existence of the generalized Pauli constraints) in the sense that skipping the antisymmetry would lead to much lower ground state energies.
The analytical study of the influence of this competition on quasi-pinning would be important, as well.
Another surprising result for the harmonic model was the similarity between the natural orbitals of the fermionic and bosonic ground state (see also \cite{CS2013NO}). Furthermore, the corresponding NONs, $\lambda_k^{(f)}$ and $\lambda_k^{(b)}$ exhibit identical exponential decay for $k\rightarrow \infty$. It would be interesting to study whether this also holds for fermionic and bosonic systems with anharmonic interactions.

As a second model we studied analytically the Hubbard model for three electrons on three lattice sites. We learned that due to the translational symmetry and spin-symmetries the total Hilbert space (already quite small) splits into several even smaller subspaces. The presence of symmetries is also the mathematical reason for the surprising findings of exact pinning. By considering larger settings we numerically verified that the ground state pinning vanishes. This is already the case for all settings with five lattice sites. Although the results of exact pinning seem to be not relevant for macroscopic systems in condensed matter physics they may be important from different viewpoints. The Hubbard model of e.g.~three electrons on three lattice sites can be realized in the lab, e.g.~by optical traps, and it also serves as a prime model for the description of some molecules as e.g.~benzene. In that sense the pinning effect can have a strong influence on the behavior of some concrete physical systems. One of the hopes would be to find a physical quantity for those systems which is monotonously related to pinning. Then, the generalized Pauli constraints would impose a ``magical'' (kinematic) bound on this quantity.

Although the validity of Pauli's exclusion principle is extremely well verified in experiments the antisymmetry of the corresponding fermionic wave function under particle exchange is only motivated theoretically and measured indirectly. The reason for this is that measuring occupation numbers is much easier than measuring two-particle amplitudes. However, thanks to Klyachko's work, there is now a chance to reduce the task of verifying the antisymmetry to measuring occupation numbers. Indeed, all measured tuples of occupation numbers are not only expected to lie inside of the ``Pauli hypercube'' $1\leq \lambda_i \leq 1$, but should be restricted to the corresponding polytope. Confirming this experimentally, would provide very strong evidence for the correctness of the antisymmetry. Alternatively, if one accepts the antisymmetry as necessity following from the indistinguishable character of the particles, such a measurement could verify that e.g.~those particle, which we call electrons and which seem to be identical from our state of knowledge are identical, indeed. This would rule out the possibility of additional physical properties (as e.g.~``flavors'', which are not know yet) leading to different types of electrons, which couldn't be distinguished in experiments so far.

In general, the main result of quasi-pinning for the harmonic system is quite surprising. Due to the expected relation of quasi-pinning and structure of the corresponding $N$-fermion state the generalized Pauli constraints play a very important role at least for the ground states of some fermionic systems. The central question stimulated by this thesis is whether the effect of ground state quasi-pinning is generic for all few-fermion systems trapped by some confinement potential. Does it also show up for $2$- or $3$-dimensional systems? Moreover, we still did not yet sufficiently understand the mechanism behind quasi-pinning. A strong candidate which still needs to be analyzed is the conflict of antisymmetry and energy minimization mentioned above. It is in particular well-known that for identical particle systems confined by some external trap the ground state is symmetric under particle exchange and that the corresponding bosonic ground state energy is much lower than the corresponding fermionic one. In that sense it clear that the antisymmetry strongly limits the ground state properties of fermionic quantum systems. However, it still needs to be explored whether this is just due to Pauli's exclusion principle or whether additional fermionic features (as e.g.~the generalized Pauli constraints) have also an influence.

If the suggested quasi-pinning effect turns out to be generic for the regime of not too strong interaction and thus holds for most fermionic quantum systems confined by some trap one could improve numerical methods by making use of the important structural implications of pinning on the corresponding $N$-fermion quantum state. A first example is the generalization of the Hartree-Fock method, which was already mentioned above. A second, much more challenging and speculative idea would be to use the knowledge of generic quasi-pinning for DMRG methods. There, one would intend to truncate the local Hilbert spaces to the subspaces corresponding to exact pinning. An alternative idea would be to study certain macroscopic systems from a few-fermion viewpoint. E.g.~for localized or itinerant electronic systems one could consider a small region inside the material and study the few electrons contained there in the field generated by the nuclei and all the other electrons. This would give rise to an external potential similar to that of a harmonic trap. If the dynamics of those few electrons is sufficiently slow we expect that their quantum state is (approximately) the ground state of that effective trap and therefore (quasi-)pinned. From a broader viewpoint this means to generalize the elementary $1$-particle picture for weakly correlated quantum systems to a picture where the smallest unit subsystem is not given by a single electron anymore but by a small group of a few electrons described by a quasi-pinned state. Similar techniques can also be found in quantum information theory. There, for a given many spin state, an important task is to extract as many Bell pairs as possible. In that case the Bell pair of two spins would define this smallest unit and their quantum state is given by the famous Bell state.

\begin{appendix}
\chapter{Supporting Material for Chapter \ref{chap:Math}}\label{app:chapMath}
\section{Some Additional Mathematical Concepts}\label{rest}
\subsection{Basic definitions}\label{app:definitions}
\begin{defn}\label{quotienttop}
Let $(X,\mathfrak{T}_{E})$ be a topological space, $F$ a set and $\varphi: X \rightarrow F$ a map. $\varphi$ induces a topology $\mathfrak{T}_F$ on $F$ by:
$U \subseteq F$ open if and only if $\varphi^{-1}(U) \in \mathfrak{T}_E$. This topology is called quotient topology on $F$ induced by $\varphi$. It is the strongest topology on $F$ under the condition that $\varphi$ is continuous.
\end{defn}

\begin{defn}
Let $\K$ be a field. The \textbf{polynomial ring $\K[X]$} is defined as the free (additive) abelian group generated by the `independent' symbols $X^0, X^1, X^2, \ldots$ with coefficients in $\K$, equipped with an abelian multiplicative structure defined by the rule $X^k X^l := X^{k+l}$, which is then linearly extended to the product of arbitrary elements of $\K[X]$. The elements of $\K[X]$ are called polynomials of the variable $X$ over the field $\K$.
\end{defn}
We abbreviate $X^0 \equiv 1$ and $X^1 \equiv X$. Moreover, we can extend the definition of a polynomial ring in one variable $X$ in the natural way to the polynomial ring $\K[X_1, X_2, \ldots, X_n]$ of independent variables $X_1, X_2, \ldots, X_n$.
\begin{defn} The \textbf{roots} of a polynomial $p(X) \in \K[X]$ is the set of solutions of the equation $p(X)=0$, where $X$ now takes concrete values in $\K$.
\end{defn}

\begin{defn} A field $\K$ is said to be \textbf{algebraically closed} iff every polynomial $p(X) \in \K[X]$ with one variable of degree at least $1$ has a root in $\K$.
\end{defn}
This definition is equivalent with the property that every polynomial $p(X) \in \K[X]$ with one variable of degree at least one
could be written as product of degree $1$ polynomials.
We give some examples:\\
\itemize
\item The field $\mathbb{C}$ of complex numbers is algebraically closed
\item The field $\mathbb{R}$ is not algebraically closed since $X^2 +1 = 0 $ has no solution in $\mathbb{R}$.
\item Every subfield of $\mathbb{R}$ is not algebraically closed
\item Every finite field $\K = \{k_1,k_2,\ldots,k_n\}$ is not algebraically closed since $p(X) = (X-k_1)\cdot \ldots \cdot (X-k_n) +1$ has no root in $\K$
\item The field of algebraic numbers is algebraically closed. An algebraic number is the root in $\mathbb{C}$ of a nonzero polynomial in  $\mathbb{Q}[X]$ or equivalently $\mathbb{Z}[X]$. Even the fact that these numbers form a field is not trivial.
\begin{defn} A \textbf{ring} $(R,+,\cdot)$ is a triple of a set $R$ and two binary operations $+$ (addition) and $\cdot$ (multiplication) with the following properties
\begin{enumerate}
\item  $(R,+)$ is an (additive) abelian group
\item $(R,\cdot)$ is closed under multiplication, the associativity law holds and there exists a neutral element $1$.
\item $\forall r, s, t \in R:\, r (s +t) = r s + r t$ and $ (s + t) r = s r +  t r $
\end{enumerate}
\end{defn}
In some definitions of a ring the existence of an identity is skipped and they use the name unitary ring for a ring with identity.

\begin{defn} An (two-sided) \textbf{ideal $\mathcal{I}$} of a ring $(R,+,\cdot)$ is a subset of $R$ with the properties
\begin{enumerate}
\item $\mathcal{I}$ is a subgroup of $(R,+)$
\item $\forall r \in R, x \in \mathcal{I}: r x \in \mathcal{I}$
\end{enumerate}
\end{defn}
An ideal is not necessary a subring of $R$ since it does not automatically include the neutral element $1$ of $(R,\cdot)$.
Every subset $S \subseteq R$ generates in the natural way an ideal $\mathcal{I}(S)$, which also could be defined as the intersection of all ideals of $R$ that contain $S$.
We add ideals $\mathcal{I}$ and $\mathcal{J}$ of a ring $R$ by
\begin{equation}
\mathcal{I}+\mathcal{J} = \{x +y | x \in \mathcal{I},  y \in \mathcal{J}\}
\end{equation}
and we multiply them by taking the ideal generated by the set
\begin{equation}
 S= \{x \cdot y |x \in \mathcal{I},  y \in \mathcal{J}\}
\end{equation}

\begin{defn} The \textbf{zero locus $Z[S]$} of an subset $\mathcal{S} \subseteq \K[X_1,\ldots,X_n]$ of a polynomial ring $\K[X_1,\ldots,X_n]$ is defined as
\begin{equation}
Z[S] := \{x \in \K^n | p(x) = 0 \,\,\forall p \in \K[X_1,\ldots,X_n]\}
\end{equation}
\end{defn}
We find $Z[S] = Z[\mathcal{I}(S)]$ and hence we refer in the following zero loci only to ideals.
\subsection{Cell decomposition and CW complex}\label{CW-complex}
In this section we introduce the concept of cellular decompositions of topological spaces. Later we can use these concepts to easily calculate the homology and cohomology of decomposable topological spaces.

\begin{defn}\label{cells}
A closed \emph{$n$-cell} is a topological space that is homeomorphic to an Euclidean $n$-ball $E^{(n)}$.
\end{defn}

\begin{defn}\label{def:cellstructure}
Let $X$ be a set. The pair $(X,\Phi)$ consisting of $X$ and a family $\Phi = \{\varphi_{\alpha}\}$ of maps $\varphi_{\alpha}: E^{(n_{\alpha})} \rightarrow X$ is called a cell structure if
\begin{enumerate}
\item $\forall \alpha: \varphi|_{\mathring{E}^{(n_{\alpha})}}$ is injective
\item $\{\varphi_{\alpha}(\mathring{E}^{(n_{\alpha})})\}_{\alpha}$ is a partition of $X$, i.e.~$\bigcup_{\alpha}  \varphi_{\alpha}(\mathring{E}^{(n_{\alpha})}) = X$ and $\forall \alpha, \beta, \alpha \neq \beta: \varphi_{\alpha}(\mathring{E}^{(n_{\alpha})}) \cap \varphi_{\beta}(\mathring{E}^{(n_{\beta})}) = \emptyset$
\item $\forall \alpha: {\varphi_{\alpha}}(\partial E^{(n_{\alpha})}) \subseteq \bigcup_{\beta, n_{\beta}< n_{\alpha}}{\varphi_{\beta}}(\mathring{E}^{(n_{\beta})}) $
\end{enumerate}
\end{defn}
The sets $\sigma_{\alpha}:= \varphi|_{E^{(n_{\alpha})}}$ are called \textbf{\emph{cells}}, the maps $\varphi_{\alpha}:\varphi|_{E^{(n_{\alpha})}}\rightarrow \sigma_{\alpha} $ \emph{characteristic maps} and the union $X^{(n)} := \bigcup \{\varphi_{\beta}(\mathring{E}^{(n_{\beta})})|\mbox{dom}({\varphi_{\beta}})= E^{(n_{\beta})} \wedge n_\beta \leq n \}$ the \emph{$n$-skeleton}.
We find
\begin{lem}
Let $(X,\Phi)$ be a cell structure. Then
\begin{enumerate}
\item $\forall \alpha: \varphi_{\alpha}(\mathring{E}^{(n_{\alpha})}) = \varphi_{\alpha}(E^{(n_{\alpha})}) \setminus  \varphi_{\alpha}(\partial E^{(n_{\alpha})})$
\item $\forall \alpha:\sigma_{\alpha} \subseteq X^{(n_{\alpha})}$, i.e.~every $n$-cell of $(X,\Phi)$ is a subset of $X^{(n)}$
\item $X^{(k)} = \bigcup \{\varphi_{\beta}(\sigma_{\beta})|\mbox{dom}({\varphi_{\beta}})= E^{(n_{\beta})} \wedge n_\beta \leq k \}$
\end{enumerate}
\begin{proof}
\begin{enumerate}
\item
Due to condition $3$ in Definition \ref{def:cellstructure} $\varphi_{\alpha}(\partial E^{(n_{\alpha})})$ is contained in $X^{(n_{\alpha}-1)}$ and since $\{\varphi_{\alpha}(\mathring{E}^{(n_{\alpha})})\}_{\alpha}$ is a partition of $X$, $\varphi_{\alpha}(E^{(n_{\alpha})})= \varphi_{\alpha}(\mathring{E}^{(n_{\alpha})}) \cup \varphi_{\alpha}(\partial E^{(n_{\alpha})})$ is a disjoint union. Hence, $\varphi_{\alpha}(\mathring{E}^{(n_{\alpha})}) = \varphi_{\alpha}(E^{(n_{\alpha})}) \setminus  \varphi_{\alpha}(\partial E^{(n_{\alpha})})$.
\item By definition $\sigma_{\alpha} \setminus \partial \sigma_{\alpha} \in X^{(n_{\alpha})}$ and due to property $3$ in Definition \ref{def:cellstructure}, $\sigma_{\alpha} \subseteq X^{(n_{\alpha})}$
\item This point is clear in the spirit of the last two points.
\end{enumerate}
\end{proof}
\end{lem}

Moreover, the family $\Phi$ induces a topology $\mathfrak{T}^{(CW)}$ on $X$ according to
\begin{enumerate}
\item $\forall \alpha, \sigma_{\alpha}$ is made into a topological space by choosing the quotient topology induced $\varphi_{\alpha}$ (see also \ref{quotienttop}), namely the strongest topology $\mathfrak{T}_{\alpha}$ such that $\varphi_{\alpha}$ is still continuous
\item The topology $\mathfrak{T}^{(CW)}$ is defined as the weakest topology with respect to the subsets ${\sigma_{\alpha}}$, i.e.~a set $A \subseteq X$ is closed in $(X,\mathfrak{T}^{(CW)})$ if and only if $A \cap \sigma_{\alpha}$ is closed in $(\sigma_{\alpha},\mathfrak{T}_{\alpha})$ for all $\alpha$.
\end{enumerate}
The triple $(X,\Phi,\mathfrak{T}^{(CW)})$ is called CW complex.

We can prove
\begin{lem}\label{lemhomeo1}
Let $(X,\Phi,\mathfrak{T}^{(CW)})$ be a CW complex and $\mathfrak{T}^{(CW)}$ the CW topology induced by $\Phi$. Then
$\forall \alpha, \tilde{\varphi}_{\alpha}:=\varphi_{\alpha}|_{\mathring{E}^{(n_{\alpha})}}$ is a homeomorphism onto $\varphi_{\alpha}({\mathring{E}^{(n_{\alpha})}})$ with respect to the quotient topologies $\tilde{\mathfrak{T}}_{\alpha}$ on $\varphi_{\alpha}(\mathring{E}^{(\alpha)})$ and $\tilde{\mathfrak{T}}_{eucl}^{(n_{\alpha})}$ on $\mathring{E}^{(n_{\alpha})}$, induced by the topologies $\mathfrak{T}_{\alpha}$ on $\sigma_{\alpha}$ and $\mathfrak{T}_{eucl}^{(n_{\alpha})}$ on $E^{(n_{\alpha})}$.
\end{lem}
\begin{proof}
Let $V = U \cap \varphi_{\alpha}(\mathring{E}^{(\alpha)})\in \tilde{\mathfrak{T}}_{\alpha} $ for some $U \in \mathfrak{T}_{\alpha}$. Hence
\begin{eqnarray}
\tilde{\varphi}_{\alpha}^{-1}(V)&=& \tilde{\varphi}_{\alpha}^{-1}(U) \cap \tilde{\varphi}_{\alpha}^{-1}(\varphi_{\alpha}(\mathring{E}^{(n_{\alpha})}))\nonumber \\
&=& (\varphi_{\alpha}^{-1}(U)\cap  \mathring{E}^{(n_{\alpha})} ) \cap \mathring{E}^{(n_{\alpha})}  \nonumber \\
&=& \varphi_{\alpha}^{-1}(U)\cap  \mathring{E}^{(n_{\alpha})}
\end{eqnarray}
and since $\varphi_{\alpha}$ continuous with respect to $\mathfrak{T}_{\alpha}$ and $\mathfrak{T}_{eucl}^{(n_{\alpha})}$, $\tilde{\varphi}_{\alpha}^{-1}(V) \in \tilde{\mathfrak{T}}_{eucl}^{(n_{\alpha})}$. Since $\tilde{\varphi}_{\alpha}$ bijective onto $\varphi_{\alpha}({\mathring{E}^{(n_{\alpha})}})$, it is also a homeomorphism.
\end{proof}

%\begin{proof}
%\begin{enumerate}
%Let $\alpha$ be arbitrary but fixed.
%\item Let $A \subseteq E^{(n_{\alpha})}$. Since $E^{(n_{\alpha})}$ is compact, $A$ is also compact. Continuity (given per definition of $\mathfrak{T}_{\alpha}$) of $\varphi_{\alpha}$ yields compactness of $\varphi_{\alpha}(A)$ in $(\sigma_{\alpha}, \mathfrak{T}_{\alpha})$. Since
%\item
%\end{enumerate}
%\end{proof}

%First let us recall the definition of a disjoint union of topological spaces. Let $(X_i,\mathcal{T})$ a topological space.

%\item $\forall \alpha, \sigma_{\alpha}$ is closed in $(X,\mathfrak{T}^{(CW)})$

Now we define CW complexes for topological spaces:

\begin{defn}\label{defcwcomplex2}
Let $(X,\mathfrak{T})$ be a topological space and $(X,\Phi)$ a cell structure.
Then, if every characteristic map $\varphi_{\alpha}|{\mathring{E}^{(n_{\alpha})}}$ is a homeomorphism from $\mathring{E}^{(n_{\alpha})}$ onto $\varphi_{\alpha}({\mathring{E}^{(n_{\alpha})}})$ in the spirit of Definition \ref{lemhomeo1}
we call $(X,\Phi,\mathfrak{T}^{(CW)})$ a CW complex.
\end{defn}

The following important theorem holds
\begin{thm}\label{thmCW}
Let $(X,\mathfrak{T})$ be a topological space and $(X,\Phi,\mathfrak{T}^{(CW)})$ its CW complex according Definition \ref{defcwcomplex2}.
Then $\mathfrak{T} = \mathfrak{T}^{(CW)}$
\end{thm}
\begin{proof}
We prove the theorem by showing
\begin{enumerate}
\item $\forall x \in X, U \in \mathfrak{T}^{(CW)}(x) \exists V \in \mathfrak{T}(x)$ such that $V \subseteq U$
\item $\forall x \in X, U \in \mathfrak{T} \exists V \in \mathfrak{T}^{(CW)}(x)(x)$ such that $U \subseteq V$
\end{enumerate}
Therefore, we first observe $\forall A \subseteq X, A = \mathring{\cup}_{\alpha} {A\cap \mathring{\sigma_{\alpha}}}$.
Let $x \in X$ arbitrary but fixed. Moreover choose $U \in \mathfrak{T}^{(CW)}(x)$ and decompose it according this observation, namely $U = \mathring{\cup}_{\beta} U_{\beta},$ where we introduced $U_{\beta}= U \cap \mathring{\sigma}_{\beta}$. There exists a unique $\alpha$ such that $x \in U_{\alpha}$.  According to the construction of our topologies we find that $U_{\beta} \in \tilde{\mathfrak{T}}_{\alpha}$. Since $\tilde{\varphi}_{\alpha}$ is a homeomorphism with respect to $\tilde{\mathfrak{T}}_{\alpha}$, but also with respect to $\mathfrak{T}_{rel,\alpha}$, $U_{\beta} \in \mathfrak{T}_{rel,\alpha}$. Hence, $U \in \mathfrak{T}(x)$. The same argument also works for the other direction, i.e.~for every $U \in \mathfrak{T} \exists V \in \mathfrak{T}^{(CW)}(x)(x)$ such that $U \subseteq V$.
\end{proof}
We give some general remarks on cellular decompositions. Consider a topological space $(X,\mathfrak{T})$ and assume that there exists a CW complex structure $(X,\Phi,\mathfrak{T})$ on $(X,\mathfrak{T})$. Theorem \ref{thmCW} states that the CW topology is equal to the given topology $\mathfrak{T}$. Every other family $\Phi'$ given by $\Phi$ where we may change the characteristic maps $\varphi_{\alpha}$ on the boundary $\partial E^{(n_{\beta})}$ of $E^{(n_{\beta})}$ under the condition 3 of Def.~\ref{def:cellstructure} will lead to the same induced CW topology, namely $\mathfrak{T}$.Therefore, to verify that a topological space could be interpreted as a CW complex it is enough to find a disjunct decomposition in `open' cells, which are homeomorphic to $\mathring{E}^{(n)}$.

\subsection{Algebraic varieties}\label{appvarieties}
\begin{defn} Let $\K$ be an algebraically closed field and $\mathcal{I} \subseteq \K[X_1,\ldots,X_n]$ an ideal of the polynomial ring $\K[X_1,\ldots,X_n]$. A subset $V \subseteq \K^n$ is called an affine algebraic subset of $\K^n$ if $V = Z[\mathcal{I}]$ for some ideal $\mathcal{I}$ of $\K[X_1,\ldots,X_n]$.
\end{defn}

\begin{defn}
An affine algebraic subset $V$ of $\K^n$ is called an affine algebraic variety of $\K^n$ if $V$ is irreducible, i.e.~$V$ is not the union of two proper affine algebraic subsets.
\end{defn}
There are also definitions that are calling every algebraic set an affine algebraic variety.

\section{Some simple Proofs of Trivialities}\label{trivial}
We prove Lem.~\ref{KyFan}.
\begin{proof}
Let $\mathcal{H}$ be a complex, $d$-dimensional Hilbert space and $k < d$ a fixed integer, choose an orthonormal set $\{v_1, \ldots, v_k\}$, set $V = \langle v_1,\ldots,v_k \rangle$ and denote the orthogonal projection operator onto $V$ by $P_V$. Moreover we choose a unitary matrix $U$ such that $v_j = \sum_{i=1}^d U_{j i} e_i$, where $e_j$ is the eigenvector of $\rho$ corresponding to $\lambda_j$. We find ($\lambda_i$ arranged non-increasing)
\begin{eqnarray}
\mbox{Tr}[P_V \rho] &=& \sum_{j=1}^k \langle v_j,\rho v_j\rangle \nonumber \\
 &=& \sum_{j=1}^k \sum_{i=1}^d |U_{ji}|^2 \lambda_i \nonumber \\
&\leq& \sum_{j=1}^k \sum_{i=1}^k |U_{ji}|^2 \lambda_i + \sum_{j=1}^k \sum_{i=k+1}^d |U_{ji}|^2 \lambda_{k+1} \nonumber \\
&=& \sum_{i=1}^k \sum_{j=1}^k |U_{ji}|^2 \lambda_i + \sum_{i=1}^k \sum_{j=k+1}^d |U_{ji}|^2 \lambda_{k+1} \nonumber
\end{eqnarray}
\begin{eqnarray}
\qquad \qquad \,\,\,&\leq& \sum_{i=1}^k \sum_{j=1}^k |U_{ji}|^2 \lambda_i + \sum_{i=1}^k \sum_{j=k+1}^d |U_{ji}|^2 \lambda_{i} \nonumber \\
&=& \sum_{i=1}^k \lambda_i \,.
\end{eqnarray}
On the other hand, if we choose $V = \langle e_1,\ldots e_k \rangle$ we find $\mbox{Tr}[P_V \rho] = \sum_{i=1}^k \lambda_i$, which leads to (\ref{KyFan1}).
The second variational principle (\ref{KyFan2}) follows immediately by applying (\ref{KyFan1}) to the state $\tilde{\rho}:= -\rho$.
\end{proof}

\section{Schubert Calculus}\label{Schubertcalculus}
We verify a statement used in Sec.~\ref{flagvarieties} which we state here as lemma:
\begin{lem}
Let $B \subset Gl(n)$ the set of regular complex upper triangle matrices. Then $B$ is a subgroup (w.r.t matrix multiplication).
\end{lem}
\begin{proof}
Obviously, the product of two regular upper triangle matrices is again a regular upper triangle matrix and $B$ also contains the identity $\mathds{1}$.
We still have to show the existence of inverse elements in $B$.
Let $b \in B$.  Hence, b is invertible in $Gl(n)$, i.e.~$\exists a \in Gl(n)$ such that $b a = a b = \mathds{1}$. We have an explicit formula for the matrix $a$, namely
\begin{equation}
a_{k l} = \frac{(-1)^{i+j}}{\mbox{det}(b)} \left| \begin{array}{cccccc}  b_{1 1}&& b_{1 k-1}& b_{1 k+1}&& b_{1 d}\\ &\ddots & \vdots& \vdots& $\reflectbox{$\ddots$}$ &\\  b_{l-1 1}& \cdots & b_{l-1 k-1}& b_{l-1 k+1}&\hdots &  b_{l-1 d} \\  b_{l+1 1}&\cdots & b_{l+1 k-1}& b_{l+1 k+1}&\hdots &  b_{l+1 d}\\ &$\reflectbox{$\ddots$}$ &\vdots&\vdots&\ddots &\\ b_{1 d}&& b_{d k-1}& b_{d k+1}&& b_{d d}\end{array} \right| \,.
\end{equation}
If $k>l$ we find $a_{k l}=0$. This means  that $a = b^{-1}$ is also an upper triangle matrix. Since a product of $b_1, b_2$ is also an upper triangle matrix, $B$ is a subgroup of $Gl(n)$ and the left cosets $g B$ are well defined.
\end{proof}

To construct a more concrete representation of flags in particular to obtain \ref{flagvarietyGL3} we still have to verify
\begin{lem}
All elements $g$ of a given equivalence class/left coset $g_0 B$ lead to the same CEF or stated differently for an matrix $g^c$ with CEF the product $g^c \, b$ with $b \in B$ has also of CEF iff $b = \mathds{1}$.
\end{lem}
This verifies the $1-1$ correspondence between flags and CEF used in Sec.~\ref{flagvarieties}.
\begin{proof}
Given a matrix $g = \left(\vec{g}_1,\ldots,\vec{g}_d\right)\in Gl(d)$. It is clear that the algorithm described in Sec.~\ref{flagvarieties}
yields a unique column echelon form $g^{C} \in Gl(d)$. We show that $g^{(C)} b$ has CEF iff $b = \mathds{1}$.
We partially calculate $h = g^{C} b$ and by assuming $h$ to have CEF we derive $b = \mathds{1}$.  Let the position of the pivots of $g^{C}$  be ${(\alpha_k,k)}_{k=1}^d$. For $\alpha < \alpha_1$ we find $h_{\alpha 1}= \vec{g}_{\alpha} \cdot \vec{b}_1=0$. Since $h_{\alpha_1 1}= \vec{g}_{\alpha_1} \cdot \vec{b}_1= b_{11}\neq 0$, we conclude $b_{11}=1$. Since for all $k>1$, $0 = h_{\alpha_1 k} = g_{\alpha_1} \cdot b_k = b_{1 k} $ we also find $b_{1 k}=0$ for $k>1$. Analyzing the matrix element $h_{\alpha_2 2}$ then leads in the same way as before $b_{22}=1$. Repeating this procedure leads finally to $b = \mathds{1}$.
\end{proof}

In the proof of Lem.~\ref{manifoldatlas} we have used the following lemma:
\begin{lem}
Given a reference orthonormal basis $\{e_1,\ldots, e_d\}$. The open set of flags transversal to the standard flag $E_{\bullet}$, that has linear subspaces $E_k = \langle e_1,\ldots,e_k \rangle $ is represented in the spirit of (\ref{flagvarietyGL2}) by the set of CEF matrices of the form
\begin{equation}
\left( \begin{array}{cccc} \ast & \cdots & \ast & 1 \\ \vdots & &  $\reflectbox{$\ddots$}$  & 0 \\ \ast & 1 &   & \vdots \\ 1 & 0 & \cdots & 0 \end{array}\right) \,, \label{transversestandard}
\end{equation}
where every star represents one complex variable.
\end{lem}
\begin{proof}
First, we observe that any flag represented by a CEF of type (\ref{transversestandard}) is indeed transverse to the standard flag $E_{\bullet}$. Consider a regular matrix $g = (\vec{g}_1,\ldots,\vec{g}_d)$ with a CEF and transverse to $E_{\bullet}$. We will conclude that its pivots are on the same positions as the ones in (\ref{transversestandard}). Lets first analyze $\vec{g}_1$ and its span $\langle \vec{g}_1\rangle$ in $\mathbb{C}^d$. Transversality implies that
$\vec{g}_1 \not \in \langle e_1,\ldots,e_{d-1}\rangle$. This is only fulfilled if the $d$-th component of $\vec{g}_1$ does not vanish, i.e.~the pivot of $\vec{g}_1-$ is on position $d$. Hence, all the other vectors $\vec{g_2},\ldots, \vec{g_d}$ have a vanishing last component. Let's go on to $\vec{g}_2$. Transversality then states that $\vec{g}_2 \not \in \langle e_1,\ldots,e_{d-2}\rangle $, which implies that the $(d-1)$-th component of $\vec{g}_2$ cannot vanish (since the $d$-th has to vanish). Thus, the corresponding pivot has position $d-1$. Continuing this procedure yields the CEF type presented in (\ref{transversestandard}).
\end{proof}

We prove Lem.~\ref{zero}.
\begin{proof}
In the first step we show that every image $\theta(U)$ is a zero of the system described in Lem.~\ref{zero} and then in the second step we verify that every point in $\mathbb{P}[\bigwedge^d V]$ given by homogeneous coordinates meeting these conditions is the image $\theta(U)$ of some $U \in Gr_{d,n}$.
For the first part, let us expand one of these conditions by the use of $A = (a_{i j})$ and $A_{i_1,\ldots,i_{d-1}}^{\hat{m}}$ has the $m^{th}$ column deleted:
\begin{eqnarray}
\lefteqn{\sum_{k=1}^{d+1} (-1)^k p_{i_1,\ldots,i_{d-1},j_k}(A)\, p_{j_1,\ldots,\hat{j_k},\ldots,j_{d+1}}(A)} \nonumber\\
&=& \sum_{k=1}^{d+1} (-1)^k \sum_{m=1}^{d} (-1)^{d+m} a_{j_k m} \,\mbox{det}(A_{i_1,\ldots,i_{d-1}}^{\hat{m}}) \,\mbox{det}(A_{j_1,\ldots,\hat{j_k},\ldots,j_{d+1}}) \nonumber \\
&=& \sum_{m=1}^{d} (-1)^{d+m}\,\mbox{det}(A_{i_1,\ldots,i_{d-1}}^{\hat{m}}) \nonumber \\
&&\cdot \left(\sum_{k=1}^{d+1} (-1)^k a_{j_k m} \,\mbox{det}(A_{j_1,\ldots,\hat{j_k},\ldots,j_{d+1}})\right)
\end{eqnarray}
The last factor of the right side in the last equation is the determinant of a $(d+1)\times (d+1)$ matrix $\tilde{A}_m$, expanded along the first column. This matrix reads
\begin{equation}
\tilde{A}_m = \left(\begin{array}{cccc}
a_{j_1,m}& a_{j_1,1} & \cdots &a_{j_1,d} \\
\vdots & \vdots&&\vdots \\
a_{j_{d+1},m}& a_{j_{d+1},1} & \cdots &a_{j_{d+1},d}
\end{array}\right)\,.
\end{equation}
Since the first and $(m+1)^{th}$ column are the same, the determinant of $\tilde{A}_m$ vanishes and every element $U \in Gr_{d,n}$ leads to a zero of the system in Lem.~\ref{zero}. For the second part let us choose $q = (q_{\underline i})$ representing a point in $\mathbb{P}[\bigwedge^d V]$ and w.l.o.g. there exists a $\underline l=(l_1,\ldots,l_d)$ such that $q_{\underline l}=1$ and we define $a_{i k} := q_{l_1,\ldots,l_{k-1},i,l_{k+1},\ldots,l_d}$. By the use of decreasing induction on $|l\cap k|$ we prove the claim
\begin{equation}\label{Plrelations}
p_{\underline k}(A) = q_{\underline k}
\end{equation}
for all $\underline  k$ such that $k_1 < k_2< \ldots < k_d  $. To begin, let $\underline k = \underline l$. Then $A_{\underline l}= \mathds{1}_d$ and hence $p_{\underline k}(A) = \mbox{det}(A_{\underline l}) = 1$. Moreover let $\underline k = (l_1,\ldots, l_{k-1},m,l_{k+1},\ldots,l_{d})$. Also then we find $p_{\underline k} = \mbox{det}(A_{l_1,\ldots,l_{k-1},m,l_{k+1},\ldots,\l_{d}}) = q_{\underline  l}$ and the claim also holds for $|\underline k \cap \underline l| = d-1$.

Now assume claim (\ref{Plrelations}) holds for all $\underline{k}^{\prime}$ fulfilling $\#(\underline{k}^{\prime} \cap \underline{l})\ge m$, let $\underline{k}$ be an arbitrary, such that $1\leq k_1 <k_2<\ldots<k_d \leq n$, $\#(\underline{k} \cap \underline{l}) = m-1 $ and define $\underline i = (\l_1,\ldots,\l_{d-1})$ and $\underline j=(l_d,k_1,\ldots,k_d)$. W.l.o.g. we assume $l_d \notin \{k_1,\ldots,k_d\}$.
Since $(q_{\underline i})$ is a zero, we have
\begin{equation}
q_{\underline l} q_{\underline k} + \sum \pm q_{\underline l ^\prime} q_{\underline k ^\prime} = 0
\end{equation}
For each part of the sum above either $\#(\underline{l})\cap \underline{k}^{\prime} = m$  or $q_{\underline{l}^{\prime}}=0$ holds. Hence, we can substitute $q_{\underline{k}^{\prime}} = p_{\underline{k}^{\prime}}(A)$. Since $\underline{l}^{\prime}$ and $\underline{l}$ differ only by one entry, $q_{\underline{l}^{\prime}} = p_{\underline{l}^{\prime}}(A)$. Moreover, since the matrix $A$ as defined above has full rank,
\begin{equation}
p_{\underline l}(A) p_{\underline k}(A) + \sum \pm p_{\underline l ^\prime}(A) p_{\underline k ^\prime}(A) = 0 \,.
\end{equation}
Hence, we find, $p_{\underline l}(A)  p_{\underline k}(A) = q_{\underline l} q_{\underline k}  $ and since $p_{\underline l}(A) = q_{\underline l}=1$,
$p_{\underline k}(A) = q_{\underline k}$. This completes the induction and finishes the proof.
\end{proof}
\end{appendix}

To apply Schubert calculus to the QMP Lem.~\ref{intersection1varieties} and \ref{intersection2varieties} are both important steps. We prove them here.
First, we observe for $V,W \in \mbox{Gr}_{k,d}$ and a density operator $\rho = \sum_{k=1}^d\lambda_k |k\rangle \langle k|$ that
\begin{eqnarray}\label{matrixmetric1}
|\mbox{Tr}[\rho P_V]-\mbox{Tr}[\rho P_W]| &= &|\mbox{Tr}[\rho (P_V-P_W)]| \nonumber\\
&\leq& \sum_{k=1}^d\lambda_k | \langle k| P_V -P_W  |k\rangle | \nonumber \\
& \leq & \mbox{Tr}[\rho] \, \|P_V-P_W\|_{\mbox{Op}} \,.
\end{eqnarray}
Analogously, for given test spectrum $c$ , flags $F_{\bullet}$ and $G_{\bullet}$ defining hermitian operators $A$ and $B$ in the spirit of Definition \ref{flagoperator} and a density operator $\rho = \sum_{k=1}^d\lambda_k |k\rangle \langle k|$ we find
\begin{eqnarray}\label{matrixmetric2}
|\mbox{Tr}[\rho A]-\mbox{Tr}[\rho B]| &= &|\mbox{Tr}[\rho (A-B)]| \nonumber\\
&\leq& \sum_{k=1}^d\lambda_k | \langle k| A-B  |k\rangle | \nonumber \\
& \leq & \mbox{Tr}[\rho] \, \|A-B\|_{\mbox{Op}} \,.
\end{eqnarray}

The crucial aspect for the proof of Lem.~\ref{intersection1varieties} and \ref{intersection2varieties} is to show (recall deviation (\ref{deviationspecineq1}) and (\ref{deviationspecineq2})) that the replacement of a Schubert cell by its closure, $S_{\pi}^{\circ} \rightarrow S_{\pi}$ and $X_{\alpha}^{\circ} \rightarrow X_{\alpha}$  , does not change the expression
\begin{equation}
\min_{V \in S_{\pi}^{\circ}(\rho)}(\mbox{Tr}[\rho P_{V}]) \qquad \mbox{and} \qquad \min \limits_{\begin{array}{c}F_{\bullet}(A) \in X_{\alpha}^{\circ}(\rho)\\ spec(A)=a\end{array}}(\mbox{Tr}[\rho A])\,,
\end{equation}
respectively.

Therefore, it is important to show that a point on the boundary of a Schubert cell has zero distance $\|P_V-P_W\|_{\mbox{Op}}$ and $\|A-B\|_{\mbox{Op}}$, respectively to the Schubert cell. This is obviously expected since the topology for the Grassmannian and the flag variety are induced by the topology of regular matrices, but it has to be proven. Let $X$ be a topological space and $Y= X/\sim$ a quotient space, whose topology is induced by the quotient map
\begin{equation}
\pi: X \rightarrow Y\,,
\end{equation}
namely as the finest topology on $Y$ such that $\pi$ is still continuous. Let us now consider $M\subset Y$ open and denote its closure by $N = \overline{M}$.
Since $\pi$ is continuous, the preimage $\pi^{-1}(M)$ is open in $X$ and $\pi^{-1}(N)$ and $\pi^{-1}(\partial M)$ are both closed. We find
\begin{eqnarray}
\pi^{-1}(\partial M) &=& \pi^{-1}(N \setminus M) \nonumber \\
&=& \pi^{-1}(N) \setminus \pi^{-1}(M) \nonumber \\
&=& \pi^{-1}(\,\cap\,\{V\subset Y\,|\, M \subset V\, \mbox{and}\, V \, \mbox{closed}\}) \setminus \pi^{-1}(M) \nonumber \\
&=& \left(\,\cap\,\{\pi^{-1}(V) \,|\, M \subset V \subset Y\, \mbox{and}\, V \, \mbox{closed}\})\right) \setminus \pi^{-1}(M) \nonumber \\
&\supset& (\,\cap\,\{W\subset X\,|\, \pi^{-1}(M) \subset W\, \mbox{and}\, W \, \mbox{closed}\}) \setminus \pi^{-1}(M) \nonumber \\
&=& \overline{\pi^{-1}(M)} \setminus \pi^{-1}(M) \nonumber \\
&=& \partial \pi^{-1}(M)\,.
\end{eqnarray}
In the third and sixth step we used the definition of the closure of a set and the remaining steps are all elementary.
Let's choose $W \in \partial S_{\pi}^{\circ}(\rho)$. Recall (\ref{matrixgras}) which defines a quotient map $\pi$ from the set $\mathcal{M}_{k,d}$ of $d\times k$-matrices with $k$ linearly independent column vectors to the Grassmannian,
\begin{equation}
\pi: \mathcal{M}_{k,d} \rightarrow \mbox{Gr}_{k,d}\,.
\end{equation}
Moreover we choose $R \in \overline{\pi^{-1}(S_{\pi}^{\circ}(\rho))}$ such that $R$ represents $W$. Then by using the estimate (\ref{matrixmetric1}) we find
\begin{eqnarray}
\lefteqn{0 < \min \limits_{V \in S_{\sigma}^{\circ}(\rho)}(\mbox{Tr}[P_V \rho])-\mbox{Tr}[P_W \rho])}&& \nonumber \\
&=& \min \limits_{V \in S_{\sigma}^{\circ}(\rho)}(|\mbox{Tr}[\rho (P_V -P_W)]|) \nonumber \\
&\leq& \inf \limits_{V \in S_{\sigma}^{\circ}(\rho)} \mbox{Tr}[\rho] \|P_V -P_W\|_{\mbox{Op}} \nonumber \\
&=& \mbox{Tr}[\rho] \inf_{T \in \pi^{-1}(S_{\sigma}^{\circ}(\rho))}(\|P_{\pi(T)}-P_{\pi(R)}\|_{\mbox{Op}}) \nonumber \\
&=& 0
\end{eqnarray}
since $R \in \overline{\pi^{-1}(S_{\pi}^{\circ}(\rho))} $ and this means zero distance to $\pi^{-1}(S_{\sigma}^{\circ}(\rho))$ w.r.t.~the metric $\|P_{\pi(T)}-P_{\pi(R)}\|_{\mbox{Op}}$.
The proof of Lem.~\ref{intersection2varieties} works in the same way by using (\ref{matrixmetric2}) instead of (\ref{matrixmetric1}) and we skip it here.

\chapter{Stability of the Selection Rule}\label{app:BDSelRulestable}
In this chapter we will prove Thm.~\ref{thm:BDSelRulestable}.
In Example \ref{ex:PintoStructBD} we concluded that any $|\Psi_3\rangle \in \wedge^3[\mathcal{H}_1^{(6)}]$ has the form
\begin{eqnarray}\label{ansatz}
|\Psi_3\rangle &=& \alpha |1,2,3\rangle+ \beta |1,2,4\rangle+ \gamma |1,3,5\rangle+ \delta |2,3,6\rangle \nonumber \\
&&+\nu |1,4,5\rangle+\mu |2,4,6\rangle + \xi |3,5,6\rangle+\zeta |4,5,6\rangle \,,
\end{eqnarray}
with natural orbitals $\{|k\rangle\}_{k=1}^6$.
Since the corresponding $1-$RDO is diagonal w.r.t. $\{|k\rangle\}_{k=1}^6$,
\begin{equation}\label{diagonal}
\langle k |\rho_1|l\rangle = \delta_{k l} \,\lambda_k\,,
\end{equation}
we find
\begin{eqnarray}
\lambda_4 &=& |\beta|^2+|\nu|^2+|\mu|^2+|\zeta|^2 \\
\lambda_5 &=& |\gamma|^2+|\nu|^2+|\xi|^2+|\zeta|^2 \\
\lambda_6 &=& |\delta|^2+|\mu|^2+|\xi|^2+|\zeta|^2\,.
\end{eqnarray}
The goal is now to show that the coefficients $\beta, \gamma, \delta,\xi$ and $\zeta$ are small, i.e.
\begin{equation}
\|P \Psi\|_{L^2}^2 = |\alpha|^2+ |\mu|^2+|\nu|^2 = 1-\left( |\beta|^2+ |\gamma|^2+|\delta|^2+|\xi|^2+ |\zeta|^2\right)
\end{equation}
is close to $1$, whenever constraint $D^{(3,6)}(\cdot) \geq 0$ (\ref{set36}) is approximately saturated. For the given state $|\Psi_3\rangle$ (\ref{ansatz})  the saturation $D^{(3,6)}(\vec{\lambda})$ reads
\begin{equation}\label{distancegreek}
D^{(3,6)}(\vec{\lambda}) = -|\beta|^2+|\gamma|^2+|\delta|^2+ 2|\xi|^2+|\zeta|^2\,.
\end{equation}
First we observe
\begin{eqnarray}
\|P \Psi_3\|_{L^2}^2 &\leq& 1- \frac{1}{2}\left( |\beta|^2+ |\gamma|^2+|\delta|^2+ 2|\xi|^2+ |\zeta|^2\right)\nonumber\\
&\leq& 1- \frac{1}{2}\left( - |\beta|^2+ |\gamma|^2+|\delta|^2+ 2|\xi|^2+ |\zeta|^2\right)\nonumber\\
&=& 1-\frac{1}{2}\,D^{(3,6)}(\vec{\lambda})\,,
\end{eqnarray}
which is the upper bound for $\|P \Psi_3\|_{L^2}^2$ in Thm.~\ref{thm:BDSelRulestable}.

To derive the lower bound note that (\ref{diagonal}) in particular implies
\begin{eqnarray}
0 = \langle 4|\rho_1|3\rangle = \overline{\alpha} \beta +  \overline{\gamma} \nu +  \overline{\delta} \mu +  \overline{\xi} \zeta\,,
\end{eqnarray}
which leads by the triangle inequality, the identity $(A+B+C)^3 \leq 3 \,(A^2+B^2 + C^3)$ and $|\mu|^2, |\nu|^2,|\xi|^2,|\zeta|^2 \leq 1-|\alpha|^2$ to
\begin{eqnarray}\label{betaestimate}
|\beta|^2 &\leq& \frac{1}{|\alpha|^2}\,\left(|\gamma| \,|\nu| + |\delta| \,|\mu| +|\xi| \,|\zeta|   \right)^2                  \nonumber \\
&\leq& \frac{3}{|\alpha|^2}\,\left(|\gamma|^2 \,|\nu|^2 + |\delta|^2 \,|\mu|^2 +|\xi|^2 \,|\zeta|^2   \right)\nonumber \\
 &\leq& \frac{3(1-|\alpha|^2)}{|\alpha|^2}\,\left(|\gamma|^2  + |\delta|^2  + \frac{1}{3} (2|\xi|^2 +|\zeta|^2) \right)\,.
\end{eqnarray}
Now, for all $s,r \geq0$ we find by using (\ref{betaestimate})
\begin{eqnarray}\label{lowerboundest}
\lefteqn{|\beta|^2+ |\gamma|^2+|\delta|^2+|\xi|^2+ |\zeta|^2} &&\nonumber \\
&\leq& (1-r)|\beta|^2+ |\gamma|^2+|\delta|^2+ (1+s)(2 |\xi|^2+ |\zeta|^2) + r |\beta|^2 \nonumber \\
&\leq& (1-r)|\beta|^2+ |\gamma|^2+|\delta|^2+(1+s)(2 |\xi|^2+ |\zeta|^2) \nonumber \\
&& + \frac{3 r (1-|\alpha|^2)}{|\alpha|^2}\,\left(|\gamma|^2  + |\delta|^2  + \frac{1}{3} (2|\xi|^2 +|\zeta|^2) \right) \nonumber \\
&=& (1-r)|\beta|^2 + \left(1+ \frac{3r(1-|\alpha|^2)}{|\alpha|^2}\right)\,\left(|\gamma|^2 + |\delta|^2\right) \nonumber \\
&& +  \left( 1+ s+\frac{r(1-|\alpha|^2)}{|\alpha|^2}\right) \left(2 |\xi|^2+ |\zeta|^2\right)\,.
\end{eqnarray}
We can choose the parameters $s,r$ such that the last expression in (\ref{lowerboundest}) coincides with $D^{(3,6)}(\vec{\lambda})$ up to a global factor $\chi$.
For this we solve (recall (\ref{distancegreek}))
\begin{equation}
-(1-r) = \left(1+\frac{3r(1-|\alpha|^2)}{|\alpha|^2}\right) = 1 +s+ \frac{r(1-|\alpha|^2)}{|\alpha|^2} \,.
\end{equation}
The solution reads
\begin{eqnarray}
r&=& \frac{2 |\alpha|^2}{4 |\alpha|^2-3}\\
s&=& \frac{4 (1-|\alpha|^2)}{4|\alpha|^2-3}\,.
\end{eqnarray}
Both parameters are non-negative as long as $|\alpha|^2\geq \frac{3}{4}$.
Finally, this leads to
\begin{eqnarray}
\|P \Psi_3\|_{L^2}^2 &=& 1-(|\beta|^2+ |\gamma|^2+|\delta|^2+|\xi|^2+ |\zeta|^2) \nonumber \\
&\geq& 1-(r-1)D^{(3,6)}(\vec{\lambda}) \nonumber \\
&\geq& 1-\chi_{1-|\alpha|^2} D^{(3,6)}(\vec{\lambda})\,,
\end{eqnarray}
with
\begin{equation}
\chi_{1-|\alpha|^2}\equiv r-1 = \frac{3-2 |\alpha|^2}{4 |\alpha|^2-3} =  \frac{1 + 2(1- |\alpha|^2)}{1-4(1-|\alpha|^2)}\,.
\end{equation}
Lem.~\ref{lem:HFNONtoSlater} states $|\alpha|^2 \geq 1-\delta$ and since $\chi$ is monotonously increasing, $\chi_{1-|\alpha|^2}\leq \chi_{\delta}$,
which finishes the proof.

\chapter{Technical Results for $N$-Harmonium}\label{app:NHarm}
\section{Bosonic results}\label{app:bosonic}
In the following we calculate the $1$-RDO for the bosonic ground state $\Psi_0^{(b)}$ (recall (\ref{gsbosons})):
\begin{eqnarray}
\rho_1^{(b)}(x,y)&=& \int \!\mathrm{dx_2}\ldots\mathrm{dx_N} \,\Psi_0^{(b)}(x,x_2,\ldots,x_N) \nonumber \\
&& \cdot \Psi_0^{(b)}(y,x_2,\ldots,x_N)^\ast\\
&=& \mathcal{N} ^2 e^{-(A-B_N)(x^2+y^2)} \int \!\mathrm{dx_2}\ldots\mathrm{dx_N} \,e^{-2 A (x_2^2+\ldots +x_N^2)}\nonumber \\
&& e^{2 B_N (x_2+\ldots +x_N)^2}e^{2 B_N (x+ y)(x_2+\ldots +x_N) }\nonumber
\end{eqnarray}
Here we resort to the Hubbard-Stratonovich identity,
\begin{equation}\label{hubbardstrat}
\mbox{e}^{a\xi^2}= \sqrt{\frac{a}{\pi}} \int_{-\infty}^{\infty} \!\mathrm{d}y \,\mbox{e}^{-a y^2 +2 a y \xi}
\end{equation}
for $a \in \mathbb{C}$ such that $\mbox{Re}(a)>0$.
%By analytical continuation (set $\xi \mapsto i \xi$) this identity can be extended to $\mbox{Re}(a)<0$.
With $a=2 B_N$ and $\xi = (x_2+\ldots + x_N)$, this leads to (for the case $B_N<0$ use a modified version of Eq.~(\ref{hubbardstrat}) with $\xi \mapsto i \xi$)
\begin{eqnarray}
\rho_1^{(b)}(x,y)&=&\mathcal{N} ^2 \sqrt{\frac{2 B_N}{\pi}} e^{-(A-B_N)(x^2+y^2)} \int_{-\infty}^{\infty} \!\mathrm{d}z \, \mbox{e}^{-2 B_N z^2} \int_{-\infty}^{\infty}\!\mathrm{dx_2}\ldots\mathrm{dx_N} \,e^{-2 A (x_2^2+\ldots +x_N^2)}\nonumber \\
&&\times \, e^{2 B_N (x+y+2 z)(x_2+\ldots +x_N) }\nonumber\\
&=& \mathcal{N} ^2 \sqrt{\frac{2 B_N}{\pi}} e^{-(A-B_N)(x^2+y^2)} \int_{-\infty}^{\infty} \!\mathrm{d}z \,\mbox{e}^{-2 B_N z^2} \Big(\int\!\mathrm{du} \, e^{-2 A (u-\frac{B_N}{2 A}(x+ y+2 z))^2}\Big)^{N-1}\nonumber \\
&&\times \, e^{(N-1)\frac{B_N^2}{2A} (x+ y+2 z)^2}\nonumber\\
&=&  \mathcal{N} ^2 \sqrt{\frac{2 B_N}{\pi}} \left(\frac{\pi}{2 A}\right)^{\frac{N-1}{2}}e^{-(A-B_N)(x^2+ y^2)} \int_{-\infty}^{\infty} \!\mathrm{d}z \,\mbox{e}^{-2 B_N z^2} e^{(N-1)\frac{B_N^2}{2A} (x+ y+2 z)^2} \nonumber \\
&& \nonumber
\end{eqnarray}
Since
\begin{eqnarray}
\lefteqn{\int_{-\infty}^{\infty} \!\mathrm{d}z \,\mbox{e}^{-2B_N z^2} e^{(N-1)\frac{B_N^2}{2A} (x+ y+2z)^2}}&&\nonumber \\
&=& \sqrt{\pi}\sqrt{\frac{A C_N}{(N-1)B_N^3}} e^{B_N (x+y)^2}
\end{eqnarray}
with
\begin{equation}\label{CN}
C_N = \frac{(N-1) \frac{B_N^2}{2}} {A-(N-1)B_N}\,,
\end{equation}
we find
\begin{eqnarray}\label{reddensity}
\rho_1^{(b)}(x,y)&=& \tilde{\mathcal{N}}\,e^{-(A-B_N-C_N)(x^2+y^2)+2 C_N x y}\,,
\end{eqnarray}
where $\tilde{\mathcal{N}}$ follows from the normalization of $\rho_1^{(b)}(x,y)$.
Moreover we observe with Eqs.~(\ref{CN}), (\ref{parameterab}) that
\begin{equation}
A-B_N-C_N =a_N\,\,,\,C_N = \frac{1}{2} b_N\,.
\end{equation}
Therefore, the exponent in Eq.~(\ref{reddensity}) is identical to the one in Eq.~(\ref{1RDOb}).

In Sec.~\ref{sec:1RDObosons} we have diagonalized $\rho_1^{(b)}$ by equating it with the Gibbs state of an effective harmonic oscillator. This is equivalent to apply Mehler's formula to the expression in (\ref{reddensity}). This means to use \cite{Rob}
\begin{eqnarray}\label{Mehler}
\lefteqn{e^{-\frac{1}{4}(c^2+d^2)(z^2+\tilde{z}^2)-\frac{1}{2}(c^2-d^2) z \tilde{z}}}\nonumber \\
 &=& \sqrt{\pi}\, l (1-q^2)^{\frac{1}{2}} \sum_{k=0}^{\infty} q^k \varphi_k^{(l)}(z)\varphi_k^{(l)}(\tilde{z}) \,,
\end{eqnarray}
with $l=(c d)^{-\frac{1}{2}}$ and $q= \frac{d-c}{d+c}$. From (\ref{Mehler}) and (\ref{reddensity}) we obtain
\begin{eqnarray}
c &=&\sqrt{2(A-B_N -2 C_N)} = \sqrt{\frac{N}{\left((N-1) l_+^{\,2}+l_-^{\,2}\right)}} \nonumber \\
d &=&\sqrt{2(A-B_N)} = \sqrt{\frac{(N-1) l_-^{\,2}+ l_+^{\,2}}{ N l_-^{\,2} l_+^{\,2}}}  \nonumber \\
l &=& \sqrt{l_-  l_+} \left(\frac{(N-1)l_+^{\,2} + l_-^{\,2}}{(N-1)l_-^{\,2} + l_+^{\,2}} \right)^{\frac{1}{4}}\,.
\end{eqnarray}
Comparing with the form in Eq.~(\ref{1RDOb}) yields immediately the concrete expressions for the parameters $b_N$, $a_N$ and $L_N$ in Eq.~(\ref{parameterab}).
After all, the NONs $\lambda_k^{(b)}$ (their sum is normalized to the particle number $N$) are given by
\begin{equation}\label{NON}
\lambda_k^{(b)} = N(1-q)\,q^k \,.
\end{equation}

\section{Fermionic results}\label{app:fermionic}
\subsection{$\rho_1^{(f)}(x,y)$}\label{app:fermionicspatial}
In this section we calculate the $1$-RDO $\rho_1^{(f)}(x,y)$ of the fermionic ground state $\Psi_0^{(f)}$ in spatial representation.
Below it will prove convenient to first rearrange the Vandermonde determinant
\begin{eqnarray}
V(\vec{x})&=& \prod_{1\leq i<j\leq N} (x_i-x_j) \nonumber \\
&=& \prod_{1\leq i<j\leq N} [(x_i- s)-(x_j-s)] \nonumber \\
&=& l^{\binom{N}{2}}\,\prod_{1\leq i<j\leq N} (z_i-z_j)\qquad, z_i \equiv \frac{x_i-s}{l} \nonumber \\
&=& l^{\binom{N}{2}}\,\left|\begin{array}{lll}\,1&\ldots&\,1\\z_1&\ldots&z_N\\ \,\vdots& &\,\vdots\\ z_1^{N-1}&\ldots& z_N^{N-1} \end{array}\right|\\
&=& \left(\frac{l}{2}\right)^{\binom{N}{2}}\,\left|\begin{array}{lll}H_0(z_1)&\ldots&H_0(z_N)\\H_1(z_1)&\ldots&H_1(z_N)\\ \,\vdots& &\,\vdots\\ H_{N-1}(z_1)&\ldots&H_{N-1}(z_N) \end{array}\right|\,
\end{eqnarray}
for all $s, l \in \mathbb{C}$, where $H_k(z)$ is the $k$-th Hermite polynomial and in the last step we used the invariance of determinants under changes of a column by just linear combinations of the other ones.
Moreover, by using the orthonormalized Hermite functions $\varphi_k^{(l)}(z)$,
\begin{equation}
\varphi_k^{(l)}(z) = \frac{1}{\sqrt{2^k k!}}\,\pi^{-\frac{1}{4}}\,l^{-\frac{1}{2}}\,H_k\left(\frac{z}{l}\right) \,e^{-\frac{z^2}{2 l^2}}
\end{equation}
we find
\begin{equation}\label{VandermondeHermite}
V(\vec{x}) = const\times\left|\begin{array}{lll}\varphi_0^{(1)}(z_1)&\ldots&\varphi_0^{(1)}(z_N)\\ \varphi_1^{(1)}(z_1)&\ldots&\varphi_1^{(1)}(z_N)\\ \,\vdots& &\,\vdots\\ \varphi_{N-1}^{(1)}(z_1)&\ldots&\varphi_{N-1}^{(1)}(z_N) \end{array}\right|\,\prod_{j=1}^N \,e^{\frac{z_j^2}{2}}\,,
\end{equation}
where $z_j = z_j(x_j)$.
Note that the determinant on the rhs is nothing else but a Slater determinant.
In the following, to obtain the $1$-RDO in spatial representation we integrate out $N-1$ particle coordinates. The essential simplification used is to decouple the coordinates $x_2,\ldots,x_N$ in the exponent of the exponential function in ground state wave function (cf.~Eq.~(\ref{gsfermions})) by resorting to the Hubbard-Stratonovich identity and than afterwards using the orthogonality of the Hermite functions to make the integration trivial. In order not to confuse the reader we do not care about global constants, collect and represent them just by symbols $\mathcal{N}^{(i)},i=1,,\ldots$ and normalize the final expression for the $1$-RDO at the end.
We find
\begin{eqnarray}
\rho_1^{(f)}(x,y)&=& \int \!\mathrm{d}x_2\ldots \mathrm{d}x_N \, \Psi_N(x,x_2,\ldots,x_N) \Psi_N(y,x_2,\ldots,x_N)^\ast \nonumber \\
&=& \mathcal{N}^{(1)} \, e^{-(A-B_N)(x^2+y^2)} \, \int \!\mathrm{d}x_2\ldots \mathrm{d}x_N \, V(x,x_2,\ldots,x_N) V(y,x_2,\ldots,x_N) \nonumber \\
&&\cdot e^{-2A (x_2^2+\ldots+x_N^2)} \, e^{2B_N (x_2+\ldots+x_N)^2}\,e^{2B_N (x+y)(x_2+\ldots+x_N)}\,.
\end{eqnarray}
Now we use the Hubbard-Stratonovich identity (\ref{hubbardstrat}) with
\begin{equation}
a\equiv 2B_N\qquad,\, \xi\equiv x_2+\ldots +x_N
\end{equation}
to decouple the mixed terms in the exponent $(x_2+\ldots +x_N)^2$. This yields
\begin{eqnarray}
\rho_1^{(f)}(x,y)&=& \mathcal{N}^{(2)}\,e^{-(A-B_N)(x^2+y^2)}\, \int \!\mathrm{d}z \, \int \!\mathrm{d}x_2\ldots \mathrm{d}x_N \, V(x,x_2,\ldots,x_N) V(y,x_2,\ldots,x_N)\nonumber \\
&&\cdot\, e^{-2A (x_2^2+\ldots+x_N^2)} \,e^{2B_N (x+y)(x_2+\ldots+x_N)} \, e^{-2B_N z^2}\, e^{4B_N(x_2+\ldots+x_N)z}\nonumber \\
&=& \mathcal{N}^{(2)}\,e^{-(A-B_N)(x^2+y^2)}\, \int \!\mathrm{d}z \, e^{-2B_N z^2} \int \!\mathrm{d}x_2\ldots \mathrm{d}x_N \, V(x,x_2,\ldots,x_N) \nonumber \\
&& \cdot\, V(y,x_2,\ldots,x_N)\,\prod_{j=2}^N \,e^{-2A x_j^2  + \left(2B_N (x+y)+4B_N z\right) x_j} \nonumber \\
&=& \mathcal{N}^{(2)}\,e^{-(A-B_N)(x^2+y^2)}\, \int \!\mathrm{d}z \, e^{-2B_N z^2} \int \!\mathrm{d}x_2\ldots \mathrm{d}x_N \, V(x,x_2,\ldots,x_N) \nonumber \\
&& \cdot\, V(y,x_2,\ldots,x_N)\,\prod_{j=2}^N \,e^{-2A \big(x_j - \frac{B_N}{2A} (x+y+2z)\big)^2} \, e^{\frac{B_N^2}{2A} (x+y+2z)^2} \,.
\end{eqnarray}
Now we fix $s$ introduced above. For $\,j=2,3,\ldots,N$ we use
\begin{equation}\label{Zvariables}
z_j \equiv \frac{x_j-s}{l} = \sqrt{2A} \,\big(x_j - \frac{B_N}{2A} (x+y+2z)\big)
\end{equation}
with
\begin{equation}
l \equiv \frac{1}{\sqrt{2A}} \qquad,\, s \equiv  \frac{B_N}{2A}(x+y+2z) \,.
\end{equation}
Thus, by using (\ref{VandermondeHermite}) and $z_1^{(X)}\equiv \left(\frac{x-s}{l}\right)$, $z_1^{(Y)}\equiv \left(\frac{y-s}{l}\right)$,  we find
\begin{eqnarray}\label{Slatertrick}
\rho_1^{(f)}(x,y)&=& \mathcal{N}^{(3)}\,e^{-(A-B_N)(x^2+y^2)}\, \int \!\mathrm{d}z \, e^{-2B_N z^2} \,e^{\frac{B_N^2}{2A}(N-1)(x+y+2z)^2} \,e^{\frac{\left(z_1^{(X)}\right)^2+\left(z_1^{(Y)}\right)^2}{2}} \nonumber \\
&& \cdot \int \!\mathrm{d}z_2\ldots \mathrm{d}z_N   \left|\begin{array}{llll}\varphi_0^{(1)}\left(z_1^{(X)}\right)&\varphi_0^{(1)}(z_2)&\ldots&\varphi_0^{(1)}(z_N)\\ \varphi_1^{(1)}\left(z_1^{(X)}\right)&\varphi_1^{(1)}(z_2)&\ldots&\varphi_1^{(1)}(z_N)\\ \,\vdots&& &\,\vdots\\ \varphi_{N-1}^{(1)}\left(z_1^{(X)}\right)&\varphi_{N-1}^{(1)}(z_2)&\ldots&\varphi_{N-1}^{(1)}(z_N) \end{array}\right| \nonumber \\ &&\cdot \left|\begin{array}{llll}\varphi_0^{(1)}\left(z_1^{(Y)}\right)&\varphi_0^{(1)}(z_2)&\ldots&\varphi_0^{(1)}(z_N)\\ \varphi_1^{(1)}\left(z_1^{(Y)}\right)&\varphi_1^{(1)}(z_2)&\ldots&\varphi_1^{(1)}(z_N)\\ \,\vdots&& &\,\vdots\\ \varphi_{N-1}^{(1)}\left(z_1^{(Y)}\right)&\varphi_{N-1}^{(1)}(z_2)&\ldots&\varphi_{N-1}^{(1)}(z_N) \end{array}\right| \,.
\end{eqnarray}
The orthogonality of the Hermite functions makes the $z_2,\ldots,z_N$ integrals trivial and we find
\begin{eqnarray}
\rho_1^{(f)}(x,y) &=& \mathcal{N}^{(4)}\,e^{-(A-B_N)(x^2+y^2)}\, \int \!\mathrm{d}z \, e^{-2B_N z^2} \, e^{\frac{B_N^2}{2A}(N-1)(x+y+2z)^2} \nonumber \\
&&\cdot  \sum_{k=0}^{N-1}\,\frac{1}{2^k k!}\, H_k\left(z_1^{(X)}\right)  H_k\left(z_1^{(Y)}\right)\,.
\end{eqnarray}
Finally, we simplify the $z$-integral.
We rearrange
\begin{eqnarray}
\lefteqn{2B_N z^2- \frac{B_N^2}{2A} (N-1)(x+y+2z)^2} \nonumber \\
&=& \left(2 B_N - \frac{2B_N^2}{A}(N-1)\right)\,z^2 - 2 \frac{B_N^2}{A}(N-1)(x+y)\,z- \frac{B_N^2}{2A}(N-1)(x+y)^2 \nonumber \\
&\equiv& r \, z^2- 2t\,z+v\nonumber \\
&=& r\,\left(z-\frac{t}{r}\right)^2-\frac{t^2}{r}+v
\end{eqnarray}
with
\begin{equation}\label{parameterr}
r\equiv 2B_N\,\left(1 - \frac{B_N}{A}(N-1)\right)\,\,,\, t\equiv \frac{B_N^2}{A}(N-1)(x+y)\,\,,\, v\equiv - \frac{B_N^2}{2A}(N-1)(x+y)^2\,.
\end{equation}
From Eq.~(\ref{parameterr}) it follows with Eq.~(\ref{CN})
\begin{eqnarray}
\frac{t^2}{r}-v &=& \frac{B_N^3(N-1)^2}{2A \left(A-B_N(N-1)\right)}\,(x+y)^2+\frac{B_N^2}{2A}\,(x+y)^2\nonumber \\
&=& C_N\,(x+y)^2 \nonumber \\
\frac{t}{r} &=& \frac{B_N(N-1)}{2\left(A-B_N(N-1)\right)}\,(x+y) = \frac{C_N}{B_N}\,(x+y)
\end{eqnarray}
and we obtain
\begin{eqnarray}\label{1RDOfUint}
\rho_1^{(f)}(x,y)&=& \mathcal{N}^{(5)}\,e^{-\left(A-B_N-C_N\right)\,(x^2+y^2)+ 2 C_N\,x y} \\
&&\cdot  \int \!\mathrm{d}u \, e^{- u^2}\,\sum_{k=0}^{N-1}\,\frac{1}{2^k k!}\, H_k(p u+q(x,y))  H_k(p u+q(y,x))\,, \nonumber
\end{eqnarray}
where we defined
\begin{equation}
p\equiv \sqrt{\frac{B_N}{A-B_N(N-1)}}\qquad,\,q(x,y)= \sqrt{2A}\left[x-\frac{B_N}{2\left(A-B_N(N-1)\right)}\,(x+y)\right]\,.
\end{equation}
Note that the exponential factor in the first line Eq.~(\ref{1RDOfUint}) is identical to that in Eq.~(\ref{reddensity}) for $\rho_1^{(b)}(x,y)$.
From the fact that only even order terms in $u$ are relevant for $u$-integration in Eq.~(\ref{1RDOfUint}) and due to the structure of the Hermite polynomials it is clear that the $1$-RDO has the form
\begin{equation}
\rho_1^{(f)}(x,y) = F_N(x,y)\, \exp{\left[-a_N (x^2+y^2) +b_N x y\right]},
\end{equation}
with
\begin{equation}
F_N(x,y) = \sum_{\nu=0}^{N-1} \sum_{\mu=0}^{2 \nu} \,c_{\nu,\mu}\, x^{2\nu-\mu} y^{\mu}\,.
\end{equation}
The coefficients $c_{\nu,\mu}$ depend on the model parameters and fulfill $c_{\nu,\mu} = c_{\nu,2\nu-\mu}$ and $a_N, b_N$ are given by Eq.~(\ref{parameterab}).
\subsection{$\rho_1^{(f)}$ as matrix}\label{app:fermionicmatrix}
In this section we calculate analytically the $1$-RDO $\rho_1^{(f)}$ for $N=3$ represented w.r.t.~to the basis of bosonic NOs, the Hermite functions with natural length scale $L_3$ (recall Eq.~(\ref{parameterbhm})).
To calculate these matrix elements,
\begin{eqnarray}
\left(\rho_1^{(f)}\right)_{nm}  &\equiv&  \langle \varphi_{n}^{(L_N)}, \rho_1^{(f)} \varphi_{m}^{(L_N)} \rangle \nonumber \\
&=& \int \!\mathrm{d}x \mathrm{d}y\, \varphi_{n}^{(L_N)}(x) \rho_1^{(f)}(x,y) \varphi_{m}^{(L_N)}(y)\,,
\end{eqnarray}
it is instructive to rescale the coordinates by writing
\begin{eqnarray}\label{1RDOfpolyresc}
\tilde{x} &:=& \frac{x}{L_3}\qquad,\, \tilde{y} := \frac{y}{L_3} \nonumber \\
F_3(x,y) &=& \tilde{F}_3\big(\tilde{x},\tilde{L_3}\big) \nonumber \\
    &=:& \tilde{C}_1 (\tilde{x}^4 +\tilde{y}^4)+\tilde{C}_2 (\tilde{x}^3 \tilde{y} +\tilde{x} \tilde{y}^3)+\tilde{C}_3 \tilde{x}^2 \tilde{y}^2+\tilde{C}_4 (\tilde{x}^2 +\tilde{y}^2) \nonumber \\
    &&+\tilde{C}_5 \tilde{x} \tilde{y} +\tilde{C}_6 \,,
\end{eqnarray}
with $\tilde{C}_k = {L_3}^4\cdot C_k$ for $k=1,2,3$, $\tilde{C}_k = {L_3}^2\cdot C_k$ for $k=4,5$ and $\tilde{C}_6 = C_6$. Moreover, we
consider every monomial in $\tilde{F}_3$ separately. Consequently, we expand $\left(\rho_1^{(f)}\right)_{nm}$ as
\begin{equation}
\left(\rho_1^{(f)}\right)_{nm} = d_3\cdot \sum_{k=1}^6\, \tilde{C}_k \left(\rho_{1,k}^{(f)}\right)_{nm}
\end{equation}
Powers of the spatial coordinates $\tilde{x}$ and $\tilde{y}$ can then be combined with the Hermite functions $ \varphi_{n}^{(L_N)}(x) = \varphi_{n}^{(1)}(\tilde{x}) $ and  $\varphi_{m}^{(L_N)}(y) = \varphi_{m}^{(1)}(\tilde{y})$. For this, we calculate
\begin{eqnarray}
\tilde{x} H_n(\tilde{x}) &=& \frac{1}{2}H_{n+1}(\tilde{x}) + n H_{n-1}(\tilde{x})  \\
\tilde{x}^2 H_n(\tilde{x}) &=& \frac{1}{4}H_{n+2}(\tilde{x}) + \frac{1}{2} (2n+1) H_{n}(\tilde{x}) + n(n-1) H_{n-2}(\tilde{x}) \nonumber \\
\tilde{x}^3 H_n(\tilde{x}) &=& \frac{1}{8}H_{n+3}(\tilde{x}) + \frac{3}{4} (n+1) H_{n+1}(\tilde{x}) + \frac{3}{2} n^2 H_{n-1}(\tilde{x})\nonumber\\
&&+n(n-1)(n-2) H_{n-3}(\tilde{x}) \nonumber \\
\tilde{x}^4 H_n(\tilde{x}) &=& \frac{1}{16}H_{n+4}(\tilde{x}) + \frac{1}{4} (2n+3) H_{n+2}(\tilde{x}) + \frac{3}{4} (2n^2+2n+1) H_{n}(\tilde{x}) \nonumber \\
&&+n(n-1)(2n-1) H_{n-2}(\tilde{x})+n(n-1)(n-2)(n-3)H_{n-4}(\tilde{x}) \nonumber\,.
\end{eqnarray}
Using these identities for the coordinate $\tilde{x}$ and $\tilde{y}$ as well, the matrix elements $\left(\rho_{1,k}^{(f)}\right)_{nm}$ can symbolically by calculated by referring to the orthonormality of the Hermite functions. Exemplary, we do this for $\left(\rho_{1,5}^{(f)}\right)_{nm}$. We find
\begin{eqnarray}
\left(\rho_{1,5}^{(f)}\right)_{nm} &=& \mathcal{N}\,\int \!\mathrm{d}\tilde{x} \mathrm{d}\tilde{y}\, \tilde{x}\tilde{y} \,\varphi_{n}^{(1)}(\tilde{x}) \varphi_{m}^{(1)}(\tilde{y}) \nonumber \\
&& \cdot \sum_{k=0}^\infty {q_3}^k\,\varphi_k^{(1)}(\tilde{x})  \varphi_k^{(1)}(\tilde{y}) \,,
\end{eqnarray}
where we used the results from Sec.~\ref{sec:1RDObosons} and $\mathcal{N}= \sqrt{\pi} L_3^3 \sqrt{1-{q_3}^2}$ is the normalization constant. Then, we obtain
\begin{eqnarray}
\left(\rho_{1,5}^{(f)}\right)_{nm} &=& \mathcal{N}\,(2^{n+m} n! m!)^{-\frac{1}{2}}\,\int \!\mathrm{d}\tilde{x} \mathrm{d}\tilde{y}\, \sum_{k=0}^\infty {q_3}^k\,\varphi_k^{(1)}(\tilde{x})  \varphi_k^{(1)}(\tilde{y}) \nonumber\\
&& \big[\frac{1}{2}(2^{n+1} (n+1)!)^{\frac{1}{2}}\varphi_{n+1}^{(1)}(\tilde{x}) + n (2^{n-1} (n-1)!)^{\frac{1}{2}}\varphi_{n-1}^{(1)}(\tilde{x})\big] \nonumber \\
&& \cdot\big[n\leftrightarrow m, \tilde{x}\leftrightarrow \tilde{y} \big] \,,
\end{eqnarray}
where $\big[n\leftrightarrow m, \tilde{x}\leftrightarrow \tilde{y} \big]$ means to take the expression from the previous square bracket and swap $n$ with $m$ and also both spatial variables. The orthonormality of $\varphi_{k}^{(l)}$ finally yields
\begin{eqnarray}
\left(\rho_{1,5}^{(f)}\right)_{nm} &=& \mathcal{N}\,(2^{n+m} n! m!)^{-\frac{1}{2}}\cdot\big[ {q_3}^{N+1}  \frac{1}{2}(2^{n+1} (n+1)!)^{\frac{1}{2}} \frac{1}{2}(2^{m+1} (m+1)!)^{\frac{1}{2}}\,\delta_{nm}\nonumber\\
&& + {q_3}^{N-1} n (2^{n-1} (n-1)!)^{\frac{1}{2}}  m (2^{m-1} (m-1)!)^{\frac{1}{2}}\delta_{nm} \nonumber \\
&& +  {q_3}^{N+1} \frac{1}{2}(2^{n+1} (n+1)!)^{\frac{1}{2}}   m (2^{m-1} (m-1)!)^{\frac{1}{2}}\delta_{n m+2} \nonumber \\
&& + {q_3}^{N-1} n (2^{n-1} (n-1)!)^{\frac{1}{2}}   \frac{1}{2}(2^{m+1} (m+1)!)^{\frac{1}{2}}\delta_{n m-2} \big]\,.
% &=& \delta_{nm}\,\mathcal{N}\,(2^{n+m} n! m!)^{-1}\cdot\big[ {q_3}^{N+1}  \frac{1}{4}(2^{n+1} (n+1)!) \nonumber\\
%&& + {q_3}^{N-1} n^2 (2^{n-1} (n-1)!)  \big]\,.
\end{eqnarray}
This expression, together with those for the other indices $k\neq 5$, which are even more complicated, are then be used for a computer program.

\subsection{Natural occupation numbers via perturbation theory}\label{app:fermionicPT}
In this section we apply degenerate Rayleigh-Schr\"odinger perturbation theory to the $1$-RDO represented as matrix w.r.t.~the bosonic NOs and determine the behavior (\ref{spectrum}) of the NONs up to tenth order in the coupling strength $\delta$. By checking lecture notes and standard text books on quantum mechanics we surprisingly realized that most of the algorithms provided there are either wrong or useless. We adapt the very general strategy of J.Fr\"ohlich \cite{JFrQM2}, which we will first present in detail and afterwards apply it to our problem.

Consider a perturbed operator (just for simplicity we call it Hamiltonian) of the form
\begin{equation}
H(\delta) = H^{(0)} + \delta H^{(1)} + \delta^2 H^{(2)} + \ldots \,,
\end{equation}
acting on some separable Hilbert space $\mathcal{H}$. Let $E_{\alpha}^{(0)}$ be an isolated and degenerate eigenvalue of the unperturbed Hamiltonian $H^{(0)}$ and denote the corresponding orthogonal projection operator on that subspace by $P_{\alpha}^{(0)}$. To simplify the notation we introduce the short notation $P \equiv P_{\alpha}^{(0)}$ and $\overline{P} \equiv \mathds{1}-P$. The first step is to block-diagonalize the perturbed Hamiltonian w.r.t.~$P$ and $\overline{P}$. For this we construct a unitary operator $U(\delta)$ expressed as $e^{S(\delta)}$ with an anti-self-adjoint operator $S(\delta)$ and $U(0) = \mathds{1}$,
\begin{equation}\label{ptblockdiag}
 e^{S(\delta)} H (\delta) e^{-S(\delta)} = \tilde{H}(\delta) = P \tilde{H}(\delta) P + \overline{P} \tilde{H}(\delta) \overline{P}\,.
\end{equation}
Moreover we define for any operator $A$ the restrictions to the diagonal blocks and off-diagonal blocks,
\begin{equation}\label{ptOprestrict}
A_d \equiv P A P +\overline{P} A \overline{P}\,\,,\,\,A_{od} \equiv P A \overline{P} +\overline{P} A P\,.
\end{equation}
As an ansatz we assume $S(\delta) = S(\delta)_{od}$ (which will be justified afterwards).
We expand the lhs in Eq.~(\ref{ptblockdiag}) in a Lie-Schwinger series and find
\begin{equation}\label{PTLieSchwinger}
H(\delta) + [S(\delta),H(\delta)] + \frac{1}{2}[S(\delta),[S(\delta),H(\delta)]]+\ldots \,=\, \tilde{H}(\delta)\,.
\end{equation}
By expanding $S(\delta)$ and $\tilde{H}(\delta)$ in a Taylor series,
\begin{eqnarray}
S(\delta) &=& \delta S^{(1)} + \delta^2 S^{(2)} +\ldots \nonumber \\
\tilde{H}(\delta) &=& \tilde{H}^{(0)} + \delta \tilde{H}^{(1)} + \delta^2 \tilde{H}^{(2)} + \ldots\,,
\end{eqnarray}
and plugging in those expansions in Eq.~(\ref{PTLieSchwinger}) we can determine succinctly all orders $S^{(k)}$ and $\tilde{H}^{(k)}$ by comparing different orders in $\delta$. Exemplary, we do this for the first two orders. The off-diagonal part of Eq.~(\ref{PTLieSchwinger}) yields to following two conditions
\begin{eqnarray}\label{PTkorders}
0 &=& H^{(1)}_{od} + [S^{(1)},H^{(0)}]_{od} = H^{(1)}_{od} + [S^{(1)},H^{(0)}] \nonumber \\
0 &=& H^{(2)}_{od} + [S^{(2)},H^{(0)}]_{od}+ [S^{(1)},H^{(1)}]_{od} +\frac{1}{2} [S^{(1)},[S^{(1)},H^{(0)}]]_{od} \nonumber \\
&=& H^{(2)}_{od} + [S^{(2)},H^{(0)}] + [S^{(1)},H^{(1)}_d] +\frac{1}{2} [S^{(1)},[S^{(1)},H^{(0)}_{od}]]  \,,
\end{eqnarray}
where we used $H^{(0)} = H^{(0)}_d$. The ansatz $S^{(k)} = S^{(k)}_{od}$ and the higher order conditions can be obtained in a straightforward way.
From such conditions we can determine succinctly all orders $S^{(k)}$.
To illustrate this we define the map
\begin{equation}\label{adjointinvers}
ad_A^{\,-1}(B) \equiv \iint_{\lambda\neq \lambda'}\,(\lambda-\lambda')^{-1}\,dP_A(\lambda)\,B\,dP_A(\lambda')\,,
\end{equation}
where $dP_A(\cdot)$ is the projection valued measure of the self-adjoint operator $A$. It can easily be seen that (\ref{adjointinvers}) inverts the so-called adjoint, $ad_A(B) \equiv [A,B]$, whenever $P_A(\lambda) B P_A(\lambda) = 0$ for all $\lambda \in \mbox{spec}(A)$.
Since Eq.~(\ref{PTkorders}) contains per construction only off-diagonal parts, this is the case here and $S^{(k)}$ can be obtained by inverting conditions (\ref{PTkorders}) w.r.t.~the highest order $S^{(k)}$.
For the first two orders we find
\begin{eqnarray}\label{PTSk}
S^{(1)} &=& ad_{H^{(0)}}^{\,-1}(H^{(1)}_{od}) \nonumber \\
S^{(2)} &=& ad_{H^{(0)}}^{\,-1}\big(H^{(2)}_{od} + [S^{(1)},H^{(1)}_d] +\frac{1}{2} [S^{(1)},[S^{(1)},H^{(0)}_{od}]]  \big)\,.
\end{eqnarray}
The corresponding orders of the new Hamiltonian $\tilde{H}(\delta)$ follow immediately form the diagonal parts of Eq.~(\ref{PTLieSchwinger})
by plugging in the required orders of $S^{(k)}$. Notice that the ansatz $S(\delta) = S(\delta)_{od}$ was justified since we succeeded in block-diagonalizing $H(\delta)$ and $S(\delta)$ is indeed anti-self-adjoint, which follows from the fact that $ad^{-1}$ maps self-adjoint operators to anti-self-adjoint ones and all the arguments of $ad_{H^{(0)}}^{\,-1}$ in conditions of the type (\ref{PTSk}) are indeed self-adjoint.

In a second step to find the corrections of the degenerate eigenvalue $E_{\alpha}^{(0)} \equiv E_{\alpha}(\delta= 0)$ and the corresponding eigenfunctions we should consider the block Hamiltonian $P\tilde{H}(\delta)P$ and then use standard algorithms to diagonalize it (we do not explain those here).

Now, we apply this general strategy to the $1$-RDO represented as matrix $\rho_1^{(f)}(\delta)$. To simplify the notation we skip the index $1$ and the superscript $(f)$.
The spectrum of $\rho^{(0)} \equiv \rho(\delta=0)$ is just $\{0,1\}$. We denote the projection operator onto the three-dimensional eigenspace of $1$
by $P$ and the one onto the $0$-eigenspace by $\overline{P} = \mathds{1}-P$. Due to this simplified form of the spectrum of the unperturbed ``Hamiltonian'' $\rho^{(0)}$ the inverse of the adjoint (\ref{adjointinvers}) is much simpler and we find
\begin{equation}\label{adjointinversrho}
ad_{\rho^{(0)}}^{\,-1}(B_{od}) = P B_{od} \overline{P} - \overline{P} B_{od} P\,.
\end{equation}
As an additional simplification we have here $\rho_{nm} = 0$, whenever $n+m$ odd. Consequently, $\rho(\delta)$ splits into two density operators, $\rho_o(\delta)$ and $\rho_e(\delta)$, containing just solely the odd and even entries of $\rho(\delta)$, respectively,
\begin{equation}
\rho(\delta)= \rho_o(\delta)\oplus\rho_e(\delta)
\end{equation}
and we can apply the diagonalization procedure to both $\rho_{e/o}(\delta)$ separately.
Finally, both matrices $\rho_{e/o}(\delta)$ restricted to the first $k$ orders in $\delta$ are just finite. E.g.~by considering the first ten orders in $\delta$ we find $\rho_o(\delta) \in \R^{5\times 5}$ and $\rho_e(\delta) \in \R^{4\times 4}$. This reflects our
great choice of the reference basis for the $1$-RDO. To block-diagonalize $\rho_{e/o}(\delta)$ for the first $k$ orders we write a Mathematica program which performs the required steps to block-diagonalize $\rho_{o/e}(\delta)$. The corresponding unitarily transformed parts  $\tilde{\rho}_{o/e}(\delta)$ (recall Eq.~(\ref{ptunitary})) of the density operator can be diagonalized just by brute force due to their small dimension. This then yields the NONs up to corrections of order $\delta^{12}$, where we presented them up to order $\delta^{10}$ in (\ref{spectrum}).

\subsection{Generalized Pauli constraints for larger settings}\label{app:GPClargsettings}
In this section we present the generalized Pauli constraints for the setting $\wedge^3[\mathcal{H}_1^{(8)}]$.
The eight NONs $(\lambda_1,\ldots,\lambda_8)$ are normalized to the particle number $N=3$ and arranged decreasingly,
\begin{eqnarray}
&&\lambda_1+\lambda_2+\ldots +\lambda_8 = 3 \nonumber \\
&& \lambda_1\geq \lambda_2\geq \ldots \geq \lambda_8 \geq 0\,.
\end{eqnarray}
The generalized Pauli constraints found by Klyachko \cite{Kly3} read
\begin{eqnarray}\label{set38}
\,\,\,\,&&D^{(3,8)}_1 := 2-(\lambda_2+\lambda_3+\lambda_4 +\lambda_5) \geq 0 \label{D8.1} \\
&&D^{(3,8)}_2 := 2-(\lambda_1+\lambda_2+\lambda_4 +\lambda_7) \geq 0\label{D8.2} \\
&&D^{(3,8)}_3 := 2-(\lambda_1+\lambda_3+\lambda_4 +\lambda_6) \geq 0 \label{D8.3} \\
&&D^{(3,8)}_4 := 2-(\lambda_1+\lambda_2+\lambda_5 +\lambda_6) \geq 0 \label{D8.4} \\
&&D^{(3,8)}_5 := 1-(\lambda_1+\lambda_2-\lambda_3 ) \geq 0 \label{D8.5} \\
&&D^{(3,8)}_6 := 1-(\lambda_2+\lambda_5-\lambda_7 ) \geq 0 \label{D8.6} \\
&&D^{(3,8)}_7 := 1-(\lambda_1+\lambda_6-\lambda_7 ) \geq 0 \label{D8.7} \\
&&D^{(3,8)}_8 := 1-(\lambda_2+\lambda_4-\lambda_6 ) \geq 0 \label{D8.8} \\
&&D^{(3,8)}_9 := 1-(\lambda_1+\lambda_4-\lambda_5 ) \geq 0 \label{D8.9} \\
&&D^{(3,8)}_{10} := 1-(\lambda_3+\lambda_4-\lambda_7 ) \geq 0 \label{D8.10} \\
&&D^{(3,8)}_{11} := 1-(\lambda_1+\lambda_8) \geq 0 \label{D8.11} \\
&&D^{(3,8)}_{12} := -\lambda_2+\lambda_3+\lambda_6 +\lambda_7 \geq 0 \label{D8.12} \\
&&D^{(3,8)}_{13} := -\lambda_4+\lambda_5+\lambda_6 +\lambda_7 \geq 0\label{D8.13} \\
&&D^{(3,8)}_{14} := -\lambda_1+\lambda_3+\lambda_5 +\lambda_7 \geq 0 \label{D8.14} \\
&&D^{(3,8)}_{15} := 2-(\lambda_2+\lambda_3+2\lambda_4 -\lambda_5-\lambda_7 +\lambda_8) \geq 0 \label{D8.15} \\
&&D^{(3,8)}_{16} := 2-(\lambda_1+\lambda_3+2\lambda_4 -\lambda_5-\lambda_6 +\lambda_8) \geq 0\label{D8.16} \\
&&D^{(3,8)}_{17} := 2-(\lambda_1+2\lambda_2-\lambda_3 +\lambda_4-\lambda_5 +\lambda_8) \geq 0 \label{D8.17} \\
&&D^{(3,8)}_{18} := 2-(\lambda_1+2\lambda_2 - \lambda_3 +\lambda_5-\lambda_6 +\lambda_8) \geq 0 \label{D8.18} \\
&&D^{(3,8)}_{19} := -\lambda_1 - \lambda_2 + 2 \lambda_3 + \lambda_4 + \lambda_5  \geq 0 \label{D8.19} \\
&&D^{(3,8)}_{20} :=  -\lambda_1 +\lambda_2 + \lambda_3 -\lambda_6 +2\lambda_7 \geq 0 \label{D8.20}
\end{eqnarray}
\begin{eqnarray}
&&D^{(3,8)}_{21} := -\lambda_1+ \lambda_3+\lambda_4 +\lambda_5-\lambda_8 \geq 0 \label{D8.21} \\
&&D^{(3,8)}_{22} := -\lambda_1+ \lambda_2 + \lambda_3 +\lambda_7-\lambda_8 \geq 0 \label{D8.22} \\
&&D^{(3,8)}_{23} := 1-(2\lambda_1-\lambda_2+\lambda_4 -2\lambda_5-\lambda_6 +\lambda_8) \geq 0 \label{D8.23} \\
&&D^{(3,8)}_{24} := 1-(\lambda_3+2\lambda_4-2\lambda_5 -\lambda_6-\lambda_7 +\lambda_8) \geq 0\label{D8.24} \\
&&D^{(3,8)}_{25} := 1-(2\lambda_1-\lambda_2-\lambda_4 +\lambda_6-2\lambda_7 +\lambda_8) \geq 0 \label{D8.25} \\
&&D^{(3,8)}_{26} := 1-(2 \lambda_1+ \lambda_2 - 2 \lambda_3 - \lambda_4-\lambda_6 +\lambda_8) \geq 0 \label{D8.26} \\
&&D^{(3,8)}_{27} := 1-(\lambda_1 + 2\lambda_2 - 2 \lambda_3 - \lambda_5-\lambda_6 +\lambda_8) \geq 0 \label{D8.27} \\
&&D^{(3,8)}_{28} := -2\lambda_1 +2\lambda_2+\lambda_3 +\lambda_4-\lambda_6+3\lambda_7 -\lambda_8 \geq 0 \label{D8.28} \\
&&D^{(3,8)}_{29} :=  \lambda_1 -\lambda_3-2\lambda_4 +3\lambda_5+2\lambda_6+\lambda_7 -\lambda_8\geq 0\label{D8.29} \\
&&D^{(3,8)}_{30} :=  -2\lambda_1 -\lambda_2+3\lambda_3 +2\lambda_4+\lambda_5+\lambda_6 -\lambda_8 \geq 0 \label{D8.30} \\
&&D^{(3,8)}_{31} :=  -\lambda_1 -2\lambda_2+3\lambda_3 +\lambda_4+2\lambda_5+\lambda_6 -\lambda_8 \geq 0 \label{D8.31}
\end{eqnarray}

\subsection{Eigenvalue equation for the fermionic matrix $(\langle \varphi_m|\rho_1^{(f)}|\varphi_n\rangle)$}\label{app:NOsfermionic}
With $\hat{x}$ the position operator and recalling the representation $\rho_1^{(f)}(x,y) = \langle x| e^{-\beta_N H_{eff}}|y\rangle$ we get from Eq.~(\ref{1RDOf})
\begin{equation}\label{1RDOfposition}
\rho_1^{(f)} = \sum_{\nu=0}^{N-1} \sum_{\mu=0}^{2\nu} \,c_{\nu,\mu}\,\hat{x}^{2\nu-\mu}\, e^{-\beta_N H_{eff}} \,\hat{x}^{\mu},
\end{equation}
which is hermitian due to $c_{\nu,\mu} = c_{\nu,2\nu-\mu}$. Since $H_{eff}$ describes a harmonic oscillator with characteristic length scale $L_N$ (see Sec.~\ref{sec:1RDObosons}), $\hat{x}$ and $H_{eff}$ can elegantly be expressed by the corresponding creation and annihilation operators
\begin{eqnarray}
\hat{x}&=& \sqrt{\frac{L_N}{2}}\,(a+a^\dagger)\nonumber \\
H_{eff} &=& \hbar \Omega_N (a^\dagger a + \frac{1}{2})\,.
\end{eqnarray}
Then, $\rho_1^{(f)}$ takes the form
\begin{eqnarray}\label{1RDOfladderop}
\rho_1^{(f)} &=& \sum_{\nu=0}^{N-1}\left(\frac{L_N}{2}\right)^{\nu} \sum_{\mu=0}^{2\nu} \,c_{\nu,\mu}\,(a+a^\dagger)^{2\nu-\mu}\nonumber \\
&&\cdot e^{-\beta_N \hbar \Omega_N (a^\dagger a +\frac{1}{2})} \,(a+a^\dagger)^{\mu}\,.
\end{eqnarray}
To determine the NOs $|\chi^{(f)}\rangle$ of $\rho_1^{(f)}$ we expand them w.r.t.~the bosonic NOs, the Hermite states $|m\rangle$ with natural length scale $L_N$ ($\varphi_m^{(L_N)}(x)\equiv \langle x |m\rangle$):
\begin{equation}
|\chi^{(f)}\rangle = \sum_{m=0}^{\infty}\,\zeta_m\,|m\rangle\,.
\end{equation}
Since $a^\dagger a |m\rangle = m |m\rangle$, we find for $\mu$ fixed and $m$ sufficiently large
\begin{equation}
(a+a^\dagger)^{\mu}|m\rangle = m^{\frac{\mu}{2}}\,\left(1+O\left(\frac{1}{m}\right)\right) \sum_{\kappa=0}^{\mu}\binom{\mu}{\kappa}|m-\mu-\kappa\rangle\,.
\end{equation}
Using this asymptotic result we get for $N$ fixed and $m\rightarrow \infty$
\begin{eqnarray}\label{coefrelNOf}
\rho_1^{(f)}|m\rangle &\rightarrow& m^{N-1} e^{-\beta_N\hbar \Omega_N (m+\frac{1}{2})}\sum_{\nu=0}^{N-1}\left(\frac{L_N}{2}\right)^{\nu}\sum_{\mu=0}^{2 \nu} c_{\nu,\mu} \sum_{\kappa=0}^{\mu} \binom{\mu}{\kappa} \nonumber \\
&&e^{\beta_N \hbar \Omega_N (\mu-2\kappa)}\sum_{\tau=0}^{2\nu-\mu} \binom{2\nu-\mu}{\tau} |m-2\underbrace{(\nu-\kappa-\tau)}_{:=r}\rangle \nonumber \\
&=& m^{N-1} e^{-\beta_N\hbar \Omega_N (m+\frac{1}{2})} \sum_{r=-(N-1)}^{N-1} h_{m,m-2r} |m-2r\rangle\,,
\end{eqnarray}
where the real coefficients $h_{m,m-2r}$ depend on $L_N$ and $\beta_N \hbar \Omega_N$, but not \emph{explicitly} on $m$.

\subsection{Cubic eigenvalue problem}\label{app:cubicEigenProb}
In this section we solve the cubic eigenvalue problems (\ref{eigenvalueQ}) and (\ref{rho1down}) from Sec.\ref{sec:hubbardAnalyt}.
The corresponding characteristic polynomial $P_u(E)$ of problem (\ref{eigenvalueQ}) reads
\begin{equation}
P_u(E)= 6 u+ E \left(u^2-9\right)-2 E^2 u+E^3\,.
\end{equation}
This means nothing else but solving a cubic equation. For the canonic form $x^3+ax^2+bx+c=0$ we first consider the quantities $Q$ and $R$, defined as
\begin{equation}\label{cubicParQ}
Q \equiv \frac{a^2-3b}{9}\qquad,\,\,  R\equiv\frac{2 a^3-9ab+27c}{54}\,.
\end{equation}
Here, we have
\begin{equation}
a=-2u\,\,\,,\qquad\,b=u^2-9\,\,\,,\qquad c= 6 u\,,
\end{equation}
which leads to
\begin{equation}
Q = \frac{u^2}{9}+3\qquad,\,\,  R = \frac{u^3}{27}\,.
\end{equation}
By defining
\begin{equation}\label{cubicParTh}
\Theta \equiv \arccos{\left(R/\sqrt{Q^3}\right)}
\end{equation}
the three real roots (energy eigenvalues) of $P_u(E) = 0$ are
\begin{eqnarray}
E_1&=& -2 \sqrt{Q}\,\cos{\left(\frac{\Theta}{3}\right)} -\frac{a}{3}  \nonumber \\
&=& \frac{2}{3} \left[u-\sqrt{u^2+27} \cos \left(\frac{1}{3} \arccos\left(u^3/\sqrt{\left(u^2+27\right)^3}\right)\right)\right] \\
E_2&=& -2 \sqrt{Q}\,\cos{\left(\frac{\Theta+2\pi}{3}\right)} -\frac{a}{3} \nonumber \\
   &=& \frac{2}{3} \left[u-\sqrt{u^2+27} \cos \left(\frac{1}{3} \arccos\left(u^3/\sqrt{\left(u^2+27\right)^3}\right)-\frac{2\pi}{3}\right)\right] \nonumber \\
E_3&=& -2 \sqrt{Q}\,\cos{\left(\frac{\Theta-2\pi}{3}\right)} -\frac{a}{3} \nonumber \\
   &=& \frac{2}{3} \left[u-\sqrt{u^2+27} \cos \left(\frac{1}{3} \arccos\left(u^3/\sqrt{\left(u^2+27\right)^3}\right)+\frac{2\pi}{3}\right)\right]\nonumber\,.
\end{eqnarray}
Here, we have
\begin{equation}\label{energysymetries}
E_1(-u)= -E_2(u)\qquad \mbox{and}\qquad E_3(-u)= -E_3(u)\,.
\end{equation}
In the regime of weak interaction, i.e.~small Hubbard-$u$ we find the behavior
\begin{eqnarray}
E_1&=& -3+\frac{2 u}{3}-\frac{u^2}{18}+O\left(u^3\right)\\
E_2&=& 3+\frac{2 u}{3}+\frac{u^2}{18}+O\left(u^3\right)\\
E_3&=& \frac{2 u}{3}+\frac{2 u^3}{243}+O\left(u^4\right)
\end{eqnarray}
and the leading behavior in the regime of strong interaction
\begin{eqnarray}
E_1&\sim& -\frac{6}{u}+\frac{18}{u^3}\qquad \,\,\,\,\,\,,u\rightarrow \infty\\
E_2&\sim& u + \sqrt {3} + \frac {3} {u} \qquad ,u\rightarrow \infty\\
E_3&\sim&  u - \sqrt {3} + \frac {3} {u}  \qquad ,u\rightarrow \infty \,,
\end{eqnarray}

%The ground state vector (for $(S,M,K) = (1/2,1/2,1)$) is given by the $3$-vector
%\tiny\begin{equation}
%\left(
%\begin{array}{c}
% \frac{1}{3} \sqrt{2} \left(2 \sqrt{u^2+27} \cos \left(\frac{1}{3} \cos ^{-1}\left(\frac{u^3}{\sqrt{\left(u^2+27\right)^3}}\right)\right)+u+9\right) \\
% -\frac{i}{9 \sqrt{3}} \left(-2 (2 u-9) \sqrt{u^2+27} \cos \left(\frac{1}{3} \cos ^{-1}\left(\frac{u^3}{\sqrt{\left(u^2+27\right)^3}}\right)\right)+4 \left(u^2+27\right) \cos
%   \left(\frac{2}{3} \cos ^{-1}\left(\frac{u^3}{\sqrt{\left(u^2+27\right)^3}}\right)\right)-18 u+27\right) \\
% \frac{1}{3} \left(2 \sqrt{u^2+27} \cos \left(\frac{1}{3} \cos ^{-1}\left(\frac{u^3}{\sqrt{\left(u^2+27\right)^3}}\right)\right)-2 u+9\right)
%\end{array}
%\right)\,.
%\end{equation}
%\normalsize
%
The corresponding unnormalized eigenvectors $\vec{v}_j\equiv (\tilde{\alpha}_j,\tilde{\beta}_j,\tilde{\gamma}_j), j=1,2,3$ follow directly from (\ref{eigenvalueQ}). For their three unnormalized coefficients we find
\begin{eqnarray}
\tilde{\alpha}_j &=& u+3-E_j \nonumber \\
\tilde{\beta}_j &=& u -3-E_j \nonumber \\
\tilde{\gamma}_j &=&u-4 E_j+\frac{3}{u} \left(E_j^2-9\right)\,.
\end{eqnarray}

For the ground state (index $1$, which we will skip in the following) we study the asymptotic behavior of the normalized coefficients $\alpha, \beta$ and $\gamma$ in the regimes $u\approx 0$ and $u\rightarrow \pm\infty$. We find
\\
\\
$
\begin{array}{lllll}
\qquad&\alpha_1(u)&\sim& \frac{1}{\sqrt{3}}+\sqrt{3}\,\frac{1}{u} +O\left(\left(\frac{1}{u}\right)^2\right) &\qquad,\,u\rightarrow +\infty \nonumber \\
\qquad&\beta_1(u)&\sim& \frac{1}{\sqrt{3}}-\sqrt{3}\,\frac{1}{u} +O\left(\left(\frac{1}{u}\right)^2\right)&\qquad,\,  u\rightarrow +\infty \nonumber \\
\qquad&\gamma_1(u)&\sim& \frac{1}{\sqrt{3}}-3 \sqrt{3}\,\frac{1}{u^2} +O\left(\left(\frac{1}{u}\right)^3\right)&\qquad,\, u\rightarrow +\infty \nonumber \\
\qquad&\alpha_1(u)&\sim&\frac{1}{6} \left(3+\sqrt{3}\right)-\frac{1}{4} \left(3+\sqrt{3}\right)\,\frac{1}{u} +O\left(\left(\frac{1}{u}\right)^2\right)&\qquad,\, u\rightarrow -\infty \nonumber \\
\qquad&\beta_1(u)&\sim& \frac{1}{6} \left(\sqrt{3}-3\right)+\frac{1}{4} \left(\sqrt{3}-3\right)\,\frac{1}{u}+O\left(\left(\frac{1}{u}\right)^2\right)&\qquad,\, u\rightarrow -\infty \nonumber \\
\qquad&\gamma_1(u)&\sim& -\frac{1}{\sqrt{3}}-\frac{3}{2}\,\frac{1}{u} +O\left(\left(\frac{1}{u}\right)^2\right)&\qquad,\, u\rightarrow -\infty \nonumber \\
\qquad&\alpha_1(u)&\sim& 1-\frac{5 u^2}{648} +O\left(u^3\right)& \qquad,\, u\rightarrow 0 \nonumber \\
\qquad&\beta_1(u)&\sim& \frac{u}{18}+\frac{u^2}{162}+O\left(u^3\right)&\qquad,\, u\rightarrow 0 \nonumber \\
\qquad&\gamma_1(u)&\sim& \frac{u}{9} +\frac{u^2}{162}+O\left(u^3\right)&\qquad,\, u\rightarrow 0\qquad.
\end{array}
$
%Moreover, for the squares of the coefficient we find
%\\
%\\
%$
%\begin{array}{lllll}
%\qquad&|\alpha_1(u)|^2&\sim& \frac{1}{3}+2\,\frac{1}{u} +O\left(\left(\frac{1}{u}\right)^2\right) &\qquad,\,u\rightarrow +\infty \nonumber \\
%\qquad&|\beta_1(u)|^2&\sim& \frac{1}{3}- 2\,\frac{1}{u} +O\left(\left(\frac{1}{u}\right)^2\right)&\qquad,\,  u\rightarrow +\infty \nonumber \\
%\qquad&|\gamma_1(u)|^2&\sim&\frac{1}{3} +O\left(\left(\frac{1}{u}\right)^3\right)&\qquad,\, u\rightarrow +\infty \nonumber \\
%\qquad&|\alpha_1(u)|^2&\sim& \frac{1}{6} \left(2+\sqrt{3}\right)-\frac{2+\sqrt{3}}{2}\,\frac{1}{u}+O\left(\left(\frac{1}{u}\right)^2\right)&\qquad,\, u\rightarrow -\infty \nonumber \\
%\qquad&|\beta_1(u)|^2&\sim&\frac{1}{6} \left(2-\sqrt{3}\right)-\frac{\sqrt{3}-2}{2}\,\frac{1}{u}+O\left(\left(\frac{1}{u}\right)^2\right)&\qquad,\, u\rightarrow -\infty \nonumber \\
%\qquad&|\gamma_1(u)|^2&\sim& \frac{1}{3}+\sqrt{3}\,\frac{1}{u}  +O\left(\left(\frac{1}{u}\right)^2\right)&\qquad,\, u\rightarrow -\infty \nonumber \\
%\qquad&|\alpha_1(u)|^2&\sim& 1-\frac{5}{324}\,u^2 +O\left(u^3\right)& \qquad,\, u\rightarrow 0 \nonumber \\
%\qquad&|\beta_1(u)|^2&\sim& \frac{1}{324}\,u^2+O\left(u^3\right)&\qquad,\, u\rightarrow 0 \nonumber \\
%\qquad&|\gamma_1(u)|^2&\sim&  \frac{1}{81}\,u^2+O\left(u^3\right)&\qquad,\, u\rightarrow 0\qquad.
%\end{array}
%$
%\\
%\\
\\
\\
Determining the corresponding NONs for the ground state $|\Psi_1\rangle$ is trivial.

If we instead omit the symmetry breaking between the wavenumbers $K=\pm 1$, determine the NONs is not trivial anymore.
As explained in Sec.~\ref{sec:hubbardAnalyt} we need to diagonalize the $3\times3$ matrix $\rho_1^{\downarrow}$ (\ref{rho1down}) which we will do here.

The characteristic polynomial $p_u(\lambda)$ of $\rho_1^{\downarrow}$ reads
\begin{eqnarray}
p_u(\lambda) &=& \lambda^3 + c_2\,\lambda^2+c_1\,\lambda+c_0 \,=\,0\\
c_2 &=& - 1 \nonumber \\
c_1 &=& |\alpha(u)|^2\,|\beta(u)|^2 +  |\alpha(u)|^2\,|\gamma(u)|^2 + |\beta(u)|^2\,|\gamma(u)|^2\nonumber \\
&&-(2-3 |\gamma(u)|^2)\, |\gamma(u)|^2 \,|\zeta|^2\,|\xi|^2 \nonumber \\
c_0 &=& - |\alpha(u)|^2\,|\beta(u)|^2\,|\gamma(u)|^2\,\left(1-3|\zeta|^2\,|\xi|^2 - \zeta^3 \,(\xi^*)^3-  (\zeta^*)^3 \,\xi^3\right)\,.\nonumber
\end{eqnarray}
Here, one can again pursue the elementary methods to solve cubic equations as done to solve Eq.~(\ref{eigenvalueQ}).
This then yields the eigenvalues $n_1, n_2$ and $n_3$ of block $\rho_1^{\downarrow}$ (recall (\ref{cubicParQ}) and (\ref{cubicParTh}))
\begin{eqnarray}
n_1&=& -2 \sqrt{Q}\,\cos{\left(\frac{\Theta}{3}\right)} -\frac{c_2}{3}  \nonumber \\
n_2&=& -2 \sqrt{Q}\,\cos{\left(\frac{\Theta+2\pi}{3}\right)} -\frac{c_2}{3} \nonumber \\
n_3&=& -2 \sqrt{Q}\,\cos{\left(\frac{\Theta-2\pi}{3}\right)} -\frac{c_2}{3} \,.
\end{eqnarray}
where $Q$ and the `angel' $\Theta$ depend via $c_2,c_1$ and $c_0$ on $\zeta, \xi$ and $u$. Independent of the concrete form of $c_1$ and $c_0$, just by using $c_2=-1$, we find that
\begin{equation}
n_1+n_2+n_3 = - c_2 = 1 \qquad,\,\forall \Theta, Q\,,
\end{equation}
as it should be.

\bibliography{bibliographyPhDthesis}
\bibliographystyle{alpha}
\end{document}